\pdfoutput=1
\documentclass[aps,prx,10pt,twocolumn,floatfix,amsmath,amssymb,notitlepage,nofootinbib,superscriptaddress,longbibliography]{revtex4-2}

\usepackage[T1]{fontenc}
\usepackage{graphicx}
\usepackage{bm}
\usepackage{booktabs}
\usepackage{braket}
\usepackage{amsmath,amssymb,amsfonts,amsthm}
\usepackage{mathtools}
\usepackage{dsfont}
\usepackage{enumitem}
\usepackage{multirow}
\usepackage{threeparttable}
\usepackage{xcolor}
\usepackage{algpseudocode}
\usepackage[colorlinks=true,linkcolor=blue!60!black,citecolor=blue!60!black,urlcolor=blue!60!black]{hyperref}

\graphicspath{{figs/}}

\newcommand{\FF}{\mathbb{F}_2}

\newcommand{\EE}{\mathbb{E}}
\newcommand{\Var}{\mathrm{Var}}
\newcommand{\Cov}{\mathrm{Cov}}

\newcommand{\Ker}{\mathrm{ker}\,}

\newcommand{\supp}{\mathrm{supp}\,}

\newcommand{\HH}{\bm{H}}

\newcommand{\WW}{\bm{W}}

\newcommand{\Amat}{\bm{A}}
\newcommand{\RR}{\mathbb{R}}

\newcommand{\Fs}{F}
\newcommand{\Xset}{\mathcal{X}}

\newtheorem{theorem}{Theorem}
\newtheorem{proposition}{Proposition}
\newtheorem{corollary}{Corollary}
\newtheorem{lemma}{Lemma}

\newtheorem{remark}{Remark}

\newcounter{algorithm}
\renewcommand{\thealgorithm}{\arabic{algorithm}}

\newcommand{\DD}{\mathcal{D}}
\newcommand{\Drt}{\DD_{\mathrm{rt}}}

\DeclareMathOperator{\Tr}{Tr}

\newcommand{\CC}{\mathbb{C}}
\newcommand{\cI}{\mathcal{I}}
\newcommand{\cP}{\mathcal{P}}
\newcommand{\cE}{\mathcal{E}}
\newcommand{\cL}{\mathcal{L}}
\newcommand{\cH}{\mathcal{H}}
\newcommand{\E}{\mathbb{E}}
\newcommand{\U}{\mathcal{U}}
\newcommand{\M}{\mathcal{M}}
\newcommand{\B}{\mathcal{B}}
\newcommand{\Id}{\mathcal{I}}
\newcommand{\OVEC}{\widehat{O}_{\rm VEC}}
\newcommand{\normd}[1]{\left\|#1\right\|_{\diamond}}

\makeatletter
\newcommand{\aliaslabel}[2]{\begingroup\phantomsection\protected@edef\@currentlabel{#2}\label{#1}\endgroup}

\makeatother

\DeclareRobustCommand{\vecMethodsSec}{Supplementary~Section~\ref{sec:methods_numerical}}

\DeclareRobustCommand{\vecLearnSec}{Supplementary~Section~\ref{sec:learning_noise_model}}
\DeclareRobustCommand{\vecTheMethodsSec}{Supplementary~Section~\ref{sec:methods_numerical}}
\DeclareRobustCommand{\vecSimSubA}{Supplementary~Section~\ref{subsec:methods_circuits_noise}}
\DeclareRobustCommand{\vecSimSubB}{(Supplementary~Section~\ref{subsec:methods_decoders_acceleration})}
\DeclareRobustCommand{\vecSimSubCD}{Supplementary~Sections~\ref{subsec:methods_noise_learning_benchmark} and~\ref{subsec:methods_sampling_uncertainty}}

\begin{document}

\title{Measure Now, Mitigate Later\texorpdfstring{:}{:} Virtual Error Cancellation\texorpdfstring{\\}{} for Logical Quantum Circuits}

\author{Oles Shtanko}
\email{shtankoo@google.com}
\affiliation{Google Quantum AI}
\author{Saswata Roy}
\affiliation{Google Quantum AI}
\affiliation{Department of Physics, Cornell University, Ithaca, NY, 14853, USA}
\author{Shashwat Kumar}
\affiliation{Google Quantum AI}
\author{Zlatko K.~Minev}
\email{zminev@google.com}
\affiliation{Google Quantum AI}
\date{\today}

\begin{abstract}
  Early fault-tolerant quantum processors will accumulate logical errors as circuits grow wider or deeper. With a fixed budget of physical qubits, increasing the code distance suppresses these errors but leaves fewer logical qubits and magic-state factories. Logical error mitigation can improve accuracy within that budget, though existing methods demand modified circuits, prior accurate noise characterization, and many additional runs.
Here we introduce virtual error cancellation (VEC), which learns a noise model from the syndrome records of the unmodified circuit and cancels logical bias entirely in classical post-processing.
VEC yields unbiased estimates under that model without extra qubits or separate calibration experiments, even after real-time decoding and feedforward decisions have been committed.
Our dual-decoder VEC scheme can overcome the traditional $1/F^2$ sampling-cost floor
for physical error mitigation, and approaches the
ultimate fundamental limit of $1/F$, where $F$ is the probability of no logical
error under real-time decoding.
In end-to-end numerical benchmarks under standard noise models and sampling budgets, VEC suppresses logical error rates by multiple orders of magnitude, significantly increasing the effective code distance, and operates in the presence of drifting noise at the $10^6$-cycle regime. 
Our results suggest that VEC can bring larger quantum applications, otherwise inaccessible, within reach of early logical processors and beyond.

\end{abstract}

\maketitle

Quantum error correction protects encoded quantum information by continuously measuring stabilizer checks and decoding the resulting syndrome records~\cite{kitaev2003fault,fowler2012surface}.
Recently, experiments have crossed the fault-tolerance threshold, demonstrating that scaling the code distance $d$ suppresses logical errors~\cite{google2023suppressing,bluvstein2024logical,google2025below,lacroix2025scaling,bluvstein2025architecture,lavasani2026color}.
Yet on any processor of fixed physical size, increasing the code distance trades directly against the number of logical qubits and magic-state factories available for computation~\cite{gidney2021how,gidney2024magic}.
Early fault-tolerant quantum computers will therefore operate at moderate code distances, where residual logical errors remain the primary barrier to running deep algorithms.

\begin{figure*}[t]
    \centering
    \includegraphics[width=\textwidth]{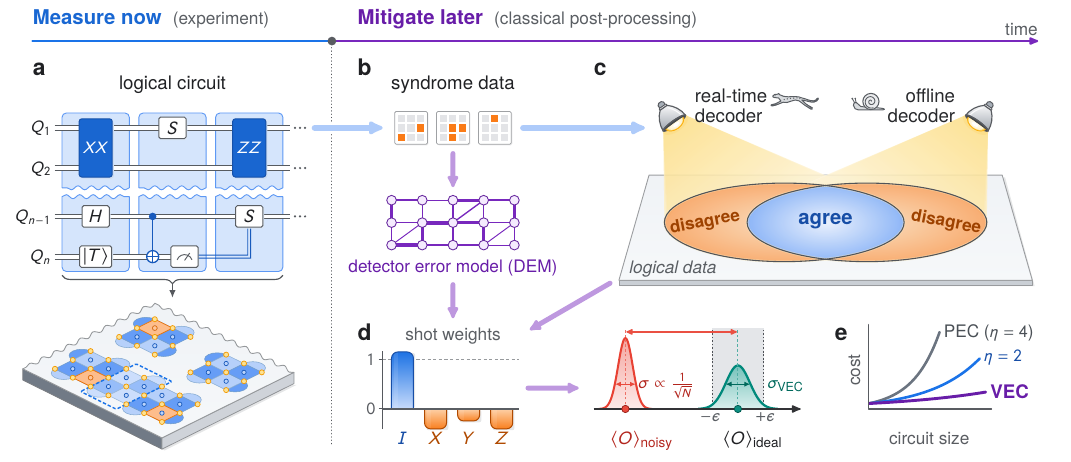}
    \caption{\textbf{Virtual error cancellation.}
    The protocol removes residual logical errors from an unmodified fault-tolerant circuit in classical post-processing by combining syndrome-based noise learning with shot-by-shot weights conditioned on the syndrome.
    \textbf{(a)}~Encoded qubits, such as surface-code patches, evolve across parallel layers (light-blue columns) combining fault-tolerant operations (e.g., lattice surgery, single-qubit Clifford gates, and $T$-state cultivation and teleportation).
    \textbf{(b)}~Stabilizer checks generate syndrome data that stream to the real-time decoder and are stored for offline post-processing. From these syndromes alone, classical computation learns a detector error model (purple graph, with edge widths showing fault rates) without calibration circuits.
    \textbf{(c)}~A fast real-time decoder applies corrections during circuit execution, whereas a slower, higher-accuracy offline decoder re-evaluates the syndromes in post-processing. Comparing their recoveries on the logical data classifies each shot into consensus (agree, blue) or Pauli disagreement (disagree, orange) sectors.
    \textbf{(d)}~Using the per-sector error probabilities predicted by the detector error model, the algorithm assigns a signed shot weight to each sector. Over $N$ runs, forming this linear combination across sectors removes the bias of the unmitigated mean $\langle O\rangle_{\mathrm{noisy}}$ (red, standard error $\sigma \propto 1/\sqrt{N}$) and centers the estimate on $\langle O\rangle_{\mathrm{ideal}}$ within $\pm\epsilon$ (green), at the cost of a wider spread $\sigma_{\mathrm{VEC}}$.
    \textbf{(e)}~As circuit size grows (lowering the probability of an error-free circuit run $F$), the VEC sampling cost $C \sim 1/F^{\eta}$ (purple, approaching $\eta = 1$ as the relative accuracy of the offline decoder grows at lower noise and larger distances) scales more gently than the gap-conditioned single-decoder benchmark ($\eta = 2$, blue) and unconditioned logical PEC ($\eta = 4$, gray).
    All panels are schematic.}
    \label{fig:main_schematic}
\end{figure*}

Logical quantum error mitigation offers a route to suppress these errors without extra physical qubits by trading residual logical bias for an increase in circuit repetitions~\cite{piveteau2021error,lostaglio2021error,suzuki2022quantum,zhang2025demonstrating,umbrarescu2026infinite,xie2025noise}.
Unlike physical-level mitigation~\cite{temme2017error,li2017efficient,endo2018practical,vandenberg2023probabilistic,cai2023quantum}, logical mitigation gains extra information from the syndrome stream generated during every run of the computation, helping identify shots likely to contain logical errors.
For example, syndrome-based postselection uses decoder confidence metrics, such as the complementary gap~\cite{bombin2024fault,gidney2025yoked,meister2024efficient,gidney2024magic}, to discard ambiguous runs~\cite{smith2024mitigating,dinca2025confidence,zhong2025combining,kumar2026codesigning,zhou2025soft}. However, it still leaves residual unheralded logical errors in the accepted subset.
In contrast, probabilistic error cancellation (PEC) can remove this bias entirely by inserting quasi-probability recovery operations during execution~\cite{temme2017error,li2017efficient,piveteau2021error,suzuki2022quantum} with or without conditioning them on syndrome soft information~\cite{zhou2025soft,aharonov2025salem,tsubouchi2026advantages}.
However, PEC requires expensive noise characterization, modified circuits, and steep sampling overheads. Moreover, universal logical circuits are challenging for cancellation methods due to fundamental ambiguities in characterizing noisy non-Clifford operations~\cite{huang2022foundations} and temporal error correlations that cannot be canceled as simple local error channels.

Here we introduce virtual error cancellation (VEC), a framework that runs entirely in classical post-processing
on an unmodified original logical circuit, without gate injections, noise amplification, or characterization circuits (Fig.~\ref{fig:main_schematic}).
By combining linear superoperator inversion~\cite{temme2017error,endo2018practical,cai2023quantum,suchsland2026quantumerrormitigationinformation} with syndrome-based
noise learning~\cite{spitz2018adaptive,wagner2023learning,zheng2026learning,takou2026logical}, VEC
cancels local Pauli logical errors using only syndrome
data recorded during
computation (Fig.~\ref{fig:main_schematic}b).
Notably, our method applies to universal circuits that use non-Clifford gate
teleportation with hardware feedforward and real-time recovery, where
direct decoding in post-processing is impossible.
VEC sorts runs into sectors in post-processing, for example by comparing
the recovery committed in real time with that proposed by a more accurate
offline decoder~\cite{bravyi2014efficient,roffe2020decoding,shutty2024efficient}.
It then takes a signed linear combination of the sector outputs
(Fig.~\ref{fig:main_schematic}c,d) to reduce the sampling overhead of logical
error cancellation. 
We show that, when offline decoding substantially outperforms real-time decoding, the upper bound on the VEC sampling overhead approaches $1/F$
at low noise and large code distances, where $F$ is the probability of no
logical error in the observable's light cone under real-time decoding.
In comparison, conventional logical PEC scales as $1/F^4$, while
gap-conditioned PEC obeys a $1/F^2$ floor~\cite{tsubouchi2026advantages}
(Fig.~\ref{fig:main_schematic}e).

We validate VEC in circuit-level memory simulations of a rotated surface code~\cite{dennis2002topological}.
For a distance-$7$ code under non-uniform depolarizing noise below threshold, learning the disordered spacetime noise model solely from unlabeled syndromes already suppresses the logical memory error by up to $400\times$ within $50{,}000$ total circuit samples, reaching over $10^4\times$ suppression for larger sample budgets.
We show how to extend VEC analytically to multi-layer logical circuits with non-Clifford gate-teleportation gadgets~\cite{bravyi2005universal,litinski2019game,gidney2024magic}.
By properly assigning faults to circuit layers, local product weights suppress unaccounted interlayer correlation errors down to an $\mathcal{O}(p^{\lfloor d/2\rfloor})$ fraction of the logical error under non-uniform circuit-level noise, 
reaching below a $10^{-3}$ fraction in below-threshold simulations for distances $d\geq 3$, including both memory and lattice-surgery operations. We also confirm noise-learning stability and unbiasedness of the end-to-end algorithm on memory simulations exceeding $10^6$ syndrome-measurement cycles with smoothly drifting circuit-level noise. \medskip

\textbf{The method}. Our method applies to fault-tolerant quantum circuits on $n$ logical qubits decomposed into physical magic-state preparation and logical Clifford layers, which can be connected by feedforward operations (Fig.~\ref{fig:main_schematic}a).
Under Pauli circuit and state-preparation noise, the logical errors remaining after decoding each layer can be modeled as a product of $\mathcal{O}(n)$ local logical error channels, analogous to a sparse Pauli--Lindblad model~\cite{vandenberg2023probabilistic}. Correlations between channels are strongly suppressed below threshold (see the Full-circuit integration section below).
Given a syndrome record $\bm{s}$ (effectively restricted to a local spacetime decoding window around the layer), each local channel acts on the input state $\rho$ as
\begin{equation}
\rho \mapsto \Lambda_{\bm{s}}(\rho) = \Big(1 - \sum_{\alpha=1}^M p_\alpha(\bm{s})\Big) \rho + \sum_{\alpha =1}^M p_\alpha(\bm{s}) \mathcal{E}_\alpha(\rho),
\label{eq:conditional_logical_channel}
\end{equation}
where $\mathcal{E}_\alpha(\rho) \equiv P_\alpha \rho P_\alpha$ denotes one of $M = \mathcal{O}(1)$ active non-trivial logical Pauli errors occurring with conditional probability $p_\alpha(\bm{s})$.
Our objective is to eliminate this noise and recover the ideal logical operations across all initial states and observables.

We achieve this inversion on average in classical post-processing by weighting each shot's measurement outcome with a real, syndrome-dependent filter $w(\bm{s})$. Across $N$ experimental runs, the processor records the syndrome $\bm{s}_i$ and projective readout $m_i \in \{+1, -1\}$ of a Pauli logical observable $O$, forming the linear estimator
\begin{equation}
  \hat{\langle O \rangle}_{\mathrm{VEC}} = \frac{1}{N} \sum_{i=1}^N w(\bm{s}_i) m_i.
\label{eq:estimator}
\end{equation}
This estimator is unbiased, $\EE[\hat{\langle O \rangle}_{\mathrm{VEC}}] = \langle O\rangle_{\mathrm{ideal}}$, for all states and observables if and only if the syndrome-averaged logical channel reduces to the identity, $\EE_{\bm{s}}[w(\bm{s})\Lambda_{\bm{s}}] = \mathcal{I}$. This condition imposes $M + 1$ linear constraints on $w(\bm{s})$ by the linear independence of the Pauli basis (see Supplementary Section~\ref{sec:supp_operator_channel_inversion}),
\begin{align}
\label{eq:pauli_cancellation_constraint}
    \begin{aligned}
\EE_{\bm{s}}\big[ w(\bm{s}) \big] &= 1 \quad \text{(Trace Preservation)}, \\
\EE_{\bm{s}}\big[ w(\bm{s}) p_\alpha(\bm{s}) \big] &= 0 \quad \text{(Error Cancellation)}.
    \end{aligned}
\end{align}
Satisfying Eq.~\eqref{eq:pauli_cancellation_constraint} typically requires $w(\bm{s}) < 0$ on at least some syndromes, which removes logical bias offline at the expense of a wider statistical spread (Fig.~\ref{fig:main_schematic}d) and an additional sampling cost to achieve a target precision.

This sampling cost is governed directly by the variance of the filter weights. Because $m_i^2 = 1$, the estimator variance obeys $\Var[\hat{\langle O \rangle}_{\mathrm{VEC}}] \le N^{-1}\EE_{\bm{s}}[w(\bm{s})^2]$. A worst-case root-mean-square precision $\epsilon$ relative to the noiseless budget $N_{\mathrm{ideal}} = \lceil 1/\epsilon^2 \rceil$ is therefore achieved with $N_{\mathrm{VEC}} = \lceil C(w)/\epsilon^2 \rceil$ shots, with a sampling overhead multiplier
\begin{equation}
C(w) \equiv \EE_{\bm{s}}\big[ w(\bm{s})^2 \big] = 1 + \Var_{\bm{s}}\big[ w(\bm{s}) \big].
\label{eq:vpec_sampling_overhead_def}
\end{equation}
\aliaslabel{eq:vpec_shot_budget_scaling}{\ref{eq:vpec_sampling_overhead_def}}%
The following sections develop post-processing protocols that minimize $C(w)$ subject to Eq.~\eqref{eq:pauli_cancellation_constraint} and practical constraints.\medskip

\textbf{Syndrome-resolved protocol}. We first formulate a protocol that assigns an individual weight $w(\bm{s})$ to every syndrome $\bm{s}$ to minimize overall sampling overhead, establishing the minimum sampling overhead among all state- and observable-independent channel-inverting filters.
Let $\bm{p}(\bm{s}) = \{p_\alpha(\bm{s})\}$ denote the conditional logical error probability vector in Eq.~\eqref{eq:conditional_logical_channel}, with ensemble mean $\bar{\bm{p}} \equiv \EE_{\bm{s}}[\bm{p}(\bm{s})]$ and covariance matrix $\Sigma_{\alpha\beta} \equiv \EE_{\bm{s}}[(p_\alpha(\bm{s}) - \bar{p}_\alpha)(p_\beta(\bm{s}) - \bar{p}_\beta)]$. For any $\bm{\Sigma} \succ 0$, the unique filter that minimizes $C(w)$ in Eq.~\eqref{eq:vpec_sampling_overhead_def} is
\begin{align}
&w^*(\bm{s}) = 1 - \bar{\bm{p}}^\top \bm{\Sigma}^{-1} \big( \bm{p}(\bm{s}) - \bar{\bm{p}} \big), \label{eq:optimal_weights} \\
&C_{\mathrm{syn}}^* \equiv C(w^*) = 1 + \bar{\bm{p}}^\top \bm{\Sigma}^{-1} \bar{\bm{p}}, \label{eq:optimal_sampling_overhead}
\end{align}
\aliaslabel{eq:cs_variance_lower_bound}{\ref{eq:optimal_sampling_overhead}}%
\aliaslabel{eq:overhead_mahalanobis}{\ref{eq:optimal_sampling_overhead}}%
so that any unbiased filter obeys $C(w) \ge C_{\mathrm{syn}}^*$, with equality if and only if $w(\bm{s}) = w^*(\bm{s})$ (see Supplementary Section~\ref{sec:appendix_proof_variational_cs_bound}).

The sampling cost of the mitigation is ultimately governed by the properties of the code and the respective decoder.
In the fault-tolerant regime at low physical error rate $p$, the logical error rate decays as $p_L \equiv \|\bar{\bm{p}}\|_1 = \Theta(p^r)$, set by the minimum uncorrectable fault weight $r$.
Each leading failure syndrome $\bm{s}$ with weight gap $\delta(\bm{s}) \in [0, r)$ between the failing and lowest-weight cosets occurs with probability $\Theta(p^{r-\delta(\bm{s})})$ and carries a conditional logical error rate $\Theta(p^{\delta(\bm{s})})$, contributing $\Theta(p^{r+\delta(\bm{s})})$ to the covariance $\bm{\Sigma}$. Consequently, the smallest dominant-subspace eigenvalue of $\bm{\Sigma}$ is bounded below and scales as $\Theta(p^{r+\Delta})$, where $\Delta \in [0, r)$ is the smallest gap over syndromes spanning the dominant error subspace. Combining this with $\|\bar{\bm{p}}\|_2^2 = \Theta(p^{2r})$ in Eq.~\eqref{eq:optimal_sampling_overhead} yields
\begin{equation}
C_{\mathrm{syn}}^* = 1 + \Theta(p^{r-\Delta}) = 1 + \Theta(p_L / p^\Delta).
\label{eq:overhead_gap_law}
\end{equation}
When leading failure syndromes exhibit minimum-weight coset ties that span all dominant logical error directions ($\Delta = 0$), e.g., as in a rotated surface code subject to depolarizing noise,
the overhead reaches $C_{\mathrm{syn}}^* = 1 + \Theta(p_L)$ (for general coset-gap and untied \mbox{$\Delta > 0$} regimes, see Supplementary Section~\ref{sec:appendix_spectral_scaling_mahalanobis}).
Thus, in the small-$p_L$ regime, $C_{\mathrm{syn}}^* - 1$ scales linearly with $p_L$, akin to logical PEC ($C_{\mathrm{PEC}} - 1 \approx 4p_L$), with a prefactor that depends on the code and real-time decoder.

Despite its manageable sampling cost, computing exact per-syndrome posteriors $\bm{p}(\bm{s})$ is classically intractable on general spacetime decoding graphs~\cite{iyer2015hardness}. This motivates a more practical coarse-grained protocol.\medskip

\aliaslabel{eq:discrete_trace_preservation}{\ref{eq:discrete_optimal_weights0}}%
\aliaslabel{eq:discrete_error_cancellation}{\ref{eq:discrete_optimal_weights0}}%
\aliaslabel{eq:between_sector_covariance_def}{\ref{eq:discrete_optimal_weights0}}%
\aliaslabel{eq:dd_sector_channels}{\ref{eq:discrete_optimal_weights0}}%
\aliaslabel{eq:discrete_overhead_hierarchy}{\ref{eq:discrete_optimal_weights0}}%
\aliaslabel{eq:dd_variance_decomposition}{\ref{eq:discrete_optimal_weights0}}%
\aliaslabel{eq:minimal_k_inequality}{\ref{eq:discrete_optimal_weights0}}%

\textbf{Coarse-grained protocol}. To make virtual error cancellation computationally tractable, we group syndromes into a small number of discrete sectors with constant weights. Let us partition the syndrome space $\Omega$ into $K$ disjoint sectors $\mathcal{S}_k$ each characterized by probability of occurrence $\pi_k>0$, conditional error rates $\bm{p}_k$, and weights $w_k$. Then Eq.~\eqref{eq:pauli_cancellation_constraint} reduces to $M + 1$ linear equations: $\sum_k \pi_k w_k = 1$ and $\sum_k \pi_k w_k \bm{p}_k = \bm{0}$.
Minimizing the weight variance mirrors Eqs.~\eqref{eq:optimal_weights}--\eqref{eq:optimal_sampling_overhead} with $\bm{p}_k$ and $\bm{\Sigma}_S \equiv \sum_k \pi_k (\bm{p}_k - \bar{\bm{p}})(\bm{p}_k - \bar{\bm{p}})^\top$ in place of $\bm{p}(\bm{s})$ and $\bm{\Sigma}$,
\begin{equation}
  \label{eq:discrete_optimal_weights0}
  \begin{aligned}
w_k^* &= 1 - \bar{\bm{p}}^\top \bm{\Sigma}_S^{-1} \big( \bm{p}_k - \bar{\bm{p}} \big), \\
C_S^* &= 1 + \bar{\bm{p}}^\top \bm{\Sigma}_S^{-1} \bar{\bm{p}}.
\end{aligned}
\end{equation}
Full-rank inversion across all $M$ logical error channels is available only for positive-definite $\bm{\Sigma}_S \succ 0$, which requires at least $K \ge M + 1$ sectors. However, when $\bm{\Sigma}_S$ is rank-deficient with $\bar{\bm{p}} \in \operatorname{range}\bm{\Sigma}_S$, one can replace $\bm{\Sigma}_S^{-1}$ with the pseudoinverse $\bm{\Sigma}_S^+$.
Because coarse-graining averages over within-sector fluctuations ($\bm{\Sigma} \succeq \bm{\Sigma}_S$), any unbiased filter with constant weight per sector obeys $C(w) \ge C_S^* \ge C_{\mathrm{syn}}^*$ (see Supplementary Section~\ref{sec:appendix_proofs_discrete_protocol}). Coarse-graining therefore reduces the exponentially many per-syndrome weights to $K$ distinct values, with a small increase in sampling overhead when the sectors strongly separate the $M$ logical error channels. In practice, sectors can be defined using single-decoder metrics such as syndrome weight or the complementary gap, or from disagreements between two decoders, as we develop next.\medskip

\aliaslabel{eq:dd_overhead_pLprime}{\ref{eq:dd_overhead_eta_formula}}%
\aliaslabel{eq:dd_residual_taylor}{\ref{eq:four_sector_distance_suppression2}}%

\begin{figure*}[t]
    \centering
    \includegraphics[width=\linewidth]{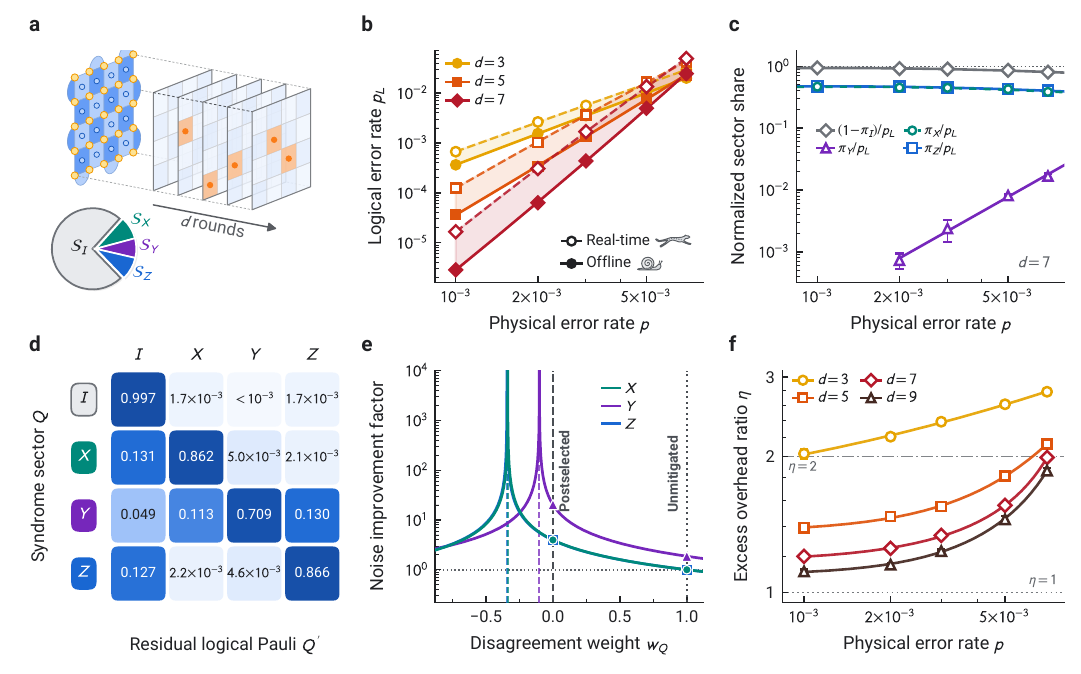}
\caption{\textbf{Dual-decoder virtual error cancellation.}
Pairing a fast real-time decoder with a stronger offline decoder sorts shots into sectors whose reweighting cancels logical errors at low sampling cost.
\textbf{(a)}~Comparing the real-time and offline recoveries partitions shots into consensus ($\mathcal{S}_I$) and disagreement sectors ($\mathcal{S}_X, \mathcal{S}_Y, \mathcal{S}_Z$), labeled by the Pauli difference between the two predictions.
\textbf{(b)}~Logical error rates for real-time minimum-weight perfect matching ($p_L$, open) and offline ensemble decoding ($p_L'$, filled), showing a widening accuracy gap (shaded) at lower $p$ and larger $d$.
\textbf{(c)}~Normalized disagreement shares $\pi_Q/p_L$ at $d = 7$: total disagreements ($1 - \pi_I$) track $p_L$, dominated by $\mathcal{S}_X$ and $\mathcal{S}_Z$.
\textbf{(d)}~Conditional probability of actual residual logical error $Q'$ after real-time decoding (columns) given sector $Q$ (rows) at $d = 7, p = 5\times 10^{-3}$; the dominant diagonal shows each sector isolates its corresponding error.
\textbf{(e)}~Noise improvement factor (unmitigated-to-residual error ratio per Pauli channel) versus $w_Q$ at $d = 7, p = 5\times 10^{-3}$, holding the other two weights at $w_{Q'}^*$. Tuning $w_Q$ to $w_Q^* < 0$ (dashed lines) cancels bias far beyond postselection ($w_Q = 0$); $X$ and $Z$ overlap.
\textbf{(f)}~Excess overhead ratio $\eta = (C - 1)/p_L$ (extra shot cost per unit $p_L$). For $d \ge 5$ below threshold, $\eta$ drops below syndrome-optimized PEC for a single real-time decoder ($\eta = 2$) toward an error-free offline decoder ($\eta = 1$).
Simulations model a $d$-round rotated surface-code memory under uniform SI1000 circuit noise with noiseless preparation and terminal readout; error bars show $1\text{ s.e.m.}$ over $30$ jackknife groups where larger than markers (see \vecMethodsSec).}
\label{fig:dual_decoder}
\end{figure*}

\textbf{Dual-decoder protocol}. Consider a local channel with $M$ active non-trivial logical Pauli errors $\mathcal{Q} = \{Q_1, \dots, Q_M\}$ and write $\bar{\mathcal{Q}} \equiv \{I\} \cup \mathcal{Q}$.
To partition $\Omega$, we pair the real-time decoder $\mathcal{D}_{\mathrm{rt}}(\bm{s})$ with a more accurate offline reference decoder $\mathcal{D}_{\mathrm{off}}(\bm{s})$.
Crucially, the two decoders serve only to partition $\Omega$ into sectors, so the protocol remains unbiased regardless of decoder optimality, with decoder accuracy affecting only the sampling overhead.
Comparing the two recoveries sorts active runs into $K = M + 1$ disjoint sectors,
\begin{equation}
\mathcal{S}_Q \equiv \big\{\bm{s} \in \Omega:\; \mathcal{D}_{\mathrm{rt}}(\bm{s})\cdot \mathcal{D}_{\mathrm{off}}(\bm{s}) \sim Q\big\},
\label{eq:dd_consensus_indicator}
\end{equation}
where $Q \in \bar{\mathcal{Q}}$ and $\sim$ denotes equality up to phase. Each sector $\mathcal{S}_Q$ has probability share $\pi_Q$ and conditional error centroid $\bm{p}_Q \equiv \EE_{\bm{s}}[\bm{p}(\bm{s}) \mid \bm{s} \in \mathcal{S}_Q]$ in the committed real-time frame. Decoder agreement defines the consensus sector $\mathcal{S}_I$, which captures most runs, whereas disagreement along $Q \in \mathcal{Q}$ sorts the remaining shots into $M$ directional sectors (as illustrated in Fig.~\ref{fig:dual_decoder}a). When the two decoders rarely disagree (or no stronger offline decoder is available), $\Omega$ can alternatively be partitioned using single-decoder metrics such as the complementary gap or syndrome weight, provided the sectors still distinguish the $M$ error directions.

\begin{figure*}[t]
    \centering
    \includegraphics[width=\linewidth]{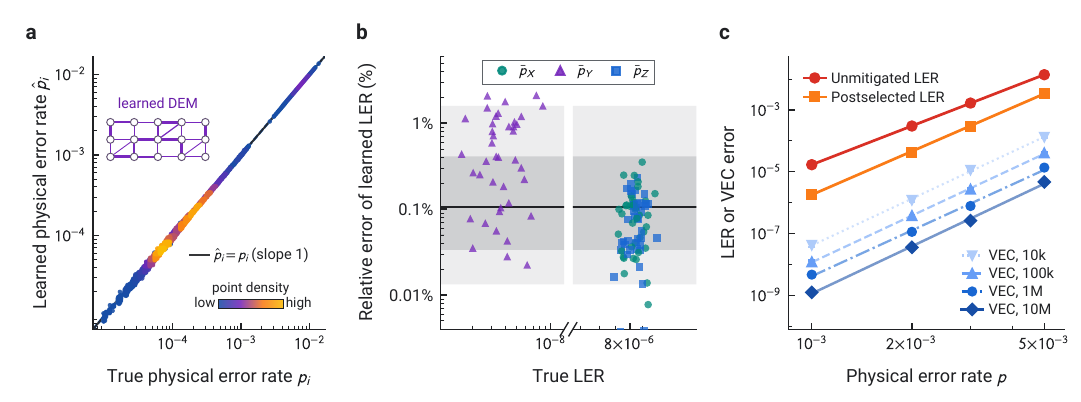}
\caption{\textbf{Syndrome-based noise learning and error cancellation.}
A detector error model (DEM) learned from recorded syndromes calibrates the post-processing weights in a digital twin; simulations model $d = 7$ surface-code memory decoding under disordered SI1000 circuit noise with noiseless initial preparation and terminal data readout.
\textbf{(a)}~Learned versus true physical error rates across all $6{,}351$ fault mechanisms ($0.76\%$ median relative error) at $p = 10^{-3}$ and base sample count $N_0 = 10^7$, for $N_{\mathrm{learn}} = N_0(5\times 10^{-3}/p)$ syndrome runs.
\textbf{(b)}~Relative prediction error of the logical error rate (LER) versus true LER at $p = 10^{-3}$ and $N_0 = 10^7$, with the $0.10\%$ median (black line) alongside $68\%$ and $95\%$ bands.
\textbf{(c)}~Mean learning-limited residual VEC error $\langle\|\bm{p}_{\mathrm{res}}\|_1\rangle$ (evaluated via trajectory likelihood-ratio reweighting) for $N_0 = 10^4, 10^5, 10^6,$ and $10^7$ ($\text{10k}$, $\text{100k}$, $\text{1M}$, and $\text{10M}$) compared with the unmitigated real-time LER and decoder-agreement postselected LER ($w_Q = 0$); lines are power-law fits in $p$ at each $N_0$.
Results summarize $40$ disorder realizations per error rate ($20$ in \textbf{(c)} at $p = 10^{-3}$), with $1\text{ s.e.m.}$ error bars smaller than markers (see \vecMethodsSec).}
\label{fig:vlpec_validation}
\end{figure*}

Let us illustrate the protocol with a single binary error channel ($M = 1$). Across an experimental dataset, the two decoders agree on nearly all runs (forming the consensus sector $\mathcal{S}_I$) and disagree on a small, error-enriched subset (forming the disagreement sector $\mathcal{S}_X$); subtracting a weighted contribution of $\mathcal{S}_X$ from $\mathcal{S}_I$ cancels the residual errors remaining in $\mathcal{S}_I$---including errors unrecoverable by both decoders---in expectation, 
provided that we can predict the size of each sector and its logical error precisely. Quantitatively, this cancellation is governed by two dimensionless ratios: the consensus-error fraction $\alpha \equiv p_I/p_L$ of real-time errors missed by both decoders, and the failure-asymmetry ratio $b \equiv (p_L' - \pi_I p_I)/(p_L - \pi_I p_I)$, where $p_L$ and $p_L'$ are the real-time and offline decoders' logical error rates, respectively. No independence or stochasticity of decoder errors is assumed; in particular, strictly nested failures correspond to $b = 0$. Enforcing trace preservation and error cancellation fixes the sector weights and leading sampling overhead as
\begin{align}
&w_I = \frac{1}{1 - (1 + b)p_L}, \qquad w_X = -\frac{\alpha}{1 - \alpha}, \nonumber \\
&C_S^* = 1 + \eta_0 p_L + \mathcal{O}(p_L^2), \qquad
\eta_0 = \frac{1 + b}{1 - \alpha}.
\label{eq:dd_overhead_eta_formula}
\end{align}
Whereas decoder-agreement postselection ($w_X = 0$) leaves the consensus logical error rate $p_I = \alpha p_L$ unmitigated, VEC assigns a negative weight $w_X < 0$ whenever $\alpha > 0$ to cancel that residual error. Moreover, while two equally accurate decoders ($p_L' = p_L$) obey $\eta_0 \ge 2$, pairing a fixed real-time decoder with a much more accurate offline reference ($p_L'/p_L \to 0$) drives both $b \to 0$ and $\alpha \to 0$, so $w_X \to 0$ and the excess-overhead slope of $C_S^*$ approaches $1$,
 reducing smoothly to postselection on decoder agreement. Because the total excess overhead scales as $\eta_0 p_L$, minimizing sampling cost requires keeping both $p_L$ and $p_L'/p_L$ small: degrading $\mathcal{D}_{\mathrm{rt}}$ raises $p_L$ directly, whereas as $\mathcal{D}_{\mathrm{rt}}$ approaches $\mathcal{D}_{\mathrm{off}}$ in accuracy, disagreements vanish ($\alpha \to 1$) and $C_S^*$ diverges, in which case single-decoder partitioning is preferable.

Extending this two-sector balance to a general local channel with $M$ active Pauli errors $\mathcal{Q}$ replaces scalar contrast with a square matrix inversion. Because the $K = M + 1$ sectors in Eq.~\eqref{eq:dd_consensus_indicator} match the number of linear constraints in Eq.~\eqref{eq:pauli_cancellation_constraint}, an invertible $M \times M$ contrast matrix $\bm{V} \equiv (\bm{p}_{Q_1} - \bm{p}_I, \dots, \bm{p}_{Q_M} - \bm{p}_I)$ uniquely fixes the sector weights by direct linear inversion,
\begin{align}
w_I &= \frac{1 + \bm{1}_M^\top \bm{V}^{-1}\bm{p}_I}{\pi_I}, \nonumber \\
w_Q &= -\frac{(\bm{V}^{-1}\bm{p}_I)_Q}{\pi_Q}, \qquad Q \in \mathcal{Q},
\label{eq:four_sector_closed_form_weights}
\end{align}
where $\bm{1}_M \equiv (1, \dots, 1)^\top$. Provided $\bm{V}$ stays well-conditioned, the disagreement weights $|w_Q| = \mathcal{O}(1)$ remain uniformly bounded, and the occupied disagreement sectors ($\pi_Q > 0$) carry total share $\sum_{Q \in \mathcal{Q}} \pi_Q = \mathcal{O}(p_L)$, the coarse-grained overhead expands as
\begin{equation}
C_S^* = 1 + \sum_{Q \in \mathcal{Q}} \frac{\big[(\bm{V}^{-1}\bar{\bm{p}})_Q\big]^2}{\pi_Q} + \mathcal{O}(p_L^2),
\label{eq:four_sector_overhead_matrix_form}
\end{equation}
\aliaslabel{eq:four_sector_overhead_asymptotic}{\ref{eq:four_sector_overhead_matrix_form}}%
so $C_S^* - 1 = \mathcal{O}(p_L)$ inherits the exponential suppression of $p_L$ with code distance $d$ when these bounds hold uniformly (see Supplementary Section~\ref{sec:supp_dual_decoder_protocol}).

\begin{figure*}[t]
    \centering
    \includegraphics[width=\linewidth]{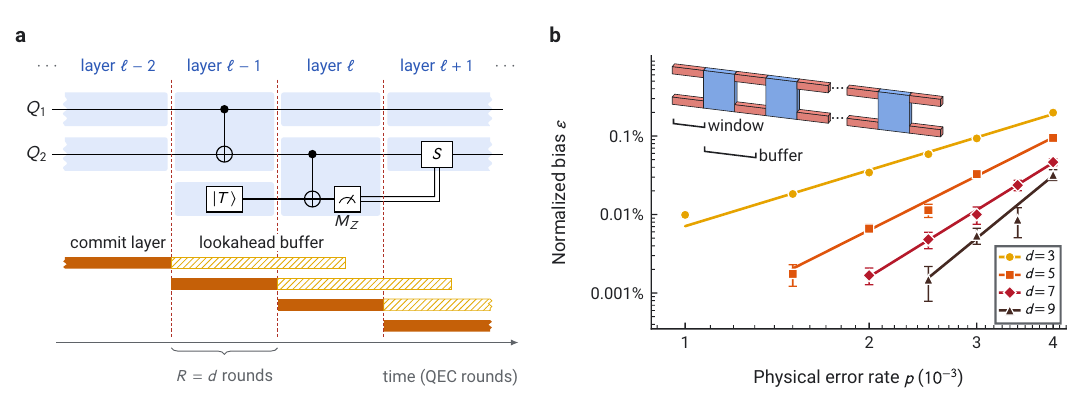}
    \caption{\textbf{Sliding-window decoding and multi-window bias.}
    Overlapping decoding windows resolve time correlations across circuit layers.
    \textbf{(a)}~Simple schematic example of a multi-layer logical circuit and its sliding-window schedule.
    Logical qubits evolve across $R$-round layers of parallel Clifford blocks (light blue), which may include $T$-gate teleportation gadgets ($|T\rangle$ preparation, CNOT gates, readout $M_Z$, and a conditioned $S$ gate). Sequential windowed decoding spans an $R$-round commit layer (solid orange) with a lookahead buffer (hatched yellow) and commits faults in the commit layer before advancing by $R$ rounds, resolving measurement error before feedforward operations.
    \textbf{(b)}~Normalized mitigation bias $\epsilon$ per block, evaluated in the nearest-window approximation (keeping only correlations between neighboring decoding windows), versus physical error rate $p$ across $11$ sliding windows of a five-merge two-qubit $ZZ$ lattice-surgery ladder under disordered SI1000 circuit noise, with windowed real-time matching and global correlated offline matching restricted to retain only local discrepancy clusters (see \vecMethodsSec). Weighted power-law fits $\epsilon \propto p^k$ (solid lines) yield empirical slopes $k = 2.37(10)$, $3.92(16)$, $4.8(3)$, and $6.4(8)$ for $d = 3, 5, 7,$ and $9$ (each above $\lfloor d/2\rfloor$).
    The inset shows alternating $d$-round single-patch (red) and merged (blue) blocks, with brackets marking the commit window ($R = d$) and lookahead buffer ($b = 2d - 3$).
    Markers in \textbf{(b)} show exact-model cross-fitted window-boundary splice averages across $80$ disorder realizations per point, with $1\text{ s.e.m.}$ error bars (see \vecMethodsSec).}
    \label{fig:circuit_integration}
\end{figure*}

We benchmark this four-sector protocol on a $d$-round rotated surface-code memory under uniform SI1000 circuit noise (Fig.~\ref{fig:dual_decoder}a), pairing real-time minimum-weight perfect matching (MWPM) with an offline ensemble of belief propagation with ordered statistics decoding (BP+OSD)~\cite{roffe2020decoding}, correlated matching, and MWPM (see \vecMethodsSec). Below threshold, resolving correlated spacetime faults gives $\mathcal{D}_{\mathrm{off}}$ a steeper distance suppression than MWPM, widening $p_L/p_L'$ as $d$ grows (Fig.~\ref{fig:dual_decoder}b).
With $\pi_{X,Z} = \mathcal{O}(p_L)$ and $\pi_Y = o(p_L)$ (Fig.~\ref{fig:dual_decoder}c), each disagreement sector isolates its corresponding residual fault to yield a strongly diagonal matrix $p_{Q,Q'}$ (Fig.~\ref{fig:dual_decoder}d). Sweeping each weight $w_Q$ reveals a sharp cancellation peak that boosts the noise improvement factor far above postselection on the results of the two decoders (Fig.~\ref{fig:dual_decoder}e).
 Across $d \ge 5$ in the sub-threshold regime,
 the growing offline advantage pushes $\eta = (C_S^* - 1)/p_L$ below the $\eta = 2$ syndrome-optimized PEC limit for a single real-time decoder and toward unity with distance (Fig.~\ref{fig:dual_decoder}f; see Supplementary Section~\ref{subsec:salem_surface_benchmark}), while obeying $\eta p_L \ge 2 p_L'$~\cite{tsubouchi2026advantages}. 

Notably, this coarse-grained architecture also makes VEC intrinsically resilient to finite-sample calibration errors. Because negative weights act only on rare disagreement sectors to cancel $\bm{p}_I$, estimating occupied sector parameters to within relative error $\varepsilon_{\max}$ suppresses the residual channel bias $\bm{p}_{\mathrm{res}}$, normalized by the mean weight, to (see Supplementary Section~\ref{subsec:app_proof_distance_suppression})
\begin{equation}
\|\bm{p}_{\mathrm{res}}\|_1 = \mathcal{O}\big(\varepsilon_{\max}\|\bm{p}_I\|_1\big) = \mathcal{O}(\varepsilon_{\max}\alpha p_L),
\label{eq:four_sector_distance_suppression2}
\end{equation}
attenuating relative calibration errors by $\alpha \equiv \|\bm{p}_I\|_1/p_L < 1$. \medskip

\textbf{Noise learning}. Evaluating the sector weights $w_Q$ in Eq.~\eqref{eq:four_sector_closed_form_weights} requires knowledge of the parameters $\pi_Q$ and $\bm{p}_Q$.
Rather than measuring these parameters through dedicated logical tomography circuits that must sample rare $\mathcal{O}(p_L)$ logical failures, we learn a stationary, independent-Bernoulli physical detector error model (DEM)~\cite{gidney2021stim} for an identifiable noise family directly from a split of the stabilizer syndromes recorded during the target computation~\cite{spitz2018adaptive,harper2020efficient,wagner2022pauli,wagner2023learning,zheng2026learning,takou2026logical}
(Supplementary Section~\ref{sec:supp_noise_learning_identifiability}). This allows us to simulate circuit Clifford layers using a classical model (``digital twin'') of the target circuit, as described in \vecLearnSec. Because DEM learning relies on physical detector flips occurring at rate $\mathcal{O}(p)$ in every run, its quantum shot budget scales as $\mathcal{O}(1/p)$ rather than $\mathcal{O}(1/p_L)$. The digital twin then propagates the learned DEM through $\mathcal{D}_{\mathrm{rt}}$ and $\mathcal{D}_{\mathrm{off}}$ in classical post-processing to compute $\hat{\pi}_Q$ and $\hat{\bm{p}}_Q$, using at most $\mathcal{O}((\alpha p_L \varepsilon_{\max}^2)^{-1})$ classical screening runs. This count can be improved by importance-sampling methods, while stratified subsampling keeps $\mathcal{D}_{\mathrm{off}}$ calls at $\mathcal{O}((\alpha \varepsilon_{\max}^2)^{-1})$.

Simulations of a rotated surface-code memory under spatially randomized, stationary SI1000 circuit noise with noiseless initial preparation and terminal data readout confirm the accuracy and sample efficiency of this in-family calibration pipeline (Fig.~\ref{fig:vlpec_validation}). Even at $p = 10^{-3}$ and a sampling budget of $N_{\mathrm{learn}} = 50{,}000$, where $d$-round memory runs contain $\approx 0.8$ expected disagreements, the digital twin resolves the disagreement sectors offline via stratified classical Monte Carlo sampling and signed likelihood-ratio reweighting of shared trajectories (see \vecMethodsSec). The syndrome-only estimator reconstructs DEM mechanism rates (Fig.~\ref{fig:vlpec_validation}a) and unmitigated logical Pauli error rates (Fig.~\ref{fig:vlpec_validation}b) to within sub-percent relative error.
Substituting the learned sector parameters into Eq.~\eqref{eq:four_sector_closed_form_weights} suppresses the normalized residual
channel bias as $\|\bm{p}_{\mathrm{res}}\|_1\sim 1/\sqrt{N_{\mathrm{learn}}}$, achieving multiple orders of magnitude of suppression relative to both the unmitigated logical error rate and agreement postselection (Fig.~\ref{fig:vlpec_validation}c; see Supplementary Section~\ref{sec:supp_noise_learning_identifiability}).

\medskip

\textbf{Full-circuit integration}.
Scaling VEC to deep computations requires modeling circuit noise as a sequence of local channels across layers.
In an $n$-qubit circuit of $L$ sequential layers ($\ell = 1, \dots, L$), each layer spans $R$ stabilizer-measurement rounds of parallel Clifford gates (some conditioned on earlier measurements) and magic-state preparation gadgets~\cite{skoric2023parallel} (Fig.~\ref{fig:circuit_integration}a). Across consecutive unconditioned Clifford layers, $\mathcal{D}_{\mathrm{rt}}$ can decode continuously without committing after every layer.
We only assume that before feedforward, unmeasured logical qubits continue through a buffer of at least $b\geq d$ rounds before $\mathcal{D}_{\mathrm{rt}}$ commits that earlier layer's recovery and feedforward (see \vecMethodsSec), matching the processor's native real-time schedule.

We can then model the noise of each layer $\ell$ in post-processing through $K_\ell = \mathcal{O}(n)$ nearly independent local logical error channels $(\ell, c)$ ($c = 1, \dots, K_\ell$), totaling $V = \sum_{\ell=1}^{L} K_\ell$ channels. Naively truncating $\bm{s}$ at every layer boundary would cut multi-round error chains and introduce $\mathcal{O}(1)$ relative time-correlated errors. Instead,
$\mathcal{D}_{\mathrm{off}}$ advances over overlapping windows of an $R$-round commit region spanning layer $\ell$ and a $b$-round buffer region.
The decoder must commit errors to layer $\ell$, and their detection events are removed (modulo $2$) from the buffer; decoding then continues to window $\ell + 1$. Because the accumulated $\mathcal{D}_{\mathrm{off}}$ and $\mathcal{D}_{\mathrm{rt}}$ recoveries both match $\bm{s}$, their differences have trivial syndrome and, below threshold, form localized spacetime cycles with well-defined logical Pauli actions. Assigning each cluster intact to a local channel $(\ell, c)$ at its latest layer then sets the disagreement sector $Q_{\ell,c}(\bm{s})$ ($Q_{\ell,c} = I$ on local agreement).

In parallel, the digital twin draws physical faults from the syndrome-learned DEM, generates their syndromes, and runs the same real-time and offline windowed decoders. Comparing the injected faults against the simulated real-time recovery identifies any residual logical error; tallying these outcomes by sector $Q_{\ell,c}$ estimates the local sector shares and conditional error rates, which determine $w_{\ell,c}(Q)$ via Eq.~\eqref{eq:four_sector_closed_form_weights}.
Multiplying these weights forms the full-circuit shot weight $\widetilde{w}(\bm{s}) = \prod_{\ell=1}^{L}\prod_{c=1}^{K_\ell} w_{\ell,c}(Q_{\ell,c}(\bm{s}))$ using at most $\mathcal{O}(n)$ parameters per layer.

Because independent channel errors factorize and cancel automatically under this product weight, residual bias arises only from correlated events that affect more than one attributed channel. Suppose that each local channel is calibrated to cancel its marginal error with bounded weights satisfying $w_{\ell,c}(I) = 1 + \mathcal{O}(p_L)$, and that fewer than $r + k$ physical faults produce a disagreement or residual error in at most one channel $(\ell, c)$, where $r$ is the minimum uncorrectable error weight and $k\geq 1$.
The bias per channel (normalized by the worst-case unmitigated scale $2 V p_L$) for any observable $\|O\|_\infty \le 1$ in the limit $p \to 0$ obeys
\begin{equation}
 \epsilon = \frac{\bigl|\EE[\hat{\langle O\rangle}_{\mathrm{VEC}}]
       -\langle O\rangle_{\mathrm{ideal}}\bigr|}{2 V p_L}
 =\mathcal{O}\!\left(p^{\min(r,k)}\right).
 \label{eq:full_circuit_relative_bias}
\end{equation}

\begin{figure}[t!]
    \centering
    \includegraphics[width=\columnwidth]{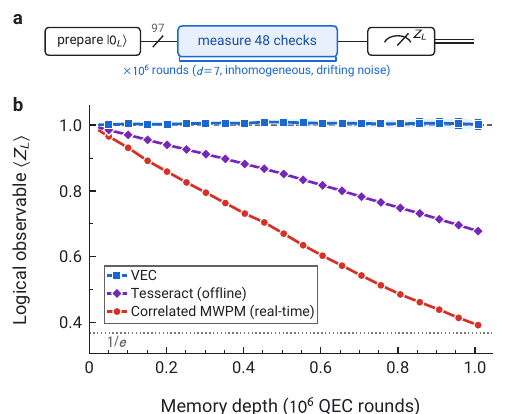}
    \caption{\textbf{Deep surface-code memory mitigation.}
    Offline slice-by-slice noise learning tracks smooth intra-circuit drift across $10^6$ stabilizer rounds without recalibrating the decoders.
    \textbf{(a)}~A distance-$7$ rotated surface-code logical qubit prepared in $|0_L\rangle$ undergoes $1{,}008{,}000$ rounds of $48$ stabilizer checks under spatially disordered, smoothly drifting SI1000 circuit noise ($p_0 = 8.85\times 10^{-4}$, chip-mean drift $p(t)/p_0 \in [0.925, 1.056]$, $2.5\times 10^4$ shots; see \vecMethodsSec).
    \textbf{(b)}~Logical expectation $\langle Z_L \rangle$ across $>10^6$ rounds.
    Sliding-window VEC ($R = 28$, $b = 14$, $36{,}000$ committed windows; blue squares) pairs fixed $t = 0$ real-time correlated MWPM (red circles) and offline search-based Tesseract decoding~\cite{beni2025tesseract} (purple diamonds) with weights calibrated from ten chronological syndrome-learned slices, restoring $\langle Z_L \rangle = 1.002 \pm 0.015$ at overhead $\widehat{C}_{\mathrm{tot}} = 2.21$.
    Error bars and shading show $1\text{ s.e.m.}$}
    \label{fig:memory_1m}
\end{figure}

For fault-tolerant Clifford circuits of circuit-level distance $d \ge 2$ ($r = \lceil d/2 \rceil$) and exact minimum-weight decoders as $p \to 0$, a buffer of $b \ge 2d - 3$ rounds guarantees $k = \lfloor d/2 \rfloor$ and $\epsilon = \mathcal{O}(p^{\lfloor d/2 \rfloor})$ (see Supplementary Section~\ref{sec:supp_full_circuit_integration}).
Fig.~\ref{fig:circuit_integration}b illustrates this error scaling under weighted real-time and correlated offline matching (the latter retaining only small local discrepancy clusters; see \vecMethodsSec):
the bias, evaluated in the nearest-window approximation, has fitted exponents above $\lfloor d/2 \rfloor$ across distances.
We also demonstrate unbiased recovery across more than $10^6$ stabilizer rounds on a $d = 7$ surface-code memory under smoothly drifting physical noise, where real-time decoding decays toward $1/e$ (Fig.~\ref{fig:memory_1m}). Even when both decoders are configured for a static initial noise model, VEC tracks the drift purely in post-processing by dividing the syndrome stream into chronological slices, re-learning the noise parameters per slice, and recomputing the corresponding weights (Supplementary Section~\ref{subsec:fc-million-round-drift}).

Non-Clifford gate-teleportation gadgets can be treated in the same way as Clifford ones by scheduling a $b$-round buffer between the destructive measurement and the conditional $S$ gate.
This buffer separates measurement errors from faults during the $S$ gate. Consequently, errors in the conditional $S$ gate can be removed by canceling the bit-flip error on the preceding measurement outcome, rendering the feedforward operation effectively noiseless (see Proposition~\ref{prop:fc-seam}, Supplementary Section~\ref{sec:fc-seam}).

Removing logical bias incurs a quantum sampling cost $C_{\mathrm{tot}} \equiv \EE_{\bm{s}}[\widetilde{w}(\bm{s})^2]$ derived from Eq.~\eqref{eq:vpec_sampling_overhead_def}. At fixed $V$ as $p \to 0$ under the single-channel fault condition, weight correlations remain subleading, so $\EE[\prod_{\ell,c} w_{\ell,c}^2] \approx \prod_{\ell,c} \EE[w_{\ell,c}^2]$ up to $\mathcal{O}(p_L p^{\min(r,k)})$ error.
For operation-uniform rate $p_L$ and $C_{\ell,c} = 1 + \eta p_L + \mathcal{O}(p_L^2)$,
\begin{align}
C_{\mathrm{tot}} &\approx \prod_{\ell,c} C_{\ell,c} = \big(1 + \eta p_L + \mathcal{O}(p_L^2)\big)^V \nonumber \\
&= \exp\!\big(\eta V p_L + \mathcal{O}(V p_L^2)\big) = \frac{1}{F^\eta}\big[1 + \mathcal{O}(V p_L^2)\big],
\label{eq:full_circuit_overhead_scaling}
\end{align}
where $F \equiv \exp(-V p_L) \simeq (1 - p_L)^V$ constitutes the probability of an error-free circuit run under real-time decoding for $V p_L^2 \ll 1$.
When $\eta\to 1$, $C_{\mathrm{tot}}$ reduces toward $1/F$.

\medskip

\textbf{Discussion}. Virtual error cancellation (VEC) turns syndrome information into a resource for quantum error mitigation. By learning the physical fault model directly from the syndrome stream in a classical digital twin and comparing the committed real-time recovery against an accurate offline decoder, VEC suppresses residual logical bias while driving sampling overhead toward the inverse of the probability of an error-free circuit run under real-time decoding. For identifiable circuit-level noise, the syndrome records already generated during execution suffice to self-calibrate and cancel logical errors offline.

The framework combines with standard noise-tailoring tools across spacetime decoding architectures. Our analysis assumes an effective Pauli logical channel and an identifiable noise family, so non-Pauli logical errors leave residual bias if unaddressed. In practice, repeated stabilizer measurements suppress logical coherence with code distance~\cite{bravyi2018correcting}, while residual coherent rotations or leakage can be converted into Pauli or heralded faults using randomized compiling~\cite{wallman2016noise,vandenberg2023probabilistic,kim2023scalable}, erasure heralding, and leakage-reduction units~\cite{google2023suppressing,google2025below}. Alternatively, VEC can be generalized directly to non-Pauli logical channels by expanding the linear constraints in Eq.~\eqref{eq:pauli_cancellation_constraint} to a broader channel basis, provided those error components can be characterized. Beyond the ensemble decoders~\cite{roffe2020decoding,higgott2025blossom} benchmarked here, offline decoding~\cite{fowler2012surface,gidney2021stim} can incorporate multi-round ensembled matching~\cite{shutty2024efficient}, tensor-network, or neural-network decoders. Whether universal scaling near threshold~\cite{dennis2002topological} can bypass the identifiability and sample limits of microscopic detector-error learning by compressing calibration into a few macroscopic parameters at large distances remains an open question.

More broadly, offline syndrome-conditioned mitigation opens a new design axis for fault-tolerant compilation. Optional lightweight probe circuits can sharpen parameter identifiability when needed~\cite{harper2020efficient}, while embedding VEC overhead models into quantum compilers~\cite{gidney2021how} enables joint optimization of code distance, decoder pairing, and shot budget. On finite-size processors, trading a step in code distance for offline mitigation frees physical qubits for additional logical registers and magic-state factories, expanding the algorithmic reach of early fault-tolerant hardware.\medskip

\medskip\noindent\textbf{Note added}.---Concurrent with this work, related results on offline syndrome-based error mitigation were reported by Shutman \textit{et al.}~\cite{shutman2026error}.\medskip

\noindent\textbf{Acknowledgements}. The authors thank R.~Babbush, S.~Boixo, E.~Campbell, Y.~Chen, W.~Huggins, C.~Jones, D.~Kafri, J.~Kelly, B.~Kobrin, A.~Morvan, H.~Neven, M.~Newman, T.~E.~O'Brien, V.~Sivak, V.~Smelyanskiy, A.~Zalcman, and the broader Google Quantum AI team for valuable input and inspiring discussions.

\makeatletter
\def\bibsection{%
  \par
  \begingroup
  \let\addcontentsline\@gobblethree
  \section*{\refname}%
  \@nobreaktrue
  \endgroup
}%
\makeatother

\bibliography{references}

\begin{thebibliography}{57}%
\makeatletter
\providecommand \@ifxundefined [1]{%
 \@ifx{#1\undefined}
}%
\providecommand \@ifnum [1]{%
 \ifnum #1\expandafter \@firstoftwo
 \else \expandafter \@secondoftwo
 \fi
}%
\providecommand \@ifx [1]{%
 \ifx #1\expandafter \@firstoftwo
 \else \expandafter \@secondoftwo
 \fi
}%
\providecommand \natexlab [1]{#1}%
\providecommand \enquote  [1]{``#1''}%
\providecommand \bibnamefont  [1]{#1}%
\providecommand \bibfnamefont [1]{#1}%
\providecommand \citenamefont [1]{#1}%
\providecommand \href@noop [0]{\@secondoftwo}%
\providecommand \href [0]{\begingroup \@sanitize@url \@href}%
\providecommand \@href[1]{\@@startlink{#1}\@@href}%
\providecommand \@@href[1]{\endgroup#1\@@endlink}%
\providecommand \@sanitize@url [0]{\catcode `\\12\catcode `\$12\catcode
  `\&12\catcode `\#12\catcode `\^12\catcode `\_12\catcode `\%12\relax}%
\providecommand \@@startlink[1]{}%
\providecommand \@@endlink[0]{}%
\providecommand \url  [0]{\begingroup\@sanitize@url \@url }%
\providecommand \@url [1]{\endgroup\@href {#1}{\urlprefix }}%
\providecommand \urlprefix  [0]{URL }%
\providecommand \Eprint [0]{\href }%
\providecommand \doibase [0]{https://doi.org/}%
\providecommand \selectlanguage [0]{\@gobble}%
\providecommand \bibinfo  [0]{\@secondoftwo}%
\providecommand \bibfield  [0]{\@secondoftwo}%
\providecommand \translation [1]{[#1]}%
\providecommand \BibitemOpen [0]{}%
\providecommand \bibitemStop [0]{}%
\providecommand \bibitemNoStop [0]{.\EOS\space}%
\providecommand \EOS [0]{\spacefactor3000\relax}%
\providecommand \BibitemShut  [1]{\csname bibitem#1\endcsname}%
\let\auto@bib@innerbib\@empty
\bibitem [{\citenamefont {Kitaev}(2003)}]{kitaev2003fault}%
  \BibitemOpen
  \bibfield  {author} {\bibinfo {author} {\bibfnamefont {A.~Y.}\ \bibnamefont
  {Kitaev}},\ }\bibfield  {title} {\bibinfo {title} {{Fault-tolerant quantum
  computation by anyons}},\ }\href
  {https://doi.org/10.1016/S0003-4916(02)00018-0} {\bibfield  {journal}
  {\bibinfo  {journal} {Ann. Phys.}\ }\textbf {\bibinfo {volume} {303}},\
  \bibinfo {pages} {2} (\bibinfo {year} {2003})}\BibitemShut {NoStop}%
\bibitem [{\citenamefont {Fowler}\ \emph {et~al.}(2012)\citenamefont {Fowler},
  \citenamefont {Mariantoni}, \citenamefont {Martinis},\ and\ \citenamefont
  {Cleland}}]{fowler2012surface}%
  \BibitemOpen
  \bibfield  {author} {\bibinfo {author} {\bibfnamefont {A.~G.}\ \bibnamefont
  {Fowler}}, \bibinfo {author} {\bibfnamefont {M.}~\bibnamefont {Mariantoni}},
  \bibinfo {author} {\bibfnamefont {J.~M.}\ \bibnamefont {Martinis}},\ and\
  \bibinfo {author} {\bibfnamefont {A.~N.}\ \bibnamefont {Cleland}},\
  }\bibfield  {title} {\bibinfo {title} {{Surface codes: Towards practical
  large-scale quantum computation}},\ }\href
  {https://doi.org/10.1103/PhysRevA.86.032324} {\bibfield  {journal} {\bibinfo
  {journal} {Phys. Rev. A}\ }\textbf {\bibinfo {volume} {86}},\ \bibinfo
  {pages} {032324} (\bibinfo {year} {2012})}\BibitemShut {NoStop}%
\bibitem [{\citenamefont {{Google Quantum AI}}(2023)}]{google2023suppressing}%
  \BibitemOpen
  \bibfield  {author} {\bibinfo {author} {\bibnamefont {{Google Quantum AI}}},\
  }\bibfield  {title} {\bibinfo {title} {{Suppressing quantum errors by scaling
  a surface code logical qubit}},\ }\href
  {https://doi.org/10.1038/s41586-022-05434-1} {\bibfield  {journal} {\bibinfo
  {journal} {Nature}\ }\textbf {\bibinfo {volume} {614}},\ \bibinfo {pages}
  {676} (\bibinfo {year} {2023})}\BibitemShut {NoStop}%
\bibitem [{\citenamefont {Bluvstein}\ \emph {et~al.}(2024)\citenamefont
  {Bluvstein}, \citenamefont {Evered}, \citenamefont {Geim} \emph
  {et~al.}}]{bluvstein2024logical}%
  \BibitemOpen
  \bibfield  {author} {\bibinfo {author} {\bibfnamefont {D.}~\bibnamefont
  {Bluvstein}}, \bibinfo {author} {\bibfnamefont {S.~J.}\ \bibnamefont
  {Evered}}, \bibinfo {author} {\bibfnamefont {A.~A.}\ \bibnamefont {Geim}},
  \emph {et~al.},\ }\bibfield  {title} {\bibinfo {title} {{Logical quantum
  processor based on reconfigurable atom arrays}},\ }\href
  {https://doi.org/10.1038/s41586-023-06927-3} {\bibfield  {journal} {\bibinfo
  {journal} {Nature}\ }\textbf {\bibinfo {volume} {626}},\ \bibinfo {pages}
  {58} (\bibinfo {year} {2024})}\BibitemShut {NoStop}%
\bibitem [{\citenamefont {{Google Quantum AI and
  Collaborators}}(2025)}]{google2025below}%
  \BibitemOpen
  \bibfield  {author} {\bibinfo {author} {\bibnamefont {{Google Quantum AI and
  Collaborators}}},\ }\bibfield  {title} {\bibinfo {title} {{Quantum error
  correction below the surface code threshold}},\ }\href
  {https://doi.org/10.1038/s41586-024-08449-y} {\bibfield  {journal} {\bibinfo
  {journal} {Nature}\ }\textbf {\bibinfo {volume} {638}},\ \bibinfo {pages}
  {920} (\bibinfo {year} {2025})}\BibitemShut {NoStop}%
\bibitem [{\citenamefont {Lacroix}\ \emph {et~al.}(2025)\citenamefont
  {Lacroix}, \citenamefont {Bourassa}, \citenamefont {Heras}, \citenamefont
  {Zhang} \emph {et~al.}}]{lacroix2025scaling}%
  \BibitemOpen
  \bibfield  {author} {\bibinfo {author} {\bibfnamefont {N.}~\bibnamefont
  {Lacroix}}, \bibinfo {author} {\bibfnamefont {A.}~\bibnamefont {Bourassa}},
  \bibinfo {author} {\bibfnamefont {F.~J.~H.}\ \bibnamefont {Heras}}, \bibinfo
  {author} {\bibfnamefont {L.~M.}\ \bibnamefont {Zhang}}, \emph {et~al.},\
  }\bibfield  {title} {\bibinfo {title} {{Scaling and logic in the color code
  on a superconducting quantum processor}},\ }\href
  {https://doi.org/10.1038/s41586-025-09061-4} {\bibfield  {journal} {\bibinfo
  {journal} {Nature}\ }\textbf {\bibinfo {volume} {645}},\ \bibinfo {pages}
  {614} (\bibinfo {year} {2025})},\ \Eprint {https://arxiv.org/abs/2412.14256}
  {arXiv:2412.14256 [quant-ph]} \BibitemShut {NoStop}%
\bibitem [{\citenamefont {Bluvstein}\ \emph {et~al.}(2026)\citenamefont
  {Bluvstein}, \citenamefont {Geim}, \citenamefont {Li}, \citenamefont {Evered}
  \emph {et~al.}}]{bluvstein2025architecture}%
  \BibitemOpen
  \bibfield  {author} {\bibinfo {author} {\bibfnamefont {D.}~\bibnamefont
  {Bluvstein}}, \bibinfo {author} {\bibfnamefont {A.~A.}\ \bibnamefont {Geim}},
  \bibinfo {author} {\bibfnamefont {S.~H.}\ \bibnamefont {Li}}, \bibinfo
  {author} {\bibfnamefont {S.~J.}\ \bibnamefont {Evered}}, \emph {et~al.},\
  }\bibfield  {title} {\bibinfo {title} {{A fault-tolerant neutral-atom
  architecture for universal quantum computation}},\ }\href
  {https://doi.org/10.1038/s41586-025-09848-5} {\bibfield  {journal} {\bibinfo
  {journal} {Nature}\ }\textbf {\bibinfo {volume} {649}},\ \bibinfo {pages}
  {39} (\bibinfo {year} {2026})},\ \Eprint {https://arxiv.org/abs/2506.20661}
  {arXiv:2506.20661 [quant-ph]} \BibitemShut {NoStop}%
\bibitem [{\citenamefont {Lavasani}\ \emph {et~al.}(2026)\citenamefont
  {Lavasani}, \citenamefont {Reed}, \citenamefont {Dreiling}, \citenamefont
  {Park} \emph {et~al.}}]{lavasani2026color}%
  \BibitemOpen
  \bibfield  {author} {\bibinfo {author} {\bibfnamefont {A.}~\bibnamefont
  {Lavasani}}, \bibinfo {author} {\bibfnamefont {A.~P.}\ \bibnamefont {Reed}},
  \bibinfo {author} {\bibfnamefont {J.~M.}\ \bibnamefont {Dreiling}}, \bibinfo
  {author} {\bibfnamefont {A.~J.}\ \bibnamefont {Park}}, \emph {et~al.},\
  }\href@noop {} {\bibinfo {title} {{Quantum error correction well below the
  color code threshold}}} (\bibinfo {year} {2026}),\ \Eprint
  {https://arxiv.org/abs/2610.24077} {arXiv:2610.24077 [quant-ph]} \BibitemShut
  {NoStop}%
\bibitem [{\citenamefont {Gidney}\ and\ \citenamefont
  {Eker{\aa}}(2021)}]{gidney2021how}%
  \BibitemOpen
  \bibfield  {author} {\bibinfo {author} {\bibfnamefont {C.}~\bibnamefont
  {Gidney}}\ and\ \bibinfo {author} {\bibfnamefont {M.}~\bibnamefont
  {Eker{\aa}}},\ }\bibfield  {title} {\bibinfo {title} {{How to factor 2048 bit
  RSA integers in 8 hours using 20 million noisy qubits}},\ }\href
  {https://doi.org/10.22331/q-2021-04-15-433} {\bibfield  {journal} {\bibinfo
  {journal} {Quantum}\ }\textbf {\bibinfo {volume} {5}},\ \bibinfo {pages}
  {433} (\bibinfo {year} {2021})}\BibitemShut {NoStop}%
\bibitem [{\citenamefont {Gidney}\ \emph {et~al.}(2024)\citenamefont {Gidney},
  \citenamefont {Shutty},\ and\ \citenamefont {Jones}}]{gidney2024magic}%
  \BibitemOpen
  \bibfield  {author} {\bibinfo {author} {\bibfnamefont {C.}~\bibnamefont
  {Gidney}}, \bibinfo {author} {\bibfnamefont {N.}~\bibnamefont {Shutty}},\
  and\ \bibinfo {author} {\bibfnamefont {C.}~\bibnamefont {Jones}},\
  }\href@noop {} {\bibinfo {title} {{Magic state cultivation: growing T states
  as cheap as CNOT gates}}} (\bibinfo {year} {2024}),\ \Eprint
  {https://arxiv.org/abs/2409.17595} {arXiv:2409.17595 [quant-ph]} \BibitemShut
  {NoStop}%
\bibitem [{\citenamefont {Piveteau}\ \emph {et~al.}(2021)\citenamefont
  {Piveteau}, \citenamefont {Sutter}, \citenamefont {Bravyi}, \citenamefont
  {Gambetta},\ and\ \citenamefont {Temme}}]{piveteau2021error}%
  \BibitemOpen
  \bibfield  {author} {\bibinfo {author} {\bibfnamefont {C.}~\bibnamefont
  {Piveteau}}, \bibinfo {author} {\bibfnamefont {D.}~\bibnamefont {Sutter}},
  \bibinfo {author} {\bibfnamefont {S.}~\bibnamefont {Bravyi}}, \bibinfo
  {author} {\bibfnamefont {J.~M.}\ \bibnamefont {Gambetta}},\ and\ \bibinfo
  {author} {\bibfnamefont {K.}~\bibnamefont {Temme}},\ }\bibfield  {title}
  {\bibinfo {title} {{Error mitigation for universal gates on encoded
  qubits}},\ }\href {https://doi.org/10.1103/PhysRevLett.127.200505} {\bibfield
   {journal} {\bibinfo  {journal} {Phys. Rev. Lett.}\ }\textbf {\bibinfo
  {volume} {127}},\ \bibinfo {pages} {200505} (\bibinfo {year}
  {2021})}\BibitemShut {NoStop}%
\bibitem [{\citenamefont {Lostaglio}\ and\ \citenamefont
  {Ciani}(2021)}]{lostaglio2021error}%
  \BibitemOpen
  \bibfield  {author} {\bibinfo {author} {\bibfnamefont {M.}~\bibnamefont
  {Lostaglio}}\ and\ \bibinfo {author} {\bibfnamefont {A.}~\bibnamefont
  {Ciani}},\ }\bibfield  {title} {\bibinfo {title} {{Error mitigation and
  quantum-assisted simulation in the error corrected regime}},\ }\href
  {https://doi.org/10.1103/PhysRevLett.127.200506} {\bibfield  {journal}
  {\bibinfo  {journal} {Phys. Rev. Lett.}\ }\textbf {\bibinfo {volume} {127}},\
  \bibinfo {pages} {200506} (\bibinfo {year} {2021})}\BibitemShut {NoStop}%
\bibitem [{\citenamefont {Suzuki}\ \emph {et~al.}(2022)\citenamefont {Suzuki},
  \citenamefont {Endo}, \citenamefont {Fujii},\ and\ \citenamefont
  {Tokunaga}}]{suzuki2022quantum}%
  \BibitemOpen
  \bibfield  {author} {\bibinfo {author} {\bibfnamefont {Y.}~\bibnamefont
  {Suzuki}}, \bibinfo {author} {\bibfnamefont {S.}~\bibnamefont {Endo}},
  \bibinfo {author} {\bibfnamefont {K.}~\bibnamefont {Fujii}},\ and\ \bibinfo
  {author} {\bibfnamefont {Y.}~\bibnamefont {Tokunaga}},\ }\bibfield  {title}
  {\bibinfo {title} {{Quantum error mitigation as a universal error reduction
  technique: Applications from the NISQ to the fault-tolerant quantum computing
  eras}},\ }\href {https://doi.org/10.1103/PRXQuantum.3.010345} {\bibfield
  {journal} {\bibinfo  {journal} {PRX Quantum}\ }\textbf {\bibinfo {volume}
  {3}},\ \bibinfo {pages} {010345} (\bibinfo {year} {2022})}\BibitemShut
  {NoStop}%
\bibitem [{\citenamefont {Zhang}\ \emph {et~al.}(2026)\citenamefont {Zhang},
  \citenamefont {Xie}, \citenamefont {Gao}, \citenamefont {Yang}, \citenamefont
  {Bao} \emph {et~al.}}]{zhang2025demonstrating}%
  \BibitemOpen
  \bibfield  {author} {\bibinfo {author} {\bibfnamefont {A.}~\bibnamefont
  {Zhang}}, \bibinfo {author} {\bibfnamefont {H.}~\bibnamefont {Xie}}, \bibinfo
  {author} {\bibfnamefont {Y.}~\bibnamefont {Gao}}, \bibinfo {author}
  {\bibfnamefont {J.-N.}\ \bibnamefont {Yang}}, \bibinfo {author}
  {\bibfnamefont {Z.}~\bibnamefont {Bao}}, \emph {et~al.},\ }\bibfield  {title}
  {\bibinfo {title} {{Demonstrating quantum error mitigation on logical
  qubits}},\ }\href {https://doi.org/10.1038/s41467-025-67768-4} {\bibfield
  {journal} {\bibinfo  {journal} {Nat. Commun.}\ }\textbf {\bibinfo {volume}
  {17}},\ \bibinfo {pages} {1021} (\bibinfo {year} {2026})},\ \Eprint
  {https://arxiv.org/abs/2501.09079} {arXiv:2501.09079} \BibitemShut {NoStop}%
\bibitem [{\citenamefont {Umbrarescu}\ \emph {et~al.}(2026)\citenamefont
  {Umbrarescu}, \citenamefont {Higgott},\ and\ \citenamefont
  {Browne}}]{umbrarescu2026infinite}%
  \BibitemOpen
  \bibfield  {author} {\bibinfo {author} {\bibfnamefont {G.}~\bibnamefont
  {Umbrarescu}}, \bibinfo {author} {\bibfnamefont {O.}~\bibnamefont
  {Higgott}},\ and\ \bibinfo {author} {\bibfnamefont {D.~E.}\ \bibnamefont
  {Browne}},\ }\href@noop {} {\bibinfo {title} {{Infinite distance
  extrapolation: How error mitigation can enhance quantum error correction}}}
  (\bibinfo {year} {2026}),\ \Eprint {https://arxiv.org/abs/2603.11285}
  {arXiv:2603.11285 [quant-ph]} \BibitemShut {NoStop}%
\bibitem [{\citenamefont {Xie}\ \emph {et~al.}(2026)\citenamefont {Xie},
  \citenamefont {Yoshioka}, \citenamefont {Tsubouchi},\ and\ \citenamefont
  {Li}}]{xie2025noise}%
  \BibitemOpen
  \bibfield  {author} {\bibinfo {author} {\bibfnamefont {H.}~\bibnamefont
  {Xie}}, \bibinfo {author} {\bibfnamefont {N.}~\bibnamefont {Yoshioka}},
  \bibinfo {author} {\bibfnamefont {K.}~\bibnamefont {Tsubouchi}},\ and\
  \bibinfo {author} {\bibfnamefont {Y.}~\bibnamefont {Li}},\ }\bibfield
  {title} {\bibinfo {title} {{Noise-Agnostic Unbiased Quantum Error Mitigation
  for Logical Qubits}},\ }\href {https://doi.org/10.1103/2n23-4qg8} {\bibfield
  {journal} {\bibinfo  {journal} {Phys. Rev. Lett.}\ }\textbf {\bibinfo
  {volume} {136}},\ \bibinfo {pages} {010603} (\bibinfo {year} {2026})},\
  \Eprint {https://arxiv.org/abs/2504.12864} {arXiv:2504.12864} \BibitemShut
  {NoStop}%
\bibitem [{\citenamefont {Temme}\ \emph {et~al.}(2017)\citenamefont {Temme},
  \citenamefont {Bravyi},\ and\ \citenamefont {Gambetta}}]{temme2017error}%
  \BibitemOpen
  \bibfield  {author} {\bibinfo {author} {\bibfnamefont {K.}~\bibnamefont
  {Temme}}, \bibinfo {author} {\bibfnamefont {S.}~\bibnamefont {Bravyi}},\ and\
  \bibinfo {author} {\bibfnamefont {J.~M.}\ \bibnamefont {Gambetta}},\
  }\bibfield  {title} {\bibinfo {title} {{Error mitigation for short-depth
  quantum circuits}},\ }\href {https://doi.org/10.1103/PhysRevLett.119.180509}
  {\bibfield  {journal} {\bibinfo  {journal} {Phys. Rev. Lett.}\ }\textbf
  {\bibinfo {volume} {119}},\ \bibinfo {pages} {180509} (\bibinfo {year}
  {2017})}\BibitemShut {NoStop}%
\bibitem [{\citenamefont {Li}\ and\ \citenamefont
  {Benjamin}(2017)}]{li2017efficient}%
  \BibitemOpen
  \bibfield  {author} {\bibinfo {author} {\bibfnamefont {Y.}~\bibnamefont
  {Li}}\ and\ \bibinfo {author} {\bibfnamefont {S.~C.}\ \bibnamefont
  {Benjamin}},\ }\bibfield  {title} {\bibinfo {title} {{Efficient variational
  quantum simulator incorporating active error minimisation}},\ }\href
  {https://doi.org/10.1103/PhysRevX.7.021050} {\bibfield  {journal} {\bibinfo
  {journal} {Phys. Rev. X}\ }\textbf {\bibinfo {volume} {7}},\ \bibinfo {pages}
  {021050} (\bibinfo {year} {2017})}\BibitemShut {NoStop}%
\bibitem [{\citenamefont {Endo}\ \emph {et~al.}(2018)\citenamefont {Endo},
  \citenamefont {Benjamin},\ and\ \citenamefont {Li}}]{endo2018practical}%
  \BibitemOpen
  \bibfield  {author} {\bibinfo {author} {\bibfnamefont {S.}~\bibnamefont
  {Endo}}, \bibinfo {author} {\bibfnamefont {S.~C.}\ \bibnamefont {Benjamin}},\
  and\ \bibinfo {author} {\bibfnamefont {Y.}~\bibnamefont {Li}},\ }\bibfield
  {title} {\bibinfo {title} {{Practical quantum error mitigation for
  near-future applications}},\ }\href
  {https://doi.org/10.1103/PhysRevX.8.031027} {\bibfield  {journal} {\bibinfo
  {journal} {Phys. Rev. X}\ }\textbf {\bibinfo {volume} {8}},\ \bibinfo {pages}
  {031027} (\bibinfo {year} {2018})}\BibitemShut {NoStop}%
\bibitem [{\citenamefont {van~den Berg}\ \emph {et~al.}(2023)\citenamefont
  {van~den Berg}, \citenamefont {Minev}, \citenamefont {Kandala},\ and\
  \citenamefont {Temme}}]{vandenberg2023probabilistic}%
  \BibitemOpen
  \bibfield  {author} {\bibinfo {author} {\bibfnamefont {E.}~\bibnamefont
  {van~den Berg}}, \bibinfo {author} {\bibfnamefont {Z.~K.}\ \bibnamefont
  {Minev}}, \bibinfo {author} {\bibfnamefont {A.}~\bibnamefont {Kandala}},\
  and\ \bibinfo {author} {\bibfnamefont {K.}~\bibnamefont {Temme}},\ }\bibfield
   {title} {\bibinfo {title} {{Probabilistic error cancellation with sparse
  Pauli--Lindblad models on noisy quantum processors}},\ }\href
  {https://doi.org/10.1038/s41567-023-02042-2} {\bibfield  {journal} {\bibinfo
  {journal} {Nat. Phys.}\ }\textbf {\bibinfo {volume} {19}},\ \bibinfo {pages}
  {1116} (\bibinfo {year} {2023})}\BibitemShut {NoStop}%
\bibitem [{\citenamefont {Cai}\ \emph {et~al.}(2023)\citenamefont {Cai},
  \citenamefont {Babbush}, \citenamefont {Benjamin}, \citenamefont {Endo},
  \citenamefont {Huggins}, \citenamefont {Li}, \citenamefont {McClean},\ and\
  \citenamefont {O'Brien}}]{cai2023quantum}%
  \BibitemOpen
  \bibfield  {author} {\bibinfo {author} {\bibfnamefont {Z.}~\bibnamefont
  {Cai}}, \bibinfo {author} {\bibfnamefont {R.}~\bibnamefont {Babbush}},
  \bibinfo {author} {\bibfnamefont {S.~C.}\ \bibnamefont {Benjamin}}, \bibinfo
  {author} {\bibfnamefont {S.}~\bibnamefont {Endo}}, \bibinfo {author}
  {\bibfnamefont {W.~J.}\ \bibnamefont {Huggins}}, \bibinfo {author}
  {\bibfnamefont {Y.}~\bibnamefont {Li}}, \bibinfo {author} {\bibfnamefont
  {J.~R.}\ \bibnamefont {McClean}},\ and\ \bibinfo {author} {\bibfnamefont
  {T.~E.}\ \bibnamefont {O'Brien}},\ }\bibfield  {title} {\bibinfo {title}
  {{Quantum error mitigation}},\ }\href
  {https://doi.org/10.1103/RevModPhys.95.045005} {\bibfield  {journal}
  {\bibinfo  {journal} {Rev. Mod. Phys.}\ }\textbf {\bibinfo {volume} {95}},\
  \bibinfo {pages} {045005} (\bibinfo {year} {2023})}\BibitemShut {NoStop}%
\bibitem [{\citenamefont {Bomb{\'{\i}}n}\ \emph {et~al.}(2024)\citenamefont
  {Bomb{\'{\i}}n}, \citenamefont {Pant}, \citenamefont {Roberts},\ and\
  \citenamefont {Seetharam}}]{bombin2024fault}%
  \BibitemOpen
  \bibfield  {author} {\bibinfo {author} {\bibfnamefont {H.}~\bibnamefont
  {Bomb{\'{\i}}n}}, \bibinfo {author} {\bibfnamefont {M.}~\bibnamefont {Pant}},
  \bibinfo {author} {\bibfnamefont {S.}~\bibnamefont {Roberts}},\ and\ \bibinfo
  {author} {\bibfnamefont {K.~I.}\ \bibnamefont {Seetharam}},\ }\bibfield
  {title} {\bibinfo {title} {{Fault-tolerant postselection for low-overhead
  magic state preparation}},\ }\href
  {https://doi.org/10.1103/PRXQuantum.5.010302} {\bibfield  {journal} {\bibinfo
   {journal} {PRX Quantum}\ }\textbf {\bibinfo {volume} {5}},\ \bibinfo {pages}
  {010302} (\bibinfo {year} {2024})},\ \Eprint
  {https://arxiv.org/abs/2212.00813} {arXiv:2212.00813 [quant-ph]} \BibitemShut
  {NoStop}%
\bibitem [{\citenamefont {Gidney}\ \emph {et~al.}(2025)\citenamefont {Gidney},
  \citenamefont {Newman}, \citenamefont {Brooks},\ and\ \citenamefont
  {Jones}}]{gidney2025yoked}%
  \BibitemOpen
  \bibfield  {author} {\bibinfo {author} {\bibfnamefont {C.}~\bibnamefont
  {Gidney}}, \bibinfo {author} {\bibfnamefont {M.}~\bibnamefont {Newman}},
  \bibinfo {author} {\bibfnamefont {P.}~\bibnamefont {Brooks}},\ and\ \bibinfo
  {author} {\bibfnamefont {C.}~\bibnamefont {Jones}},\ }\bibfield  {title}
  {\bibinfo {title} {{Yoked surface codes}},\ }\href
  {https://doi.org/10.1038/s41467-025-59714-1} {\bibfield  {journal} {\bibinfo
  {journal} {Nat. Commun.}\ }\textbf {\bibinfo {volume} {16}},\ \bibinfo
  {pages} {4498} (\bibinfo {year} {2025})},\ \Eprint
  {https://arxiv.org/abs/2312.04522} {arXiv:2312.04522 [quant-ph]} \BibitemShut
  {NoStop}%
\bibitem [{\citenamefont {Meister}\ \emph {et~al.}(2024)\citenamefont
  {Meister}, \citenamefont {Pattison},\ and\ \citenamefont
  {Preskill}}]{meister2024efficient}%
  \BibitemOpen
  \bibfield  {author} {\bibinfo {author} {\bibfnamefont {N.}~\bibnamefont
  {Meister}}, \bibinfo {author} {\bibfnamefont {C.~A.}\ \bibnamefont
  {Pattison}},\ and\ \bibinfo {author} {\bibfnamefont {J.}~\bibnamefont
  {Preskill}},\ }\href@noop {} {\bibinfo {title} {{Efficient soft-output
  decoders for the surface code}}} (\bibinfo {year} {2024}),\ \Eprint
  {https://arxiv.org/abs/2405.07433} {arXiv:2405.07433 [quant-ph]} \BibitemShut
  {NoStop}%
\bibitem [{\citenamefont {Smith}\ \emph {et~al.}(2024)\citenamefont {Smith},
  \citenamefont {Brown},\ and\ \citenamefont {Bartlett}}]{smith2024mitigating}%
  \BibitemOpen
  \bibfield  {author} {\bibinfo {author} {\bibfnamefont {S.~C.}\ \bibnamefont
  {Smith}}, \bibinfo {author} {\bibfnamefont {B.~J.}\ \bibnamefont {Brown}},\
  and\ \bibinfo {author} {\bibfnamefont {S.~D.}\ \bibnamefont {Bartlett}},\
  }\bibfield  {title} {\bibinfo {title} {{Mitigating errors in logical
  qubits}},\ }\href {https://doi.org/10.1038/s42005-024-01883-4} {\bibfield
  {journal} {\bibinfo  {journal} {Commun. Phys.}\ }\textbf {\bibinfo {volume}
  {7}},\ \bibinfo {pages} {386} (\bibinfo {year} {2024})},\ \Eprint
  {https://arxiv.org/abs/2405.03766} {arXiv:2405.03766 [quant-ph]} \BibitemShut
  {NoStop}%
\bibitem [{\citenamefont {Dinc{\u{a}}}\ \emph {et~al.}(2025)\citenamefont
  {Dinc{\u{a}}}, \citenamefont {Chan},\ and\ \citenamefont
  {Benjamin}}]{dinca2025confidence}%
  \BibitemOpen
  \bibfield  {author} {\bibinfo {author} {\bibfnamefont {M.}~\bibnamefont
  {Dinc{\u{a}}}}, \bibinfo {author} {\bibfnamefont {T.}~\bibnamefont {Chan}},\
  and\ \bibinfo {author} {\bibfnamefont {S.~C.}\ \bibnamefont {Benjamin}},\
  }\href@noop {} {\bibinfo {title} {{Error mitigation for logical circuits
  using decoder confidence}}} (\bibinfo {year} {2025}),\ \Eprint
  {https://arxiv.org/abs/2512.15689} {arXiv:2512.15689 [quant-ph]} \BibitemShut
  {NoStop}%
\bibitem [{\citenamefont {Zhong}\ \emph {et~al.}(2025)\citenamefont {Zhong},
  \citenamefont {Munizzi}, \citenamefont {Chen},\ and\ \citenamefont
  {de~Jong}}]{zhong2025combining}%
  \BibitemOpen
  \bibfield  {author} {\bibinfo {author} {\bibfnamefont {D.}~\bibnamefont
  {Zhong}}, \bibinfo {author} {\bibfnamefont {W.}~\bibnamefont {Munizzi}},
  \bibinfo {author} {\bibfnamefont {H.}~\bibnamefont {Chen}},\ and\ \bibinfo
  {author} {\bibfnamefont {W.~A.}\ \bibnamefont {de~Jong}},\ }\href@noop {}
  {\bibinfo {title} {{Combining error detection and mitigation: A hybrid
  protocol for near-term quantum simulation}}} (\bibinfo {year} {2025}),\
  \Eprint {https://arxiv.org/abs/2510.01181} {arXiv:2510.01181 [quant-ph]}
  \BibitemShut {NoStop}%
\bibitem [{\citenamefont {Kumar}\ \emph {et~al.}(2026)\citenamefont {Kumar},
  \citenamefont {Tsunoda}, \citenamefont {Xue}, \citenamefont {Li},
  \citenamefont {Schoelkopf},\ and\ \citenamefont
  {Ding}}]{kumar2026codesigning}%
  \BibitemOpen
  \bibfield  {author} {\bibinfo {author} {\bibfnamefont {R.~S.}\ \bibnamefont
  {Kumar}}, \bibinfo {author} {\bibfnamefont {T.}~\bibnamefont {Tsunoda}},
  \bibinfo {author} {\bibfnamefont {S.~H.}\ \bibnamefont {Xue}}, \bibinfo
  {author} {\bibfnamefont {D.}~\bibnamefont {Li}}, \bibinfo {author}
  {\bibfnamefont {R.~J.}\ \bibnamefont {Schoelkopf}},\ and\ \bibinfo {author}
  {\bibfnamefont {Y.}~\bibnamefont {Ding}},\ }\href@noop {} {\bibinfo {title}
  {{Co-designing error mitigation and error detection for logical qubits}}}
  (\bibinfo {year} {2026}),\ \Eprint {https://arxiv.org/abs/2604.19871}
  {arXiv:2604.19871 [quant-ph]} \BibitemShut {NoStop}%
\bibitem [{\citenamefont {Zhou}\ \emph {et~al.}(2025)\citenamefont {Zhou},
  \citenamefont {Pexton}, \citenamefont {Kubica},\ and\ \citenamefont
  {Ding}}]{zhou2025soft}%
  \BibitemOpen
  \bibfield  {author} {\bibinfo {author} {\bibfnamefont {Z.}~\bibnamefont
  {Zhou}}, \bibinfo {author} {\bibfnamefont {S.}~\bibnamefont {Pexton}},
  \bibinfo {author} {\bibfnamefont {A.}~\bibnamefont {Kubica}},\ and\ \bibinfo
  {author} {\bibfnamefont {Y.}~\bibnamefont {Ding}},\ }\href@noop {} {\bibinfo
  {title} {{Error Mitigation of Fault-Tolerant Quantum Circuits with Soft
  Information}}} (\bibinfo {year} {2025}),\ \Eprint
  {https://arxiv.org/abs/2512.09863} {arXiv:2512.09863 [quant-ph]} \BibitemShut
  {NoStop}%
\bibitem [{\citenamefont {Aharonov}\ \emph {et~al.}(2025)\citenamefont
  {Aharonov}, \citenamefont {Atia}, \citenamefont {Bairey}, \citenamefont
  {Brakerski}, \citenamefont {Cohen}, \citenamefont {Golan}, \citenamefont
  {Gurwich}, \citenamefont {Lindner},\ and\ \citenamefont
  {Shutman}}]{aharonov2025salem}%
  \BibitemOpen
  \bibfield  {author} {\bibinfo {author} {\bibfnamefont {D.}~\bibnamefont
  {Aharonov}}, \bibinfo {author} {\bibfnamefont {Y.}~\bibnamefont {Atia}},
  \bibinfo {author} {\bibfnamefont {E.}~\bibnamefont {Bairey}}, \bibinfo
  {author} {\bibfnamefont {Z.}~\bibnamefont {Brakerski}}, \bibinfo {author}
  {\bibfnamefont {I.}~\bibnamefont {Cohen}}, \bibinfo {author} {\bibfnamefont
  {O.}~\bibnamefont {Golan}}, \bibinfo {author} {\bibfnamefont
  {I.}~\bibnamefont {Gurwich}}, \bibinfo {author} {\bibfnamefont {N.~H.}\
  \bibnamefont {Lindner}},\ and\ \bibinfo {author} {\bibfnamefont
  {M.}~\bibnamefont {Shutman}},\ }\href@noop {} {\bibinfo {title} {{Syndrome
  aware mitigation of logical errors}}} (\bibinfo {year} {2025}),\ \Eprint
  {https://arxiv.org/abs/2512.23810} {arXiv:2512.23810 [quant-ph]} \BibitemShut
  {NoStop}%
\bibitem [{\citenamefont {Tsubouchi}\ \emph {et~al.}(2026)\citenamefont
  {Tsubouchi}, \citenamefont {Kwon}, \citenamefont {Jiang},\ and\ \citenamefont
  {Yoshioka}}]{tsubouchi2026advantages}%
  \BibitemOpen
  \bibfield  {author} {\bibinfo {author} {\bibfnamefont {K.}~\bibnamefont
  {Tsubouchi}}, \bibinfo {author} {\bibfnamefont {H.}~\bibnamefont {Kwon}},
  \bibinfo {author} {\bibfnamefont {L.}~\bibnamefont {Jiang}},\ and\ \bibinfo
  {author} {\bibfnamefont {N.}~\bibnamefont {Yoshioka}},\ }\href@noop {}
  {\bibinfo {title} {{Quantum advantages for syndrome-aware noisy logical
  observable estimation}}} (\bibinfo {year} {2026}),\ \Eprint
  {https://arxiv.org/abs/2603.05145} {arXiv:2603.05145 [quant-ph]} \BibitemShut
  {NoStop}%
\bibitem [{\citenamefont {Huang}\ \emph {et~al.}(2022)\citenamefont {Huang},
  \citenamefont {Flammia},\ and\ \citenamefont
  {Preskill}}]{huang2022foundations}%
  \BibitemOpen
  \bibfield  {author} {\bibinfo {author} {\bibfnamefont {H.-Y.}\ \bibnamefont
  {Huang}}, \bibinfo {author} {\bibfnamefont {S.~T.}\ \bibnamefont {Flammia}},\
  and\ \bibinfo {author} {\bibfnamefont {J.}~\bibnamefont {Preskill}},\
  }\href@noop {} {\bibinfo {title} {{Foundations for learning from noisy
  quantum experiments}}} (\bibinfo {year} {2022}),\ \Eprint
  {https://arxiv.org/abs/2204.13691} {arXiv:2204.13691 [quant-ph]} \BibitemShut
  {NoStop}%
\bibitem [{\citenamefont {Suchsland}\ \emph {et~al.}(2026)\citenamefont
  {Suchsland}, \citenamefont {Schuster}, \citenamefont {Rudolph}, \citenamefont
  {Noll}, \citenamefont {O'Brien},\ and\ \citenamefont
  {Minev}}]{suchsland2026quantumerrormitigationinformation}%
  \BibitemOpen
  \bibfield  {author} {\bibinfo {author} {\bibfnamefont {P.}~\bibnamefont
  {Suchsland}}, \bibinfo {author} {\bibfnamefont {T.}~\bibnamefont {Schuster}},
  \bibinfo {author} {\bibfnamefont {M.~S.}\ \bibnamefont {Rudolph}}, \bibinfo
  {author} {\bibfnamefont {N.}~\bibnamefont {Noll}}, \bibinfo {author}
  {\bibfnamefont {T.~E.}\ \bibnamefont {O'Brien}},\ and\ \bibinfo {author}
  {\bibfnamefont {Z.~K.}\ \bibnamefont {Minev}},\ }\href@noop {} {\bibinfo
  {title} {{Quantum error mitigation from information dynamics}}} (\bibinfo
  {year} {2026}),\ \Eprint {https://arxiv.org/abs/2609.31934} {arXiv:2609.31934
  [quant-ph]} \BibitemShut {NoStop}%
\bibitem [{\citenamefont {Spitz}\ \emph {et~al.}(2018)\citenamefont {Spitz},
  \citenamefont {Tarasinski}, \citenamefont {Beenakker},\ and\ \citenamefont
  {O'Brien}}]{spitz2018adaptive}%
  \BibitemOpen
  \bibfield  {author} {\bibinfo {author} {\bibfnamefont {S.~T.}\ \bibnamefont
  {Spitz}}, \bibinfo {author} {\bibfnamefont {B.}~\bibnamefont {Tarasinski}},
  \bibinfo {author} {\bibfnamefont {C.~W.~J.}\ \bibnamefont {Beenakker}},\ and\
  \bibinfo {author} {\bibfnamefont {T.~E.}\ \bibnamefont {O'Brien}},\
  }\bibfield  {title} {\bibinfo {title} {{Adaptive Weight Estimator for Quantum
  Error Correction in a Time-Dependent Environment}},\ }\href
  {https://doi.org/10.1002/qute.201800012} {\bibfield  {journal} {\bibinfo
  {journal} {Adv. Quantum Technol.}\ }\textbf {\bibinfo {volume} {1}},\
  \bibinfo {pages} {1800012} (\bibinfo {year} {2018})}\BibitemShut {NoStop}%
\bibitem [{\citenamefont {Wagner}\ \emph {et~al.}(2023)\citenamefont {Wagner},
  \citenamefont {Kampermann}, \citenamefont {Bru{\ss}},\ and\ \citenamefont
  {Kliesch}}]{wagner2023learning}%
  \BibitemOpen
  \bibfield  {author} {\bibinfo {author} {\bibfnamefont {T.}~\bibnamefont
  {Wagner}}, \bibinfo {author} {\bibfnamefont {H.}~\bibnamefont {Kampermann}},
  \bibinfo {author} {\bibfnamefont {D.}~\bibnamefont {Bru{\ss}}},\ and\
  \bibinfo {author} {\bibfnamefont {M.}~\bibnamefont {Kliesch}},\ }\bibfield
  {title} {\bibinfo {title} {{Learning logical Pauli noise in quantum error
  correction}},\ }\href {https://doi.org/10.1103/PhysRevLett.130.200601}
  {\bibfield  {journal} {\bibinfo  {journal} {Phys. Rev. Lett.}\ }\textbf
  {\bibinfo {volume} {130}},\ \bibinfo {pages} {200601} (\bibinfo {year}
  {2023})}\BibitemShut {NoStop}%
\bibitem [{\citenamefont {Zheng}\ \emph {et~al.}(2026)\citenamefont {Zheng},
  \citenamefont {Chu}, \citenamefont {Chen}, \citenamefont {{Giannisis Manes}},
  \citenamefont {Lee}, \citenamefont {Zhou},\ and\ \citenamefont
  {Jiang}}]{zheng2026learning}%
  \BibitemOpen
  \bibfield  {author} {\bibinfo {author} {\bibfnamefont {H.}~\bibnamefont
  {Zheng}}, \bibinfo {author} {\bibfnamefont {C.-T.}\ \bibnamefont {Chu}},
  \bibinfo {author} {\bibfnamefont {S.}~\bibnamefont {Chen}}, \bibinfo {author}
  {\bibfnamefont {A.}~\bibnamefont {{Giannisis Manes}}}, \bibinfo {author}
  {\bibfnamefont {S.}~\bibnamefont {Lee}}, \bibinfo {author} {\bibfnamefont
  {S.}~\bibnamefont {Zhou}},\ and\ \bibinfo {author} {\bibfnamefont
  {L.}~\bibnamefont {Jiang}},\ }\href@noop {} {\bibinfo {title} {{Efficient
  learning of logical noise from syndrome data}}} (\bibinfo {year} {2026}),\
  \Eprint {https://arxiv.org/abs/2601.22286} {arXiv:2601.22286 [quant-ph]}
  \BibitemShut {NoStop}%
\bibitem [{\citenamefont {Takou}\ \emph {et~al.}(2026)\citenamefont {Takou},
  \citenamefont {Benito}, \citenamefont {Vezvaee}, \citenamefont {Lidar},\ and\
  \citenamefont {Brown}}]{takou2026logical}%
  \BibitemOpen
  \bibfield  {author} {\bibinfo {author} {\bibfnamefont {E.}~\bibnamefont
  {Takou}}, \bibinfo {author} {\bibfnamefont {C.}~\bibnamefont {Benito}},
  \bibinfo {author} {\bibfnamefont {A.}~\bibnamefont {Vezvaee}}, \bibinfo
  {author} {\bibfnamefont {D.~A.}\ \bibnamefont {Lidar}},\ and\ \bibinfo
  {author} {\bibfnamefont {K.~R.}\ \bibnamefont {Brown}},\ }\href@noop {}
  {\bibinfo {title} {{Logical error estimation from syndrome data of
  surface-code experiments}}} (\bibinfo {year} {2026}),\ \Eprint
  {https://arxiv.org/abs/2606.11496} {arXiv:2606.11496 [quant-ph]} \BibitemShut
  {NoStop}%
\bibitem [{\citenamefont {Bravyi}\ \emph {et~al.}(2014)\citenamefont {Bravyi},
  \citenamefont {Suchara},\ and\ \citenamefont {Vargo}}]{bravyi2014efficient}%
  \BibitemOpen
  \bibfield  {author} {\bibinfo {author} {\bibfnamefont {S.}~\bibnamefont
  {Bravyi}}, \bibinfo {author} {\bibfnamefont {M.}~\bibnamefont {Suchara}},\
  and\ \bibinfo {author} {\bibfnamefont {A.}~\bibnamefont {Vargo}},\ }\bibfield
   {title} {\bibinfo {title} {{Efficient algorithms for maximum likelihood
  decoding in the surface code}},\ }\href
  {https://doi.org/10.1103/PhysRevA.90.032326} {\bibfield  {journal} {\bibinfo
  {journal} {Phys. Rev. A}\ }\textbf {\bibinfo {volume} {90}},\ \bibinfo
  {pages} {032326} (\bibinfo {year} {2014})}\BibitemShut {NoStop}%
\bibitem [{\citenamefont {Roffe}\ \emph {et~al.}(2020)\citenamefont {Roffe},
  \citenamefont {White}, \citenamefont {Burton},\ and\ \citenamefont
  {Campbell}}]{roffe2020decoding}%
  \BibitemOpen
  \bibfield  {author} {\bibinfo {author} {\bibfnamefont {J.}~\bibnamefont
  {Roffe}}, \bibinfo {author} {\bibfnamefont {D.~R.}\ \bibnamefont {White}},
  \bibinfo {author} {\bibfnamefont {S.}~\bibnamefont {Burton}},\ and\ \bibinfo
  {author} {\bibfnamefont {E.~T.}\ \bibnamefont {Campbell}},\ }\bibfield
  {title} {\bibinfo {title} {{Decoding across the quantum low-density
  parity-check code landscape}},\ }\href
  {https://doi.org/10.1103/PhysRevResearch.2.043423} {\bibfield  {journal}
  {\bibinfo  {journal} {Phys. Rev. Res.}\ }\textbf {\bibinfo {volume} {2}},\
  \bibinfo {pages} {043423} (\bibinfo {year} {2020})}\BibitemShut {NoStop}%
\bibitem [{\citenamefont {Shutty}\ \emph {et~al.}(2026)\citenamefont {Shutty},
  \citenamefont {Newman},\ and\ \citenamefont
  {Villalonga}}]{shutty2024efficient}%
  \BibitemOpen
  \bibfield  {author} {\bibinfo {author} {\bibfnamefont {N.}~\bibnamefont
  {Shutty}}, \bibinfo {author} {\bibfnamefont {M.}~\bibnamefont {Newman}},\
  and\ \bibinfo {author} {\bibfnamefont {B.}~\bibnamefont {Villalonga}},\
  }\bibfield  {title} {\bibinfo {title} {{Efficient near-optimal decoding of
  the surface code through ensembling}},\ }\href
  {https://doi.org/10.1103/77j6-32xx} {\bibfield  {journal} {\bibinfo
  {journal} {Phys. Rev. Lett.}\ }\textbf {\bibinfo {volume} {136}},\ \bibinfo
  {pages} {070603} (\bibinfo {year} {2026})},\ \Eprint
  {https://arxiv.org/abs/2401.12434} {arXiv:2401.12434 [quant-ph]} \BibitemShut
  {NoStop}%
\bibitem [{\citenamefont {Dennis}\ \emph {et~al.}(2002)\citenamefont {Dennis},
  \citenamefont {Kitaev}, \citenamefont {Landahl},\ and\ \citenamefont
  {Preskill}}]{dennis2002topological}%
  \BibitemOpen
  \bibfield  {author} {\bibinfo {author} {\bibfnamefont {E.}~\bibnamefont
  {Dennis}}, \bibinfo {author} {\bibfnamefont {A.}~\bibnamefont {Kitaev}},
  \bibinfo {author} {\bibfnamefont {A.}~\bibnamefont {Landahl}},\ and\ \bibinfo
  {author} {\bibfnamefont {J.}~\bibnamefont {Preskill}},\ }\bibfield  {title}
  {\bibinfo {title} {{Topological quantum memory}},\ }\href
  {https://doi.org/10.1063/1.1499754} {\bibfield  {journal} {\bibinfo
  {journal} {J. Math. Phys.}\ }\textbf {\bibinfo {volume} {43}},\ \bibinfo
  {pages} {4452} (\bibinfo {year} {2002})}\BibitemShut {NoStop}%
\bibitem [{\citenamefont {Bravyi}\ and\ \citenamefont
  {Kitaev}(2005)}]{bravyi2005universal}%
  \BibitemOpen
  \bibfield  {author} {\bibinfo {author} {\bibfnamefont {S.}~\bibnamefont
  {Bravyi}}\ and\ \bibinfo {author} {\bibfnamefont {A.}~\bibnamefont
  {Kitaev}},\ }\bibfield  {title} {\bibinfo {title} {{Universal quantum
  computation with ideal Clifford gates and noisy ancillas}},\ }\href
  {https://doi.org/10.1103/PhysRevA.71.022316} {\bibfield  {journal} {\bibinfo
  {journal} {Phys. Rev. A}\ }\textbf {\bibinfo {volume} {71}},\ \bibinfo
  {pages} {022316} (\bibinfo {year} {2005})}\BibitemShut {NoStop}%
\bibitem [{\citenamefont {Litinski}(2019)}]{litinski2019game}%
  \BibitemOpen
  \bibfield  {author} {\bibinfo {author} {\bibfnamefont {D.}~\bibnamefont
  {Litinski}},\ }\bibfield  {title} {\bibinfo {title} {{A game of surface
  codes: Large-scale quantum computing with lattice surgery}},\ }\href
  {https://doi.org/10.22331/q-2019-03-05-128} {\bibfield  {journal} {\bibinfo
  {journal} {Quantum}\ }\textbf {\bibinfo {volume} {3}},\ \bibinfo {pages}
  {128} (\bibinfo {year} {2019})}\BibitemShut {NoStop}%
\bibitem [{\citenamefont {Iyer}\ and\ \citenamefont
  {Poulin}(2015)}]{iyer2015hardness}%
  \BibitemOpen
  \bibfield  {author} {\bibinfo {author} {\bibfnamefont {P.}~\bibnamefont
  {Iyer}}\ and\ \bibinfo {author} {\bibfnamefont {D.}~\bibnamefont {Poulin}},\
  }\bibfield  {title} {\bibinfo {title} {{Hardness of decoding quantum
  stabilizer codes}},\ }\href {https://doi.org/10.1109/TIT.2015.2422294}
  {\bibfield  {journal} {\bibinfo  {journal} {IEEE Trans. Inf. Theory}\
  }\textbf {\bibinfo {volume} {61}},\ \bibinfo {pages} {5209} (\bibinfo {year}
  {2015})},\ \Eprint {https://arxiv.org/abs/1310.3235} {arXiv:1310.3235}
  \BibitemShut {NoStop}%
\bibitem [{\citenamefont {Gidney}(2021)}]{gidney2021stim}%
  \BibitemOpen
  \bibfield  {author} {\bibinfo {author} {\bibfnamefont {C.}~\bibnamefont
  {Gidney}},\ }\bibfield  {title} {\bibinfo {title} {{Stim: a fast stabilizer
  circuit simulator}},\ }\href {https://doi.org/10.22331/q-2021-07-06-497}
  {\bibfield  {journal} {\bibinfo  {journal} {Quantum}\ }\textbf {\bibinfo
  {volume} {5}},\ \bibinfo {pages} {497} (\bibinfo {year} {2021})}\BibitemShut
  {NoStop}%
\bibitem [{\citenamefont {Harper}\ \emph {et~al.}(2020)\citenamefont {Harper},
  \citenamefont {Flammia},\ and\ \citenamefont
  {Wallman}}]{harper2020efficient}%
  \BibitemOpen
  \bibfield  {author} {\bibinfo {author} {\bibfnamefont {R.}~\bibnamefont
  {Harper}}, \bibinfo {author} {\bibfnamefont {S.~T.}\ \bibnamefont
  {Flammia}},\ and\ \bibinfo {author} {\bibfnamefont {J.~J.}\ \bibnamefont
  {Wallman}},\ }\bibfield  {title} {\bibinfo {title} {{Efficient learning of
  quantum noise}},\ }\href {https://doi.org/10.1038/s41567-020-0992-8}
  {\bibfield  {journal} {\bibinfo  {journal} {Nat. Phys.}\ }\textbf {\bibinfo
  {volume} {16}},\ \bibinfo {pages} {1184} (\bibinfo {year}
  {2020})}\BibitemShut {NoStop}%
\bibitem [{\citenamefont {Wagner}\ \emph {et~al.}(2022)\citenamefont {Wagner},
  \citenamefont {Kampermann}, \citenamefont {Bru{\ss}},\ and\ \citenamefont
  {Kliesch}}]{wagner2022pauli}%
  \BibitemOpen
  \bibfield  {author} {\bibinfo {author} {\bibfnamefont {T.}~\bibnamefont
  {Wagner}}, \bibinfo {author} {\bibfnamefont {H.}~\bibnamefont {Kampermann}},
  \bibinfo {author} {\bibfnamefont {D.}~\bibnamefont {Bru{\ss}}},\ and\
  \bibinfo {author} {\bibfnamefont {M.}~\bibnamefont {Kliesch}},\ }\bibfield
  {title} {\bibinfo {title} {{Pauli channels can be estimated from syndrome
  measurements in quantum error correction}},\ }\href
  {https://doi.org/10.22331/q-2022-09-19-809} {\bibfield  {journal} {\bibinfo
  {journal} {Quantum}\ }\textbf {\bibinfo {volume} {6}},\ \bibinfo {pages}
  {809} (\bibinfo {year} {2022})},\ \Eprint {https://arxiv.org/abs/2107.14252}
  {arXiv:2107.14252} \BibitemShut {NoStop}%
\bibitem [{\citenamefont {Skoric}\ \emph {et~al.}(2023)\citenamefont {Skoric},
  \citenamefont {Browne}, \citenamefont {Barnes}, \citenamefont {Gillespie},\
  and\ \citenamefont {Campbell}}]{skoric2023parallel}%
  \BibitemOpen
  \bibfield  {author} {\bibinfo {author} {\bibfnamefont {L.}~\bibnamefont
  {Skoric}}, \bibinfo {author} {\bibfnamefont {D.~E.}\ \bibnamefont {Browne}},
  \bibinfo {author} {\bibfnamefont {K.~M.}\ \bibnamefont {Barnes}}, \bibinfo
  {author} {\bibfnamefont {N.~I.}\ \bibnamefont {Gillespie}},\ and\ \bibinfo
  {author} {\bibfnamefont {E.~T.}\ \bibnamefont {Campbell}},\ }\bibfield
  {title} {\bibinfo {title} {{Parallel window decoding enables scalable fault
  tolerant quantum computation}},\ }\href
  {https://doi.org/10.1038/s41467-023-42482-1} {\bibfield  {journal} {\bibinfo
  {journal} {Nat. Commun.}\ }\textbf {\bibinfo {volume} {14}},\ \bibinfo
  {pages} {7040} (\bibinfo {year} {2023})}\BibitemShut {NoStop}%
\bibitem [{\citenamefont {Beni}\ \emph {et~al.}(2025)\citenamefont {Beni},
  \citenamefont {Higgott},\ and\ \citenamefont {Shutty}}]{beni2025tesseract}%
  \BibitemOpen
  \bibfield  {author} {\bibinfo {author} {\bibfnamefont {L.~A.}\ \bibnamefont
  {Beni}}, \bibinfo {author} {\bibfnamefont {O.}~\bibnamefont {Higgott}},\ and\
  \bibinfo {author} {\bibfnamefont {N.}~\bibnamefont {Shutty}},\ }\href@noop {}
  {\bibinfo {title} {{Tesseract: A Search-Based Decoder for Quantum Error
  Correction}}} (\bibinfo {year} {2025}),\ \Eprint
  {https://arxiv.org/abs/2503.10988} {arXiv:2503.10988 [quant-ph]} \BibitemShut
  {NoStop}%
\bibitem [{\citenamefont {Bravyi}\ \emph {et~al.}(2018)\citenamefont {Bravyi},
  \citenamefont {Englbrecht}, \citenamefont {K{\"o}nig},\ and\ \citenamefont
  {Peard}}]{bravyi2018correcting}%
  \BibitemOpen
  \bibfield  {author} {\bibinfo {author} {\bibfnamefont {S.}~\bibnamefont
  {Bravyi}}, \bibinfo {author} {\bibfnamefont {M.}~\bibnamefont {Englbrecht}},
  \bibinfo {author} {\bibfnamefont {R.}~\bibnamefont {K{\"o}nig}},\ and\
  \bibinfo {author} {\bibfnamefont {N.}~\bibnamefont {Peard}},\ }\bibfield
  {title} {\bibinfo {title} {{Correcting coherent errors with surface codes}},\
  }\href {https://doi.org/10.1038/s41534-018-0106-y} {\bibfield  {journal}
  {\bibinfo  {journal} {npj Quantum Inf.}\ }\textbf {\bibinfo {volume} {4}},\
  \bibinfo {pages} {55} (\bibinfo {year} {2018})}\BibitemShut {NoStop}%
\bibitem [{\citenamefont {Wallman}\ and\ \citenamefont
  {Emerson}(2016)}]{wallman2016noise}%
  \BibitemOpen
  \bibfield  {author} {\bibinfo {author} {\bibfnamefont {J.~J.}\ \bibnamefont
  {Wallman}}\ and\ \bibinfo {author} {\bibfnamefont {J.}~\bibnamefont
  {Emerson}},\ }\bibfield  {title} {\bibinfo {title} {{Noise tailoring for
  scalable quantum computation via randomized compiling}},\ }\href
  {https://doi.org/10.1103/PhysRevA.94.052325} {\bibfield  {journal} {\bibinfo
  {journal} {Phys. Rev. A}\ }\textbf {\bibinfo {volume} {94}},\ \bibinfo
  {pages} {052325} (\bibinfo {year} {2016})}\BibitemShut {NoStop}%
\bibitem [{\citenamefont {Kim}\ \emph {et~al.}(2023)\citenamefont {Kim},
  \citenamefont {Wood}, \citenamefont {Yoder}, \citenamefont {Merkel},
  \citenamefont {Gambetta}, \citenamefont {Temme},\ and\ \citenamefont
  {Kandala}}]{kim2023scalable}%
  \BibitemOpen
  \bibfield  {author} {\bibinfo {author} {\bibfnamefont {Y.}~\bibnamefont
  {Kim}}, \bibinfo {author} {\bibfnamefont {C.~J.}\ \bibnamefont {Wood}},
  \bibinfo {author} {\bibfnamefont {T.~J.}\ \bibnamefont {Yoder}}, \bibinfo
  {author} {\bibfnamefont {S.~T.}\ \bibnamefont {Merkel}}, \bibinfo {author}
  {\bibfnamefont {J.~M.}\ \bibnamefont {Gambetta}}, \bibinfo {author}
  {\bibfnamefont {K.}~\bibnamefont {Temme}},\ and\ \bibinfo {author}
  {\bibfnamefont {A.}~\bibnamefont {Kandala}},\ }\bibfield  {title} {\bibinfo
  {title} {{Scalable error mitigation for noisy quantum circuits produces
  competitive expectation values}},\ }\href
  {https://doi.org/10.1038/s41567-022-01914-3} {\bibfield  {journal} {\bibinfo
  {journal} {Nat. Phys.}\ }\textbf {\bibinfo {volume} {19}},\ \bibinfo {pages}
  {752} (\bibinfo {year} {2023})}\BibitemShut {NoStop}%
\bibitem [{\citenamefont {Higgott}\ and\ \citenamefont
  {Gidney}(2025)}]{higgott2025blossom}%
  \BibitemOpen
  \bibfield  {author} {\bibinfo {author} {\bibfnamefont {O.}~\bibnamefont
  {Higgott}}\ and\ \bibinfo {author} {\bibfnamefont {C.}~\bibnamefont
  {Gidney}},\ }\bibfield  {title} {\bibinfo {title} {{Sparse Blossom:
  correcting a million errors per core second with minimum-weight matching}},\
  }\href {https://doi.org/10.22331/q-2025-01-20-1600} {\bibfield  {journal}
  {\bibinfo  {journal} {Quantum}\ }\textbf {\bibinfo {volume} {9}},\ \bibinfo
  {pages} {1600} (\bibinfo {year} {2025})}\BibitemShut {NoStop}%
\bibitem [{\citenamefont {Shutman}\ \emph {et~al.}(2026)\citenamefont
  {Shutman}, \citenamefont {Aharonov}, \citenamefont {Arad}, \citenamefont
  {Atia}, \citenamefont {Bairey}, \citenamefont {Brakerski}, \citenamefont
  {Katzir}, \citenamefont {Lindner}, \citenamefont {Golan},\ and\ \citenamefont
  {Gurwich}}]{shutman2026error}%
  \BibitemOpen
  \bibfield  {author} {\bibinfo {author} {\bibfnamefont {M.}~\bibnamefont
  {Shutman}}, \bibinfo {author} {\bibfnamefont {D.}~\bibnamefont {Aharonov}},
  \bibinfo {author} {\bibfnamefont {I.}~\bibnamefont {Arad}}, \bibinfo {author}
  {\bibfnamefont {Y.}~\bibnamefont {Atia}}, \bibinfo {author} {\bibfnamefont
  {E.}~\bibnamefont {Bairey}}, \bibinfo {author} {\bibfnamefont
  {Z.}~\bibnamefont {Brakerski}}, \bibinfo {author} {\bibfnamefont {B.~A.}\
  \bibnamefont {Katzir}}, \bibinfo {author} {\bibfnamefont {N.~H.}\
  \bibnamefont {Lindner}}, \bibinfo {author} {\bibfnamefont {O.}~\bibnamefont
  {Golan}},\ and\ \bibinfo {author} {\bibfnamefont {I.}~\bibnamefont
  {Gurwich}},\ }\href@noop {} {\bibinfo {title} {{Error Mitigation From a
  Single Unmodified Error-Corrected Circuit}}} (\bibinfo {year} {2026}),\
  \Eprint {https://arxiv.org/abs/2610.06792} {arXiv:2610.06792 [quant-ph]}
  \BibitemShut {NoStop}%
\bibitem [{\citenamefont {Gidney}\ \emph {et~al.}(2021)\citenamefont {Gidney},
  \citenamefont {Newman}, \citenamefont {Fowler},\ and\ \citenamefont
  {Broughton}}]{gidney2021fault}%
  \BibitemOpen
  \bibfield  {author} {\bibinfo {author} {\bibfnamefont {C.}~\bibnamefont
  {Gidney}}, \bibinfo {author} {\bibfnamefont {M.}~\bibnamefont {Newman}},
  \bibinfo {author} {\bibfnamefont {A.}~\bibnamefont {Fowler}},\ and\ \bibinfo
  {author} {\bibfnamefont {M.}~\bibnamefont {Broughton}},\ }\bibfield  {title}
  {\bibinfo {title} {{A fault-tolerant honeycomb memory}},\ }\href
  {https://doi.org/10.22331/q-2021-12-20-605} {\bibfield  {journal} {\bibinfo
  {journal} {Quantum}\ }\textbf {\bibinfo {volume} {5}},\ \bibinfo {pages}
  {605} (\bibinfo {year} {2021})},\ \Eprint {https://arxiv.org/abs/2108.10457}
  {arXiv:2108.10457 [quant-ph]} \BibitemShut {NoStop}%
\bibitem [{\citenamefont {Higgott}(2022)}]{higgott2022pymatching}%
  \BibitemOpen
  \bibfield  {author} {\bibinfo {author} {\bibfnamefont {O.}~\bibnamefont
  {Higgott}},\ }\bibfield  {title} {\bibinfo {title} {{PyMatching: A Python
  package for decoding quantum error-correcting codes with minimum-weight
  perfect matching}},\ }\href {https://doi.org/10.1145/3555546} {\bibfield
  {journal} {\bibinfo  {journal} {ACM Trans. Quantum Comput.}\ }\textbf
  {\bibinfo {volume} {3}},\ \bibinfo {pages} {1} (\bibinfo {year}
  {2022})}\BibitemShut {NoStop}%
\bibitem [{\citenamefont {Bhatia}\ and\ \citenamefont
  {Davis}(2000)}]{bhatia2000better}%
  \BibitemOpen
  \bibfield  {author} {\bibinfo {author} {\bibfnamefont {R.}~\bibnamefont
  {Bhatia}}\ and\ \bibinfo {author} {\bibfnamefont {C.}~\bibnamefont {Davis}},\
  }\bibfield  {title} {\bibinfo {title} {{A better bound on the variance}},\
  }\href {https://doi.org/10.1080/00029890.2000.12005203} {\bibfield  {journal}
  {\bibinfo  {journal} {Am. Math. Mon.}\ }\textbf {\bibinfo {volume} {107}},\
  \bibinfo {pages} {353} (\bibinfo {year} {2000})}\BibitemShut {NoStop}%
\end{thebibliography}%

\clearpage
\setcounter{theorem}{0}
\setcounter{proposition}{0}
\setcounter{corollary}{0}
\setcounter{lemma}{0}
\makeatletter
\def\l@subsubsection#1#2{}
\makeatother
\setcounter{tocdepth}{2}
\clearpage \pagebreak \setcounter{page}{1} \setcounter{section}{0} \setcounter{equation}{0} \setcounter{figure}{0} \setcounter{table}{0} \makeatletter \def\p@subsection{\thesection.} \def\p@subsubsection{\thesection.\thesubsection.} \makeatother \renewcommand{\theHsection}{S.\arabic{section}} \renewcommand{\theequation}{S.\arabic{equation}} \renewcommand{\theHequation}{S.\arabic{equation}} \renewcommand{\thefigure}{S\arabic{figure}} \renewcommand{\theHfigure}{S\arabic{figure}} \renewcommand{\thetable}{S\arabic{table}} \renewcommand{\theHtable}{S\arabic{table}} \renewcommand{\theHtheorem}{S.\arabic{theorem}} \renewcommand{\theHproposition}{S.\arabic{proposition}} \renewcommand{\theHcorollary}{S.\arabic{corollary}} \renewcommand{\theHlemma}{S.\arabic{lemma}} \renewcommand*{\thepage}{S\arabic{page}}
\aliaslabel{eq:trace_preservation_constraint}{\ref{eq:pauli_cancellation_constraint}}%
\onecolumngrid
\enlargethispage{\baselineskip}
\begin{center}
  {\large \textbf{Supplementary Information for ``Measure Now, Mitigate Later:\\[0.25em]
  Virtual Error Cancellation for Logical Quantum Circuits''}}\\[0.5em]
  {\normalsize Oles Shtanko,$^{1}$ Saswata Roy,$^{1,2}$ Shashwat Kumar,$^{1}$ and Zlatko K.~Minev$^{1}$}\\[0.2em]
  {\small $^{1}$\textit{Google Quantum AI}\\
  $^{2}$\textit{Department of Physics, Cornell University, Ithaca, NY, 14853, USA}}
\end{center}
\vspace{-0.5em}
\begingroup
\small
\tableofcontents
\endgroup

\ifdefined\isnaturesupplement\else
\section{Methods}
\label{sec:methods_numerical}
\label{sec:appendix_methodology}
\label{sec:methods}
\label{sec:methods_mcmc_estimation}

This section summarizes the circuit models, decoding algorithms, syndrome-based noise-learning estimators, and statistical protocols used in Figs.~\ref{fig:dual_decoder}--\ref{fig:memory_1m}. Extended implementation details, subsampling schedules, and supplementary benchmarks appear in Supplementary Section~\ref{sec:supp_simulation_details}.

\ifdefined\isnaturemanuscript\else
The remainder of this section is organized into five parts covering circuit models, decoding algorithms, syndrome-based noise learning, uncertainty quantification, and multi-window composition.
Subsection~\ref{subsec:methods_circuits_noise} summarizes the surface-code circuit models, spatial disorder, and Pauli-frame logical readout.
Next, Subsection~\ref{subsec:methods_decoders_acceleration} describes the real-time minimum-weight perfect matching decoder, the offline ensemble decoder and failure-first Monte Carlo estimator, and exact matrix-product-state inference.
In Subsection~\ref{subsec:methods_noise_learning_benchmark}, we formulate Walsh--Hadamard moment inversion with temporal orbit pooling and signed likelihood-ratio calibration, while Subsection~\ref{subsec:methods_sampling_uncertainty} documents jackknife error estimation and estimator normalizations.
Finally, Subsection~\ref{subsec:methods_circuit_integration} details the sliding-window schedule, $11$-window boundary-splicing simulations, $400$-circuit learned-twin benchmarks, and million-round drifting memory stream in Figs.~\ref{fig:circuit_integration} and~\ref{fig:memory_1m} and Supplementary Figs.~\ref{fig:supp_ladder_learned_scaling} and~\ref{fig:supp_memory_1m_drift}.
Each subsection specifies the core parameters and statistical protocols supporting the numerical results.
\fi

\subsection{Surface-Code Circuits and Logical Error Readout}
\label{subsec:methods_circuits_noise}

We simulate planar rotated surface codes of odd distance $d \in \{3, 5, 7, 9, 11\}$ in \texttt{Stim}~\cite{gidney2021stim} under superconducting SI1000 circuit noise~\cite{gidney2021fault}.
The baseline channel assigns depolarizing errors at rate $p$ to two-qubit Clifford gates and $p/10$ to single-qubit gates and pre-round idling, alongside ancilla readout and reset bit flips at rates $5p$ and $2p$; we distinguish four named parameterizations: \emph{memory-SI1000} (omitting resonator idling during measurement and reset, Figs.~\ref{fig:dual_decoder} and~\ref{fig:vlpec_validation}), \emph{surgery-splice SI1000} (idling at $p/10$ in every gate, measurement, and reset layer and noiseless one-qubit gates, Fig.~\ref{fig:circuit_integration}b), \emph{drifting-memory SI1000} (resonator idling and one-qubit gates at $p_0/10$, Fig.~\ref{fig:memory_1m} and Supplementary Fig.~\ref{fig:supp_memory_1m_drift}), and \emph{full SI1000} (resonator idling at $2p$ and one-qubit gates at $p/10$, Supplementary Fig.~\ref{fig:supp_ladder_learned_scaling}).
Single-block memory simulations (Figs.~\ref{fig:dual_decoder} and~\ref{fig:vlpec_validation}) isolate $T = d$ noisy syndrome rounds with noiseless initial preparation and terminal data readout, and track $\bar{X}$ and $\bar{Z}$ jointly via an ideal Bell-pair reference qubit.
Whereas Fig.~\ref{fig:dual_decoder} uses spatially uniform error rates, the noise-learning and lattice-surgery benchmarks (Figs.~\ref{fig:vlpec_validation} and~\ref{fig:circuit_integration}b) perturb each operation rate by $25\%$ relative standard deviation and Pauli-rate anisotropy by $5\%$ before exact Bernoulli factorization.
The simulator records the true logical frame flip classically on each run solely to evaluate benchmark error metrics offline, never supplying it to the decoders or noise learner.
Extended circuit dimensions, single-round code-capacity models (Supplementary Figs.~\ref{fig:supp_dual_decoder_2d} and~\ref{fig:supp_vlpec_validation_2d}), and out-of-model noise considerations are detailed in Supplementary Section~\ref{sec:supp_simulation_details}.

\subsection{Real-Time and Offline Decoders}
\label{subsec:methods_decoders_acceleration}

Each simulation pairs a fast real-time graph decoder $\mathcal{D}_{\mathrm{rt}}$ with a higher-accuracy offline reference decoder $\mathcal{D}_{\mathrm{off}}$.
Real-time decoding $\mathcal{D}_{\mathrm{rt}}$ runs minimum-weight perfect matching (MWPM) in \texttt{PyMatching}~\cite{higgott2022pymatching,higgott2025blossom} on the decomposed graphlike error model with edge weights $w_e = \ln[(1 - p_e)/p_e]$.
In single-block memory (Figs.~\ref{fig:dual_decoder} and~\ref{fig:vlpec_validation}), both decoders use the homogeneous nominal model at rate $p$.
The offline ensemble $\mathcal{D}_{\mathrm{off}}$ compares serial minimum-sum belief propagation (BP) at 20 iterations and scaling 0.625 plus order-3 combination-sweep ordered statistics decoding (OSD)~\cite{roffe2020decoding}, $\mathcal{D}_{\mathrm{rt}}$, and two-pass correlated MWPM.
Recombining decomposed graphlike edge pairs into multi-detector $Y$ and hook hyperedges places all three recoveries on the unsplit hypergraph, where $\mathcal{D}_{\mathrm{off}}$ selects the syndrome-valid solution minimizing $\sum_i \ln[(1 - p_i)/p_i]$.
Every run is decoded by $\mathcal{D}_{\mathrm{rt}}$ (assigning $\bm{s} = \bm{0}$ to consensus), while $\mathcal{D}_{\mathrm{off}}$ evaluates all remaining real-time failures plus a random fraction $f$ of other non-trivial syndromes weighted by $1/f$ ($f \in [0.05, 1]$ in Fig.~\ref{fig:dual_decoder}; $f = 5\times 10^{-4}$ in Fig.~\ref{fig:vlpec_validation}c).
Whenever fewer than twenty $Y$-sector events occur, the weight solve sets $w_Y = 0$ to avoid Schur-complement noise amplification from under-sampled rare sectors (Supplementary Section~\ref{sec:supp_dual_decoder_protocol}), while keeping $Y$ errors in $p_L$ and $\|\bm{p}_{\mathrm{res}}\|_1$.
This failure-first scheme records all real-time logical failures as exact counts and estimates every four-sector table entry without bias (Supplementary Section~\ref{sec:supp_simulation_details}).

In the two-dimensional code-capacity model (Supplementary Figs.~\ref{fig:supp_dual_decoder_2d} and~\ref{fig:supp_vlpec_validation_2d}), $\mathcal{D}_{\mathrm{off}}$ performs exact maximum-likelihood inference under depolarizing noise ($p_j/3$ per non-identity Pauli error) via the coset partition functions
\begin{equation}
Z_L(s)=\sum_{E\in\mathcal C_L(s)}\prod_{j=1}^{d^2}
\left[\frac{p_j}{3(1-p_j)}\right]^{\mathbf1[E_j\ne I]},
\label{eq:methods_coset_partition}
\end{equation}
for $L\in\{I,X,Y,Z\}$ in the committed real-time frame.
We contract Eq.~\eqref{eq:methods_coset_partition} across all $d$ lattice columns via untruncated matrix-product-state sweeps~\cite{bravyi2014efficient} at bond dimension $\chi_{\max} = 2^d$ for $d \in \{3, 5, 7\}$, obtaining normalized coset posteriors $p_L(s) = Z_L(s)/\sum_{L'} Z_{L'}(s)$.
Extended contraction profiles, coset tie-breaking rules, and single-decoder comparators appear in Supplementary Section~\ref{sec:supp_simulation_details}.

Comparing real-time and offline decisions populates the four syndrome sectors $\mathcal{S}_Q$ ($Q \equiv \mathcal{D}_{\mathrm{rt}}(\bm{s})\mathcal{D}_{\mathrm{off}}(\bm{s}) \in \{I, X, Y, Z\}$) and defines the conditional transition matrix $p_{Q,Q'} = \Pr(Q' \mid \bm{s} \in \mathcal{S}_Q)$ (Fig.~\ref{fig:dual_decoder}d and Supplementary Fig.~\ref{fig:supp_dual_decoder_2d}c).
In Fig.~\ref{fig:dual_decoder}e and Supplementary Fig.~\ref{fig:supp_dual_decoder_2d}d, we sweep one disagreement weight $w_Q$ while adjusting $w_I(w_Q) = w_I^* - (\pi_Q/\pi_I)(w_Q - w_Q^*)$ to preserve unit mean weight at fixed $w_{Q'}^*$ ($Q' \notin \{I, Q\}$).
This constrained sweep defines the normalized residual ratio $R_Q(w_Q)$ and its reciprocal noise improvement factor $R_Q(w_Q)^{-1}$ (Fig.~\ref{fig:dual_decoder}e and Supplementary Fig.~\ref{fig:supp_dual_decoder_2d}d),
\begin{equation}
R_Q(w_Q) \equiv \frac{\big|(\bm{p}_{\mathrm{res}}(w_Q))_Q\big|}{\bar{p}_Q} = \frac{\pi_Q |V_{Q,Q}|}{\bar{p}_Q}\big|w_Q - w_Q^*\big|,
\label{eq:methods_residual_sweep_profile}
\end{equation}
with sector share $\pi_Q$, diagonal contrast $V_{Q,Q} = p_{Q,Q} - p_{I,Q}$, and real-time error rate $\bar{p}_Q$.
Each cusp at $w_Q = w_Q^*$ locates the zero-bias weight where channel $Q$ vanishes relative to postselection ($w_Q = 0$) and unmitigated decoding ($w_Q = 1$).

\subsection{Syndrome-Based Noise-Learning Benchmark}
\label{subsec:methods_noise_learning_benchmark}
\label{sec:learning_noise_model}
\label{sec:noise_learning_theory}
\label{sec:empirical_wls}

Calibrating mitigation weights without dedicated tomography circuits requires learning physical fault rates directly from computation syndromes (Fig.~\ref{fig:main_schematic}b).
A detector error model (DEM)~\cite{gidney2021stim} represents noise as $N_E$ independent Bernoulli mechanisms of probabilities $p_j \in [0, 1/2)$ and detector signatures $\bm{h}_j \in \{0, 1\}^{N_D}$, mapping faults $\bm{e} \in \{0, 1\}^{N_E}$ to syndromes $\bm{s} = \HH \bm{e} \pmod 2$.
For binary characters $\bm{\chi} \in \{0, 1\}^{N_D}$, Bernoulli factorization decomposes the Walsh--Hadamard syndrome fidelity~\cite{harper2020efficient,wagner2022pauli,wagner2023learning} into a linear system in the log-attenuations $\theta_j \equiv -\ln(1 - 2p_j) \ge 0$,
\begin{align}
\Fs(\bm{\chi}) &\equiv \EE_{\bm{s}}\big[(-1)^{\bm{\chi} \cdot \bm{s}}\big] = \EE_{\bm{e}}\big[(-1)^{\bm{\chi}^\top \HH \bm{e}}\big],
\label{eq:nl_syndrome_fidelity} \\
b_{\bm{\chi}} &\equiv -\ln \Fs(\bm{\chi}) = \sum_{j=1}^{N_E} A_{\bm{\chi}, j} \theta_j = (\Amat \bm{\theta})_{\bm{\chi}},
\label{eq:nl_linear_system}
\end{align}
with binary design elements $A_{\bm{\chi}, j} \equiv (\bm{\chi}^\top \bm{h}_j) \bmod 2$.
Across $N_{\mathrm{learn}}$ unlabeled syndrome runs $\{\bm{s}^{(t)}\}_{t=1}^{N_{\mathrm{learn}}}$, empirical parities estimate these moments via
\begin{equation}
\hat{\Fs}(\bm{\chi}) = \frac{1}{N_{\mathrm{learn}}} \sum_{t=1}^{N_{\mathrm{learn}}} (-1)^{\bm{\chi} \cdot \bm{s}^{(t)}},
\label{eq:methods_fidelity_estimator}
\end{equation}
converting $\mathcal{O}(p)$ stabilizer flips directly into physical rate estimates.

Circuit-level benchmarks (Fig.~\ref{fig:vlpec_validation} and Supplementary Fig.~\ref{fig:supp_ladder_learned_scaling}) invert Eq.~\eqref{eq:nl_linear_system} using bounded detector weight and temporal symmetry without non-negativity clipping.
Streaming 64-run bit-packed chunks accumulates $\hat{\Fs}(\bm{\chi})$ across all non-empty subsets of each support ($\|\bm{h}_j\|_0 \le 4$), and hierarchical M\"obius inclusion--exclusion inverts $-\ln\hat{\Fs}(\bm{\chi})$ from weight four down to one.
Shared-support faults are split proportionally to the nominal prior, and $\hat{\theta}_j$ is pooled linearly over bulk spacetime translation orbits ($4{,}368$ orbits across $6{,}351$ mechanisms at $d = 7$) or within-stage gate orbits in lattice surgery.
Mapping $\hat{p}_j = (1 - e^{-\hat{\theta}_j})/2$ without truncation eliminates positive clipping bias at finite sample counts, requiring $N_{\mathrm{learn}} = \mathcal{O}((|\mathcal{O}|\, p\, \varepsilon^2)^{-1})$ runs for orbit size $|\mathcal{O}|$ and relative error $\varepsilon$.
Weighted non-negative least-squares estimation for the two-dimensional code-capacity model (Supplementary Fig.~\ref{fig:supp_vlpec_validation_2d}) appears in Supplementary Sections~\ref{sec:supp_noise_learning_identifiability} and~\ref{sec:supp_simulation_details}.

Fig.~\ref{fig:vlpec_validation} tests this pipeline across $40$ disordered $d = 7$ memory chips per error rate $p \in \{1, 2, 3, 5\}\times 10^{-3}$ ($20$ deep chips at $p = 10^{-3}$ in Fig.~\ref{fig:vlpec_validation}c).
To hold relative Poisson precision comparable across $p$, the learning sample count scales inversely with $p$ according to
\begin{equation}
N_{\mathrm{learn}}(p) = N_0 \left(\frac{p_{\mathrm{base}}}{p}\right), \quad p_{\mathrm{base}} = 5\times 10^{-3},
\label{eq:methods_stratified_budget}
\end{equation}
for base budgets $N_0 \in \{10^4, 10^5, 10^6, 10^7\}$.
On each chip ($6\times 10^7$ to $1.2\times 10^{10}$ runs at $f = 5\times 10^{-4}$), physical fault trajectories $\bm{e} \in \{0, 1\}^{N_E}$ are reweighted to each learned model via the signed common-random-numbers (CRN) ratio
\begin{equation}
\rho(\bm{e}) = \prod_{j=1}^{N_E} \left(\frac{\hat{p}_j}{p_j}\right)^{e_j} \left(\frac{1 - \hat{p}_j}{1 - p_j}\right)^{1 - e_j}.
\label{eq:methods_signed_crn_ratio}
\end{equation}
Targeted rare-$Y$ importance sampling supplements this pass solely for $\bar{p}_Y \sim 10^{-9}$ at $p = 10^{-3}$ in Fig.~\ref{fig:vlpec_validation}b.
At $p = 10^{-3}$ and $N_0 = 10^7$, median relative errors reach $0.76\%$ on mechanism rates and $0.10\%$ on logical rates (Fig.~\ref{fig:vlpec_validation}a,b).

To separate calibration residuals from evaluation noise in Fig.~\ref{fig:vlpec_validation}c, eight groups per chip form even- and odd-half residual vectors $\bm{v}_A$ and $\bm{v}_B$ defining the half-split signal-to-noise ratio (SNR)
\begin{equation}
\mathrm{SNR}_{\mathrm{chip}} = \frac{\sqrt{\max\bigl(\overline{\bm{v}_A \cdot \bm{v}_B},\, 0\bigr)}}{\sqrt{\overline{\|\bm{v}_A - \bm{v}_B\|_2^2}/4}}.
\label{eq:methods_half_split_snr}
\end{equation}
Fifteen of the sixteen $(p, N_0)$ points satisfy $\mathrm{SNR}_{\mathrm{chip}} \ge 2$ and median relative effective sample size (ESS) of at least $0.25$ under CRN reweighting.
At $(p = 5\times 10^{-3}, N_0 = 10^4)$, where signed rate noise lowers the CRN relative ESS to $0.008$, we instead sample $6\times 10^7$ learned-twin trajectories per chip from $q_j = |\hat{p}_j|$, giving relative ESS $0.24$ and $\mathrm{SNR}_{\mathrm{chip}} = 4.9$ (Supplementary Section~\ref{sec:supp_simulation_details}).
All sixteen points thus pass $\mathrm{SNR}_{\mathrm{chip}} \ge 2$ and are reported as resolved chip means $\langle\|\bm{p}_{\mathrm{res}}\|_1\rangle$; reweighting shared trajectories isolates the learning-limited residual from classical twin sampling noise (unweighted twin screening requires $\mathcal{O}((\alpha p_L \varepsilon_{\max}^2)^{-1})$ runs, exceeding $N_{\mathrm{learn}}$ at the deepest suppression points, whereas Supplementary Fig.~\ref{fig:supp_ladder_learned_scaling} evaluates mutually independent learning, twin, and test streams).

\subsection{Statistical Estimators and Uncertainty Analysis}
\label{subsec:methods_sampling_uncertainty}

Statistical uncertainties are quantified across $B$ independent groups ($B = 30$ in Fig.~\ref{fig:dual_decoder}) via the delete-one-group jackknife estimator
\begin{equation}
\begin{aligned}
\hat{\theta}_{\mathrm{jack}} &= B\hat{\theta} - (B - 1)\hat{\theta}_{(\cdot)}, \\
\hat{\sigma}_{\mathrm{jack}}^2 &= \frac{B - 1}{B}\sum_{b=1}^B \big(\hat{\theta}_{(b)} - \hat{\theta}_{(\cdot)}\big)^2,
\end{aligned}
\label{eq:methods_jackknife_estimator}
\end{equation}
where $\hat{\theta}_{(b)}$ omits group $b$ and $\hat{\theta}_{(\cdot)}$ averages the $B$ leave-one-out replicates.
Fig.~\ref{fig:dual_decoder} plots the pooled estimate $\hat{\theta}$ with error bars $\hat{\sigma}_{\mathrm{jack}}$ ($|\hat{\theta} - \hat{\theta}_{\mathrm{jack}}| < 0.07\hat{\sigma}_{\mathrm{jack}}$ everywhere), whereas Figs.~\ref{fig:vlpec_validation}c and~\ref{fig:circuit_integration}b report standard errors across chips.
Guide curves show least-squares log-polynomial or power-law fits in Figs.~\ref{fig:dual_decoder}b,c, \ref{fig:vlpec_validation}c, and~\ref{fig:circuit_integration}b, and the rational threshold fit $(\eta_0 + a x^b)/(1 - c x^{(d+1)/2})$ with $x = p/p_{\mathrm{ref}}$ in Fig.~\ref{fig:dual_decoder}f (Supplementary Section~\ref{sec:supp_simulation_details}).

All benchmarks separate syndrome learning, weight solving, and circuit evaluation across distinct random streams while matching estimator scaling to the target metric (Supplementary Section~\ref{subsec:app_proof_distance_suppression}).
Population normalization $\bm{p}_{\mathrm{res}} = \bm{r}_{\mathrm{raw}}/b_{\mathrm{raw}}$ (Supplementary Section~\ref{subsec:app_residual_error_matrix}) isolates the traceless channel residual reported in Fig.~\ref{fig:vlpec_validation}c, whereas raw unnormalized estimation retains the $\mathcal{O}(\varepsilon_{\max} p_L)$ mean-weight shift (removable via empirical sector or sample means; Proposition~\ref{prop:sector_sample_means} in Supplementary Section~\ref{subsec:app_residual_error_matrix}).
Fig.~\ref{fig:circuit_integration}b scores the cross-fitted eleven-window product estimator in the nearest-window approximation across disjoint shard folds, while Supplementary Fig.~\ref{fig:supp_ladder_learned_scaling} evaluates the raw two-sector estimator across three mutually independent shot streams.

\subsection{Sliding-Window Schedule and Full-Circuit Bias Simulation}
\label{subsec:methods_circuit_integration}

In multi-layer circuits, real-time decoding $\mathcal{D}_{\mathrm{rt}}$ advances over overlapping windows of $R = d$ commit rounds and $b = 2d - 3$ buffer rounds (Fig.~\ref{fig:circuit_integration}a and Supplementary Section~\ref{sec:fc-protocol}).
Placing this buffer region before the conditional $S$ gate in decomposed non-Clifford gate-teleportation gadgets reduces pre-measurement faults to data-qubit Pauli errors paired with decoded outcome flips, enabling pre-feedforward channel inversion (Supplementary Section~\ref{sec:fc-seam}).

Fig.~\ref{fig:circuit_integration}b estimates the normalized bias $\epsilon$ (Eq.~\eqref{eq:full_circuit_relative_bias}) in the nearest-window approximation, which keeps the ten covariances between the weights of neighboring windows (divided by $2 V p_L$) and drops correlations between windows two or more apart and connected three-window terms (the chip-averaged signed value is negative at every point and is plotted in magnitude), on a two-patch $Z_A Z_B$ lattice-surgery ladder with $K_{\mathrm{surg}} = 5$ merges across $d \in \{3, 5, 7, 9\}$ ($L = 11$ windows, $11d + 1$ rounds).
Each $(d, p)$ point averages $80$ disordered SI1000 instances with idling at rate $p/10$ in every gate, measurement, and reset layer and noiseless single-qubit gates, converted into exact independent Bernoulli mechanisms.
Real-time $\mathcal{D}_{\mathrm{rt}}$ runs sliding-window MWPM on the $Z$-detector graph, whereas $\mathcal{D}_{\mathrm{off}}$ runs global two-pass correlated MWPM on the unsplit $X+Z$ model, retaining odd-parity discrepancy clusters whose one- and three-hop augmented weights do not exceed baseline; every other correlated-matching correction reverts to the uncorrelated $Z$-MWPM correction, which forgoes part of the accuracy gain of correlated matching.
Each retained cluster is attributed to the window of its latest active detector, suppressing weight-3 multi-window false activations to $10^{-6}$ at $d = 3$.
Per-window weights are calibrated by symmetric two-fold cross-fitting across disjoint shard folds, and all ten adjacent window boundaries are resolved via three-region Hoeffding splicing with five-stratum boundary sampling and block-restriction control variates (weighted power-law fits yield $\chi^2/\mathrm{dof} = 47.7/4$ at $d = 3$ from high-$p$ curvature and $\chi^2/\mathrm{dof} \le 1.11$ for $d \ge 5$; verified against direct eleven-window product Monte Carlo at $d = 3$, Supplementary Section~\ref{sec:fc-scaling-simulation}). The fitted exponents are empirical for this decoder pair.

Supplementary Fig.~\ref{fig:supp_ladder_learned_scaling} tests end-to-end learned-twin mitigation across $400$ round-stationary disordered SI1000 two-patch ladders ($d \in \{3, 5, 7, 9, 11\}$ across 1 to 16 merges and 10 to 176 rounds, with resonator idling at rate $2p$) tuned near $\langle O\rangle_{\mathrm{unmit}} \approx e^{-1}$.
Each circuit learns its DEM via within-stage orbit-pooled M\"obius inversion from $N_{\mathrm{learn}} = 999{,}960$ syndrome runs, calibrates global two-sector weights on $T_{\mathrm{MC}} = 10^7$ correlated-MWPM twin runs, and scores $N_{\mathrm{hw}} = 99{,}960$ independent test runs.
Together, exact-model boundary splicing and syndrome-calibrated ladder scaling support both cross-window bias suppression over the simulated noise range and practical multi-window accuracy (Supplementary Section~\ref{subsec:fc-scaling-panel-bc}).

Fig.~\ref{fig:memory_1m} and Supplementary Fig.~\ref{fig:supp_memory_1m_drift} extend continuous sliding-window mitigation to $1{,}008{,}000$ stabilizer rounds on a $d = 7$ rotated surface-code memory ($p_0 = 8.85\times 10^{-4}$ under SI1000 noise with resonator idling at $p_0/10$, $15\%$ static spatial rate spread, and noisy state preparation and readout).
Across $1{,}000$ steps of $1{,}008$ rounds, the global error scale rises non-monotonically from $0.925$ to $1.044$ (peaking at $1.056$), combining a linear $0.93 \to 1.07$ ramp, zero-mean Gaussian-smoothed Ornstein--Uhlenbeck fluctuations ($2.5\%$ nominal standard deviation, $1.2\%$ realized root-mean-square, $250$-step correlation scale), and a $\pm 3\%$ spatial gradient tilt.
Real-time two-pass correlated MWPM and unconditional offline search-based hypergraph decoding (Tesseract, beam width $10$) stream over $36{,}000$ committed windows ($R = 28$, $b = 14$) using a shared fixed $t = 0$ decoder model learned from a $5\times 10^5$-shot ($70$-round) pre-run syndrome burn-in, with connected discrepancy clusters co-attributed to their terminal committed window.
Ten chronological slice models ($100{,}800$ rounds each, reconstructed via two-fold cross-fitted orbit-pooled M\"obius inversion directly from $2.5\times 10^6$ streamed fourteen-round syndrome blocks per slice) calibrate unmodified two-sector product weights across $4.5\times 10^7$ localized eight-window importance-sampled digital-twin trajectories, with total $1\text{ s.e.m.}$ combining hardware-shot, syndrome-learning, and digital-twin leave-one-group-out jackknife variances across $2.5\times 10^4$ evaluation trajectories (Supplementary Section~\ref{subsec:fc-million-round-drift}).
\fi

\section{Superoperator Error Cancellation}
\label{sec:supp_operator_channel_inversion}

This section derives the error-channel inversion constraints of virtual error cancellation (VEC) in Eq.~\eqref{eq:pauli_cancellation_constraint} and the variance relations in Eq.~\eqref{eq:vpec_sampling_overhead_def} of the main text. In Subsection~\ref{subsec:supp_superoperator_inversion}, we prove that unbiased reconstruction across all input states and target operators is equivalent to canceling the conditional Pauli error channels in syndrome expectation. Eq.~\eqref{eq:pauli_cancellation_constraint} and the criterion for when negative weights are necessary follow directly from this identity. Subsection~\ref{subsec:supp_estimator_variance} then derives the multi-shot estimator variance bound and sampling overhead decomposition in Eq.~\eqref{eq:vpec_sampling_overhead_def}, establishing the constrained optimization framework analyzed in the subsequent sections.

\subsection{Error-Channel Cancellation Constraints}
\label{subsec:appendix_operator_channel_inversion}
\label{subsec:supp_superoperator_inversion}

We first formulate the syndrome-conditioned logical Pauli channel for an $n$-qubit encoded block and establish the notation linking integer channel indices to Pauli operator labels.
Consider an encoded block of $n$ logical qubits with Hilbert space $\cH_L \simeq (\CC^2)^{\otimes n}$. Its $n$-qubit logical Pauli group $\cP_n = \{I, X, Y, Z\}^{\otimes n}$ has cardinality $4^n$, and an active subset $\mathcal{Q} = \{Q_1, \dots, Q_M\} \subseteq \cP_n^* \equiv \cP_n \setminus \{I^{\otimes n}\}$ contains $M \le 4^n - 1$ non-trivial logical Pauli operators ($M = 4^n - 1$ when all non-trivial Pauli channels are active). In Eq.~\eqref{eq:conditional_logical_channel} and Eq.~\eqref{eq:pauli_cancellation_constraint} of the main text, these $M$ active error channels are indexed by integers $\alpha \in \{1, \dots, M\}$. In our multi-qubit and coarse-grained sector analyses, by contrast, they are indexed directly by their Pauli labels $Q \in \mathcal{Q}$. Both conventions label the same $M$-element basis. Specifically, a bijection $\alpha \leftrightarrow Q_\alpha \in \mathcal{Q}$ identifies the conjugation map $\cE_\alpha(\rho) \equiv P_\alpha \rho P_\alpha^\dagger$ with $\mathcal{X}_{Q_\alpha}(\rho) \equiv Q_\alpha \rho Q_\alpha^\dagger$ and the posterior probability $p_\alpha(\bm{s})$ with $p_{Q_\alpha}(\bm{s})$. For each syndrome record $\bm{s} \in \Omega \equiv \{0, 1\}^{N_D}$ occurring with probability $P(\bm{s})$, the residual logical channel after real-time decoding takes the Pauli form of Eq.~\eqref{eq:conditional_logical_channel},
\begin{equation}
\Lambda_{\bm{s}}(\rho) = \left(1 - \sum_{\alpha=1}^M p_\alpha(\bm{s})\right) \cI(\rho) + \sum_{\alpha=1}^M p_\alpha(\bm{s}) \cE_\alpha(\rho),
\label{eq:supp_conditional_pauli_channel}
\end{equation}
where $\cI(\rho) \equiv \rho$ is the identity superoperator on the operator space $\mathcal{B}(\cH_L)$. Here $P(\bm{s})$ is input-independent and the retained quantum instrument is $P(\bm{s})\Lambda_{\bm{s}}$, which is an explicit hypothesis of the Pauli-block argument; adaptive measured gadgets require the joint-instrument treatment in Section~\ref{sec:fc-seam}. The conditional probabilities $\bm{p}(\bm{s}) \equiv (p_1(\bm{s}), \dots, p_M(\bm{s}))^\top$ reside in the physical simplex defined by $p_\alpha(\bm{s}) \ge 0$ and $\sum_{\alpha=1}^M p_\alpha(\bm{s}) \le 1$. Their syndrome average defines the unmitigated logical error vector $\bar{\bm{p}} \equiv \EE_{\bm{s}}[\bm{p}(\bm{s})]$ with components $\bar{p}_\alpha \equiv \EE_{\bm{s}}[p_\alpha(\bm{s})]$, completely specifying the block's residual Pauli noise conditioned on $\bm{s}$.

Requiring the syndrome-weighted estimator to be unbiased across all input states and observables is equivalent to inverting this conditional channel in syndrome expectation.
For an arbitrary input state $\rho \in \mathcal{B}(\cH_L)$ and Hermitian logical operator $O \in \mathcal{B}(\cH_L)$, the syndrome-conditioned expectation value is $\langle O \rangle_{\bm{s}} = \Tr[O \Lambda_{\bm{s}}(\rho)]$. The corresponding ideal noiseless expectation value is $\langle O \rangle_{\mathrm{ideal}} = \Tr[O \rho] = \Tr[O \cI(\rho)]$. Linearity of the trace and syndrome expectation implies that the weighted estimator in Eq.~\eqref{eq:estimator} satisfies
\begin{equation}
\EE_{\bm{s}}\big[w(\bm{s})\langle O \rangle_{\bm{s}}\big] = \Tr\!\Big[O \, \EE_{\bm{s}}\big[w(\bm{s})\Lambda_{\bm{s}}\big](\rho)\Big].
\label{eq:supp_reweighted_trace_expectation}
\end{equation}
Enforcing the unbiasedness condition $\EE_{\bm{s}}[w(\bm{s})\langle O \rangle_{\bm{s}}] = \langle O \rangle_{\mathrm{ideal}}$ for the estimator in Eq.~\eqref{eq:estimator} across every state $\rho$ and Hermitian observable $O$ amounts to requiring the Hilbert--Schmidt pairing $\Tr[O\,\Delta\Lambda(\rho)] = 0$ to vanish for the difference map $\Delta\Lambda \equiv \EE_{\bm{s}}[w(\bm{s})\Lambda_{\bm{s}}] - \cI$. Because density operators and Hermitian operators each span $\mathcal{B}(\cH_L)$, non-degeneracy of the inner product $\langle A, B \rangle_{\mathrm{HS}} \equiv \Tr(A^\dagger B)$ forces $\Delta\Lambda = 0$ on $\mathcal{B}(\cH_L)$. Universal reconstruction across all states and operators therefore holds if and only if the weighted channel obeys the superoperator identity
\begin{equation}
\EE_{\bm{s}}\big[w(\bm{s})\Lambda_{\bm{s}}\big] = \cI.
\label{eq:supp_superoperator_identity_general}
\end{equation}
Substituting Eq.~\eqref{eq:supp_conditional_pauli_channel} into Eq.~\eqref{eq:supp_superoperator_identity_general} and regrouping terms around $\cI$ expands this condition in the Pauli basis,
\begin{align}
\EE_{\bm{s}}\big[w(\bm{s})\Lambda_{\bm{s}}\big] &= \EE_{\bm{s}}\!\left[w(\bm{s})\left(1 - \sum_{\alpha=1}^M p_\alpha(\bm{s})\right)\right]\cI + \sum_{\alpha=1}^M \EE_{\bm{s}}\big[w(\bm{s})p_\alpha(\bm{s})\big]\cE_\alpha \nonumber \\
&= \EE_{\bm{s}}\big[w(\bm{s})\big]\cI + \sum_{\alpha=1}^M \EE_{\bm{s}}\big[w(\bm{s})p_\alpha(\bm{s})\big](\cE_\alpha - \cI) = \cI.
\label{eq:channel_inversion_expanded}
\end{align}
Channel inversion thus reduces to balancing the identity map $\cI$ against the shifted Pauli error maps $\cE_\alpha - \cI$.

To show that Eq.~\eqref{eq:channel_inversion_expanded} enforces the dual constraints in Eq.~\eqref{eq:pauli_cancellation_constraint}, we prove the linear independence of the $4^n$ Pauli maps $\{\cE_\alpha\}_{\alpha=0}^M$, where $\cE_0 \equiv \cI$ corresponds to $P_0 \equiv I^{\otimes n}$. Any two $n$-qubit Pauli operators $P_\alpha$ and $P_\beta$ obey the commutation relation $P_\alpha P_\beta P_\alpha^\dagger = (-1)^{\langle \alpha, \beta \rangle} P_\beta$. Here $\langle \alpha, \beta \rangle \in \{0, 1\}$ is the binary symplectic inner product, which equals $0$ when $P_\alpha$ and $P_\beta$ commute and $1$ when they anticommute. Each operator $P_\beta$ is thus an eigenoperator of $\cE_\alpha$ with eigenvalue $H_{\beta\alpha} \equiv (-1)^{\langle \alpha, \beta \rangle}$. Suppose a linear combination vanishes on $\mathcal{B}(\cH_L)$, so that $\sum_{\alpha=0}^M c_\alpha \cE_\alpha = 0$ for scalars $c_\alpha \in \CC$. Evaluating this map on each basis element $P_\beta \in \cP_n$ yields
\begin{equation}
0 = \sum_{\alpha=0}^M c_\alpha \cE_\alpha(P_\beta) = \left(\sum_{\alpha=0}^M H_{\beta\alpha} c_\alpha\right) P_\beta.
\label{eq:supp_walsh_hadamard_pauli}
\end{equation}
Since $P_\beta \neq 0$, the coefficient sum $\sum_{\alpha=0}^M H_{\beta\alpha} c_\alpha$ must vanish for every $P_\beta \in \cP_n$. Because the inner product $\langle \alpha, \beta \rangle$ is non-degenerate over $\FF^{2n}$, the columns of the Walsh--Hadamard submatrix $\HH_{\mathrm{WH}} \in \{\pm 1\}^{4^n \times (M+1)}$ with entries $H_{\beta\alpha} = (-1)^{\langle \alpha, \beta \rangle}$ ($P_\beta \in \cP_n$, $\alpha \in \{0, \dots, M\}$) satisfy the orthogonality relation
\begin{equation}
\big(\HH_{\mathrm{WH}}^\top \HH_{\mathrm{WH}}\big)_{\alpha\gamma} = \sum_{P_\beta \in \cP_n} (-1)^{\langle \alpha \oplus \gamma, \beta \rangle} = 4^n \delta_{\alpha\gamma},
\label{eq:supp_walsh_hadamard_orthogonality}
\end{equation}
where $\alpha \oplus \gamma$ adds the binary vectors of $P_\alpha$ and $P_\gamma$. Eq.~\eqref{eq:supp_walsh_hadamard_orthogonality} shows that the maps $\{\cE_\alpha\}_{\alpha=0}^M$ are mutually orthogonal in the Pauli transfer matrix (PTM) representation, with left inverse $4^{-n}\HH_{\mathrm{WH}}^\top$. Multiplying Eq.~\eqref{eq:supp_walsh_hadamard_pauli} by $4^{-n}\HH_{\mathrm{WH}}^\top$ forces $c_\alpha = 0$ for all $\alpha \in \{0, \dots, M\}$. The shifted maps $\{\cI, \cE_1 - \cI, \dots, \cE_M - \cI\}$ are therefore linearly independent on $\mathcal{B}(\cH_L)$. Equating coefficients in Eq.~\eqref{eq:channel_inversion_expanded} with $\cI = 1\cdot\cI + \sum_{\alpha=1}^M 0\cdot(\cE_\alpha - \cI)$ then yields the $M + 1$ conditions of Eq.~\eqref{eq:pauli_cancellation_constraint},
\begin{align}
\EE_{\bm{s}}\big[w(\bm{s})\big] &= 1, \label{eq:supp_trace_preservation_derived} \\
\EE_{\bm{s}}\big[w(\bm{s})p_\alpha(\bm{s})\big] &= 0 \quad \text{for all } \alpha \in \{1, \dots, M\}.
\label{eq:supp_error_cancellation_derived}
\end{align}
Trace preservation and componentwise error cancellation are therefore both necessary and sufficient for universal Pauli channel inversion.

These two conditions determine the correlation structure and sign of $w(\bm{s})$. Combining Eq.~\eqref{eq:supp_trace_preservation_derived} and Eq.~\eqref{eq:supp_error_cancellation_derived} fixes the cross-covariance between any unbiased filter $w(\bm{s})$ and each error rate $p_\alpha(\bm{s})$ at a universal negative value,
\begin{equation}
\Cov_{\bm{s}}\big(w(\bm{s}), p_\alpha(\bm{s})\big) = \EE_{\bm{s}}\big[w(\bm{s})p_\alpha(\bm{s})\big] - \EE_{\bm{s}}\big[w(\bm{s})\big]\EE_{\bm{s}}\big[p_\alpha(\bm{s})\big] = -\bar{p}_\alpha.
\label{eq:supp_invariant_cross_cov}
\end{equation}
In vector notation, this relation reads $\Cov_{\bm{s}}(w(\bm{s}), \bm{p}(\bm{s})) = -\bar{\bm{p}}$. Decomposing the filter into unsigned parts $w_\pm(\bm{s}) \equiv \max(\pm w(\bm{s}), 0) \ge 0$ with $w(\bm{s}) = w_+(\bm{s}) - w_-(\bm{s})$ rewrites Eq.~\eqref{eq:supp_error_cancellation_derived} as
\begin{equation}
\EE_{\bm{s}}\big[w_+(\bm{s})p_\alpha(\bm{s})\big] = \EE_{\bm{s}}\big[w_-(\bm{s})p_\alpha(\bm{s})\big].
\label{eq:supp_general_sign_problem}
\end{equation}
A non-negative feasible filter exists if and only if the error-free syndrome event $A = \{\bm{s} : \sum_{\alpha=1}^M p_\alpha(\bm{s}) = 0\}$ has positive probability $P(A) > 0$. Specifically, $w(\bm{s}) \ge 0$ and $\EE_{\bm{s}}[w(\bm{s})\sum_{\alpha=1}^M p_\alpha(\bm{s})] = 0$ force $w(\bm{s}) = 0$ outside $A$, while $\EE_{\bm{s}}[w(\bm{s})] = 1$ requires $P(A) > 0$. When $P(A) > 0$, the postselection filter $w(\bm{s}) = \mathbf{1}_A(\bm{s})/P(A)$ is feasible and achieves the minimal second-moment cost $1/P(A)$ by the Cauchy--Schwarz inequality. Whenever every syndrome retains a non-zero residual error probability ($P(A) = 0$), any unbiased filter $w(\bm{s})$ must therefore take negative values on a subset of syndromes to cancel the positive error contributions.

\subsection{Estimator Variance and Sampling Overhead}
\label{subsec:supp_estimator_variance}

We now derive the multi-shot estimator variance bound and sampling overhead in Eq.~\eqref{eq:vpec_sampling_overhead_def} by decomposing the single-shot variance into within-syndrome readout noise and between-syndrome weight fluctuations. Consider an ensemble of $N$ measurement shots indexed by $i \in \{1, \dots, N\}$, where each run records a syndrome $\bm{s}_i \in \Omega$ and a binary readout outcome $m_i \in \{+1, -1\}$ for the target Pauli observable $O$ satisfying $O^2 = I^{\otimes n}$. At a fixed syndrome $\bm{s}_i$, averaging $m_i$ over the Born-rule readout distribution gives $\EE[m_i \mid \bm{s}_i] = \Tr[O \Lambda_{\bm{s}_i}(\rho)] = \langle O \rangle_{\bm{s}_i}$. Because $m_i \in \{+1, -1\}$ implies $m_i^2 = 1$ on every shot, the conditional second moment is $\EE[m_i^2 \mid \bm{s}_i] = 1$. Multiplying by the weight $w(\bm{s}_i)$ gives the fixed-syndrome first moment $\EE[w(\bm{s}_i)m_i \mid \bm{s}_i] = w(\bm{s}_i)\langle O \rangle_{\bm{s}_i}$ and second moment $\EE[(w(\bm{s}_i)m_i)^2 \mid \bm{s}_i] = w(\bm{s}_i)^2$. Taking the expectation over $\bm{s}_i \sim P(\bm{s})$ and invoking Eq.~\eqref{eq:supp_reweighted_trace_expectation} and Eq.~\eqref{eq:supp_superoperator_identity_general} yields the unconditional moments
\begin{align}
\EE\big[w(\bm{s}_i) m_i\big] &= \EE_{\bm{s}_i}\Big[w(\bm{s}_i)\,\EE[m_i \mid \bm{s}_i]\Big] = \EE_{\bm{s}}\big[w(\bm{s})\langle O \rangle_{\bm{s}}\big] = \langle O \rangle_{\mathrm{ideal}}, \nonumber \\
\EE\Big[\big(w(\bm{s}_i) m_i\big)^2\Big] &= \EE_{\bm{s}_i}\Big[w(\bm{s}_i)^2\,\EE\big[m_i^2 \mid \bm{s}_i\big]\Big] = \EE_{\bm{s}}\big[w(\bm{s})^2\big].
\label{eq:supp_single_shot_moments}
\end{align}
Applying the law of total variance over the joint variable $(m_i, \bm{s}_i)$ splits the variance of a single weighted outcome into readout projection noise and syndrome-weight spread,
\begin{align}
\Var\big[w(\bm{s}_i) m_i\big] &= \EE_{\bm{s}_i}\Big[\Var\big[w(\bm{s}_i) m_i \mid \bm{s}_i\big]\Big] + \Var_{\bm{s}_i}\Big(\EE\big[w(\bm{s}_i) m_i \mid \bm{s}_i\big]\Big) \nonumber \\
&= \EE_{\bm{s}}\Big[w(\bm{s})^2\big(1 - \langle O \rangle_{\bm{s}}^2\big)\Big] + \Big(\EE_{\bm{s}}\big[w(\bm{s})^2\langle O \rangle_{\bm{s}}^2\big] - \langle O \rangle_{\mathrm{ideal}}^2\Big) \nonumber \\
&= \EE_{\bm{s}}\big[w(\bm{s})^2\big] - \langle O \rangle_{\mathrm{ideal}}^2.
\label{eq:supp_total_variance_single_shot}
\end{align}
The intermediate expectation $\EE_{\bm{s}}[w(\bm{s})^2\langle O \rangle_{\bm{s}}^2]$ cancels identically between the two terms, leaving the per-shot variance governed solely by the second moment of $w(\bm{s})$ and the ideal expectation value.

Averaging over $N$ independent shots translates this single-shot variance directly into the multiplicative sampling overhead $C(w)$.
For the $N$-shot sample-mean estimator $\hat{\langle O \rangle}_{\mathrm{VEC}} = N^{-1}\sum_{i=1}^N w(\bm{s}_i) m_i$, mutual independence across shots implies $\Var[\hat{\langle O \rangle}_{\mathrm{VEC}}] = N^{-2}\sum_{i=1}^N \Var[w(\bm{s}_i)m_i]$. Substituting Eq.~\eqref{eq:supp_total_variance_single_shot} and noting that $\langle O \rangle_{\mathrm{ideal}}^2 \ge 0$ for any real expectation value $\langle O \rangle_{\mathrm{ideal}} \in [-1, 1]$ establishes the exact $N$-shot variance identity and state-independent upper bound,
\begin{equation}
\Var\big[\hat{\langle O \rangle}_{\mathrm{VEC}}\big] = \frac{1}{N}\Big(\EE_{\bm{s}}\big[w(\bm{s})^2\big] - \langle O \rangle_{\mathrm{ideal}}^2\Big) \le \frac{1}{N}\EE_{\bm{s}}\big[w(\bm{s})^2\big].
\label{eq:supp_vec_estimator_variance}
\end{equation}
Equality holds in the weak-signal regime $\langle O \rangle_{\mathrm{ideal}} = 0$. Expanding the second moment of $w(\bm{s})$ around its syndrome mean and substituting the trace-preservation condition $\EE_{\bm{s}}[w(\bm{s})] = 1$ from Eq.~\eqref{eq:supp_trace_preservation_derived} gives the sampling overhead decomposition of Eq.~\eqref{eq:vpec_sampling_overhead_def},
\begin{equation}
C(w) \equiv \EE_{\bm{s}}\big[w(\bm{s})^2\big] = \big(\EE_{\bm{s}}[w(\bm{s})]\big)^2 + \Var_{\bm{s}}\big[w(\bm{s})\big] = 1 + \Var_{\bm{s}}\big[w(\bm{s})\big].
\label{eq:supp_sampling_overhead_decomposition}
\end{equation}
For an exactly unbiased fixed-weight estimator, $N_{\mathrm{ideal}} = \lceil 1/\epsilon^2 \rceil$ is the ideal worst-case root-mean-square (RMS) shot budget, and a sufficient mitigated shot budget is
\begin{equation}
N_{\mathrm{VEC}} = \lceil C(w)/\epsilon^2 \rceil.
\label{eq:supp_shot_budget_scaling}
\end{equation}
This bound is saturated in variance when $\langle O \rangle_{\mathrm{ideal}} = 0$ (up to integer rounding) and provides a universal variance ceiling for general states, while any residual calibration bias contributes separately to the mean-square error.

The error-channel cancellation constraints of Eqs.~\eqref{eq:supp_trace_preservation_derived}--\eqref{eq:supp_error_cancellation_derived} and the variance decomposition of Eqs.~\eqref{eq:supp_sampling_overhead_decomposition}--\eqref{eq:supp_shot_budget_scaling} reduce superoperator reconstruction to a constrained variance-minimization problem over the syndrome probability space. Because both the linear conditions and the variance upper bound in Eq.~\eqref{eq:supp_vec_estimator_variance} depend only on $\bm{p}(\bm{s})$ and not on the input state $\rho$ or observable $O$, a single filter $w(\bm{s})$ cancels all $M$ logical Pauli error channels across arbitrary observables. Its sampling cost is governed entirely by $\Var_{\bm{s}}[w(\bm{s})]$. When the constraint set is nonempty, the feasible filters form a closed affine subspace; the syndrome count alone does not establish feasibility. Within this set, Supplementary Section~\ref{sec:appendix_proof_variational_cs_bound} determines the unique variance-minimizing filter $w^*(\bm{s})$ and its minimum sampling overhead $C_{\mathrm{syn}}^*$.

\section{Variational Theory of Error Cancellation and Optimal Mitigation Weights}
\label{sec:appendix_proof_variational_cs_bound}

This section establishes the variational theory of multi-process virtual error cancellation (VEC) in Eqs.~\eqref{eq:optimal_weights}--\eqref{eq:optimal_sampling_overhead} of the main text. Subsection~\ref{subsec:variational_preliminaries} formulates the space of unbiased filters over the syndrome probability space and proves the exact decoder feasibility condition (Lemma~\ref{lem:feasibility_condition_formal}). In Subsection~\ref{subsec:proof_variational_cs_bound}, we state and prove the Generalized Variational Matrix Bound (Theorem~\ref{thm:variational_cs_bound}) by two complementary methods: a matrix Cauchy--Schwarz completion of squares and a Karush--Kuhn--Tucker (KKT) functional minimization with Schur complement inversion. Subsection~\ref{subsec:bhatia_davis_ceiling} derives the single-decoder Bhatia--Davis variance ceiling and postselection limit (Corollary~\ref{cor:bhatia_davis_ceiling}), while Subsection~\ref{subsec:salem_surface_benchmark} proves the surface-code conditional Pauli probabilistic error cancellation (PEC) lower bound (Theorem~\ref{thm:salem_surface_gap}). Finally, we derive exact calibration identities and the scalar remainder in Subsection~\ref{subsec:app_residual_error_matrix} (Proposition~\ref{thm:taylor_expansion_full}) and prove exact scale invariance under global multiplicative miscalibration in Subsection~\ref{subsec:scale_invariance_derivations} (Proposition~\ref{thm:scale_invariance_theorem}).

\subsection{Mathematical Setting, Feasible Filter Geometry, and Feasibility Lemma}
\label{subsec:variational_preliminaries}

To formulate the search for the minimum-variance filter rigorously, we embed the syndrome-conditioned error rates and scalar filters in the Hilbert space of square-integrable functions over the syndrome space.
Let $(\Omega, \mathcal{F}, P)$ be a probability space, where $\bm{s} \in \Omega$ records the detector syndrome history. Real-time decoding applies a recovery frame whose residual $M$-channel Pauli error probabilities conditioned on $\bm{s}$ form the random vector $\bm{p}(\bm{s}) = (p_1(\bm{s}), \dots, p_M(\bm{s}))^\top \in L^2(\Omega, \mathcal{F}, P; \RR^M)$. Under $P$, the vector $\bm{p}(\bm{s})$ lies in the simplex $\mathcal{P} \equiv \{ \bm{p} \in \RR^M \mid p_\alpha \ge 0, \sum_{\alpha=1}^M p_\alpha \le 1 \}$ almost surely. Following Eqs.~\eqref{eq:optimal_weights}--\eqref{eq:optimal_sampling_overhead} of the main text, we define the unmitigated mean error vector and the decoder covariance matrix as
\begin{equation}
\bar{\bm{p}} \equiv \EE_{\bm{s}}[\bm{p}(\bm{s})] \in \mathcal{P} \quad \text{and} \quad \bm{\Sigma} \equiv \Cov_{\bm{s}}(\bm{p}(\bm{s})) = \EE_{\bm{s}}\big[(\bm{p}(\bm{s}) - \bar{\bm{p}})(\bm{p}(\bm{s}) - \bar{\bm{p}})^\top\big] \in \RR^{M \times M},
\label{eq:app_cs_pL_and_D_def}
\end{equation}
where $\bm{\Sigma} \succ 0$ is positive definite whenever all $M$ error channels fluctuate independently across syndromes.

We equip the real Hilbert space $L^2(\Omega, \mathcal{F}, P)$ with the inner product $\langle u, v \rangle \equiv \EE_{\bm{s}}[u(\bm{s}) v(\bm{s})]$ and induced norm $\|u\|_{L^2} \equiv \sqrt{\EE_{\bm{s}}[u(\bm{s})^2]}$. The unbiased VEC filters (when nonempty) form a closed affine subspace $\mathcal{W} \subset L^2(\Omega, \mathcal{F}, P)$ defined by the linear channel-inversion constraints of Eq.~\eqref{eq:pauli_cancellation_constraint},
\begin{align}
\EE_{\bm{s}}\big[ w(\bm{s}) \big] &= 1 \quad \text{(Trace Preservation)}, \label{eq:app_cs_trace_preservation} \\
\EE_{\bm{s}}\big[ w(\bm{s}) \bm{p}(\bm{s}) \big] &= \bm{0} \quad \text{(Vector Error Cancellation)}.
\label{eq:app_cs_pauli_cancellation}
\end{align}
For any feasible filter $w \in \mathcal{W}$, the sampling overhead multiplier in Eq.~\eqref{eq:vpec_sampling_overhead_def} equals the squared $L^2$ norm,
\begin{equation}
C(w) \equiv \|w\|_{L^2}^2 = \EE_{\bm{s}}\big[ w(\bm{s})^2 \big] = \big( \EE_{\bm{s}}[w(\bm{s})] \big)^2 + \Var_{\bm{s}}\big[ w(\bm{s}) \big] = 1 + \Var_{\bm{s}}\big[ w(\bm{s}) \big].
\label{eq:app_cs_overhead_variance_relation}
\end{equation}
Minimizing $C(w)$ over $\mathcal{W}$ is therefore equivalent to finding the minimum-norm element of the affine subspace $\mathcal{W}$. The following two lemmas characterize the cross-covariance and existence of this minimizer.

\begin{lemma}[Invariant Cross-Covariance]
\label{lem:invariant_cross_covariance}
Every measurable filter $w \in \mathcal{W}$ satisfying trace preservation $\EE_{\bm{s}}[w(\bm{s})] = 1$ in Eq.~\eqref{eq:app_cs_trace_preservation} and vector error cancellation $\EE_{\bm{s}}[w(\bm{s})\bm{p}(\bm{s})] = \bm{0}$ in Eq.~\eqref{eq:app_cs_pauli_cancellation} has the same cross-covariance vector with $\bm{p}(\bm{s})$. This vector is fixed by the negative mean error rate,
\begin{align}
\Cov_{\bm{s}}\big( w(\bm{s}), \bm{p}(\bm{s}) \big) &\equiv \EE_{\bm{s}}\big[ w(\bm{s}) \bm{p}(\bm{s}) \big] - \EE_{\bm{s}}\big[ w(\bm{s}) \big] \EE_{\bm{s}}\big[ \bm{p}(\bm{s}) \big] \nonumber \\
&= \bm{0} - (1)(\bar{\bm{p}}) = -\bar{\bm{p}}.
\label{eq:app_cs_covariance_evaluation}
\end{align}
\end{lemma}

\begin{lemma}[Decoder Feasibility for Exact Multi-Process Channel Inversion]
\label{lem:feasibility_condition_formal}
Let $(\Omega, \mathcal{F}, P)$ be the syndrome probability space with conditional error vector $\bm{p} \in L^2(\Omega, \mathcal{F}, P; \RR^M)$, mean $\bar{\bm{p}} \equiv \EE_{\bm{s}}[\bm{p}(\bm{s})]$, and covariance matrix $\bm{\Sigma} \equiv \Cov_{\bm{s}}(\bm{p}(\bm{s}))$. An unbiased filter $w \in L^2(\Omega, \mathcal{F}, P)$ satisfying $\EE_{\bm{s}}[w(\bm{s})] = 1$ and $\EE_{\bm{s}}[w(\bm{s})\bm{p}(\bm{s})] = \bm{0}$ in Eqs.~\eqref{eq:app_cs_trace_preservation}--\eqref{eq:app_cs_pauli_cancellation} exists if and only if
\begin{equation}
\bar{\bm{p}} \in \operatorname{ker}(\bm{\Sigma})^\perp = \operatorname{range}(\bm{\Sigma}).
\label{eq:app_feasibility_condition}
\end{equation}
In particular, whenever $\bm{\Sigma} \succ 0$ is positive definite, a feasible filter exists.
\end{lemma}

\begin{proof}
To prove the forward implication, suppose a feasible filter $w \in L^2(\Omega, \mathcal{F}, P)$ satisfies Eqs.~\eqref{eq:app_cs_trace_preservation}--\eqref{eq:app_cs_pauli_cancellation}. Take any null vector $\bm{v} \in \operatorname{ker}(\bm{\Sigma})$, so that $\bm{v}^\top \bm{\Sigma} \bm{v} = 0$. Because $\bm{\Sigma} = \EE_{\bm{s}}[(\bm{p}(\bm{s}) - \bar{\bm{p}})(\bm{p}(\bm{s}) - \bar{\bm{p}})^\top]$ is positive semi-definite, the projected scalar random variable has zero variance,
\begin{equation}
\bm{v}^\top \bm{\Sigma} \bm{v} = \Var_{\bm{s}}\big[ \bm{v}^\top \bm{p}(\bm{s}) \big] = 0.
\label{eq:app_feasibility_zero_var}
\end{equation}
A real random variable with zero variance equals its mean almost surely under $P$, giving $\bm{v}^\top \bm{p}(\bm{s}) = \bm{v}^\top \bar{\bm{p}}$ almost surely. Taking the inner product of $\bm{v}$ with Eq.~\eqref{eq:app_cs_pauli_cancellation} yields
\begin{equation}
0 = \bm{v}^\top \EE_{\bm{s}}\big[ w(\bm{s}) \bm{p}(\bm{s}) \big] = \EE_{\bm{s}}\big[ w(\bm{s}) \big( \bm{v}^\top \bm{p}(\bm{s}) \big) \big] = \big( \bm{v}^\top \bar{\bm{p}} \big) \EE_{\bm{s}}\big[ w(\bm{s}) \big].
\label{eq:app_feasibility_inner_prod}
\end{equation}
Substituting $\EE_{\bm{s}}[w(\bm{s})] = 1$ from Eq.~\eqref{eq:app_cs_trace_preservation} gives
\begin{equation}
\bm{v}^\top \bar{\bm{p}} = 0.
\label{eq:app_feasibility_orthogonality}
\end{equation}
Because Eq.~\eqref{eq:app_feasibility_orthogonality} holds for every $\bm{v} \in \operatorname{ker}(\bm{\Sigma})$ and $\bm{\Sigma}$ is symmetric, $\bar{\bm{p}} \in \operatorname{ker}(\bm{\Sigma})^\perp = \operatorname{range}(\bm{\Sigma})$.

To prove the converse implication, suppose $\bar{\bm{p}} \in \operatorname{range}(\bm{\Sigma})$, and let $\bm{\Sigma}^+$ denote the Moore--Penrose pseudoinverse of $\bm{\Sigma}$. Because $\bar{\bm{p}}$ lies in the column space of $\bm{\Sigma}$, the projection identity $\bm{\Sigma} \bm{\Sigma}^+ \bar{\bm{p}} = \bar{\bm{p}}$ holds. Define the affine filter
\begin{equation}
w_0(\bm{s}) \equiv 1 - \bar{\bm{p}}^\top \bm{\Sigma}^+ \big( \bm{p}(\bm{s}) - \bar{\bm{p}} \big).
\label{eq:app_feasibility_pseudoinverse_filter}
\end{equation}
Taking the expectation over $\bm{s}$ confirms trace preservation, $\EE_{\bm{s}}[w_0(\bm{s})] = 1 - \bar{\bm{p}}^\top \bm{\Sigma}^+ (\EE_{\bm{s}}[\bm{p}(\bm{s})] - \bar{\bm{p}}) = 1$. Evaluating the error-cancellation expectation with $\EE_{\bm{s}}[\bm{p}(\bm{s})(\bm{p}(\bm{s}) - \bar{\bm{p}})^\top] = \bm{\Sigma}$ yields
\begin{align}
\EE_{\bm{s}}\big[ w_0(\bm{s}) \bm{p}(\bm{s}) \big] &= \bar{\bm{p}} - \EE_{\bm{s}}\big[ \bm{p}(\bm{s}) (\bm{p}(\bm{s}) - \bar{\bm{p}})^\top \big] \bm{\Sigma}^+ \bar{\bm{p}} \nonumber \\
&= \bar{\bm{p}} - \bm{\Sigma} \bm{\Sigma}^+ \bar{\bm{p}} = \bar{\bm{p}} - \bar{\bm{p}} = \bm{0}.
\label{eq:app_feasibility_pseudoinverse_verify}
\end{align}
Thus $w_0 \in \mathcal{W}$ satisfies both constraints in Eqs.~\eqref{eq:app_cs_trace_preservation}--\eqref{eq:app_cs_pauli_cancellation}.
\end{proof}

\subsection{\texorpdfstring{Statement and Proofs of Theorem~\ref{thm:variational_cs_bound} (Generalized Variational Matrix Bound)}{Statement and Proofs of Generalized Variational Matrix Bound}}
\label{subsec:proof_variational_cs_bound}
\label{subsec:cs_covariance_proof}
\label{subsec:euler_lagrange_derivation}

We now establish the universal lower bound on the variance and sampling overhead of any unbiased VEC filter. We give two independent proofs: a direct completion of squares in Subsection~\ref{subsubsec:app_proof_completion_of_squares} and a constrained functional optimization in Subsection~\ref{subsubsec:app_proof_kkt_schur}.

\begin{theorem}[Generalized Variational Matrix Bound]
\label{thm:variational_cs_bound}
Let $(\Omega, \mathcal{F}, P)$ be the syndrome probability space with conditional error vector $\bm{p}(\bm{s})$, mean $\bar{\bm{p}} \equiv \EE_{\bm{s}}[\bm{p}(\bm{s})] \in \mathcal{P}$, and positive-definite covariance matrix $\bm{\Sigma} \equiv \Cov_{\bm{s}}(\bm{p}(\bm{s})) \succ 0$. Any filter $w \in L^2(\Omega, \mathcal{F}, P)$ satisfying the channel-inversion conditions in Eq.~\eqref{eq:pauli_cancellation_constraint} obeys the variance lower bound
\begin{equation}
\Var_{\bm{s}}\big[ w(\bm{s}) \big] \ge \bar{\bm{p}}^\top \bm{\Sigma}^{-1} \bar{\bm{p}}.
\label{eq:app_thm_cs_variance_lower_bound}
\end{equation}
The sampling overhead multiplier therefore satisfies
\begin{equation}
C(w) \equiv \EE_{\bm{s}}\big[ w(\bm{s})^2 \big] \ge C_{\mathrm{syn}}^* \equiv 1 + \bar{\bm{p}}^\top \bm{\Sigma}^{-1} \bar{\bm{p}}.
\label{eq:cs_overhead_lower_bound}
\end{equation}
Equality holds in Eqs.~\eqref{eq:app_thm_cs_variance_lower_bound}--\eqref{eq:cs_overhead_lower_bound} if and only if $w(\bm{s}) = w^*(\bm{s})$ almost surely under $P$. Here $w^*(\bm{s}) = 1 - \bar{\bm{p}}^\top \bm{\Sigma}^{-1}(\bm{p}(\bm{s}) - \bar{\bm{p}})$ is the affine filter of Eq.~\eqref{eq:optimal_weights}.
\end{theorem}

\subsubsection{Proof 1: Matrix Cauchy--Schwarz Completion of Squares}
\label{subsubsec:app_proof_completion_of_squares}

\begin{proof}[Proof 1 of Theorem~\ref{thm:variational_cs_bound}]
Let $w \in \mathcal{W}$ be any filter satisfying $\EE_{\bm{s}}[w(\bm{s})] = 1$ and $\EE_{\bm{s}}[w(\bm{s})\bm{p}(\bm{s})] = \bm{0}$ in Eqs.~\eqref{eq:app_cs_trace_preservation}--\eqref{eq:app_cs_pauli_cancellation}. For an arbitrary vector $\bm{a} \in \RR^M$, define the centered scalar random variable
\begin{equation}
\Delta_{\bm{a}}(\bm{s}) \equiv \big( w(\bm{s}) - 1 \big) - \bm{a}^\top \big( \bm{p}(\bm{s}) - \bar{\bm{p}} \big).
\label{eq:app_delta_def}
\end{equation}
Because the second moment of a real random variable is non-negative, $\EE_{\bm{s}}[\Delta_{\bm{a}}(\bm{s})^2] \ge 0$. Expanding the square gives
\begin{align}
\EE_{\bm{s}}\big[ \Delta_{\bm{a}}(\bm{s})^2 \big] &= \EE_{\bm{s}}\Big[ \big(w(\bm{s}) - 1\big)^2 \Big] - 2 \bm{a}^\top \EE_{\bm{s}}\Big[ \big(w(\bm{s}) - 1\big) \big(\bm{p}(\bm{s}) - \bar{\bm{p}}\big) \Big] \nonumber \\
&\quad + \bm{a}^\top \EE_{\bm{s}}\Big[ \big(\bm{p}(\bm{s}) - \bar{\bm{p}}\big)\big(\bm{p}(\bm{s}) - \bar{\bm{p}}\big)^\top \Big] \bm{a} \nonumber \\
&= \Var_{\bm{s}}\big[ w(\bm{s}) \big] - 2 \bm{a}^\top \Cov_{\bm{s}}\big( w(\bm{s}), \bm{p}(\bm{s}) \big) + \bm{a}^\top \bm{\Sigma} \bm{a}.
\label{eq:app_expanded_delta}
\end{align}
By Lemma~\ref{lem:invariant_cross_covariance}, every feasible filter satisfies $\Cov_{\bm{s}}(w(\bm{s}), \bm{p}(\bm{s})) = -\bar{\bm{p}}$. Substituting Eq.~\eqref{eq:app_cs_covariance_evaluation} into Eq.~\eqref{eq:app_expanded_delta} yields a quadratic inequality in $\bm{a}$,
\begin{equation}
\Var_{\bm{s}}\big[ w(\bm{s}) \big] + 2 \bm{a}^\top \bar{\bm{p}} + \bm{a}^\top \bm{\Sigma} \bm{a} \ge 0.
\label{eq:app_quadratic_in_a}
\end{equation}
Because $\bm{\Sigma} \succ 0$, the quadratic potential $g(\bm{a}) \equiv \bm{a}^\top \bm{\Sigma} \bm{a} + 2 \bm{a}^\top \bar{\bm{p}}$ is strictly convex on $\RR^M$. It therefore attains its unique global minimum where its gradient vanishes,
\begin{equation}
\nabla_{\bm{a}} g(\bm{a}) = 2 \bm{\Sigma} \bm{a} + 2 \bar{\bm{p}} = \bm{0}, \qquad \bm{a}^* = -\bm{\Sigma}^{-1} \bar{\bm{p}}.
\label{eq:app_optimal_a}
\end{equation}
Evaluating $g(\bm{a})$ at $\bm{a}^*$ gives
\begin{align}
g(\bm{a}^*) &= (-\bm{\Sigma}^{-1} \bar{\bm{p}})^\top \bm{\Sigma} (-\bm{\Sigma}^{-1} \bar{\bm{p}}) + 2 (-\bm{\Sigma}^{-1} \bar{\bm{p}})^\top \bar{\bm{p}} \nonumber \\
&= \bar{\bm{p}}^\top \bm{\Sigma}^{-1} \bar{\bm{p}} - 2 \bar{\bm{p}}^\top \bm{\Sigma}^{-1} \bar{\bm{p}} = -\bar{\bm{p}}^\top \bm{\Sigma}^{-1} \bar{\bm{p}}.
\label{eq:app_min_g_eval}
\end{align}
Substituting $g(\bm{a}^*)$ into Eq.~\eqref{eq:app_quadratic_in_a} proves Eq.~\eqref{eq:app_thm_cs_variance_lower_bound},
\begin{equation}
\Var_{\bm{s}}\big[ w(\bm{s}) \big] - \bar{\bm{p}}^\top \bm{\Sigma}^{-1} \bar{\bm{p}} \ge 0, \qquad \Var_{\bm{s}}\big[ w(\bm{s}) \big] \ge \bar{\bm{p}}^\top \bm{\Sigma}^{-1} \bar{\bm{p}}.
\label{eq:app_cs_variance_bound}
\end{equation}
Adding unity to both sides of Eq.~\eqref{eq:app_cs_variance_bound} via Eq.~\eqref{eq:app_cs_overhead_variance_relation} establishes Eq.~\eqref{eq:cs_overhead_lower_bound},
\begin{equation}
C(w) = 1 + \Var_{\bm{s}}\big[ w(\bm{s}) \big] \ge 1 + \bar{\bm{p}}^\top \bm{\Sigma}^{-1} \bar{\bm{p}} = C_{\mathrm{syn}}^*.
\label{eq:app_cs_overhead_bound}
\end{equation}

Equality holds in Eq.~\eqref{eq:app_cs_variance_bound} if and only if $\EE_{\bm{s}}[\Delta_{\bm{a}^*}(\bm{s})^2] = 0$, which requires $\Delta_{\bm{a}^*}(\bm{s}) = 0$ almost surely under $P$. Setting $\Delta_{\bm{a}^*}(\bm{s}) = 0$ in Eq.~\eqref{eq:app_delta_def} yields
\begin{equation}
w(\bm{s}) - 1 = \bm{a}^{*\top} \big( \bm{p}(\bm{s}) - \bar{\bm{p}} \big) = -\bar{\bm{p}}^\top \bm{\Sigma}^{-1} \big( \bm{p}(\bm{s}) - \bar{\bm{p}} \big).
\label{eq:app_cs_equality_condition}
\end{equation}
This equality recovers the optimal affine filter $w^*(\bm{s})$ in Eq.~\eqref{eq:optimal_weights}.
\end{proof}

\subsubsection{Proof 2: Karush--Kuhn--Tucker Functional Optimization and Schur Complement Inversion}
\label{subsubsec:app_proof_kkt_schur}

\begin{proof}[Proof 2 of Theorem~\ref{thm:variational_cs_bound}]
We cast the search for the minimum-overhead filter as a convex program in $L^2(\Omega, \mathcal{F}, P)$,
\begin{equation}
\min_{w \in L^2} \frac{1}{2} \EE_{\bm{s}}\big[ w(\bm{s})^2 \big] \quad \text{subject to} \quad \EE_{\bm{s}}[w(\bm{s})] = 1, \quad \EE_{\bm{s}}[w(\bm{s}) \bm{p}(\bm{s})] = \bm{0}.
\label{eq:app_kkt_opt_prob}
\end{equation}
We introduce a scalar Lagrange multiplier $\lambda_0 \in \RR$ for trace preservation and a vector $\bm{\lambda} \in \RR^M$ for error cancellation. These dual variables define the Lagrangian functional
\begin{equation}
\cL(w, \lambda_0, \bm{\lambda}) \equiv \frac{1}{2} \EE_{\bm{s}}\big[ w(\bm{s})^2 \big] - \lambda_0 \Big( \EE_{\bm{s}}[w(\bm{s})] - 1 \Big) - \bm{\lambda}^\top \EE_{\bm{s}}\big[ w(\bm{s}) \bm{p}(\bm{s}) \big].
\label{eq:app_kkt_lagrangian}
\end{equation}
The first Gateaux variation of $\cL$ along an arbitrary test function $h \in L^2(\Omega, \mathcal{F}, P)$ is
\begin{equation}
\delta \cL(w, \lambda_0, \bm{\lambda}; h) = \EE_{\bm{s}}\Big[ \big( w(\bm{s}) - \lambda_0 - \bm{\lambda}^\top \bm{p}(\bm{s}) \big) h(\bm{s}) \Big].
\label{eq:app_kkt_gateaux_first}
\end{equation}
Requiring $\delta \cL = 0$ for all $h \in L^2(\Omega, \mathcal{F}, P)$ yields the stationarity condition
\begin{equation}
w^*(\bm{s}) = \lambda_0 + \bm{\lambda}^\top \bm{p}(\bm{s}),
\label{eq:app_kkt_stationarity}
\end{equation}
almost surely under $P$. Taking the expectation of Eq.~\eqref{eq:app_kkt_stationarity} against $1$ and $\bm{p}(\bm{s})$ produces the KKT linear system
\begin{align}
\EE_{\bm{s}}\big[ w^*(\bm{s}) \big] &= \lambda_0 + \bar{\bm{p}}^\top \bm{\lambda} = 1, \label{eq:app_kkt_row1} \\
\EE_{\bm{s}}\big[ w^*(\bm{s}) \bm{p}(\bm{s}) \big] &= \lambda_0 \bar{\bm{p}} + \bm{M}_2 \bm{\lambda} = \bm{0},
\label{eq:app_kkt_row2}
\end{align}
where $\bm{M}_2 \equiv \EE_{\bm{s}}[\bm{p}(\bm{s}) \bm{p}(\bm{s})^\top] = \bm{\Sigma} + \bar{\bm{p}}\bar{\bm{p}}^\top$ is the uncentered second-moment matrix of $\bm{p}(\bm{s})$. In partitioned matrix form, Eqs.~\eqref{eq:app_kkt_row1}--\eqref{eq:app_kkt_row2} read
\begin{equation}
\begin{pmatrix} 1 & \bar{\bm{p}}^\top \\ \bar{\bm{p}} & \bm{M}_2 \end{pmatrix} \begin{pmatrix} \lambda_0 \\ \bm{\lambda} \end{pmatrix} = \begin{pmatrix} 1 \\ \bm{0} \end{pmatrix}.
\label{eq:app_kkt_block_system}
\end{equation}
The Schur complement of the top-left scalar entry $1$ in this $(M+1) \times (M+1)$ matrix is
\begin{equation}
\bm{S} = \bm{M}_2 - \bar{\bm{p}} (1)^{-1} \bar{\bm{p}}^\top = \big( \bm{\Sigma} + \bar{\bm{p}}\bar{\bm{p}}^\top \big) - \bar{\bm{p}}\bar{\bm{p}}^\top = \bm{\Sigma}.
\label{eq:app_kkt_schur_eval}
\end{equation}
Because $\bm{S} = \bm{\Sigma} \succ 0$ is invertible, the block-matrix inversion formula gives
\begin{equation}
\begin{pmatrix} 1 & \bar{\bm{p}}^\top \\ \bar{\bm{p}} & \bm{M}_2 \end{pmatrix}^{-1} = \begin{pmatrix} 1 + \bar{\bm{p}}^\top \bm{\Sigma}^{-1} \bar{\bm{p}} & -\bar{\bm{p}}^\top \bm{\Sigma}^{-1} \\ -\bm{\Sigma}^{-1} \bar{\bm{p}} & \bm{\Sigma}^{-1} \end{pmatrix}.
\label{eq:app_kkt_block_inverse}
\end{equation}
Multiplying Eq.~\eqref{eq:app_kkt_block_inverse} by the right-hand vector $(1, \bm{0}^\top)^\top$ determines the dual variables,
\begin{equation}
\begin{pmatrix} \lambda_0 \\ \bm{\lambda} \end{pmatrix} = \begin{pmatrix} 1 + \bar{\bm{p}}^\top \bm{\Sigma}^{-1} \bar{\bm{p}} \\ -\bm{\Sigma}^{-1} \bar{\bm{p}} \end{pmatrix}.
\label{eq:app_kkt_multipliers_solved}
\end{equation}
Substituting Eq.~\eqref{eq:app_kkt_multipliers_solved} into Eq.~\eqref{eq:app_kkt_stationarity} yields the optimal filter
\begin{align}
w^*(\bm{s}) &= \big( 1 + \bar{\bm{p}}^\top \bm{\Sigma}^{-1} \bar{\bm{p}} \big) - \big( \bm{\Sigma}^{-1} \bar{\bm{p}} \big)^\top \bm{p}(\bm{s}) \nonumber \\
&= 1 - \bar{\bm{p}}^\top \bm{\Sigma}^{-1} \big( \bm{p}(\bm{s}) - \bar{\bm{p}} \big).
\label{eq:app_cs_w_star_recovered}
\end{align}
This expression matches Eq.~\eqref{eq:optimal_weights} of the main text. We next verify global uniqueness of this stationary point. For any non-zero perturbation $h \in L^2(\Omega, \mathcal{F}, P)$, the second Gateaux variation of $\cL$ is
\begin{equation}
\delta^2 \cL(w, \lambda_0, \bm{\lambda}; h, h) = \EE_{\bm{s}}\big[ h(\bm{s})^2 \big] = \|h\|_{L^2}^2 > 0.
\label{eq:app_kkt_gateaux_second}
\end{equation}
Strict positivity of $\delta^2 \cL$ confirms that the objective is strictly convex on $\mathcal{W}$. Because $\mathcal{W}$ is a closed affine subspace of $L^2(\Omega, \mathcal{F}, P)$, the filter $w^*(\bm{s})$ is the unique global minimizer.
\end{proof}

\subsection{Physical Corollaries and Postselection Limits}
\label{subsec:bhatia_davis_ceiling}
\label{sec:bimodal_distribution}

When a single binary logical error channel is decoded by maximum likelihood relative to the retained syndrome information, its posterior error rate is bounded above by $1/2$, which imposes a universal lower bound on the sampling overhead. Merely using the same heuristic decoder to choose and evaluate a recovery does not guarantee $p_L(\bm{s})\le1/2$, and in the dual-decoder protocol the offline posterior evaluated in the committed real-time frame exceeds $1/2$ on disagreement syndromes. The following scalar corollary establishes the variance ceiling and overhead floor for any binary channel satisfying $p_L(\bm{s})\in[0,1/2]$, which is distinct from the surface-code conditional Pauli-PEC benchmark analyzed in Subsection~\ref{subsec:salem_surface_benchmark}.

\begin{corollary}[Physical Bhatia--Davis Variance Ceiling and Postselection Limit]
\label{cor:bhatia_davis_ceiling}
\label{prop:bhatia_davis_ceiling}
Let the conditional logical error rate $p_L(\bm{s})$ of a single-channel code satisfy $p_L(\bm{s}) \in [0, 1/2]$ almost surely with mean $p_L \equiv \EE_{\bm{s}}[p_L(\bm{s})] \in (0, 1/2)$ and variance $D \equiv \Var_{\bm{s}}[p_L(\bm{s})] > 0$. Then the following three properties hold:
\begin{enumerate}
\item \emph{Variance ceiling:} The posterior error variance obeys
\begin{equation}
D \le p_L \left(\frac{1}{2} - p_L\right).
\label{eq:app_bd_bhatia_davis_bound}
\end{equation}
\item \emph{Overhead floor:} The optimal sampling overhead $C_{\mathrm{syn}}^* = 1 + p_L^2/D$ of Theorem~\ref{thm:variational_cs_bound} satisfies
\begin{equation}
C_{\mathrm{syn}}^* \ge \frac{1}{1 - 2 p_L} = 1 + \frac{2 p_L}{1 - 2 p_L}, \qquad \eta \equiv \frac{C_{\mathrm{syn}}^* - 1}{p_L} \ge \frac{2}{1 - 2 p_L} \ge 2.
\label{eq:app_bd_c_opt_bhatia_davis_bound}
\end{equation}
\item \emph{Postselection limit:} Equality holds in Eqs.~\eqref{eq:app_bd_bhatia_davis_bound}--\eqref{eq:app_bd_c_opt_bhatia_davis_bound} if and only if $p_L(\bm{s}) \in \{0, 1/2\}$ almost surely. On this two-point support, the optimal filter $w^*(\bm{s})$ reduces to postselection,
\begin{equation}
w^*(0) = \frac{1}{1 - 2 p_L}, \qquad w^*\!\left(\frac{1}{2}\right) = 0.
\label{eq:app_w_star_postselection_extremal}
\end{equation}
\end{enumerate}
\end{corollary}

\begin{proof}
We apply the Bhatia--Davis inequality~\cite{bhatia2000better} to bound the variance of $p_L(\bm{s})$. For any real random variable $X$ supported on $[a, b]$ with mean $\mu = \EE[X]$, the variance obeys
\begin{equation}
\Var[X] \le (\mu - a)(b - \mu).
\label{eq:app_bd_bhatia_davis_general_ineq}
\end{equation}
Setting $X = p_L(\bm{s})$ with endpoints $a = 0$ and $b = 1/2$ and mean $\mu = p_L$ in Eq.~\eqref{eq:app_bd_bhatia_davis_general_ineq} gives
\begin{equation}
D = \Var_{\bm{s}}\big[ p_L(\bm{s}) \big] \le (p_L - 0)\left(\frac{1}{2} - p_L\right) = p_L \left(\frac{1}{2} - p_L\right).
\label{eq:app_bd_bhatia_davis_variance_step}
\end{equation}
This bound proves Eq.~\eqref{eq:app_bd_bhatia_davis_bound}. Substituting Eq.~\eqref{eq:app_bd_bhatia_davis_variance_step} into $C_{\mathrm{syn}}^* = 1 + p_L^2/D$ from Eq.~\eqref{eq:optimal_sampling_overhead} yields
\begin{align}
C_{\mathrm{syn}}^* &\ge 1 + \frac{p_L^2}{p_L(1/2 - p_L)} \nonumber \\
&= 1 + \frac{p_L}{1/2 - p_L} = \frac{1}{1 - 2 p_L} = 1 + \frac{2 p_L}{1 - 2 p_L},
\label{eq:app_bd_bhatia_davis_overhead_derivation}
\end{align}
and dividing $C_{\mathrm{syn}}^* - 1$ by $p_L$ gives $\eta \ge 2/(1 - 2p_L) \ge 2$, proving Eq.~\eqref{eq:app_bd_c_opt_bhatia_davis_bound}.

Equality holds in Eq.~\eqref{eq:app_bd_bhatia_davis_general_ineq} if and only if $X$ is supported on the endpoints $\{a, b\}$, which requires $P(p_L(\bm{s}) \in \{0, 1/2\}) = 1$. Let $\pi \equiv P(p_L(\bm{s}) = 1/2)$ denote the probability of the ambiguous outcome. The mean error rate is then
\begin{equation}
p_L = 0 \cdot (1 - \pi) + \frac{1}{2} \cdot \pi = \frac{\pi}{2},
\label{eq:app_bd_two_point_mean}
\end{equation}
so $\pi = 2 p_L$ and $1 - \pi = 1 - 2 p_L$. At the maximal variance $D = p_L(1/2 - p_L)$, the filter slope in Eq.~\eqref{eq:app_cs_w_star_recovered} equals $p_L/D = 2/(1 - 2p_L)$. Evaluating $w^*(p_L(\bm{s})) = 1 - (p_L/D)(p_L(\bm{s}) - p_L)$ at the two endpoints yields
\begin{align}
w^*(0) &= 1 - \frac{2}{1 - 2p_L}(0 - p_L) = 1 + \frac{2p_L}{1 - 2p_L} = \frac{1}{1 - 2p_L}, \label{eq:app_w_star_noiseless} \\
w^*\!\left(\frac{1}{2}\right) &= 1 - \frac{2}{1 - 2p_L}\left(\frac{1}{2} - p_L\right) = 1 - 1 = 0.
\label{eq:app_w_star_depolarized}
\end{align}
These two values establish Eq.~\eqref{eq:app_w_star_postselection_extremal}.
\end{proof}

The binary bound above applies specifically to scalar whole-channel reweighting and should be contrasted with unconditioned Pauli PEC. In the binary model, standard PEC without syndrome conditioning has cost $C_{\rm PEC}=(1-2p_L)^{-2}=1+4p_L+O(p_L^2)$. More generally, a fixed finite Pauli alphabet with total error $p_L\to0$ has inverse coefficients $1+p_L+O(p_L^2)$ on the identity and $-r_P+O(p_L^2)$ on each non-identity Pauli channel $P$. Its PEC one-norm is therefore $1+2p_L+O(p_L^2)$ and its squared sampling cost is $1+4p_L+O(p_L^2)$, corresponding to a universal baseline slope of $\eta=4$.

In the scalar VEC model, the asymptotic overhead slope is governed by the dispersion ratio $g(p)=D/p_L$. Because $C_{\rm syn}^*-1=p_L/g(p)$ and $\eta=1/g(p)$ hold exactly, convergence $g(p)\to\kappa>0$ along a fixed-code low-noise limit yields
\begin{equation}
C_{\rm syn}^*=1+\kappa^{-1}p_L+o(p_L),\qquad
\eta=\kappa^{-1}+o(1),
\label{eq:app_bd_overhead_scaling}
\end{equation}
with an $O(p_L^2)$ remainder whenever $D=\kappa p_L+O(p_L^2)$. A single-decoder posterior bounded by $1/2$ forces $\kappa\le 1/2$ and hence $\eta\ge 2$; for untied odd-distance codes such as the three-bit repetition code (where $p_L\sim3p^2$ and $D\sim3p^3$), the ratio $g(p)$ vanishes ($\kappa=0$). A dual-decoder posterior in the committed real-time frame can instead approach unity on disagreement shots, allowing $D/p_L\to1$ and driving the excess-overhead slope to $\eta\to1$.

\subsection{Surface-Code Conditional Pauli-PEC Benchmark}
\label{subsec:salem_surface_benchmark}
We now establish the $\eta \ge 2$ asymptotic lower bound (Theorem~\ref{thm:salem_surface_gap}) for single-decoder conditional Pauli PEC on odd-distance rotated surface codes retaining hard matching decisions and uniform-weight complementary gaps.
Consider the standard square rotated surface code of odd distance \(d\ge3\),
with \(n=d^2\) data qubits, one logical qubit, and perfect syndrome measurements.
Number the rows and columns from \(0\) to \(d-1\).
Fix a straight boundary logical \(Z\) on the left column \(B_Z\), and a
straight boundary logical \(X\) on the top row \(B_X\).
Each row supports a weight-\(d\) logical \(X\), and each column a
weight-\(d\) logical \(Z\). The same-type stabilizers have even weight; interchanging \(X\) and \(Z\), or reflecting the patch and the chosen
boundaries together, gives an equivalent convention.

Under code-capacity depolarizing noise, each data qubit independently experiences \(I,X,Y,Z\) with
respective probabilities \(1-p\) and \(p/3\) for each non-identity Pauli fault. Write the error, modulo phase, as
\(E=X^xZ^z\) for binary vectors \(x,z\in\{0,1\}^n\); its physical weight is \(|\supp(x)\cup\supp(z)|\), so the two binary components $x$ and $z$ are correlated through $Y$ faults.

Separate minimum-Hamming-weight component decoding (equivalently, uniform-weight component minimum-weight perfect matching (MWPM)) evaluates each syndrome component independently. For an \(X\)-component syndrome \(s_X\), define
\[
 m_{X,b}(s_X)=\min\{|u|:\ u\text{ has syndrome }s_X,\
                    |\,\supp(u)\cap B_Z\,|\equiv b\pmod2\},
 \qquad b\in\{0,1\},
\]
and define \(m_{Z,b}\) analogously using \(B_X\).
The decoder selects the hard decision and complementary weight gap
\[
 \widehat e_j=\arg\min_{b\in\{0,1\}}m_{j,b},
 \qquad
 \Delta_j=|m_{j,1}-m_{j,0}|,\qquad j=X,Z.
\]
Because opposite binary logical classes have opposite weight parity, the selected
class is unique and the gaps $\Delta_j$ are positive odd integers (equal to the standard weight gap divided by the common matching-edge weight).

The single-decoder conditional PEC comparator retains the four-component soft-decision record
\begin{equation}
 T=(\widehat e_X,\widehat e_Z,\Delta_X,\Delta_Z).
 \label{eq:salem_surface_record}
\end{equation}
The residual component error bits in the committed decoder frame are
\[
 e_X=(|\supp(x)\cap B_Z|\bmod2)\oplus\widehat e_X,
 \qquad
 e_Z=(|\supp(z)\cap B_X|\bmod2)\oplus\widehat e_Z,
\]
with total logical failure probability and $Y$-failure probability
\[
 p_L=\Pr((e_X,e_Z)\ne(0,0)),\qquad
 y_L=\Pr(e_X=e_Z=1),\qquad r=\frac{d+1}{2},
\]
so that \(p_L\) is the unmitigated logical error rate of the real-time decoder itself.

\paragraph{The PEC comparator.} Following the conditional-inversion protocol of Ref.~\cite{aharonov2025salem}, we invert the residual logical Pauli channel conditioned on each retained soft record $T=t$.
Conditional on \(T=t\), let
\(\mathcal N_t=\sum_{P\in\{I,X,Y,Z\}}r_P(t)\mathcal P\) be the residual logical
Pauli channel, where \(\mathcal P(\rho)=P\rho P\) denotes Pauli conjugation.
Its probabilities are assumed to be calibrated exactly, and the distribution of \(T\) is independent of the encoded input state because perfect stabilizer measurements depend only on the stochastic Pauli error.
On each retained class $t$, the comparator inverts \(\mathcal N_t\) by a
quasiprobability expansion over ideal logical Pauli corrections with one-norm $\gamma_t$. The cost \(C\equiv C_{\rm SALEM,PEC}\) is the minimum standard PEC second-moment overhead after optimizing a scalar multiplier over \(T\) (with singular channels discarded via $\gamma_t^{-2}=0$).
This formulation defines the conditional Pauli-PEC benchmark with ideal Pauli tracking, whose minimum sampling overhead is bounded below by the following theorem.

\begin{theorem}[Surface-code conditional Pauli-PEC benchmark]
\label{thm:salem_surface_gap}
In the model above, for each fixed odd \(d\ge3\),
\[
 p_L=\Theta(p^r),\qquad y_L=O(p^{r+1}),
\]
and
\begin{equation}
 C^{-1}\le 1-2p_L+O(p^{r+1}),
 \qquad
 \boxed{\displaystyle\liminf_{p\to0}\frac{C-1}{p_L}\ge2.}
 \label{eq:salem_surface_bound}
\end{equation}
For \(d=3\), the stronger lower-bound coefficient \(4\) holds.
The constants and the small-\(p\) range may depend on \(d\).
No pointwise condition \(r_X(t)+r_Z(t)\le1/2\) is imposed.
\end{theorem}

\subsubsection{Geometric and Decoding Facts}

\begin{lemma}[Logical \(Y\) weight]
\label{lem:salem_surface_y_weight}
The smallest physical Pauli weight in the logical \(Y\) class of the square
rotated patch is \(2d-1\).
\end{lemma}
\begin{proof}
From the \(X\) component of a logical \(Y\), select a simple primal
left-to-right logical path; from its \(Z\) component, select a simple dual
top-to-bottom logical path. Their union is contained in the support of the
original Pauli.

Place data qubits at integer \((\text{column},\text{row})\) grid sites.
At each data qubit, the primal and dual edges join the two pairs
of opposite half-integer corners around that site.
Uncontract the boundary vertices into their distinct geometric endpoints;
the left/right and top/bottom boundary positions lie at
\(-1/2\) and \(d-1/2\).
Every primal edge shifts its horizontal position by \(+1\) or \(-1\), and
every dual edge shifts its vertical position by \(+1\) or \(-1\), across a boundary separation of \(d\).
If the two oriented paths have \(b_X\) and \(b_Z\) backward steps, their lengths
are \(d+2b_X\) and \(d+2b_Z\).

A shared data qubit is a transverse crossing of the two paths.
Boundary tangencies can be displaced slightly inward without changing the
crossings, and the primal and dual boundary endpoints are disjoint.
Orienting the paths from left to right and from top to bottom fixes their signed intersection number to unity.
A crossing of two forward edges is positive, so each negative crossing uses a backward edge of at least one path.
Because distinct crossings use distinct physical edges, the number \(N_-\)
of negative crossings satisfies \(N_-\le b_X+b_Z\).
The total number \(K\) of crossings obeys
\(K=1+2N_-\le1+2(b_X+b_Z)\). Consequently, their union has size at least
\[
 (d+2b_X)+(d+2b_Z)-K\ge2d-1.
\]
Straight crossing logical strings attain this weight.
\end{proof}

\begin{lemma}[Leading failures]
\label{lem:salem_surface_leading}
All physical errors of weight below \(r\) are corrected by the separate
decoder. If a weight-\(r\) error fails in the \(X\) component, three properties hold. Its \(X\) support \(F\) has size \(r\), and every physical fault on \(F\) is \(X\) or \(Y\). Every selected minimum \(X\) recovery \(R\) has size \(r-1\), is disjoint from \(F\) (\(F\cap R=\varnothing\)), and combines with \(F\) to form a weight-\(d\) logical \(X\) operator \(L=F\cup R\). Finally, the \(X\)-component gap is one.
The analogous assertions hold for a \(Z\)-component failure, and no simultaneous component failures occur at physical weight \(r\).
\end{lemma}
\begin{proof}
Two opposite binary logical classes with the same syndrome have pairwise
distance at least \(d\). Hence a component of weight at most \(r-1\)
cannot be decoded into the opposite class. This proves correction below
physical weight \(r\) and the first assertion.

For a failed \(X\) component of weight \(r\), a selected recovery has weight
at most \(r\). Since its class has the opposite parity, its weight is at
most \(r-1\). The nontrivial logical operator \(F\mathbin{\triangle}R\)
has weight at least \(d=2r-1\). Equality must hold throughout
\[
 d\le|F\mathbin{\triangle}R|\le|F|+|R|\le r+(r-1)=d.
\]
This proves the second assertion. The true class has minimum weight exactly
\(r\), since a smaller representative together with \(R\) would violate
distance \(d\). Thus the gap is one.

If both components failed at weight \(r\), then \(x=z\) would be the full
weight-\(r\) support. Combining this physical error with the two selected
weight-\((r-1)\) component recoveries would produce a logical \(Y\) of
physical weight at most \(3r-2<4r-3=2d-1\), contradicting
Lemma~\ref{lem:salem_surface_y_weight}.
\end{proof}

Lemma~\ref{lem:salem_surface_leading} immediately determines the leading fault orders of $y_L$ and $p_L$. Simultaneous component failures require weight at least $r+1$, giving \(y_L=O(p^{r+1})\). Any \(r\)-qubit subset of a minimum-weight logical \(X\) string produces a weight-\(r\) failure under pure \(X\) faults, establishing \(p_L=\Theta(p^r)\).

\subsubsection{A Transformation Preserving the Entire Record}

To bound the conditional error probabilities on the retained record $T$, we use the crossing geometry of logical operators. Every weight-\(d\) logical \(X\) crosses each straight column logical \(Z\) an odd number of times, and therefore meets the left boundary \(B_Z\) at a unique qubit \(q\).

\begin{lemma}[Lighter correct events with the same record]
\label{lem:salem_surface_map}
Suppose \(d\ge5\), so \(r\ge3\).
At least half the weight-\(r\) failure probability belongs to records \(T\)
that can also be generated by a correctly decoded weight-\((r-1)\)
physical error.
\end{lemma}
\begin{proof}
Fix an \(X\)-failing support \(F\) of size \(|F|=r\).
Choose a selected minimum recovery \(R\), and let \(L=F\cup R\) as in
Lemma~\ref{lem:salem_surface_leading}. Let \(q=L\cap B_Z\).
The \(2^r\) assignments of \(X\) and \(Y\) to \(F\) all fail in the \(X\)
component, share the probability
\((p/3)^r(1-p)^{n-r}\), and produce no \(Z\)-component failure.
The hard \(X\)-decision is
\begin{equation}
 \widehat e_X=1\oplus\mathbf1_{\{q\in F\}}.
 \label{eq:salem_surface_hardbit}
\end{equation}

If \(q\in F\) and its physical fault is \(X\), delete that fault.
The new \(X\) support is \(F'=F\setminus\{q\}\), and the \(Z\) component
is unchanged. The new hard \(X\)-decision is zero, matching
Eq.~\eqref{eq:salem_surface_hardbit}.

If \(q\notin F\) and at least two physical faults \(a,b\in F\) are pure
\(X\), delete those two faults and insert a pure \(X\) fault at \(q\).
Now
\[
 F'=(F\setminus\{a,b\})\cup\{q\},
\]
and again the \(Z\) component is unchanged. The new hard \(X\)-decision
is one, also matching Eq.~\eqref{eq:salem_surface_hardbit}.

In either case the physical error has weight \(r-1\), is correctly decoded,
and has \(X\) support \(F'\subset L\) of size \(r-1\).
Its selected component class has minimum weight \(r-1\), and the opposite
class has minimum weight \(r\): the complement \(L\setminus F'\) realizes
the latter weight, and the distance bound forbids either minimum from
being smaller. Its gap is therefore still one.
Since the physical \(Z\) component has not changed, both its hard decision
and its gap are unchanged. Thus all four entries of \(T\) are preserved.

When \(q\in F\), the first construction applies to \(2^{r-1}\) fault patterns.
When \(q\notin F\), the second applies to
\[
 2^r-\binom r0-\binom r1=2^r-r-1\ge2^{r-1}\qquad(r\ge3)
\]
patterns. This proves the half-probability assertion for each failing
\(X\) support. Exchange \(X\) and \(Z\) for the other component.
The two leading failure families are disjoint by
Lemma~\ref{lem:salem_surface_leading}, completing the proof.
The transformation need not be injective: only the existence of a lighter
correct event in the same record is used.
\end{proof}

\subsubsection{The Posterior Second Moment}

We next bound the second moment of the parity-error posterior
\[
 h_t=\Pr(e_X\oplus e_Z=1\mid T=t)=r_X(t)+r_Z(t),
\]
whose expectation satisfies \(\E h_T=p_L-y_L\).

\begin{lemma}[Second-moment bound]
\label{lem:salem_surface_moment}
For each fixed odd \(d\ge5\),
\[
 \E[h_T^2]\le\tfrac12p_L+O(p^{r+1}).
\]
\end{lemma}
\begin{proof}
Let \(\mathcal G\) be the set of records realizable by a physical error of
weight at most \(r-1\).
For each record \(t\in\mathcal G\), the marginal probability obeys \(\Pr(T=t)=\Omega(p^{r-1})\), while the joint failure probability obeys
\(\Pr(e_X\oplus e_Z=1,T=t)=O(p^r)\).
Thus \(h_t=O(p)\) uniformly over the finite set \(\mathcal G\), and
\[
 \E[h_T^2\mathbf1_{\{T\in\mathcal G\}}]
 \le O(p)\E h_T=O(p^{r+1}).
\]
Lemma~\ref{lem:salem_surface_map} places at least half the weight-\(r\) failure
probability in \(\mathcal G\).
Errors of greater weight have total probability \(O(p^{r+1})\), so
\[
 \Pr(e_X\oplus e_Z=1,T\notin\mathcal G)
 \le\tfrac12p_L+O(p^{r+1}).
\]
On the complementary records use \(h_T^2\le h_T\), and add the two bounds.
\end{proof}

\paragraph{The case \(d=3\).}
At distance \(d=3\), every realizable record $T$ is already generated by zero or one physical fault, suppressing the posterior second moment even further.
Specifically, each component gap equals \(1\) or \(3\), and gap \(3\) requires the zero syndrome with hard decision zero. For any nonempty minimum recovery, multiplying by a weight-three logical string through one of its qubits yields the opposite class with weight at most one larger, forcing gap one.
All nine records in Eq.~\eqref{eq:salem_surface_record} are realized by faults of weight at most one. The fault-free configuration produces gaps \((3,3)\), while a single \(X\) or \(Z\) fault on or off the logical boundary produces \((1,3)\) or \((3,1)\) with either hard bit. A single \(Y\) fault likewise produces \((1,1)\) with all four hard-bit pairs across the boundary corner, single boundaries, or interior.
Because all single-qubit faults are corrected, every record $t$ has $h_t=O(p)$, yielding
\[
 h_t=O(p)\ \text{for every }t,\qquad
 \E[h_T^2]=O(p^3)=o(p_L).
\]

\subsubsection{From the Second Moment to the PEC Cost}

Finally, we connect the posterior second moment $\E[h_T^2]$ to the minimum conditional Pauli-PEC overhead $C$.
Conditional Pauli inversion followed by optimal scalar weighting over $T$ yields the harmonic-mean overhead
\begin{equation}
 C=\frac{1}{\E[\gamma_T^{-2}]}.
 \label{eq:salem_harmonic_cost}
\end{equation}
For any scalar multiplier \(A(T)\) satisfying the unbiasedness condition \(\E A=1\), the Cauchy--Schwarz bound minimizes the second-moment cost \(\E[A^2\gamma_T^2]\) at \(A(t)=\gamma_t^{-2}/\E[\gamma_T^{-2}]\), giving Eq.~\eqref{eq:salem_harmonic_cost} (with singular records assigned zero weight).

Evaluating the logical \(Y\) transfer eigenvalue $\lambda_Y(t)=1-2h_t$ of $\mathcal N_t$ relates $\gamma_t$ directly to $h_t$.
Because applying the Pauli quasiprobability inverse of $\mathcal N_t$ to $Y_L$ requires $|1/\lambda_Y(t)|\le\gamma_t$, we have (including for singular channels under the zero-weight convention)
\begin{equation}
 \gamma_t^{-2}\le(1-2h_t)^2.
 \label{eq:salem_surface_eigenbound}
\end{equation}
Combining Eqs.~\eqref{eq:salem_harmonic_cost}--\eqref{eq:salem_surface_eigenbound} with Lemma~\ref{lem:salem_surface_moment} and \(y_L=O(p^{r+1})\) yields
\[
 C^{-1}
 \le 1-4\E h_T+4\E h_T^2
 \le 1-2p_L+O(p^{r+1}).
\]
Since $p_L=\Theta(p^r)$, this bound proves Theorem~\ref{thm:salem_surface_gap}, while at $d=3$ the scaling $\E[h_T^2]=O(p^3)$ gives the stronger relation $C^{-1}\le1-4p_L+O(p^3)$.

\paragraph{Scope of the conclusion.} Whereas single-decoder conditional PEC on the record $T$ in Eq.~\eqref{eq:salem_surface_record} is lower-bounded by $\eta \ge 2$, dual-decoder VEC uses a finer offline classification in the same committed recovery frame to approach $\eta \to 1$.
The lower-bound slope \(2\) provides a rigorous reference for the four-entry soft record \(T=(\widehat e_X,\widehat e_Z,\Delta_X,\Delta_Z)\) with uniform-weight complementary gaps, plotted as an analytical benchmark guide in Fig.~\ref{fig:dual_decoder}f. Extending that bound to finer syndrome records, non-uniform matching weights, or circuit-level measurement noise requires a separate analysis. By contrast, full-syndrome conditional PEC (Section~IV of Ref.~\cite{aharonov2025salem}) deterministically inverts outcomes whose conditional failure probability exceeds $1/2$, yet requires intractable $\mathsf{\#P}$-hard exact coset inference on every shot. Even then, its excess-overhead slope remains at least two relative to its own maximum-likelihood logical error rate in the binary Bhatia--Davis regime (Corollary~\ref{cor:bhatia_davis_ceiling}). Reducing the sampling overhead below the quadratic $1/F^2$ floor relative to the committed real-time logical error rate therefore requires offline re-evaluation that overturns committed recovery errors.

\subsection{Calibration Error, Population Normalization, and Remainders}
\label{subsec:app_residual_error_matrix}\label{sec:optimal_vpec_perturbation_theory}
Evaluating the filter from finitely estimated posteriors and moments instead of exact ensemble values introduces both a residual error vector and a shift in the mean weight.
All expectations in this subsection are taken over the true syndrome distribution, conditional on fixed empirical estimates (for example, from an independent learning dataset). Denote the true mean and covariance by $\mu=\E\bm p$ and $\Sigma=\Cov(\bm p)\succ0$. We write the estimated posterior as $\hat{\bm p}=\bm p+e$, the estimated mean as $\hat\mu=\mu+d$, and the estimated second-moment matrix as $\hat\Sigma=\Sigma+E$ with symmetric perturbation $E=E^\top$. The implemented linear filter is then
\begin{equation}
\hat w(s)=1-(\mu+d)^\top(\Sigma+E)^{-1}[\bm p(s)+e(s)-\mu-d].
\label{eq:app_empirical_filter_def}
\end{equation}
Define the centering discrepancy $c=\E e-d$, the error-moment shift $G=\Cov(\bm p,e)+\mu c^\top$, and the cross-moment matrix $K=\E[\bm p(\hat{\bm p}-\hat\mu)^\top]=\Sigma+G$. These matrices determine the raw error coefficient vector and population mean weight of $\hat w(s)$,
\begin{align}
r_{\rm raw}&=\E[\hat w\bm p]=\mu-K(\Sigma+E)^{-1}(\mu+d),\label{eq:app_pres_compact_K}\\
b_{\rm raw}&=\E\hat w=1-(\mu+d)^\top(\Sigma+E)^{-1}c.
\label{eq:app_K_exact_decomp}
\end{align}
\aliaslabel{eq:app_pres_def_integral}{\ref{eq:app_pres_compact_K}}
Both identities hold without approximation. If $\hat\mu$ is the sample mean of $\hat{\bm p}$ with finite-sample fluctuation $\xi=\hat\mu-\E\hat{\bm p}$, then $d=\E e+\xi$ and $c=-\xi$ is generally non-zero, showing that empirical centering on a learning batch does not automatically enforce unit mean weight on the evaluation law.

The resulting signed raw channel $b_{\rm raw}\Id+\sum_\alpha(r_{\rm raw})_\alpha(\mathcal E_\alpha-\Id)$ has a total observable bias governed by both the mean-weight drift and the residual error norm. For any target operator satisfying $\|O\|_\infty\le1$, its bias is bounded by
\begin{equation}
|b_{\rm raw}-1|+2\|r_{\rm raw}\|_1.
\label{eq:raw_calibration_bias}
\end{equation}
Dividing by a known non-zero population mean $b_{\rm raw}$ yields the population-normalized error vector $r_{\rm raw}/b_{\rm raw}$. Alternatively, dividing by the empirical sample-mean weight defines a finite-sample ratio estimator $\bar X/\bar Y$ for $X=\hat w m$ and $Y=\hat w$. Letting $b=\E Y\ne0$ and $\theta=\E X/b$, the regular large-$N$ moments of this ratio estimator obey
\[
\operatorname{Bias}(\bar X/\bar Y)=
\frac{\theta\Var(Y)-\Cov(X,Y)}{Nb^2}+o(N^{-1}),\qquad
\Var(\bar X/\bar Y)=\frac{\Var(X-\theta Y)}{Nb^2}+o(N^{-1}).
\]
Because signed weights can in principle produce a small or vanishing empirical mean on exceptional samples, rigorous unconditional control at finite $N$ can instead be achieved by averaging outcomes within each occupied sector using the empty-sector convention below.

\begin{proposition}[Finite-sample sector means]
\label{prop:sector_sample_means}
Fix the sector partition and coefficients $c_k$ independently of $N$ independent and identically distributed evaluation shots, with $\sum_kc_k=1$ and $\pi_k>0$ for each retained sector. Let $N_k$ denote the sector counts, $\vartheta_k=\E[m\mid k]$ the conditional sector means, and $\sigma_k^2=\Var(m\mid k)$ the within-sector variances. Set the empirical sector mean to $\bar m_k=N_k^{-1}\sum_{i:\,s_i\in\mathcal S_k}m_i$ when $N_k>0$ and $\bar m_k=0$ otherwise. Then the estimator $\hat\vartheta=\sum_kc_k\bar m_k$ obeys
\begin{align}
\E\hat\vartheta&=\sum_kc_k\vartheta_k[1-(1-\pi_k)^N],\label{eq:sector_empty_expectation}\\
\Var(\hat\vartheta)&=\sum_kc_k^2\sigma_k^2\E\!\left[\frac{\mathbf1_{N_k>0}}{N_k}\right]
+\Var\!\left(\sum_kc_k\vartheta_k\mathbf1_{N_k>0}\right),\label{eq:sector_sample_variance}
\end{align}
where the fraction is defined as zero at $N_k=0$. For bounded outcomes $|m|\le1$, its bias relative to the calibrated target $\vartheta_c=\sum_kc_k\vartheta_k$ is at most $\sum_k|c_k|(1-\pi_k)^N$, while the total bias relative to the ideal expectation additionally includes $\vartheta_c-\vartheta_{\rm ideal}$.
\end{proposition}
\begin{proof}
Conditional on the sector labels, each occupied sector mean has expectation $\vartheta_k$ and variance $\sigma_k^2/N_k$, and the outcome samples in distinct sectors are independent. Applying the laws of total expectation and total variance together with $\Pr(N_k=0)=(1-\pi_k)^N$ yields Eqs.~\eqref{eq:sector_empty_expectation}--\eqref{eq:sector_sample_variance}.
\end{proof}
On evaluation batches where all required sectors are occupied, this sector-mean estimator coincides with weighting each shot by $c_k/(N_k/N)$ and automatically enforces unit sample-mean weight. Eqs.~\eqref{eq:sector_empty_expectation}--\eqref{eq:sector_sample_variance} thus provide exact unconditional finite-sample moments whenever the coefficients $c_k$ and sector partition are fixed on an independent calibration split.

\subsubsection{Exact Matrix Perturbation Control}
\label{subsubsec:app_neumann_inversion}
When the covariance perturbation $E\Sigma^{-1}$ is smaller than unity in norm, a Neumann series expansion bounds the nonlinear remainder between the exact residual $r_{\rm raw}$ and its first-order linear approximation.
Using any fixed compatible submultiplicative matrix norm, define the relative deviation matrices $A=E\Sigma^{-1}$ and $T=(E-G)\Sigma^{-1}$, and assume $\eta=\|A\|<1$. Since $(\Sigma+E)^{-1}=\Sigma^{-1}(I+A)^{-1}$, the exact residual vector and its linear expansion take the form
\begin{equation}
r_{\rm raw}=-d+T(I+A)^{-1}(\mu+d),\qquad r_1=-d+T\mu,
\label{eq:app_optimal_taylor_leading_order_recovered}
\end{equation}
\aliaslabel{eq:optimal_taylor_leading_order}{\ref{eq:app_optimal_taylor_leading_order_recovered}}
and the higher-order remainder satisfies
\begin{equation}
\|r_{\rm raw}-r_1\|\le
\frac{\|T\|}{1-\eta}(\|d\|+\eta\|\mu\|),
\label{eq:app_sigma_neumann}
\end{equation}
because $r_{\rm raw}-r_1=T(I+A)^{-1}(d-A\mu)$ and the inverse remainder obeys
\[
\|(\Sigma+E)^{-1}-\Sigma^{-1}+\Sigma^{-1}E\Sigma^{-1}\|
\le\|\Sigma^{-1}\|\frac{\eta^2}{1-\eta}.
\]
Achieving a quadratic remainder of order $O(\epsilon^2\|\mu\|)$ therefore requires uniform relative bounds $\eta,\|T\|=O(\epsilon)$ and $\|d\|=O(\epsilon\|\mu\|)$ in the chosen norm, alongside separate control of the population mean weight $b_{\rm raw}$.

\subsubsection{Scalar Calibration and the Variance-Dependent Remainder}
\label{subsec:taylor_expansion_derivations}
Specializing to a single scalar error channel makes the dependence of the calibration remainder on the posterior variance $D=\Var(q)$ explicit.
Let $q(s)=p_L(s)$ denote the true conditional error rate, with mean $\mu=p_L=\E q$ and variance $D=\Var(q)>0$. Write the estimated posterior as $\hat q=q(1+\varepsilon(s))$, the estimated mean as $\hat\mu=\mu(1+a)$, and the estimated variance as $\hat D=D(1+b)$, where $a=\varepsilon_L$ and $b=\varepsilon_D$ denote the relative mean and variance shifts, and define
\begin{equation}
c=\bar\varepsilon^{(2)}=D^{-1}\E[q^2\varepsilon(s)],\qquad h=\mu^2/D.
\label{eq:eps_bar_2_def_sec}
\end{equation}
For the implemented scalar filter $\hat w=1-\hat\mu(\hat q-\hat\mu)/\hat D$, the raw residual error rate evaluates exactly to
\begin{equation}
p_{\rm res}:=\E[\hat wq]
=\mu\left[1-\frac{(1+a)(1+c-ha)}{1+b}\right],
\label{eq:w_hat_star_def}
\end{equation}
\aliaslabel{eq:p_res_cov_rep}{\ref{eq:w_hat_star_def}}
expanding in powers of the relative estimation errors as follows.
\begin{proposition}[Scalar Taylor expansion with controlled dispersion]
\label{thm:taylor_expansion_full}
For $\epsilon=\max(\|\varepsilon\|_\infty,|a|,|b|)\le1/2$, the residual error rate decomposes as
\begin{align}
p_{\rm res}&=\mu(b-a-c)+a\mu^3/D+R_2,\label{eq:taylor_exact_expansion_sec}\\
R_2&=\mu\frac{(b-a)(c-ha-b)}{1+b},\qquad
|R_2|\le8\epsilon^2(\mu+\mu^3/D).
\label{eq:R2_bound_sec}
\end{align}
If $D\ge x_0\mu$ with fixed $x_0>0$, the remainder bound simplifies to $8(1+x_0^{-1})\mu\epsilon^2$, and the variance-dependent linear term obeys $|a|\mu^3/D\le\epsilon\mu^2/x_0$.
\end{proposition}
\begin{proof}
Subtracting the displayed linear terms from the exact rational expression in Eq.~\eqref{eq:w_hat_star_def} and collecting the numerator yields $R_2$. Because $|c|\le\epsilon\E q^2/D=\epsilon(1+h)$ and $|b-a|\le2\epsilon$, while the numerator factor obeys $|c-ha-b|\le2\epsilon(1+h)$ and the denominator obeys $1/(1+b)\le2$, the bound in Eq.~\eqref{eq:R2_bound_sec} follows immediately. The Bhatia--Davis inequality bounds $D$ from above but not from below; for $b=0$ and $\varepsilon(s)=0$, the remainder equals $a^2\mu^3/D$, so a variance-independent constant multiplying $\mu\epsilon^2$ is impossible in general.
\end{proof}
Consequently, the scalar residual error obeys the perturbative expansion
\begin{equation}
p_{\rm res}=\mu(b-a-c)+O\bigl(\epsilon\mu^3/D+\epsilon^2(\mu+\mu^3/D)\bigr),
\label{eq:taylor_leading_order_sec}
\end{equation}
where the $O(\epsilon\mu^3/D)$ term is suppressed to $O(\epsilon\mu^2)$ whenever the dispersion ratio $D/\mu$ remains bounded away from zero.

\subsection{Scale Invariance Under Global Multiplicative Miscalibration}
\label{subsec:scale_invariance_derivations}

Offline decoders can experience global scale drift in which all estimated error probabilities are multiplied by a common factor $c > 0$. We now show that when the rescaled mean $\hat{p}_L$ and variance $\hat{D}$ are computed self-consistently from the rescaled posteriors $\hat{p}_L(\bm{s}) = c\,p_L(\bm{s})$, the optimal filter $\hat{w}^*(\bm{s})$ is invariant for all $c > 0$ and cancels logical bias without perturbative remainder.

\begin{proposition}[Scale Invariance of Optimal VEC]
\label{thm:scale_invariance_theorem}
Let $w^*(\bm{s})$ be the optimal affine filter for true moments $(p_L(\bm{s}), p_L, D)$. Suppose an offline decoder outputs uniformly rescaled posteriors with scale factor $c > 0$,
\begin{equation}
\hat{p}_L(\bm{s}) = c \cdot p_L(\bm{s}) \qquad \text{for all } \bm{s} \in \Omega.
\label{eq:scale_transformation_sec}
\end{equation}
Suppose the rescaled mean $\hat{p}_L \equiv \EE_{\bm{s}}[\hat{p}_L(\bm{s})]$ and variance $\hat{D} \equiv \Var_{\bm{s}}[\hat{p}_L(\bm{s})]$ are evaluated from $\hat{p}_L(\bm{s})$. The resulting weights then match the exact weights pointwise,
\begin{equation}
\hat{w}^*(\bm{s}) = w^*(\bm{s}) \qquad \text{for all } \bm{s} \in \Omega \text{ and all } c > 0.
\label{eq:weights_identical_sec}
\end{equation}
The residual logical error under global multiplicative miscalibration therefore vanishes to all orders in $c - 1$,
\begin{equation}
p_{\mathrm{res}} = 0.
\label{eq:pres_identically_zero_sec}
\end{equation}
\end{proposition}

\begin{proof}
Under the rescaling in Eq.~\eqref{eq:scale_transformation_sec}, linearity of the expectation gives the rescaled mean
\begin{equation}
\hat{p}_L = \EE_{\bm{s}}[\hat{p}_L(\bm{s})] = \EE_{\bm{s}}[c \cdot p_L(\bm{s})] = c \cdot \EE_{\bm{s}}[p_L(\bm{s})] = c \cdot p_L,
\label{eq:scale_mean_sec}
\end{equation}
while quadratic homogeneity of the variance gives
\begin{equation}
\hat{D} = \Var_{\bm{s}}[\hat{p}_L(\bm{s})] = \Var_{\bm{s}}[c \cdot p_L(\bm{s})] = c^2 \cdot \Var_{\bm{s}}[p_L(\bm{s})] = c^2 \cdot D.
\label{eq:scale_var_sec}
\end{equation}
Substituting Eqs.~\eqref{eq:scale_mean_sec}--\eqref{eq:scale_var_sec} into Eq.~\eqref{eq:w_hat_star_def} cancels the factor $c^2$ between numerator and denominator,
\begin{align}
\hat{w}^*(\bm{s}) &= 1 - \frac{\hat{p}_L}{\hat{D}}\big(\hat{p}_L(\bm{s}) - \hat{p}_L\big) \nonumber \\
&= 1 - \frac{c \cdot p_L}{c^2 \cdot D}\big(c \cdot p_L(\bm{s}) - c \cdot p_L\big) \nonumber \\
&= 1 - \frac{p_L}{D}\big(p_L(\bm{s}) - p_L\big) = w^*(\bm{s}).
\label{eq:scale_cancellation_algebra}
\end{align}
Because $\hat{w}^*(\bm{s}) = w^*(\bm{s})$ for every $\bm{s} \in \Omega$, the error-cancellation identity $\EE_{\bm{s}}[w^*(\bm{s}) p_L(\bm{s})] = 0$ yields
\begin{equation}
p_{\mathrm{res}} = \EE_{\bm{s}}[\hat{w}^*(\bm{s}) p_L(\bm{s})] = \EE_{\bm{s}}[w^*(\bm{s}) p_L(\bm{s})] = 0,
\label{eq:pres_zero_exact_sec}
\end{equation}
establishing Eqs.~\eqref{eq:weights_identical_sec}--\eqref{eq:pres_identically_zero_sec}.
\end{proof}

\begin{corollary}[Perturbative Consistency with Taylor Expansion]
\label{cor:perturbative_consistency_sec}
Let $c = 1 + \varepsilon$ with $|\varepsilon| \ll 1$ in Eq.~\eqref{eq:scale_transformation_sec}. Then the first-order terms in Eq.~\eqref{eq:taylor_exact_expansion_sec} sum to zero.
\end{corollary}

\begin{proof}
For $c = 1 + \varepsilon$, Eqs.~\eqref{eq:scale_mean_sec}--\eqref{eq:scale_var_sec} give the relative shifts
\begin{equation}
\varepsilon_L = \varepsilon, \qquad \varepsilon_D = (1 + \varepsilon)^2 - 1 = 2\varepsilon + \varepsilon^2.
\label{eq:pert_eval_c}
\end{equation}
Because the syndrome perturbation is constant ($\varepsilon(\bm{s}) = \varepsilon$), the second-moment parameter in Eq.~\eqref{eq:eps_bar_2_def_sec} equals
\begin{equation}
\bar{\varepsilon}^{(2)} = \frac{1}{D} \EE_{\bm{s}}[p_L(\bm{s})^2 \cdot \varepsilon] = \varepsilon \frac{\EE_{\bm{s}}[p_L(\bm{s})^2]}{D} = \varepsilon \frac{D + p_L^2}{D} = \varepsilon \left(1 + \frac{p_L^2}{D}\right).
\label{eq:epsbar_eval_c}
\end{equation}
Combining Eqs.~\eqref{eq:pert_eval_c}--\eqref{eq:epsbar_eval_c} to first order in $\varepsilon$ gives
\begin{equation}
\varepsilon_D - \varepsilon_L - \bar{\varepsilon}^{(2)} = 2\varepsilon - \varepsilon - \varepsilon \left(1 + \frac{p_L^2}{D}\right) + \mathcal{O}(\varepsilon^2) = -\varepsilon \frac{p_L^2}{D}+\mathcal O(\varepsilon^2).
\label{eq:bracket_cancel_eval}
\end{equation}
Multiplying Eq.~\eqref{eq:bracket_cancel_eval} by $p_L$ and adding the variance-dependent term $\varepsilon_L p_L^3 / D = \varepsilon p_L^3 / D$ in Eq.~\eqref{eq:taylor_exact_expansion_sec} cancels the linear coefficient in $\varepsilon$,
\begin{equation}
\left[p_L (\varepsilon_D - \varepsilon_L - \bar{\varepsilon}^{(2)}) + \varepsilon_L \frac{p_L^3}{D}\right]_{\text{linear in }\varepsilon} = -\varepsilon \frac{p_L^3}{D} + \varepsilon \frac{p_L^3}{D} = 0.
\label{eq:exact_zero_sum_sec}
\end{equation}
This cancellation confirms consistency with Proposition~\ref{thm:scale_invariance_theorem} at first order in $\varepsilon$.
\end{proof}

\section{Conditional Spectral Scaling and Code-Specific Overhead Laws}
\label{sec:appendix_spectral_scaling_mahalanobis}

This section derives the low-noise scaling laws for the syndrome-resolved sampling overhead $C_{\mathrm{syn}}^*$ in Eq.~\eqref{eq:overhead_gap_law} of the main text.
We take the physical error rate $p\to0$ while holding the stabilizer code, finite fault and syndrome spaces, committed recovery, and retained syndrome record fixed. Unless uniform bounds are explicitly established, the leading constants and small-$p$ threshold may therefore depend on the code distance $d$.

\subsection{Coset Probabilities and Conditional Spectral Scaling}
\label{subsec:app_partition_function_universal_gap}
\label{subsec:app_variational_rayleigh_eves_law}
To analyze the quadratic form $C_{\rm syn}^*-1=\mu^\top\Sigma^{-1}\mu$ for $\mu=\bar{\bm p}$ and $\Sigma=\Cov(\bm p)\succ0$, we first record its spectral Rayleigh bounds and exact block Schur-complement decomposition.
\begin{proposition}[Rayleigh bounds and the exact block identity]
\label{prop:app_rayleigh_bounds}
For the eigenpairs $(\sigma_j,v_j)$ of $\Sigma$, the excess overhead satisfies
\begin{equation}
\frac{\|\mu\|_2^2}{\sigma_{\max}}\le C_{\rm syn}^*-1
=\sum_j\frac{(v_j^\top\mu)^2}{\sigma_j}
\le\frac{\|\mu\|_2^2}{\sigma_{\min}}.
\label{eq:app_rayleigh_sandwich}
\end{equation}
\aliaslabel{eq:app_rayleigh_mode_sum}{\ref{eq:app_rayleigh_sandwich}}
Partition the mean vector as $\mu=(a^\top,b^\top)^\top$ and the covariance matrix as $\Sigma=\left(\begin{smallmatrix}A&B\\B^\top&D\end{smallmatrix}\right)$. With the Schur complement $H=D-B^\top A^{-1}B\succ0$ and residual mean vector $h=b-B^\top A^{-1}a$, the excess overhead decomposes into
\begin{equation}
C_{\rm syn}^*-1=a^\top A^{-1}a+h^\top H^{-1}h,
\label{eq:app_schur_complement_lower_bound}
\end{equation}
where an empty complementary block contributes zero.
\end{proposition}
\begin{proof}
The spectral resolution of $\Sigma^{-1}$ gives the first identity and Rayleigh bounds. Completing the square on the partitioned quadratic form using the Schur complement gives the second identity. Because a logical coordinate is not generally an eigenmode of $\Sigma$, smaller complementary means alone do not control the second term in Eq.~\eqref{eq:app_schur_complement_lower_bound}.
\end{proof}

We next expand the coset partition functions and conditional posteriors in powers of the physical error rate $p$.
Assume that every supported fault $E$ occurs with probability $a_Ep^{w(E)}[1+O(p)]$ for a positive prefactor $a_E>0$ and an integer probability order $w(E)\ge0$. Under full-support depolarizing noise $w(E)$ equals the Pauli weight, whereas under independent-component noise a single-qubit $Y$ fault has order two. Fix the committed recovery (or its stabilized small-$p$ class choice) and label each coset by its residual logical Pauli $L$ in that committed frame, with $L=0$ denoting the coset selected by the recovery. The coset partition functions and posterior error rates then obey
\begin{align}
Z_L(s)&=a_L(s)p^{W_L(s)}[1+O(p)],\quad W_*(s)=\min_JW_J(s),\label{eq:app_boltzmann_partition_function}\\
P(s)&=\sum_JZ_J(s)=\Theta(p^{W_*(s)}),\qquad
p_L(s)=\frac{Z_L(s)}{P(s)}=\Theta(p^{W_L(s)-W_*(s)}),
\label{eq:app_posterior_bayes}
\end{align}
\aliaslabel{eq:app_posterior_gap_expansion}{\ref{eq:app_posterior_bayes}}
where unsupported cosets are omitted and the subscript $L$ labels a specific logical error class (while the total logical failure probability remains $p_L=\sum_{L\ne0}\mu_L$).
Define the coset-specific and overall minimum uncorrectable fault weights together with the dominant error subset as
\begin{equation}
r_L=\min_sW_L(s),\quad r=\min_{L\ne0}r_L,\quad
\mathcal L_{\rm dom}=\{L:r_L=r\}.
\label{eq:app_coset_distance_fail_weight}
\end{equation}
\aliaslabel{eq:app_dominant_subspace_manifold}{\ref{eq:app_coset_distance_fail_weight}}
By Bayes' rule $P(s)p_L(s)=Z_L(s)$ and positivity of the leading coefficients, the marginal error rates scale as $\mu_L=\Theta(p^{r_L})$ and $p_L=\Theta(p^r)$. When the committed recovery selects a lowest-weight coset ($W_0\le W_L$), the effective class distance $d_L=\min_s(W_0+W_L)$ implies $r_L\ge\lceil d_L/2\rceil$ in general, although strict equality need not hold for every multi-block logical Pauli coset. Under minimum-weight $X$ recovery, for example, three independent distance-three Steane blocks have coset distance $d_{\bar X_1\bar X_2\bar X_3}=9$, yet their minimum uncorrectable fault weight is $6$ instead of $5$.

In contrast, the second moments of the conditional posteriors retain the syndrome probability in the denominator and are governed by the coset weight gaps $W_L(s)-W_*(s)$.
Under full-support code-capacity Pauli noise and minimum-Pauli-weight recovery, the overall failure order is $r=\lceil d/2\rceil$. Any failing fault $E$ and recovery $R$ obey $d\le|ER|\le2|E|$, while splitting a minimum-weight logical operator into pieces of sizes $\lfloor d/2\rfloor$ and $\lceil d/2\rceil$ produces two indistinguishable logical classes, at least one of which fails. Summing over all syndromes yields the exact second moments
\begin{equation}
\E[p_L(s)p_M(s)]=\sum_s\frac{Z_L(s)Z_M(s)}{P(s)},\qquad
t_{LM}=\min_s[W_L(s)+W_M(s)-W_*(s)].
\label{eq:app_second_moment_gap}
\end{equation}
These moments scale as $\Theta(p^{t_{LM}})$ whenever the leading sum is non-zero.

\begin{theorem}[Conditional spectral scaling]
\label{thm:universal_gap_scaling_main}
Use the dominant coordinates $\mathcal L_{\rm dom}$ from Eq.~\eqref{eq:app_coset_distance_fail_weight} for the block entries $a$ and $A$. Suppose the exponents satisfy $r>\Delta\ge0$, the dominant mean vector scales as $a=p^r(m+o(1))$ with $m\ne0$, and the second-moment and Schur blocks obey
\begin{equation}
\E[\bm p_{\rm dom}\bm p_{\rm dom}^\top]
=p^{r+\Delta}(G+o(1)),\quad G\succ0,\qquad
h^\top H^{-1}h=O(p^{r-\Delta}),
\label{eq:app_dominant_cov_decomposition}
\end{equation}
where matrix remainders are measured in the operator norm. Then the dominant covariance eigenvalue and excess sampling overhead satisfy
\begin{equation}
\lambda_{\min}(A)=\Theta(p^{r+\Delta}),\qquad
C_{\rm syn}^*-1=\Theta(p^{r-\Delta})=\Theta(p_L/p^\Delta).
\label{eq:app_sigma_min_universal}
\end{equation}
\aliaslabel{eq:app_overhead_universal_gap}{\ref{eq:app_sigma_min_universal}}
\end{theorem}
\begin{proof}
Because $2r>r+\Delta$, subtracting the outer product $aa^\top$ gives $A=p^{r+\Delta}(G+o(1))$, and hence $a^\top A^{-1}a=p^{r-\Delta}(m^\top G^{-1}m+o(1))$. Positive definiteness of $G$ together with the non-negative, bounded Schur term in Eq.~\eqref{eq:app_dominant_cov_decomposition} proves both claims, whereas diagonal moment orders alone do not imply $G\succ0$.
\end{proof}
In practice, one verifies the leading Gram condition $G\succ0$ by checking that the rescaled dominant posterior vectors $v_s=\lim_{p\to0}p^{-(r+\Delta)/2}\sqrt{P(s)}\bm p_{\rm dom}(s)$ span the dominant subspace via $G=\sum_sv_sv_s^\top$. The Schur condition in Eq.~\eqref{eq:app_dominant_cov_decomposition} then bounds any subleading logical error cosets after projecting out their linear dependence on the dominant coordinates.
As a concrete scalar example with no complementary block, consider the three-bit repetition code under independent bit flips and perfect measurements. Its zero syndrome occurs with probability $P_0=(1-p)^3+p^3$ and posterior error rate $p^3/P_0$, while each of the three non-zero syndromes occurs with probability $p(1-p)$ and posterior error rate $p$. Evaluating the moments gives the mean $\mu=3p^2-2p^3$, variance $\Sigma=3p^3+O(p^4)$, and excess overhead $C_{\rm syn}^*-1=3p+O(p^2)$, confirming the exponents $r=2$ and $\Delta=1$.

\subsubsection{Committed Heuristic Recoveries}
\label{subsec:app_decoder_universality}
\aliaslabel{rem:app_decoder_universality}{\ref{subsec:app_partition_function_universal_gap}}
Theorem~\ref{thm:universal_gap_scaling_main} applies to any fixed committed real-time decoder once its mean, Gram, and Schur conditions are evaluated in the committed frame.
For a general heuristic decoder, the posterior exponent on syndrome $s$ is $W_L(s)-W_*(s)$ instead of $\max(0,W_L(s)-W_{\rm chosen}(s))$. A decoding failure occurs whenever a fault lands in a non-chosen coset $L\ne0$, even if the decoder selected the most probable coset $W_0(s)=W_*(s)$ (as in weight-two failures of the three-bit repetition code). Offline inference therefore always evaluates the posterior $\bm p(\bm s)$ relative to the recovery frame committed in real time.

\subsection{Independent Calderbank--Shor--Steane (CSS) Stabilizer Bases and Genuine Minimum-Weight Ties}
\label{subsec:app_algebraic_topology_syndrome_spaces}
\begin{proposition}[Overhead for independent CSS stabilizer bases]
\label{thm:app_linear_parity_invariance}
Fix a CSS code encoding one logical qubit with perfect measurements and basis distances $d_X,d_Z\ge2$. Suppose physical $X$ and $Z$ fault components are drawn independently of each other and across sites with rates $c_{j,i}p+O(p^2)$ for fixed constants $c_{j,i}>0$, and decoded by separate minimum-Hamming-weight recovery with fixed tie-breaking in each basis $j\in\{X,Z\}$. Set $r_j=\lceil d_j/2\rceil$, and let $\delta_j=0$ if some syndrome in basis $j$ has two logical-class minimum weights both equal to $r_j$, and $\delta_j=1$ otherwise. Writing $\mu_j=\E q_j$ and $D_j=\Var(q_j)$ for the scalar failure posterior $q_j$, the basis moments and total overhead satisfy
\begin{align}
\mu_j&=\Theta(p^{r_j}),\quad D_j=\Theta(p^{r_j+\delta_j}),\quad
C_j^*=1+\mu_j^2/D_j,\nonumber\\
C_{\rm syn}^*&=C_X^*C_Z^*,\qquad
C_{\rm syn}^*-1=\Theta(p^{r_X-\delta_X}+p^{r_Z-\delta_Z}).
\label{eq:app_linear_gapped_scaling}
\end{align}
An even distance $d_j$ forces $\delta_j=0$, whereas an odd distance $d_j$ with even-weight same-type stabilizers forces $\delta_j=1$; other odd-distance codes can realize either value.
\end{proposition}
\begin{proof}
Each syndrome in a single CSS basis has two logical classes with minimum weights $w_0\le w_1$ satisfying $w_0+w_1\ge d$, so $w_1\ge r$. Splitting a minimum-weight logical operator into complementary pieces of sizes $\lfloor d/2\rfloor$ and $\lceil d/2\rceil$ produces two classes whose minima sum to at least $d$ and therefore equal $(r,r)$ for even $d$ or $(r-1,r)$ for odd $d$.
Positive coset coefficients give $Pq=\Theta(p^{w_1})$ and $Pq^2=\Theta(p^{2w_1-w_0})$. Thus the smallest first-moment exponent is $r$, while the smallest second-moment exponent is $r$ if and only if a leading tie $(r,r)$ exists (and $r+1$ otherwise, attained by the $(r-1,r)$ witness). Because $r+\delta<2r$ for $d\ge2$, subtracting $\mu^2$ preserves the leading order of $D_j$.

Next, set $u=q_X$ and $v=q_Z$. Independence of the noise components, syndrome maps, and separate recoveries makes $u$ and $v$ statistically independent, so the full Pauli posterior in $(X,Z,Y)$ order is $(u(1-v),(1-u)v,uv)$. Error cancellation is then equivalent to $\E w=1$ together with $\E[wu]=\E[wv]=\E[wuv]=0$. The product of the two affine optimal scalar filters satisfies these four constraints and lies in $\operatorname{span}\{1,u,v,uv\}$; because its difference from any feasible filter is orthogonal to this span, it achieves the minimum squared $L^2$ norm and yields $C_{\rm syn}^*=C_X^*C_Z^*$. Since both excess-cost exponents are positive, their product is subleading to their sum.

Finally, any two binary vectors $e$ and $f$ satisfy $|e+f|=|e|+|f|-2|\supp(e)\cap\supp(f)|$. When the logical representative has odd weight and all same-type stabilizers have even weight, every sum of representatives from opposite classes has odd weight, forbidding a weight tie $w_0=w_1$.
\end{proof}
Proposition~\ref{thm:app_linear_parity_invariance} shows how parity constraints govern the gap exponent $\Delta\in\{0,1\}$ under independent $X$ and $Z$ noise, where $p_L=\mu_X+\mu_Z-\mu_X\mu_Z$ and physical $Y$ faults occur at order $p^2$. When both bases have $\delta_j=1$, the excess overhead scales as $C_{\rm syn}^*-1=\Theta(p_L/p)$; when both have $\delta_j=0$, it scales linearly as $\Theta(p_L)$, with the exact product factorization $C_{\rm syn}^*=C_X^*C_Z^*$ automatically controlling the subleading $Y$ coset.

\begin{proposition}[Genuine leading ties and spanning]
\label{thm:app_symplectic_gap_closing}
For a minimum-weight recovery with $r>0$, let $\mathcal T_r$ denote the set of syndromes satisfying $W_0=W_*=r$ alongside a dominant competing class of weight $r$. Each syndrome $s\in\mathcal T_r$ has probability $P(s)=c(s)p^r[1+O(p)]$ and a limiting dominant posterior vector $q(s)$ determined by the ratios of the tied leading coset coefficients. If the tie Gram matrix is positive definite,
\begin{equation}
G_{\rm tie}=\sum_{s\in\mathcal T_r}c(s)q(s)q(s)^\top\succ0,
\label{eq:app_symplectic_gapless_scaling}
\end{equation}
and the complementary Schur condition of Theorem~\ref{thm:universal_gap_scaling_main} holds with $\Delta=0$, then $C_{\rm syn}^*-1=\Theta(p_L)$.
\end{proposition}
\begin{proof}
Bayes' rule gives the stated syndrome probabilities and limiting posteriors. Restricting the second-moment sum to $\mathcal T_r$ contributes $p^r(G_{\rm tie}+O(p))$, while $\E\|\bm p_{\rm dom}\|_2^2\le p_L=O(p^r)$ bounds the full sum from above. The finite coset expansions therefore yield a full leading Gram matrix $G\succeq G_{\rm tie}\succ0$, and Theorem~\ref{thm:universal_gap_scaling_main} gives $C_{\rm syn}^*-1=\Theta(p_L)$, although a single tie direction need not span a multi-coset dominant subspace.
\end{proof}
Under depolarizing noise, single-qubit $Y$ faults of order $p$ couple the $X$ and $Z$ checks and can create genuine weight-$r$ coset ties even on odd-distance codes. For Pauli strings $E_0$ and $E_1$ modulo phase, the Pauli weights obey
\begin{equation}
|E_0|+|E_1|=|E_0E_1|+k_{\rm cross}+2k_{\rm same},
\label{eq:app_symplectic_sum_identity}
\end{equation}
where $k_{\rm cross}$ and $k_{\rm same}$ count sites with distinct and identical non-identity Pauli operators, respectively. While $k_{\rm cross}>0$ allows two weight-$r$ Pauli strings $E_0$ and $E_1$ to multiply to an odd-weight logical operator $E_0E_1$, establishing a genuine leading tie on syndrome $s$ also requires verifying that no lower-weight fault shares that syndrome. The perfect five-qubit code (with stabilizer group generated by cyclic shifts of $XZZXI$) exhausts all $15$ non-zero syndromes at weight one, preventing any weight-two tie.

An explicit calculation on the distance-three rotated surface code confirms that depolarizing noise produces genuine weight-two ties that span the dominant $(X,Z)$ error subspace and yield $\Delta=0$. Number the nine data qubits row-wise from $1$ to $9$, with logical operators $\bar X=X_1X_2X_3$ and $\bar Z=Z_1Z_4Z_7$. Its four $Z$ checks act on the bulk faces $\{1,2,4,5\}$ and $\{5,6,8,9\}$ and boundary pairs $\{2,3\}$ and $\{7,8\}$. Its four $X$ checks act on the bulk plaquettes $\{2,3,5,6\}$ and $\{4,5,7,8\}$ and perimeter pairs $\{1,4\}$ and $\{6,9\}$. Under independent and identically distributed single-qubit Pauli noise and asymptotic joint maximum-likelihood recovery, the two-fault pair $\{X_1Z_2, Y_2X_3\}$ shares the syndrome $(1000;1000)$, while the pair $\{Z_2X_6, Y_6Z_9\}$ shares $(0100;1000)$. Neither syndrome is reachable at weight zero or one. In $(I,X,Z,Y)$ order, their competing coset weight minima are $(2,2,3,3)$ and $(2,3,2,3)$ with tied multiplicities $(4,2)$. Each syndrome therefore has leading probability coefficient $c=2/3$ and dominant $(X,Z)$ posterior vector $(1/3,0)^\top$ or $(0,1/3)^\top$, yielding a positive-definite tie Gram matrix $G_{\rm tie}\succeq(2/27)I_2\succ0$.
Summing over all $256$ syndromes gives the exact leading moments
\begin{align}
a&=\tfrac{20}{3}p^2(1,1)^\top+O(p^3),&
b&=\tfrac{296}{27}p^3+O(p^4),\nonumber\\
A&=\tfrac{14}{9}p^2I_2+O(p^3),&
B&=\tfrac{152}{81}p^3(1,1)^\top+O(p^4),\nonumber\\
D&=\tfrac{16}{27}p^3+O(p^4).
\label{eq:app_surface_code_anisotropy}
\end{align}
From these blocks we obtain the Schur contribution $h^\top H^{-1}h=O(p^3)$ and the total logical error rate $p_L=(40/3)p^2+O(p^3)$, yielding the excess sampling overhead
\begin{equation}
C_{\rm syn}^*-1=\frac{400}{7}p^2+O(p^3)=\Theta(p_L).
\label{eq:app_surface_code_d3_overhead}
\end{equation}
This exact $d = 3$ result confirms Proposition~\ref{thm:app_symplectic_gap_closing} analytically and agrees with our numerical simulations across $d \in \{3, 5, 7\}$ in Supplementary Fig.~\ref{fig:supp_dual_decoder_2d}. Circuit-level spacetime fault hypergraphs, by contrast, require the separate verification criteria of Remark~\ref{thm:app_gauge_and_hypergraphs}.

\begin{remark}[Gauge and circuit-fault criteria]
\label{thm:app_gauge_and_hypergraphs}
Subsystem gauge freedom and hook-error geometry can generate candidate weight ties, but do not by themselves establish genuine minimum-weight ties or subspace spanning. Suppose two Pauli faults satisfy $E_0E_1=LG$ for a bare logical operator $L$ commuting with a gauge operator $G$. Eq.~\eqref{eq:app_symplectic_sum_identity} then gives $|E_0|+|E_1|\equiv|LG|+k_{\rm cross}(E_0,E_1)\pmod2$ with $|LG|\equiv|L|+|G|\pmod2$; because physical failure is governed by the dressed distance, replacing $|L|$ with the dressed distance requires a separate parity argument. In a binary circuit-level detector error model, each hook error is a single fault mechanism regardless of its geometric span, and the exact weight identity is $|f_0|+|f_1|=|f_0+f_1|+2|\supp(f_0)\cap\supp(f_1)|$. In either setting, the allowed fault combinations, probability orders, effective distance, genuine leading ties, and Gram/Schur covariance conditions of Theorem~\ref{thm:universal_gap_scaling_main} must be verified directly in the specified noise model.
\end{remark}

\begin{corollary}[Conditional large-distance suppression]
\label{cor:uniform_distance_overhead}
At a fixed physical error rate $p$, suppose a code family satisfies $p_{L,d}\le A_0e^{-\beta d}$ and $C_d-1\le Kp_{L,d}$ with constants $A_0,K<\infty$ and $\beta>0$ independent of the code distance $d$. Then the excess sampling overhead decays exponentially with distance, $C_d-1\le KA_0e^{-\beta d}$.
\end{corollary}
\begin{proof}
Substituting $p_{L,d}\le A_0e^{-\beta d}$ into $C_d-1\le Kp_{L,d}$ yields the claim immediately. For the syndrome-resolved overhead, sufficient uniform conditions are $A_d\succeq c p_{L,d}I$ and $h_d^\top H_d^{-1}h_d\le T p_{L,d}$, which give $K=c^{-1}+T$ because $\|a_d\|_2\le p_{L,d}$; for a coarse-grained partition, Theorem~\ref{thm:app_asymptotic_overhead} supplies the corresponding bound whenever its constants are uniform in $d$ (which must be established beyond fixed-code low-noise asymptotics).
\end{proof}

\section{Coarse-Grained Error Cancellation: Covariance Hierarchy and Resolving Sectors}
\label{sec:appendix_proofs_discrete_protocol}
\label{sec:generalized_vpec_framework}

In this section, we present the complete mathematical proofs governing the coarse-grained virtual error cancellation (VEC) protocol.
Although syndrome-resolved mitigation achieves the theoretical minimum sampling overhead $C_{\mathrm{syn}}^*$, computing exact posteriors $\bm{p}(\bm{s})$ requires summing partition functions over all fault patterns consistent with the syndrome $\bm{s}$, an operation that is $\#\mathrm{P}$-hard on general space-time decoding graphs.
Furthermore, calibrating individual likelihoods across all $2^{N_D}$ syndrome outcomes is statistically intractable on large quantum processors.
Sector-based mitigation replaces a table over all syndromes by $K$ sector averages. This compression reduces the dimension of the filter calibration, while the computational cost and statistical precision of estimating those averages remain properties of the chosen decoder and model-integration procedure.

The mathematical development proceeds across three subsections.
Subsection~\ref{subsec:app_discrete_filter_derivation} formulates the discrete optimization problem and derives the unique optimal sector weights and sampling overhead via Pythagorean expansion of residuals and strictly convex quadratic minimization.
In Subsection~\ref{subsec:app_proof_eves_law_sector_bound}, we establish Eve's law of total covariance, prove the Loewner ordering of inverse covariances on the positive-definite cone, and prove Theorem~\ref{thm:discrete_sector_bound}.
Subsection~\ref{subsec:app_proof_minimal_resolving_basis} then proves the Minimal Resolving Sector Basis Proposition (Proposition~\ref{thm:minimal_resolving_basis}), demonstrating that full-rank resolution of $M$ error directions requires at least $K\ge M+1$ sectors, and showing that at minimal cardinality $K = M + 1$ the unique algebraic solution coincides with the variational weights.
The dual-decoder realization of this framework, including L\'evy--Desplanques invertibility, condition number bounds, conditional overhead scaling, and calibration suppression, is established in Supplementary Section~\ref{sec:supp_dual_decoder_protocol}.

\subsection{Coarse-Grained Variational Derivation of Sector Weights}
\label{subsec:app_discrete_filter_derivation}
\label{subsec:app_proof_discrete_variational_weights}

We first define the coarse-grained sector partition and its induced probability and second-moment structure over the discrete syndrome space.
Let $(\Omega, \mathcal{F}, P)$ denote the discrete syndrome probability space, where each syndrome $\bm{s} \in \Omega$ occurs with probability $P(\bm{s}) > 0$ and produces a posterior logical error rate vector $\bm{p}(\bm{s}) \in [0, 1]^M$.
Let the syndrome space $\Omega$ be partitioned into $K$ mutually disjoint, non-empty sectors $\mathcal{S}_0, \mathcal{S}_1, \dots, \mathcal{S}_{K-1}$ such that $\Omega = \bigsqcup_{k=0}^{K-1} \mathcal{S}_k$.
Each sector $\mathcal{S}_k$ occurs with non-zero marginal probability
\begin{equation}
\pi_k \equiv P(\bm{s} \in \mathcal{S}_k) = \sum_{\bm{s} \in \mathcal{S}_k} P(\bm{s}) > 0, \qquad \sum_{k=0}^{K-1} \pi_k = 1.
\label{eq:app_sector_probabilities_def}
\end{equation}
Within each sector $\mathcal{S}_k$, the conditional error centroid is defined by the expectation
\begin{equation}
\bm{p}_k \equiv \EE_{\bm{s}}\big[ \bm{p}(\bm{s}) \mid \bm{s} \in \mathcal{S}_k \big] = \frac{1}{\pi_k} \sum_{\bm{s} \in \mathcal{S}_k} P(\bm{s}) \bm{p}(\bm{s}) \in [0, 1]^M.
\label{eq:app_discrete_pk_def}
\end{equation}
The law of total expectation gives the unconditional baseline error vector
\begin{equation}
\bar{\bm{p}} \equiv \EE_{\bm{s}}\big[ \bm{p}(\bm{s}) \big] = \sum_{k=0}^{K-1} \pi_k \bm{p}_k,
\label{eq:app_pbar_total_expectation}
\end{equation}
and the between-sector covariance matrix is given by the weighted outer-product sum
\begin{equation}
\bm{\Sigma}_S \equiv \sum_{k=0}^{K-1} \pi_k \big( \bm{p}_k - \bar{\bm{p}} \big) \big( \bm{p}_k - \bar{\bm{p}} \big)^\top \in \RR^{M \times M}.
\label{eq:app_discrete_Sigma_S_def}
\end{equation}
This matrix summarizes how the sector centroids vary across the $K$ sectors.

A discrete sector filter is specified by a vector of real weights $\bm{w} = (w_0, w_1, \dots, w_{K-1})^\top \in \RR^K$, assigning weight $w_k$ to every syndrome in sector $\mathcal{S}_k$.
To ensure unbiased quantum error cancellation at the logical level, the filter weights must satisfy the two linear channel-inversion conditions,
\begin{align}
\sum_{k=0}^{K-1} \pi_k w_k &= 1 \qquad \text{(Trace Preservation)}, \label{eq:app_discrete_trace_preservation} \\
\sum_{k=0}^{K-1} \pi_k w_k \bm{p}_k &= \bm{0} \qquad \text{(Vector Error Cancellation)}.
\label{eq:app_discrete_error_cancellation}
\end{align}
Under the sector probability measure, the variance of the filter is defined by
\begin{equation}
\Var_S[w] \equiv \sum_{k=0}^{K-1} \pi_k (w_k - 1)^2 = \sum_{k=0}^{K-1} \pi_k w_k^2 - 1.
\label{eq:app_filter_variance_def}
\end{equation}
The cross-covariance between the sector weights and the error centroids evaluates to an invariant quantity governed entirely by the unmitigated baseline error vector $\bar{\bm{p}}$.
Specifically, expanding the expectation gives
\begin{align}
\Cov_S(w, \bm{p}) &\equiv \sum_{k=0}^{K-1} \pi_k (w_k - 1)(\bm{p}_k - \bar{\bm{p}}) \nonumber \\
&= \sum_{k=0}^{K-1} \pi_k w_k \bm{p}_k - \bar{\bm{p}} \sum_{k=0}^{K-1} \pi_k w_k - \sum_{k=0}^{K-1} \pi_k \bm{p}_k + \bar{\bm{p}} \sum_{k=0}^{K-1} \pi_k \nonumber \\
&= \bm{0} - \bar{\bm{p}}(1) - \bar{\bm{p}} + \bar{\bm{p}}(1) = -\bar{\bm{p}},
\label{eq:app_discrete_cross_covariance}
\end{align}
where the third line follows by substituting the constraints in Eqs.~\eqref{eq:app_discrete_trace_preservation} and~\eqref{eq:app_discrete_error_cancellation}.

To bound the variance of any feasible sector filter, we complete the square using an arbitrary test vector $\bm{a} \in \RR^M$.
Specifically, we introduce the affine residuals
\begin{equation}
\Delta_{\bm{a}, k} \equiv (w_k - 1) - \bm{a}^\top (\bm{p}_k - \bar{\bm{p}}), \qquad k \in \{0, \dots, K-1\}.
\label{eq:app_residual_def}
\end{equation}
Because the squared terms $\Delta_{\bm{a}, k}^2$ are non-negative and the sector probabilities satisfy $\pi_k > 0$, their weighted sum satisfies $\sum_{k=0}^{K-1} \pi_k \Delta_{\bm{a}, k}^2 \ge 0$.
Expanding this quadratic sum yields the identity
\begin{align}
\sum_{k=0}^{K-1} \pi_k \Delta_{\bm{a}, k}^2 &= \sum_{k=0}^{K-1} \pi_k (w_k - 1)^2 - 2 \bm{a}^\top \sum_{k=0}^{K-1} \pi_k (w_k - 1)(\bm{p}_k - \bar{\bm{p}}) + \bm{a}^\top \bigg( \sum_{k=0}^{K-1} \pi_k (\bm{p}_k - \bar{\bm{p}})(\bm{p}_k - \bar{\bm{p}})^\top \bigg) \bm{a} \nonumber \\
&= \Var_S[w] - 2 \bm{a}^\top \Cov_S(w, \bm{p}) + \bm{a}^\top \bm{\Sigma}_S \bm{a}.
\label{eq:app_residual_expansion}
\end{align}
Substituting the cross-covariance identity $\Cov_S(w, \bm{p}) = -\bar{\bm{p}}$ into Eq.~\eqref{eq:app_residual_expansion} gives the universal quadratic lower bound
\begin{equation}
\Var_S[w] + 2 \bm{a}^\top \bar{\bm{p}} + \bm{a}^\top \bm{\Sigma}_S \bm{a} \ge 0 \qquad \text{for all } \bm{a} \in \RR^M.
\label{eq:app_discrete_quadratic_form}
\end{equation}
Minimizing Eq.~\eqref{eq:app_discrete_quadratic_form} over $\bm{a}$ directly yields the unique optimal sector filter in the following proposition.

\begin{proposition}[Discrete Variational Filter Minimizer]
\label{prop:app_discrete_variational_derivation}
Let the between-sector covariance matrix satisfy $\bm{\Sigma}_S \succ 0$.
The strictly convex quadratic objective function $h(\bm{a}) \equiv \bm{a}^\top \bm{\Sigma}_S \bm{a} + 2 \bm{a}^\top \bar{\bm{p}}$ achieves its unique global minimum at
\begin{equation}
\bm{a}^* = -\bm{\Sigma}_S^{-1} \bar{\bm{p}}.
\label{eq:app_optimal_a_vector}
\end{equation}
Consequently, the discrete filter variance satisfies the sharp lower bound
\begin{equation}
\Var_S[w] \ge \bar{\bm{p}}^\top \bm{\Sigma}_S^{-1} \bar{\bm{p}},
\label{eq:app_discrete_variance_lower_bound}
\end{equation}
with equality if and only if the sector filter weights are chosen as
\begin{equation}
w_k^* = 1 - \bar{\bm{p}}^\top \bm{\Sigma}_S^{-1} (\bm{p}_k - \bar{\bm{p}}), \qquad k \in \{0, \dots, K-1\}.
\label{eq:app_discrete_optimal_weights_formula}
\end{equation}
The minimal coarse-grained sampling overhead evaluates to
\begin{equation}
C_S^* \equiv 1 + \Var_S[w^*] = 1 + \bar{\bm{p}}^\top \bm{\Sigma}_S^{-1} \bar{\bm{p}}.
\label{eq:app_discrete_cs_star_formula}
\end{equation}
\end{proposition}

\begin{proof}
Because $\bm{\Sigma}_S \succ 0$, the Hessian matrix $\nabla^2 h(\bm{a}) = 2 \bm{\Sigma}_S$ is positive definite on $\RR^M$, certifying strict convexity.
Setting the gradient to zero yields
\begin{equation}
\nabla_{\bm{a}} h(\bm{a}) = 2 \bm{\Sigma}_S \bm{a} + 2 \bar{\bm{p}} = \bm{0}, \qquad \bm{a}^* = -\bm{\Sigma}_S^{-1} \bar{\bm{p}}.
\end{equation}
Evaluating $h(\bm{a})$ at this critical point gives
\begin{equation}
h(\bm{a}^*) = (-\bm{\Sigma}_S^{-1}\bar{\bm{p}})^\top \bm{\Sigma}_S (-\bm{\Sigma}_S^{-1}\bar{\bm{p}}) + 2 (-\bm{\Sigma}_S^{-1}\bar{\bm{p}})^\top \bar{\bm{p}} = \bar{\bm{p}}^\top \bm{\Sigma}_S^{-1}\bar{\bm{p}} - 2 \bar{\bm{p}}^\top \bm{\Sigma}_S^{-1}\bar{\bm{p}} = -\bar{\bm{p}}^\top \bm{\Sigma}_S^{-1}\bar{\bm{p}}.
\end{equation}
Substituting $h(\bm{a}^*)$ into Eq.~\eqref{eq:app_discrete_quadratic_form} yields
\begin{equation}
\Var_S[w] - \bar{\bm{p}}^\top \bm{\Sigma}_S^{-1} \bar{\bm{p}} \ge 0, \qquad \Var_S[w] \ge \bar{\bm{p}}^\top \bm{\Sigma}_S^{-1} \bar{\bm{p}}.
\end{equation}
Because the probabilities $\pi_k$ are strictly positive, the lower bound is saturated if and only if $\sum_k \pi_k \Delta_{\bm{a}^*, k}^2 = 0$, requiring $\Delta_{\bm{a}^*, k} = 0$ for all $k \in \{0, \dots, K-1\}$.
Setting $\Delta_{\bm{a}^*, k} = 0$ uniquely determines the optimal weights,
\begin{equation}
w_k - 1 = \bm{a}^{*\top} (\bm{p}_k - \bar{\bm{p}}) = -\bar{\bm{p}}^\top \bm{\Sigma}_S^{-1}(\bm{p}_k - \bar{\bm{p}}), \qquad w_k^* = 1 - \bar{\bm{p}}^\top \bm{\Sigma}_S^{-1}(\bm{p}_k - \bar{\bm{p}}).
\end{equation}
Adding unity to the minimal variance establishes $C_S^* = 1 + \bar{\bm{p}}^\top \bm{\Sigma}_S^{-1}\bar{\bm{p}}$, completing the proof.
\end{proof}

\subsection{Eve's Law Covariance Hierarchy and Discrete Variational Sector Bound}
\label{subsec:app_proof_eves_law_sector_bound}

We now relate the coarse-grained matrix $\bm{\Sigma}_S$ to the full decoder dispersion $\bm{\Sigma}$ through the law of total covariance, and prove Theorem~\ref{thm:discrete_sector_bound}.
Recall that $\bm{\Sigma}$ is defined over $\Omega$ by
\begin{equation}
\bm{\Sigma} \equiv \Cov_{\bm{s}}(\bm{p}(\bm{s})) = \sum_{\bm{s} \in \Omega} P(\bm{s})\big(\bm{p}(\bm{s}) - \bar{\bm{p}}\big)\big(\bm{p}(\bm{s}) - \bar{\bm{p}}\big)^\top \in \RR^{M \times M}.
\label{eq:app_full_covariance_def}
\end{equation}
Within each sector $\mathcal{S}_k$, the conditional spread around the centroid $\bm{p}_k$ is
\begin{equation}
\Cov_{\bm{s}}\big(\bm{p}(\bm{s}) \mid \bm{s} \in \mathcal{S}_k\big) \equiv \frac{1}{\pi_k} \sum_{\bm{s} \in \mathcal{S}_k} P(\bm{s})\big(\bm{p}(\bm{s}) - \bm{p}_k\big)\big(\bm{p}(\bm{s}) - \bm{p}_k\big)^\top \succeq 0,
\label{eq:app_conditional_covariance_def}
\end{equation}
and the intra-sector matrix $\bm{\Sigma}_{\mathrm{within}}$ is the convex combination of these sector terms,
\begin{equation}
\bm{\Sigma}_{\mathrm{within}} \equiv \sum_{k=0}^{K-1} \pi_k \Cov_{\bm{s}}\big(\bm{p}(\bm{s}) \mid \bm{s} \in \mathcal{S}_k\big) = \sum_{k=0}^{K-1} \sum_{\bm{s} \in \mathcal{S}_k} P(\bm{s})\big(\bm{p}(\bm{s}) - \bm{p}_k\big)\big(\bm{p}(\bm{s}) - \bm{p}_k\big)^\top,
\label{eq:app_discrete_Sigma_within_def}
\end{equation}
quantifying the posterior fluctuations lost when syndromes are averaged within each sector.

\begin{lemma}[Eve's Law of Total Covariance]
\label{lem:app_eves_law}
The full decoder covariance matrix decomposes into the sum of between-sector and intra-sector contributions,
\begin{equation}
\bm{\Sigma} = \bm{\Sigma}_S + \bm{\Sigma}_{\mathrm{within}}.
\label{eq:app_eves_law_equality}
\end{equation}
In the Loewner partial order, $\bm{\Sigma}$ therefore dominates $\bm{\Sigma}_S$,
\begin{equation}
\bm{\Sigma} \succeq \bm{\Sigma}_S.
\label{eq:app_loewner_covariance_ordering}
\end{equation}
Furthermore, if $\bm{\Sigma}_S \succ 0$, the inverse matrices satisfy the reversed Loewner ordering
\begin{equation}
\bm{\Sigma}_S^{-1} \succeq \bm{\Sigma}^{-1},
\label{eq:app_loewner_inverse_ordering}
\end{equation}
establishing the sampling overhead hierarchy $C_S^* \ge C_{\mathrm{syn}}^*$.
\end{lemma}

\begin{proof}
For any syndrome $\bm{s} \in \mathcal{S}_k$, we decompose the centered error vector by inserting the sector centroid $\bm{p}_k$,
\begin{equation}
\bm{p}(\bm{s}) - \bar{\bm{p}} = \big(\bm{p}(\bm{s}) - \bm{p}_k\big) + \big(\bm{p}_k - \bar{\bm{p}}\big).
\end{equation}
Taking the outer product gives
\begin{align}
\big(\bm{p}(\bm{s}) - \bar{\bm{p}}\big)\big(\bm{p}(\bm{s}) - \bar{\bm{p}}\big)^\top &= \big(\bm{p}(\bm{s}) - \bm{p}_k\big)\big(\bm{p}(\bm{s}) - \bm{p}_k\big)^\top + \big(\bm{p}_k - \bar{\bm{p}}\big)\big(\bm{p}_k - \bar{\bm{p}}\big)^\top \nonumber \\
&\quad + \big(\bm{p}(\bm{s}) - \bm{p}_k\big)\big(\bm{p}_k - \bar{\bm{p}}\big)^\top + \big(\bm{p}_k - \bar{\bm{p}}\big)\big(\bm{p}(\bm{s}) - \bm{p}_k\big)^\top.
\end{align}
Summing over all syndromes $\bm{s} \in \mathcal{S}_k$ weighted by $P(\bm{s})$, the cross terms vanish identically because
\begin{equation}
\sum_{\bm{s} \in \mathcal{S}_k} P(\bm{s})\big(\bm{p}(\bm{s}) - \bm{p}_k\big) = \pi_k \big(\bm{p}_k - \bm{p}_k\big) = \bm{0}.
\end{equation}
Summing the remaining terms over all $K$ sectors $k \in \{0, \dots, K-1\}$ yields
\begin{align}
\bm{\Sigma} &= \sum_{k=0}^{K-1} \sum_{\bm{s} \in \mathcal{S}_k} P(\bm{s})\big(\bm{p}(\bm{s}) - \bar{\bm{p}}\big)\big(\bm{p}(\bm{s}) - \bar{\bm{p}}\big)^\top \nonumber \\
&= \sum_{k=0}^{K-1} \sum_{\bm{s} \in \mathcal{S}_k} P(\bm{s})\big(\bm{p}(\bm{s}) - \bm{p}_k\big)\big(\bm{p}(\bm{s}) - \bm{p}_k\big)^\top + \sum_{k=0}^{K-1} \pi_k \big(\bm{p}_k - \bar{\bm{p}}\big)\big(\bm{p}_k - \bar{\bm{p}}\big)^\top \nonumber \\
&= \bm{\Sigma}_{\mathrm{within}} + \bm{\Sigma}_S,
\end{align}
proving Eq.~\eqref{eq:app_eves_law_equality}.
Because each sector matrix in Eq.~\eqref{eq:app_conditional_covariance_def} is positive semidefinite and the probabilities satisfy $\pi_k > 0$, their convex combination $\bm{\Sigma}_{\mathrm{within}}$ is also non-negative definite,
\begin{equation}
\bm{\Sigma}_{\mathrm{within}} \succeq 0, \qquad \bm{\Sigma} - \bm{\Sigma}_S = \bm{\Sigma}_{\mathrm{within}} \succeq 0, \qquad \bm{\Sigma} \succeq \bm{\Sigma}_S.
\end{equation}

When $\bm{\Sigma}_S \succ 0$, both $\bm{\Sigma}_S$ and $\bm{\Sigma}$ are strictly positive definite.
Operator anti-monotonicity of matrix inversion on the positive-definite cone reverses the Loewner partial order under inversion,
\begin{equation}
\bm{\Sigma} \succeq \bm{\Sigma}_S \succ 0, \qquad \bm{\Sigma}_S^{-1} \succeq \bm{\Sigma}^{-1}.
\end{equation}
Multiplying by $\bar{\bm{p}}$ on the left and right yields the scalar inequality
\begin{equation}
\bar{\bm{p}}^\top \bm{\Sigma}_S^{-1} \bar{\bm{p}} \ge \bar{\bm{p}}^\top \bm{\Sigma}^{-1} \bar{\bm{p}}.
\end{equation}
Adding unity to both sides establishes the overhead hierarchy
\begin{equation}
C_S^* \equiv 1 + \bar{\bm{p}}^\top \bm{\Sigma}_S^{-1}\bar{\bm{p}} \ge 1 + \bar{\bm{p}}^\top \bm{\Sigma}^{-1}\bar{\bm{p}} = C_{\mathrm{syn}}^*,
\end{equation}
concluding the proof of Lemma~\ref{lem:app_eves_law}.
\end{proof}

Combining Proposition~\ref{prop:app_discrete_variational_derivation} and Lemma~\ref{lem:app_eves_law} directly yields Theorem~\ref{thm:discrete_sector_bound}.

\begin{theorem}[Discrete Variational Sector Bound and Covariance Hierarchy]
\label{thm:discrete_sector_bound}
Let $(\Omega, \mathcal{F}, P)$ be the syndrome probability space with posterior error vector $\bm{p}(\bm{s}) \in [0, 1]^M$ and baseline mean $\bar{\bm{p}} \in [0, 1]^M$.
For any partition $\Omega = \bigsqcup_{k=0}^{K-1} \mathcal{S}_k$ into $K$ disjoint sectors with probabilities $\pi_k > 0$, conditional error centroids $\bm{p}_k$, and between-sector covariance $\bm{\Sigma}_S \succ 0$:
\begin{enumerate}
\item \emph{Total covariance decomposition:} The full decoder matrix $\bm{\Sigma}$ splits into between-sector and within-sector contributions obeying
\begin{equation}
\bm{\Sigma} = \bm{\Sigma}_S + \bm{\Sigma}_{\mathrm{within}} \succeq \bm{\Sigma}_S, \qquad \bm{\Sigma}_{\mathrm{within}} \succeq 0.
\label{eq:app_thm3_decomp}
\end{equation}
\item \emph{Loewner inverse ordering:} On the cone of positive-definite matrices, matrix inversion is strictly operator monotone decreasing, so $\bm{\Sigma} \succeq \bm{\Sigma}_S \succ 0$ implies
\begin{equation}
\bm{\Sigma}_S^{-1} \succeq \bm{\Sigma}^{-1}.
\label{eq:app_thm3_loewner}
\end{equation}
\item \emph{Universal sampling overhead hierarchy:} For any discrete filter $\bm{w} \in \RR^K$ obeying the channel-inversion constraints $\sum_k \pi_k w_k = 1$ and $\sum_k \pi_k w_k \bm{p}_k = \bm{0}$, the sampling overhead satisfies
\begin{equation}
C(w) \ge C_S^* \equiv 1 + \bar{\bm{p}}^\top \bm{\Sigma}_S^{-1} \bar{\bm{p}} \ge 1 + \bar{\bm{p}}^\top \bm{\Sigma}^{-1} \bar{\bm{p}} \equiv C_{\mathrm{syn}}^*.
\label{eq:app_thm3_hierarchy}
\end{equation}
Equality $C(w) = C_S^*$ holds if and only if $w_k = w_k^* = 1 - \bar{\bm{p}}^\top \bm{\Sigma}_S^{-1}(\bm{p}_k - \bar{\bm{p}})$ for all $k \in \{0, \dots, K-1\}$.
Furthermore, equality between coarse-grained and syndrome-resolved overheads, $C_S^* = C_{\mathrm{syn}}^*$, holds if and only if $\bm{\Sigma}_{\mathrm{within}}\bm{\Sigma}^{-1}\bar{\bm{p}} = \bm{0}$, a condition satisfied identically whenever the partition is fine enough that intra-sector fluctuations vanish ($\bm{\Sigma}_{\mathrm{within}} = \bm{0}$).
\end{enumerate}
\end{theorem}

\begin{proof}
Items 1 and 2 follow directly from Lemma~\ref{lem:app_eves_law}.
To prove Item 3, let $\bm{w} = (w_0, \dots, w_{K-1})^\top \in \RR^K$ be an arbitrary discrete filter satisfying $\sum_k \pi_k w_k = 1$ and $\sum_k \pi_k w_k \bm{p}_k = \bm{0}$.
By Proposition~\ref{prop:app_discrete_variational_derivation}, expanding the weighted sum of squared residuals $\sum_k \pi_k \Delta_{\bm{a}, k}^2 \ge 0$ for test vector $\bm{a} \in \RR^M$ and substituting the invariant cross-covariance $\Cov_S(w, \bm{p}) = -\bar{\bm{p}}$ establishes
\begin{equation}
\Var_S[w] \ge \bar{\bm{p}}^\top \bm{\Sigma}_S^{-1}\bar{\bm{p}}.
\end{equation}
Adding unity to both sides yields $C(w) \equiv 1 + \Var_S[w] \ge 1 + \bar{\bm{p}}^\top \bm{\Sigma}_S^{-1}\bar{\bm{p}} = C_S^*$.
Lemma~\ref{lem:app_eves_law} and the Loewner ordering $\bm{\Sigma}_S^{-1} \succeq \bm{\Sigma}^{-1}$ then guarantee
\begin{equation}
C_S^* = 1 + \bar{\bm{p}}^\top \bm{\Sigma}_S^{-1}\bar{\bm{p}} \ge 1 + \bar{\bm{p}}^\top \bm{\Sigma}^{-1}\bar{\bm{p}} = C_{\mathrm{syn}}^*,
\end{equation}
establishing the complete chain of inequalities in Eq.~\eqref{eq:app_thm3_hierarchy}.

Uniqueness of the minimizer follows because $\pi_k > 0$, requiring $\Delta_{\bm{a}^*, k} = 0$ for all $k \in \{0, \dots, K-1\}$ and uniquely fixing $w_k = w_k^*$.
Finally, the gap between the two optimal overheads is
\begin{equation}
C_S^* - C_{\mathrm{syn}}^* = \bar{\bm{p}}^\top \big( \bm{\Sigma}_S^{-1} - \bm{\Sigma}^{-1} \big) \bar{\bm{p}}.
\end{equation}
The matrix identity $\bm{A}^{-1} - \bm{B}^{-1} = \bm{A}^{-1}(\bm{B} - \bm{A})\bm{B}^{-1}$ with $\bm{A} = \bm{\Sigma}_S$ and $\bm{B} = \bm{\Sigma} = \bm{\Sigma}_S + \bm{\Sigma}_{\mathrm{within}}$ rewrites this gap as
\begin{equation}
C_S^* - C_{\mathrm{syn}}^* = \bar{\bm{p}}^\top \bm{\Sigma}_S^{-1} \bm{\Sigma}_{\mathrm{within}} \bm{\Sigma}^{-1} \bar{\bm{p}}.
\end{equation}
Because $\bm{\Sigma}_S^{-1} - \bm{\Sigma}^{-1} \succeq 0$, the quadratic form $\bar{\bm{p}}^\top(\bm{\Sigma}_S^{-1} - \bm{\Sigma}^{-1})\bar{\bm{p}}$ vanishes if and only if $(\bm{\Sigma}_S^{-1} - \bm{\Sigma}^{-1})\bar{\bm{p}} = \bm{0}$; multiplying by $\bm{\Sigma}_S \succ 0$ shows this is equivalent to $\bm{\Sigma}_{\mathrm{within}}\bm{\Sigma}^{-1}\bar{\bm{p}} = \bm{0}$.
When intra-sector fluctuations vanish identically ($\bm{\Sigma}_{\mathrm{within}} = \bm{0}$), each sector contains syndromes with identical error rates, so coarse-graining loses no statistical information and achieves $C_S^* = C_{\mathrm{syn}}^*$.
This completes the proof of Theorem~\ref{thm:discrete_sector_bound}.
\end{proof}

\subsection{Minimal Resolving Sector Basis Theorem}
\label{subsec:app_proof_minimal_resolving_basis}

We now establish the geometric and algebraic conditions on the syndrome sectors required for universal channel unbiasing across arbitrary $\bar{\bm{p}} \in \RR^M$.
Proposition~\ref{thm:minimal_resolving_basis} shows that achieving a full-rank matrix $\bm{\Sigma}_S \succ 0$ for an $M$-channel logical noise process requires at least $K \ge M + 1$ sectors, and at $K = M + 1$ the unique algebraic solution coincides with the variational weights $w_k^*$.

\begin{proposition}[Minimal Resolving Sector Basis Theorem]
\label{thm:minimal_resolving_basis}
\label{prop:minimal_resolving_basis}
Let an $M$-channel logical error process be mitigated by a coarse-grained sector filter $\bm{w} \in \RR^K$ over $K$ disjoint sectors with probabilities $\pi_k > 0$ and conditional error centroids $\bm{p}_k \in \RR^M$.
\begin{enumerate}
\item \emph{Linear system representation:} The discrete channel-inversion constraints $\sum_k \pi_k w_k = 1$ and $\sum_k \pi_k w_k \bm{p}_k = \bm{0}$ constitute an $(M+1) \times K$ real linear system for the reweighted sector probabilities $\tilde{w}_k \equiv \pi_k w_k$,
\begin{equation}
\begin{pmatrix} 1 & 1 & \dots & 1 \\ \bm{p}_0 & \bm{p}_1 & \dots & \bm{p}_{K-1} \end{pmatrix} \begin{pmatrix} \tilde{w}_0 \\ \tilde{w}_1 \\ \vdots \\ \tilde{w}_{K-1} \end{pmatrix} = \begin{pmatrix} 1 \\ \bm{0} \end{pmatrix}.
\label{eq:app_linear_system_K}
\end{equation}
\item \emph{Sylvester rank bound:} The centered displacement matrix $\bm{D} \equiv (\bm{p}_0 - \bar{\bm{p}}, \dots, \bm{p}_{K-1} - \bar{\bm{p}}) \in \RR^{M \times K}$ satisfies the linear relation $\sum_{k=0}^{K-1} \pi_k (\bm{p}_k - \bar{\bm{p}}) = \bm{0}$.
Writing $\bm{\Sigma}_S = \bm{D}\operatorname{diag}(\pi_0, \dots, \pi_{K-1})\bm{D}^\top$, Sylvester's rank inequality gives
\begin{equation}
\operatorname{rank}(\bm{\Sigma}_S) = \operatorname{rank}(\bm{D}) \le K - 1.
\label{eq:app_rank_bound_Sigma_S}
\end{equation}
\item \emph{Minimal sector count condition:} Universal cancellation across all $M$ error directions for arbitrary $\bar{\bm{p}} \in \RR^M$ requires $\operatorname{rank}(\bm{\Sigma}_S) = M$, establishing the strict necessity of
\begin{equation}
K \ge M + 1.
\label{eq:app_minimal_K_inequality}
\end{equation}
If $K \le M$, then $\operatorname{rank}(\bm{\Sigma}_S) \le M - 1 < M$, and the between-sector covariance has a non-trivial kernel $\operatorname{ker}(\bm{\Sigma}_S) \ne \{\bm{0}\}$.
Any logical error component with a non-zero projection onto $\operatorname{ker}(\bm{\Sigma}_S)$ cannot be mitigated by post-processing reweighting.
\item \emph{Affine simplex basis at $K = M + 1$:} When $K = M + 1$, the system matrix $\bm{A}_{\mathrm{geom}} \in \RR^{(M+1) \times (M+1)}$ in Eq.~\eqref{eq:app_linear_system_K} satisfies the determinant factorization
\begin{equation}
\det(\bm{A}_{\mathrm{geom}}) = \det \big( \bm{p}_1 - \bm{p}_0, \, \bm{p}_2 - \bm{p}_0, \, \dots, \, \bm{p}_M - \bm{p}_0 \big).
\label{eq:app_det_A_geom}
\end{equation}
A unique unbiased filter $\bm{w}^*$ exists if and only if the $M + 1$ centroids $\{\bm{p}_0, \dots, \bm{p}_M\}$ are affinely independent in $\RR^M$, forming a non-degenerate $M$-simplex enclosing $\bar{\bm{p}}$ with $\det(\bm{A}_{\mathrm{geom}}) \ne 0$.
The unique algebraic solution then coincides identically with the variational weights,
\begin{equation}
w_k^* = 1 - \bar{\bm{p}}^\top \bm{\Sigma}_S^{-1}(\bm{p}_k - \bar{\bm{p}}), \qquad k \in \{0, \dots, M\}.
\label{eq:app_thm4_variational_equivalence}
\end{equation}
\end{enumerate}
\end{proposition}

\begin{proof}
Item 1 is verified directly by multiplying the block matrix $\bm{A}_{\mathrm{geom}}$ by $\tilde{\bm{w}} = (\pi_0 w_0, \dots, \pi_{K-1} w_{K-1})^\top$: the top row yields $\sum_k \pi_k w_k = 1$ and the lower $M$ rows yield $\sum_k \pi_k w_k \bm{p}_k = \bm{0}$.

To prove Item 2, observe that the set of $K$ points $\{\bm{p}_0, \dots, \bm{p}_{K-1}\} \subset \RR^M$ spans an affine subspace of dimension at most $K - 1$.
The centered difference vectors $\bm{d}_k \equiv \bm{p}_k - \bar{\bm{p}}$ lie in the linear subspace parallel to this affine hull.
Furthermore, by definition of the baseline mean $\bar{\bm{p}} = \sum_k \pi_k \bm{p}_k$, these vectors satisfy the constraint
\begin{equation}
\sum_{k=0}^{K-1} \pi_k \bm{d}_k = \sum_{k=0}^{K-1} \pi_k (\bm{p}_k - \bar{\bm{p}}) = \bar{\bm{p}} - \bar{\bm{p}} = \bm{0}.
\end{equation}
Because the coefficients $\pi_k > 0$ are non-zero, at most $K - 1$ of the vectors $\{\bm{d}_0, \dots, \bm{d}_{K-1}\}$ can be linearly independent.
Expressing the between-sector covariance matrix as $\bm{\Sigma}_S = \bm{D}\bm{Q}\bm{D}^\top$ with $\bm{D} = (\bm{d}_0, \dots, \bm{d}_{K-1}) \in \RR^{M \times K}$ and diagonal matrix $\bm{Q} = \operatorname{diag}(\pi_0, \dots, \pi_{K-1}) \succ 0$, Sylvester's rank inequality implies
\begin{equation}
\operatorname{rank}(\bm{\Sigma}_S) = \operatorname{rank}(\bm{D}\bm{Q}\bm{D}^\top) = \operatorname{rank}(\bm{D}) \le K - 1.
\end{equation}
This establishes Item 2.

To prove Item 3, universal cancellation across all $M$ linearly independent error directions for arbitrary $\bar{\bm{p}} \in \RR^M$ requires the between-sector covariance matrix $\bm{\Sigma}_S$ to have full rank $M$ on $\RR^M$ so that $\operatorname{range}(\bm{\Sigma}_S) = \RR^M$.
Combining this rank requirement with the bound of Item 2 gives $M = \operatorname{rank}(\bm{\Sigma}_S) \le K - 1$, and hence
\begin{equation}
K \ge M + 1.
\end{equation}
If $K \le M$, then $\operatorname{rank}(\bm{\Sigma}_S) \le M - 1 < M$, so $\bm{\Sigma}_S$ is singular with a non-trivial null space $\operatorname{ker}(\bm{\Sigma}_S) \ne \{\bm{0}\}$.
For a fixed model, a rank-deficient partition is feasible if and only if $\bar{\bm p}\in\operatorname{range}(\Sigma_S)$ (Lemma~\ref{lem:feasibility_condition_formal}). Collinearity $p_k=c_k\bar{\bm p}\ne0$ alone is insufficient without $\Var_S(c_k)>0$, so the sector-count condition $K \ge M + 1$ concerns full-rank resolution across arbitrary $\bar{\bm{p}}$.
Any error vector with a non-zero projection onto $\operatorname{ker}(\bm{\Sigma}_S)$ cannot be mitigated by post-processing reweighting, making $K \ge M + 1$ strictly necessary for full-rank channel inversion.

To prove Item 4, consider the minimal cardinality case $K = M + 1$.
The linear system in Eq.~\eqref{eq:app_linear_system_K} becomes a square $(M+1) \times (M+1)$ system $\bm{A}_{\mathrm{geom}}\tilde{\bm{w}} = (1, \bm{0})^\top$.
Subtracting the first column from each subsequent column $k \in \{1, \dots, M\}$ via elementary column operations preserves the determinant,
\begin{equation}
\det(\bm{A}_{\mathrm{geom}}) = \det \begin{pmatrix} 1 & 0 & 0 & \dots & 0 \\ \bm{p}_0 & \bm{p}_1 - \bm{p}_0 & \bm{p}_2 - \bm{p}_0 & \dots & \bm{p}_M - \bm{p}_0 \end{pmatrix} = \det \big( \bm{p}_1 - \bm{p}_0, \, \bm{p}_2 - \bm{p}_0, \, \dots, \, \bm{p}_M - \bm{p}_0 \big).
\end{equation}
The $M + 1$ vertices in $\RR^M$ are affinely independent if and only if the $M$ edge vectors $\{\bm{p}_k - \bm{p}_0\}_{k=1}^M$ form a basis of $\RR^M$, which is equivalent to $\det(\bm{A}_{\mathrm{geom}}) \ne 0$.
Under this condition, the coefficient matrix $\bm{A}_{\mathrm{sys}} \equiv \bm{A}_{\mathrm{geom}}\operatorname{diag}(\pi_0, \dots, \pi_M)$ satisfies
\begin{equation}
\det(\bm{A}_{\mathrm{sys}}) = \bigg( \prod_{k=0}^M \pi_k \bigg) \det(\bm{A}_{\mathrm{geom}}) \ne 0,
\end{equation}
because all sector probabilities are strictly positive ($\pi_k > 0$).
Hence $\bm{A}_{\mathrm{sys}}$ is non-singular, guaranteeing that the linear system has a unique solution $\bm{w} \in \RR^{M+1}$.

To verify that this unique solution coincides with the variational weights $w_k^* = 1 - \bar{\bm{p}}^\top \bm{\Sigma}_S^{-1}(\bm{p}_k - \bar{\bm{p}})$, we evaluate the constraints directly on $w_k^*$.
For trace preservation, we have
\begin{equation}
\sum_{k=0}^M \pi_k w_k^* = \sum_{k=0}^M \pi_k \Big( 1 - \bar{\bm{p}}^\top \bm{\Sigma}_S^{-1}(\bm{p}_k - \bar{\bm{p}}) \Big) = 1 - \bar{\bm{p}}^\top \bm{\Sigma}_S^{-1} \sum_{k=0}^M \pi_k (\bm{p}_k - \bar{\bm{p}}) = 1 - \bm{0} = 1.
\end{equation}
For vector error cancellation, we obtain
\begin{align}
\sum_{k=0}^M \pi_k w_k^* \bm{p}_k &= \sum_{k=0}^M \pi_k \Big( 1 - \bar{\bm{p}}^\top \bm{\Sigma}_S^{-1}(\bm{p}_k - \bar{\bm{p}}) \Big) \bm{p}_k \nonumber \\
&= \bar{\bm{p}} - \sum_{k=0}^M \pi_k \bm{p}_k (\bm{p}_k - \bar{\bm{p}})^\top \bm{\Sigma}_S^{-1} \bar{\bm{p}} \nonumber \\
&= \bar{\bm{p}} - \bigg( \sum_{k=0}^M \pi_k (\bm{p}_k - \bar{\bm{p}})(\bm{p}_k - \bar{\bm{p}})^\top \bigg) \bm{\Sigma}_S^{-1} \bar{\bm{p}} \nonumber \\
&= \bar{\bm{p}} - \bm{\Sigma}_S \bm{\Sigma}_S^{-1} \bar{\bm{p}} = \bar{\bm{p}} - \bar{\bm{p}} = \bm{0}.
\end{align}
Because the square linear system $\bm{A}_{\mathrm{sys}}\bm{w} = (1, \bm{0})^\top$ possesses a unique solution, the algebraic solution must coincide identically with $w_k^*$, completing the proof of Proposition~\ref{thm:minimal_resolving_basis}.
\end{proof}

\section{Analytical Foundations of Dual-Decoder Error Cancellation}
\label{sec:supp_dual_decoder_protocol}
\label{sec:appendix_derivations}

Dual-decoder virtual error cancellation (VEC) exploits the latency gap between fast hardware decoding and higher-accuracy offline inference to cancel residual logical errors.
Fault-tolerant quantum computing architectures require low-latency real-time decoding for feedforward operations within coherence deadlines, typically employing fast minimum-weight perfect matching (MWPM) or union-find algorithms.
Asynchronous offline processing with tensor-network, belief-propagation, ensemble, or neural-network decoders supplies more accurate reference recoveries without delaying execution.
For an $n$-qubit logical block with Pauli group $\mathcal{P}_n$ and $M \le 4^n - 1$ active non-trivial error channels $\mathcal{Q} = \{Q_1, \dots, Q_M\} \subseteq \mathcal{P}_n \setminus \{I\}$ (setting $\bar{\mathcal{Q}} \equiv \{I\} \cup \mathcal{Q}$), we compare the real-time correction $\mathcal{D}_{\mathrm{rt}}(\bm{s}) \in \mathcal{P}_n$ against an offline reference correction $\mathcal{D}_{\mathrm{off}}(\bm{s}) \in \mathcal{P}_n$. This comparison partitions active syndromes into $K = M + 1$ disjoint sectors,
\begin{equation}
\mathcal{S}_Q \equiv \big\{\bm{s} \in \Omega :\; \mathcal{D}_{\mathrm{rt}}(\bm{s})\cdot \mathcal{D}_{\mathrm{off}}(\bm{s}) \sim Q\big\}, \qquad Q \in \bar{\mathcal{Q}},
\label{eq:app_dd_sector_definition}
\end{equation}
where $\sim$ denotes equality up to overall phase. Any rare out-of-support discrepancy $Q \notin \bar{\mathcal{Q}}$ is assigned to an overflow sector $\mathcal{S}_\perp$ with weight $w_\perp = 0$, matching the failure convention of Section~\ref{sec:fc-totality}, whereas group closure guarantees $\mathcal{S}_\perp = \emptyset$ whenever $\bar{\mathcal{Q}}$ forms a Pauli subgroup containing both decoder outputs.
Agreement defines the consensus sector $\mathcal S_I$, whose high occupancy $\pi_I=1-O(p_L)$ is an explicit hypothesis for the overhead scaling below, not a consequence of agreement alone.
Disagreement along an active logical Pauli direction $Q \in \mathcal{Q}$ sorts the remaining runs into $M$ contested sectors $\mathcal{S}_Q$ with elevated conditional error centroids $\bm{p}_Q \equiv \EE_{\bm{s}}[\bm{p}(\bm{s}) \mid \bm{s} \in \mathcal{S}_Q]$.
These $M + 1$ coarse-grained sectors provide sufficient contrast to invert the residual logical Pauli channel linearly.

The mathematical analysis proceeds across three subsections.
Subsection~\ref{subsec:app_proof_pauli_4sector_invertibility} establishes affine independence and L\'evy--Desplanques invertibility for the minimal four-sector Pauli noise model ($n = 1$ and $M = 3$). We first state Hypotheses~H1--H3, analyze the simplex geometry of the displacement matrix, and introduce a parameterized Calderbank--Shor--Steane (CSS) model together with the L\'evy--Desplanques theorem and Varah bound (Lemma~\ref{lem:app_levy_desplanques} and Theorem~\ref{thm:app_levy_desplanques_pauli}). We then derive the congruent covariance factorization (Proposition~\ref{thm:app_covariance_factorization}) via Sherman--Morrison inversion (Lemma~\ref{lem:app_sherman_morrison_categorical}), prove that the variational solution collapses to closed-form sector weights (Proposition~\ref{thm:app_variational_collapse}), and establish fault-tolerant overhead scaling in both general matrix form (Theorem~\ref{thm:app_asymptotic_overhead}) and the uncoupled check basis (Corollary~\ref{cor:app_overhead_uncoupled}).
In Subsection~\ref{subsec:app_proof_distance_suppression}, we analyze residual logical error under calibration perturbations and establish consensus attenuation. We relate the perturbed parameters $\hat{\pi}_k$ and $\hat{\bm{p}}_k$ to the trace-preservation normalization $\bar{w}$ and residual error vector $\bm{p}_{\mathrm{res}}$, and derive the exact first-order differential and binary reduction (Proposition~\ref{thm:app_exact_differential}). We also prove common coordinate rescaling invariance (Lemma~\ref{lem:app_uniform_shift_invariance}) and finite-error consensus suppression (Theorem~\ref{thm:app_distance_suppression}), and compare the protocol with postselection and physical probabilistic error cancellation.
Finally, Subsection~\ref{subsec:app_decoder_asymmetry_derivation} derives the binary asymmetry formulas from joint real-time and offline error events. Writing $J=\pi_I p_I$ for the joint failure probability, the failure-asymmetry ratio satisfies $b=(p_L'-J)/(p_L-J)$ and the disagreement error rate satisfies $p_X=1/(1+b)$, yielding the controlled small-$p_L$ slope $(1+b)/(1-\alpha)$ with $\alpha=p_I/p_L$.

\subsection{\texorpdfstring{Affine Independence and L\'evy--Desplanques Invertibility for the 4-Sector Pauli Model}{Affine Independence and Levy-Desplanques Invertibility for the 4-Sector Pauli Model}}
\label{subsec:app_proof_pauli_4sector_invertibility}

We now provide the complete proof suite for the minimal four-sector logical Pauli mitigation architecture.
Consider arbitrary single-qubit logical Pauli noise ($M = 3$, with alphabet $\mathcal{Q} \equiv \{Q_1, Q_2, Q_3\} = \{X, Z, Y\}$).
Full-rank resolution of $M=3$ error directions requires at least four sectors. We label consensus by $0$ and disagreements along $X,Z,Y$ by $1,2,3$. Write $p_{k,a}$ for the probability of logical error $a$ conditional on sector $k$ (with the sector index written first), $p_L=\|\bar{\bm p}\|_1$ for the total logical error rate, and $u=V^{-1}p_0$ for the normalized consensus coordinate vector. The following conditions are explicit hypotheses on the chosen code, decoders, noise model, and asymptotic regime; they do not follow merely from operation below threshold.

\subsubsection{Foundational Sector Hypotheses for the Four-Sector Protocol}
\label{subsubsec:app_foundational_hypotheses}
\begin{enumerate}
\item[\emph{H1.}] The displacement matrix $V$ is invertible, with minimum singular value $\sigma_{\min}(V)\ge v_0>0$ uniformly in the parameters over which a uniform claim is made. Pointwise invertibility suffices for exact inversion, whereas this quantitative lower bound controls calibration sensitivity.
\item[\emph{H2.}] All retained sector probabilities are positive, and
\begin{equation}
|u_k|\le L\pi_k\quad(k\ge1)
\label{eq:app_H2_per_channel}
\end{equation}
for a uniform constant $L$ (for which the stronger condition $|u_k|\le L_0\alpha\pi_k$ is sufficient). This bound provides separate control of consensus error share relative to each rare sector and is not implied by Hypothesis~H1.
\item[\emph{H3.}] The total contested probability satisfies $t=\sum_{k\ge1}\pi_k\le Bp_L$ for a uniform constant $B$, and the consensus error norm satisfies $\|p_0\|_1=\alpha p_L$ with $0\le\alpha<1$. This relation defines a norm ratio and does not assert vector collinearity $p_0=\alpha\bar{\bm p}$; a large-distance scaling claim additionally requires a uniform exponential bound on $p_L$ and the associated constants.
\end{enumerate}
The exact algebraic identities below require only positive sector probabilities and an invertible contrast matrix $V$, whereas the uniform overhead and sensitivity bounds invoke the additional hypotheses stated in each theorem.

\subsubsection{Simplex Geometry and Differential Sensitivity}
\label{subsubsec:app_simplex_geometry}

Geometrically, the four conditional error centroid vectors $\{\bm{p}_0, \bm{p}_1, \bm{p}_2, \bm{p}_3\}$ form the vertices of a 3-simplex $\Delta \subset [0, 1]^3$ in the three-dimensional logical Pauli error space.
With rows indexed by the logical error channel and columns by the non-reference sector, the displacement matrix takes the form
\begin{equation}
V=(p_1-p_0,p_2-p_0,p_3-p_0)=
\begin{pmatrix}
p_{1,1}-p_{0,1}&p_{2,1}-p_{0,1}&p_{3,1}-p_{0,1}\\
p_{1,2}-p_{0,2}&p_{2,2}-p_{0,2}&p_{3,2}-p_{0,2}\\
p_{1,3}-p_{0,3}&p_{2,3}-p_{0,3}&p_{3,3}-p_{0,3}
\end{pmatrix}.
\label{eq:app_displacement_matrix_V}
\end{equation}
Its columns represent the simplex edges directed from the reference vertex $p_0$, so the simplex volume is $\operatorname{Vol}(\Delta)=|\det V|/6$. Invertibility of $V$ certifies affine independence and unique channel cancellation, although it does not by itself fix the signs of all sector weights. The diagonal entries $V_{k,k}=p_{k,k}-p_{0,k}$ measure directional contrast, whereas the off-diagonal entries $V_{a,k}=p_{k,a}-p_{0,a}$ ($a\ne k$) measure cross-talk; a uniform diagonal-dominance margin provides a sufficient stability certificate.

\subsubsection{Concrete Parameterized Model: Symmetric CSS Codes under Correlated Pauli Noise}
\label{subsubsec:app_concrete_parameterized_model}

To illustrate how a well-conditioned contrast matrix $\bm{V}$ arises physically, we consider a parameterized centroid model for a symmetric CSS code under correlated Pauli noise. Suppose the first two disagreement sectors have leading centroids $(a,0,0)$ and $(0,b,0)$, with directional confidences $a=p_{c,X}>0$ and $b=p_{c,Z}>0$. In the third contested sector, suppose a fraction $\rho$ of shots arises from genuine $Y$ faults with error probability $c=p_{c,Y}>0$. The remaining fraction $1-\rho$ arises from independent binary $X$ and $Z$ faults with conditional probabilities $a$ and $b$. Under this mixture model, the third-sector centroid components are
\[
p_{3,1}=(1-\rho)a(1-b)+O(p_L),\quad
p_{3,2}=(1-\rho)(1-a)b+O(p_L),\quad
p_{3,3}=(1-\rho)ab+\rho c+O(p_L),
\]
where all background and reference centroid corrections are assumed to be $O(p_L)$ in the limit considered.

Subtracting the reference centroid $\bm{p}_0 = \mathcal{O}(p_L)$ yields the coordinate contrast matrix $\bm{V} = \bm{V}_{\mathrm{LO}} + \Delta\bm{V}$, where the leading-order matrix is
\begin{equation}
\bm{V}_{\mathrm{LO}} \equiv \begin{pmatrix} p_{c,X} & 0 & (1 - \rho) p_{c,X}(1 - p_{c,Z}) \\ 0 & p_{c,Z} & (1 - \rho)(1 - p_{c,X}) p_{c,Z} \\ 0 & 0 & (1 - \rho) p_{c,X} p_{c,Z} + \rho p_{c,Y} \end{pmatrix},
\label{eq:app_V_LO_def}
\end{equation}
and $\Delta\bm{V} = \mathcal{O}(p_L)$.
Because $\bm{V}_{\mathrm{LO}}$ is upper triangular, its determinant is given directly by the product of its diagonal elements,
\begin{equation}
\det(\bm{V}_{\mathrm{LO}}) = p_{c,X} p_{c,Z} \big[ (1 - \rho) p_{c,X} p_{c,Z} + \rho p_{c,Y} \big].
\label{eq:app_det_V_LO}
\end{equation}
For all physical classifier confidences $p_{c,X}, p_{c,Z}, p_{c,Y} > 0$ and all correlation fractions $\rho \in [0, 1]$, the diagonal entries are strictly positive, guaranteeing $\det(\bm{V}_{\mathrm{LO}}) > 0$.
At the canonical graph-decoder tie-break baseline ($p_{c,X} = p_{c,Z} = p_{c,Y} = 1/2$), the entries evaluate to $V_{13} = V_{23} = (1-\rho)/4$ and $V_{33} = (1+\rho)/4$, where the denominator $4$ arises from $(p_c)^2 = (1/2)^2 = 1/4$.
In the idealized oracle limit ($p_{c,Q} \to 1$), the off-diagonal entries vanish identically and $\lim_{p_{c,Q} \to 1} \bm{V}_{\mathrm{LO}} = \bm{I}_{3 \times 3}$ for all $\rho \in [0, 1]$.

Two distinct small parameters control the structure of $\bm{V}$.
The perturbation $\Delta\bm{V} = \bm{V} - \bm{V}_{\mathrm{LO}} = \mathcal{O}(p_L)$ measures departure from the leading-order form, whereas channel decoupling is governed by the off-diagonal component
\begin{equation}
\Delta\bm{V}_{\mathrm{off}} \equiv \bm{V} - \operatorname{diag}(V_{1,1}, V_{2,2}, V_{3,3}).
\label{eq:app_delta_V_off_def}
\end{equation}
The two parameters satisfy the norm bound
\begin{equation}
\|\Delta\bm{V}_{\mathrm{off}}\|_\infty \le \|\Delta\bm{V}\|_\infty + (1-\rho)\max\{p_{c,X}(1-p_{c,Z}),(1-p_{c,X})p_{c,Z}\},
\label{eq:app_delta_V_off_bound}
\end{equation}
showing that the leading cross-talk vanishes in either the correlated ($\rho\to1$) or high-confidence ($p_{c,Q}\to1$) limit once the background shift $\|\Delta\bm V\|_\infty$ is controlled.

\subsubsection{Affine Independence, L\'evy--Desplanques Diagonal Dominance, and Condition Number Bounds}
\label{subsubsec:app_levy_desplanques}
\label{subsec:app_levy_desplanques_stability}

By the L\'evy--Desplanques theorem and Varah's inverse norm bound, strict row-wise diagonal dominance of the displacement matrix $\bm{V}$ in Eq.~\eqref{eq:app_displacement_matrix_V} guarantees non-singularity; a uniformly positive margin and bounded matrix norm additionally give $O(1)$ spectral conditioning.

\begin{lemma}[L\'evy--Desplanques Theorem and Varah Bound]
\label{lem:app_levy_desplanques}
\label{lem:levy_desplanques_dominance}
Let $\bm{M} = (M_{ij}) \in \RR^{n \times n}$ be a strictly diagonally dominant matrix whose row margins satisfy
\begin{equation}
d_i \equiv |M_{ii}| - \sum_{j \neq i} |M_{ij}| > 0 \qquad \text{for all } i \in \{1, \dots, n\}.
\label{eq:app_levy_desplanques_condition}
\end{equation}
Then $\bm{M}$ is non-singular ($\det(\bm{M}) \neq 0$), and its inverse satisfies Varah's infinity-norm bound
\begin{equation}
\|\bm{M}^{-1}\|_\infty \le \frac{1}{\min_{1 \le i \le n} d_i}.
\label{eq:app_varah_bound}
\end{equation}
\end{lemma}

\begin{proof}
For any non-zero vector $\bm{x} \in \RR^n \setminus \{\bm{0}\}$, choose an index $k \in \{1, \dots, n\}$ attaining the maximum component magnitude $|x_k| = \|\bm{x}\|_\infty > 0$.
By the reverse triangle inequality on row $k$ of $\bm{M}\bm{x}$,
\begin{equation}
\|\bm{M}\bm{x}\|_\infty \ge |(\bm{M}\bm{x})_k| = \left| M_{kk} x_k + \sum_{j \neq k} M_{kj} x_j \right| \ge |M_{kk}| |x_k| - \sum_{j \neq k} |M_{kj}| |x_j| \ge d_k \|\bm{x}\|_\infty \ge \left(\min_{1 \le i \le n} d_i\right) \|\bm{x}\|_\infty > 0.
\end{equation}
Hence $\bm{M}\bm{x} \neq \bm{0}$ for all $\bm{x} \neq \bm{0}$, certifying $\det(\bm{M}) \neq 0$.
Substituting $\bm{x} = \bm{M}^{-1}\bm{y}$ proves Eq.~\eqref{eq:app_varah_bound}.
\end{proof}

\begin{theorem}[Quantitative diagonal-dominance certificate]
\label{thm:app_levy_desplanques_pauli}\label{thm:levy_desplanques_margin}
For the leading-order family in Eq.~\eqref{eq:app_V_LO_def}, write $a=p_{c,X}$ for the $X$-sector parameter, $b=p_{c,Z}$ for the $Z$-sector parameter, and $c=p_{c,Y}$ for the $Y$-sector confidence in $(0,1]$, and define the minimum margin $\delta=\min\{ab,c\}>0$. Its three row margins are
\begin{equation}
a[b+\rho(1-b)],\quad b[a+\rho(1-a)],\quad(1-\rho)ab+\rho c,
\label{eq:app_strict_diagonal_dominance}
\end{equation}
each of which is bounded below by $\delta$. If the perturbation satisfies $\epsilon_V=\|\Delta V\|_\infty<\delta$, then $V=V_{\rm LO}+\Delta V$ is invertible and obeys $\|V^{-1}\|_\infty\le(\delta-\epsilon_V)^{-1}$. In particular, $\delta\ge\delta_0>0$ and $\epsilon_V\le\delta_0/2$ imply
\[
\sigma_{\min}(V)\ge\delta_0/(2\sqrt3),\qquad
\kappa_2(V)\le6\sqrt3/\delta_0.
\]
Uniformity in a noise or distance limit requires $\delta_0$ and the deviation bound to hold uniformly in that limit.
\end{theorem}
\begin{proof}
Subtracting the off-diagonal entries from the diagonal entries gives the three row margins in Eq.~\eqref{eq:app_strict_diagonal_dominance}, where the third row is bounded below by $\min(ab,c)$. Under a perturbation of infinity norm $\epsilon_V$, each row margin remains at least $\delta-\epsilon_V$, so Lemma~\ref{lem:app_levy_desplanques} and the norm inequality $\|V^{-1}\|_2\le\sqrt3\|V^{-1}\|_\infty$ yield the stated bound on $\sigma_{\min}(V)$. Because the centroids lie in the probability simplex, $\|V\|_2\le\|V\|_F\le3$, proving the condition-number bound. The linear path $V_{\rm LO}+t\Delta V$ remains non-singular for all $0\le t\le1$, ensuring that
\begin{equation}
\det V>0,\qquad \det V_{\rm LO}=ab[(1-\rho)ab+\rho c]>0.
\label{eq:app_det_V_bound}
\end{equation}
Consequently, the four sector centroids are affinely independent. When $c\ge ab$ (including the equal-confidence case), the common row margin simplifies to $ab$. Near the oracle point $p_{c,Q}\to1$, the condition number satisfies $\kappa_2(V)=1+O(\max_Q|1-p_{c,Q}|+\|\Delta V\|_2)$. Operating below threshold does not by itself establish a uniform margin or Hypothesis~H2, nor does failure of this sufficient diagonal-dominance condition imply infeasibility when $V$ remains invertible.
\end{proof}

\subsubsection{Between-Sector Covariance Factorization and Closed-Form Sector Weights}
\label{subsubsec:app_covariance_factorization_weights}

To connect the variational covariance formula in Eq.~\eqref{eq:discrete_optimal_weights0} with the direct matrix-inversion weights in Eq.~\eqref{eq:four_sector_closed_form_weights}, we express the centered sector shifts using the contrast matrix $\bm{V}$.
Let $\bar{\bm{p}} \equiv \sum_{k=0}^M \pi_k \bm{p}_k$ denote the unconditional centroid.
Because $\sum_{k=0}^M \pi_k = 1$, the consensus shift satisfies
\begin{equation}
\bm{p}_0 - \bar{\bm{p}} = \bm{p}_0 - \sum_{k=0}^M \pi_k \bm{p}_k = (1 - \pi_0)\bm{p}_0 - \sum_{k=1}^M \pi_k \bm{p}_k = -\sum_{k=1}^M \pi_k (\bm{p}_k - \bm{p}_0) = -\bm{V} \bm{\pi}_{1:M},
\label{eq:app_p0_minus_pbar}
\end{equation}
where $\bm{\pi}_{1:M} \equiv (\pi_1, \dots, \pi_M)^\top \in \RR^M$.
For each directional sector $k \in \{1, \dots, M\}$, the centroid displacement evaluates to
\begin{equation}
\bm{p}_k - \bar{\bm{p}} = (\bm{p}_k - \bm{p}_0) + (\bm{p}_0 - \bar{\bm{p}}) = \bm{v}_k - \bm{V}\bm{\pi}_{1:M} = \bm{V}(\bm{e}_k - \bm{\pi}_{1:M}),
\label{eq:app_pk_minus_pbar}
\end{equation}
where $\bm{e}_k$ is the $k$-th standard basis vector in $\RR^M$, factoring $\bm{V}$ out of every sector difference vector.

\begin{proposition}[Between-Sector Covariance Factorization]
\label{thm:app_covariance_factorization}
When $K = M+1$, the between-sector covariance matrix $\bm{\Sigma}_S \equiv \sum_{k=0}^M \pi_k (\bm{p}_k - \bar{\bm{p}})(\bm{p}_k - \bar{\bm{p}})^\top \in \RR^{M \times M}$ admits an exact congruent factorization,
\begin{equation}
\bm{\Sigma}_S = \bm{V} \bm{\Sigma}_\pi \bm{V}^\top,
\label{eq:app_sigma_S_factorization}
\end{equation}
where $\bm{\Sigma}_\pi \in \RR^{M \times M}$ is the categorical dispersion matrix of the directional sector indicators,
\begin{equation}
\bm{\Sigma}_\pi \equiv \pi_0 \bm{\pi}_{1:M}\bm{\pi}_{1:M}^\top + \sum_{k=1}^M \pi_k (\bm{e}_k - \bm{\pi}_{1:M})(\bm{e}_k - \bm{\pi}_{1:M})^\top.
\end{equation}
Furthermore, the inverse matrix $\bm{\Sigma}_S^{-1}$ evaluates to
\begin{equation}
\bm{\Sigma}_S^{-1} = \bm{V}^{-\top} \bm{\Sigma}_\pi^{-1} \bm{V}^{-1} = \bm{V}^{-\top} \left( \bm{D}_\pi^{-1} + \frac{\bm{1}_M \bm{1}_M^\top}{\pi_0} \right) \bm{V}^{-1},
\label{eq:app_sigma_S_inverse_exact}
\end{equation}
where $\bm{D}_\pi \equiv \operatorname{diag}(\pi_1, \dots, \pi_M)$ and $\bm{1}_M \equiv (1, \dots, 1)^\top \in \RR^M$.
\end{proposition}

\begin{proof}
Substituting Eqs.~\eqref{eq:app_p0_minus_pbar} and~\eqref{eq:app_pk_minus_pbar} into the between-sector covariance definition yields
\begin{align}
\bm{\Sigma}_S &\equiv \sum_{k=0}^M \pi_k (\bm{p}_k - \bar{\bm{p}})(\bm{p}_k - \bar{\bm{p}})^\top \nonumber \\
&= \pi_0 \big(-\bm{V}\bm{\pi}_{1:M}\big)\big(-\bm{V}\bm{\pi}_{1:M}\big)^\top + \sum_{k=1}^M \pi_k \big[\bm{V}(\bm{e}_k - \bm{\pi}_{1:M})\big]\big[\bm{V}(\bm{e}_k - \bm{\pi}_{1:M})\big]^\top \nonumber \\
&= \bm{V} \left[ \pi_0 \bm{\pi}_{1:M}\bm{\pi}_{1:M}^\top + \sum_{k=1}^M \pi_k (\bm{e}_k - \bm{\pi}_{1:M})(\bm{e}_k - \bm{\pi}_{1:M})^\top \right] \bm{V}^\top = \bm{V} \bm{\Sigma}_\pi \bm{V}^\top.
\end{align}
Since $\bm{V}$ is non-singular by Theorem~\ref{thm:app_levy_desplanques_pauli} and $\bm{\Sigma}_\pi$ is invertible by Lemma~\ref{lem:app_sherman_morrison_categorical}, inverting both sides gives $\bm{\Sigma}_S^{-1} = \bm{V}^{-\top} \bm{\Sigma}_\pi^{-1} \bm{V}^{-1}$.
Substituting Eq.~\eqref{eq:app_sigma_q_inverse} completes the proof.
\end{proof}

\begin{lemma}[Sherman--Morrison Inversion of the Categorical Covariance Matrix]
\label{lem:app_sherman_morrison_categorical}
If $K = M+1$, the categorical covariance matrix simplifies to the canonical diagonal-minus-rank-one form,
\begin{equation}
\bm{\Sigma}_\pi = \bm{D}_\pi - \bm{\pi}_{1:M}\bm{\pi}_{1:M}^\top,
\label{eq:app_sigma_q_canonical}
\end{equation}
where $\bm{D}_\pi \equiv \operatorname{diag}(\pi_1, \dots, \pi_M)$, and its exact inverse is given by
\begin{equation}
\bm{\Sigma}_\pi^{-1} = \bm{D}_\pi^{-1} + \frac{\bm{1}_M \bm{1}_M^\top}{\pi_0}.
\label{eq:app_sigma_q_inverse}
\end{equation}
\end{lemma}

\begin{proof}
Expanding the sum over $k \in \{1, \dots, M\}$ gives
\begin{align}
\sum_{k=1}^M \pi_k (\bm{e}_k - \bm{\pi}_{1:M})(\bm{e}_k - \bm{\pi}_{1:M})^\top &= \sum_{k=1}^M \pi_k \bm{e}_k \bm{e}_k^\top - \Big(\sum_{k=1}^M \pi_k \bm{e}_k\Big)\bm{\pi}_{1:M}^\top - \bm{\pi}_{1:M}\Big(\sum_{k=1}^M \pi_k \bm{e}_k\Big)^\top + \Big(\sum_{k=1}^M \pi_k\Big)\bm{\pi}_{1:M}\bm{\pi}_{1:M}^\top \nonumber \\
&= \operatorname{diag}(\bm{\pi}_{1:M}) - (1 + \pi_0)\bm{\pi}_{1:M}\bm{\pi}_{1:M}^\top.
\end{align}
Adding $\pi_0 \bm{\pi}_{1:M}\bm{\pi}_{1:M}^\top$ reduces $\bm{\Sigma}_\pi$ to Eq.~\eqref{eq:app_sigma_q_canonical}.
Applying the Sherman--Morrison rank-1 matrix inversion formula gives
\begin{equation}
\big(\bm{D}_\pi - \bm{\pi}_{1:M}\bm{\pi}_{1:M}^\top\big)^{-1} = \bm{D}_\pi^{-1} + \frac{\bm{D}_\pi^{-1}\bm{\pi}_{1:M}\bm{\pi}_{1:M}^\top \bm{D}_\pi^{-1}}{1 - \bm{\pi}_{1:M}^\top \bm{D}_\pi^{-1}\bm{\pi}_{1:M}}.
\end{equation}
Observing that $\bm{D}_\pi^{-1}\bm{\pi}_{1:M} = \bm{1}_M \equiv (1, \dots, 1)^\top \in \RR^M$ and $1 - \bm{\pi}_{1:M}^\top \bm{D}_\pi^{-1}\bm{\pi}_{1:M} = 1 - \sum_{k=1}^M \pi_k = \pi_0$, the inverse matrix evaluates to Eq.~\eqref{eq:app_sigma_q_inverse}.
\end{proof}

\begin{proposition}[Variational Collapse to Closed-Form Sector Weights and Overhead]
\label{thm:app_variational_collapse}
For $K = M+1$ syndrome sectors, the variational filter weights $w_k^* = 1 - \bar{\bm{p}}^\top \bm{\Sigma}_S^{-1} (\bm{p}_k - \bar{\bm{p}})$ collapse identically to the unique algebraic solution of $\bm{A}_{\mathrm{sys}}\bm{w} = \bm{e}_1$,
\begin{equation}
w_0^* = \frac{1 + \bm{1}_M^\top \bm{V}^{-1}\bm{p}_0}{\pi_0}, \qquad w_k^* = -\frac{(\bm{V}^{-1}\bm{p}_0)_k}{\pi_k} \qquad (k \in \{1, \dots, M\}),
\label{eq:app_weights_closed_form}
\end{equation}
and the minimal sampling overhead evaluates to the exact rational expression
\begin{equation}
C_S^* = \frac{\big(1 + \bm{1}_M^\top \bm{V}^{-1}\bm{p}_0\big)^2}{\pi_0} + \sum_{k=1}^M \frac{\big[(\bm{V}^{-1}\bm{p}_0)_k\big]^2}{\pi_k}.
\label{eq:app_overhead_exact}
\end{equation}
\end{proposition}

\begin{proof}
Theorem~\ref{thm:discrete_sector_bound} shows that the optimal discrete sector filter is given by the variational expression
\begin{equation}
w_k^* = 1 - \bar{\bm{p}}^\top \bm{\Sigma}_S^{-1} (\bm{p}_k - \bar{\bm{p}}).
\label{eq:app_variational_weights_start}
\end{equation}
Because $\bar{\bm{p}} = \bm{p}_0 + \bm{V}\bm{\pi}_{1:M}$, we define the normalized consensus coordinate vector
\begin{equation}
\bm{u}_0 \equiv \bm{V}^{-1}\bm{p}_0 \in \RR^M, \qquad \bm{V}^{-1}\bar{\bm{p}} = \bm{u}_0 + \bm{\pi}_{1:M}.
\end{equation}
Evaluating the product $\bm{\Sigma}_\pi^{-1} \bm{V}^{-1}\bar{\bm{p}}$ using Eq.~\eqref{eq:app_sigma_q_inverse} gives
\begin{equation}
\bm{\Sigma}_\pi^{-1}\bm{V}^{-1}\bar{\bm{p}} = \left( \bm{D}_\pi^{-1} + \frac{\bm{1}_M \bm{1}_M^\top}{\pi_0} \right) (\bm{u}_0 + \bm{\pi}_{1:M}) = \bm{D}_\pi^{-1}\bm{u}_0 + \frac{1 + \bm{1}_M^\top \bm{u}_0}{\pi_0} \bm{1}_M.
\label{eq:app_sigma_q_inv_V_inv_pbar}
\end{equation}
For the consensus sector ($k = 0$), the shift is $\bm{p}_0 - \bar{\bm{p}} = -\bm{V}\bm{\pi}_{1:M}$, so Eq.~\eqref{eq:app_variational_weights_start} becomes
\begin{equation}
w_0^* = 1 + \bm{\pi}_{1:M}^\top \bm{\Sigma}_\pi^{-1} \bm{V}^{-1}\bar{\bm{p}} = 1 + \bm{1}_M^\top \bm{u}_0 + \frac{1 - \pi_0}{\pi_0} \big(1 + \bm{1}_M^\top \bm{u}_0\big) = \frac{1 + \bm{1}_M^\top \bm{V}^{-1}\bm{p}_0}{\pi_0}.
\label{eq:app_w0_exact}
\end{equation}
For directional sectors ($k \in \{1, \dots, M\}$), the centroid shift is $\bm{p}_k - \bar{\bm{p}} = \bm{V}(\bm{e}_k - \bm{\pi}_{1:M})$.
Using Eq.~\eqref{eq:app_w0_exact}, we obtain
\begin{equation}
w_k^* = w_0^* - \bm{e}_k^\top \left( \bm{D}_\pi^{-1}\bm{u}_0 + \frac{1 + \bm{1}_M^\top \bm{u}_0}{\pi_0} \bm{1}_M \right) = w_0^* - \left( \frac{u_{0,k}}{\pi_k} + w_0^* \right) = -\frac{(\bm{V}^{-1}\bm{p}_0)_k}{\pi_k}.
\label{eq:app_wk_exact}
\end{equation}
Direct substitution confirms $\sum_{k=0}^M \pi_k w_k^* = 1$ and $\sum_{k=0}^M \pi_k w_k^* \bm{p}_k = \bm{0}$.
Evaluating the minimal sampling overhead $C_S^* = \sum_{k=0}^M \pi_k (w_k^*)^2$ by substituting Eq.~\eqref{eq:app_weights_closed_form} yields Eq.~\eqref{eq:app_overhead_exact}, completing the proof.
\end{proof}

The exact weights in Eq.~\eqref{eq:app_weights_closed_form} satisfy $w_k=-u_k/\pi_k$ and $w_0=(1+\sum_k u_k)/\pi_0$, clarifying the sign regime and postselection limit of the four-sector filter. In the standard regime where $u_k\ge0$, the consensus weight satisfies $w_0\ge1/\pi_0$ and the disagreement weights satisfy $w_k\le0$ (with $w_k<0$ wherever $u_k>0$), although the algebraic formulas remain valid for mixed signs whenever $\bm{V}$ is invertible. When the consensus sector is completely error-free ($p_0=0$, or asymptotically $u_k/\pi_k\to0$), the filter reduces to agreement-only filtering with consensus weight $w_0=1/\pi_0$, contested weights $w_k=0$, and overhead $C_S^*=1/\pi_0$.

The closed-form sector weights extend directly to reduced active error sets ($M < 4^n - 1$) under overflow gating ($w_\perp = 0$). When out-of-support discrepancies $\mathcal{S}_\perp$ occur with share $\pi_\perp \in [0, 1)$, the retained sector probabilities sum to $1 - \pi_\perp$ rather than unity, so Lemma~\ref{lem:app_sherman_morrison_categorical} does not apply directly. Expressing trace preservation and active-error cancellation in terms of the probability-weighted masses $c_k \equiv \pi_k w_k$ instead yields the linear system
\begin{equation}
c_0 + \bm{1}_M^\top \bm{c}_{1:M} = 1, \qquad \bm{p}_0 + \bm{V}\bm{c}_{1:M} = \bm{0}.
\label{eq:app_overflow_mass_system}
\end{equation}
Solving Eq.~\eqref{eq:app_overflow_mass_system} recovers the closed-form weights and minimal sampling overhead in Eqs.~\eqref{eq:app_weights_closed_form}--\eqref{eq:app_overhead_exact} with retained consensus share $\pi_0 = 1 - \pi_\perp - t$, replacing $1 - t$ by $\pi_0$ in Eq.~\eqref{eq:app_overhead_sum_completed_squares}. On retained runs, any unmodeled multi-fault logical error $P \notin \mathcal{Q} \cup \{I\}$ leaves a residual channel coefficient
\begin{equation}
\EE_{\bm{s}}\big[w^*(\bm{s}) p_P(\bm{s})\big] = c_0 p_{0,P} + \sum_{k=1}^M c_k p_{k,P},
\label{eq:app_overflow_collision_residual}
\end{equation}
which scales as $\mathcal{O}(\alpha p_L^2)$ whenever sparse multi-fault collisions obey $p_{0,P} = \mathcal{O}(\alpha^2 p_L^2)$, $p_{k,P} = \mathcal{O}(p_L)$, and $\|\bm{u}_0\|_1 = \mathcal{O}(\alpha p_L)$. Finally, because the modeled error entries in each column of $\bm{V}$ sum to at most $1 - p_{k,I}$ within sector $\mathcal{S}_k$, the column diagonal-dominance margin satisfies
\begin{equation}
\delta_k^{\mathrm{col}} \equiv |V_{k,k}| - \sum_{a \ne k} |V_{a,k}| \ge 2p_{k,k} + p_{k,I} - 1 - \|\bm{p}_0\|_1.
\label{eq:app_column_margin_M_free}
\end{equation}
Consequently, the Varah inverse bound $\|\bm{V}^{-1}\|_1 \le (\min_k \delta_k^{\mathrm{col}})^{-1}$ remains independent of the active-channel count $M$, and overflow gating preserves both bounded invertibility and quadratic collision suppression.

\subsubsection{Contrast-Matrix Truncation at Finite Calibration Budgets}
\label{subsubsec:app_contrast_truncation}

Overflow gating is required for numerical stability whenever the digital-twin sample budget scales inversely with the logical error rate $p_L$.
For a local one- or two-qubit block, we split the non-identity logical Pauli errors $\mathcal{P}_n \setminus \{I\}$ into an active set $\mathcal{Q}$ and a rare complement $\mathcal{Q}_\perp$.
The active set collects the leading fault classes whose sector shares scale linearly with $p_L$.
On a single CSS patch these are $\{X, Z\}$, while on a two-qubit $\mathrm{CNOT}$ gate they form the six single-chain classes $\{X_1, X_2, X_1 X_2, Z_1, Z_2, Z_1 Z_2\}$.
Because Eq.~\eqref{eq:app_overflow_mass_system} applies to any active subset once out-of-support disagreements receive zero weight, the union $\{I\} \cup \mathcal{Q}$ need not form a Pauli subgroup.
By contrast, multi-chain classes in $\mathcal{Q}_\perp$ have suppressed shares of order $\mathcal{O}(p_L^2)$.
Every conditional entry in a rare sector divides one second-order event rate by another and remains of order unity, as in the third column of Eq.~\eqref{eq:app_V_LO_def}.
Estimating those entries to relative accuracy $\varepsilon$ requires a calibration budget satisfying
\begin{equation}
N_{\mathrm{cal}} = \mathcal{O}\!\left(\frac{1}{\varepsilon^2 \pi_j}\right) = \mathcal{O}\!\left(\frac{1}{\varepsilon^2 p_L^2}\right),
\label{eq:app_rare_sector_sample_cost}
\end{equation}
for every $j \in \mathcal{Q}_\perp$.
This cost is quadratically larger than the $\mathcal{O}((\varepsilon^2 p_L)^{-1})$ budget needed for active sectors.

Retaining under-sampled rare columns in the full linear solve corrupts the active weights through block inversion.
At a practical budget of order $1/p_L$, each class in $\mathcal{Q}_\perp$ accumulates only $\mathcal{O}(p_L)$ expected samples, so its empirical columns fluctuate at order unity and its diagonal block can become singular.
Partitioning the untruncated balance system of Proposition~\ref{thm:app_variational_collapse} over $\mathcal{Q}$ and $\mathcal{Q}_\perp$ gives
\begin{equation}
\begin{pmatrix} \bm{V}_{\mathcal{Q}\mathcal{Q}} & \bm{V}_{\mathcal{Q}\mathcal{Q}_\perp} \\ \bm{V}_{\mathcal{Q}_\perp\mathcal{Q}} & \bm{V}_{\mathcal{Q}_\perp\mathcal{Q}_\perp} \end{pmatrix}
\begin{pmatrix} \bm{u}_{\mathcal{Q}} \\ \bm{u}_{\mathcal{Q}_\perp} \end{pmatrix}
=
\begin{pmatrix} \bm{p}_{0,\mathcal{Q}} \\ \bm{p}_{0,\mathcal{Q}_\perp} \end{pmatrix}.
\label{eq:app_V_block_partition}
\end{equation}
Eliminating $\bm{u}_{\mathcal{Q}_\perp}$ yields the exact block solution
\begin{align}
\bm{S}_{\mathcal{Q}} &\equiv \bm{V}_{\mathcal{Q}\mathcal{Q}} - \bm{V}_{\mathcal{Q}\mathcal{Q}_\perp}\bm{V}_{\mathcal{Q}_\perp\mathcal{Q}_\perp}^{-1}\bm{V}_{\mathcal{Q}_\perp\mathcal{Q}},
\label{eq:app_schur_complement_def} \\
\bm{u}_{\mathcal{Q}} &= \bm{S}_{\mathcal{Q}}^{-1}\Bigl(\bm{p}_{0,\mathcal{Q}} - \bm{V}_{\mathcal{Q}\mathcal{Q}_\perp}\bm{V}_{\mathcal{Q}_\perp\mathcal{Q}_\perp}^{-1}\bm{p}_{0,\mathcal{Q}_\perp}\Bigr),
\label{eq:app_schur_contamination} \\
\bm{u}_{\mathcal{Q}_\perp} &= \bm{V}_{\mathcal{Q}_\perp\mathcal{Q}_\perp}^{-1}\Bigl(\bm{p}_{0,\mathcal{Q}_\perp} - \bm{V}_{\mathcal{Q}_\perp\mathcal{Q}}\bm{u}_{\mathcal{Q}}\Bigr),
\label{eq:app_schur_rare_coords}
\end{align}
where $\bm{S}_{\mathcal{Q}}$ is the Schur complement of the rare block $\bm{V}_{\mathcal{Q}_\perp\mathcal{Q}_\perp}$.
Because the inverse block $\bm{V}_{\mathcal{Q}_\perp\mathcal{Q}_\perp}^{-1}$ enters both $\bm{S}_{\mathcal{Q}}$ and the cross-term multiplying $\bm{p}_{0,\mathcal{Q}_\perp}$, sampling noise or near-singularity in the rare columns distorts the active solution $\bm{u}_{\mathcal{Q}}$.
The same inverse also magnifies errors in $\bm{u}_{\mathcal{Q}_\perp}$, where dividing by a small share $\pi_j$ inflates both the sector weight and the sampling overhead.
Dropping sparse rare sectors before inversion removes both instability paths.

Truncating the system to the active block $\bm{V}_{\mathcal{Q}\mathcal{Q}}$ via Eq.~\eqref{eq:app_overflow_mass_system} keeps the inversion compact and well-conditioned.
Local operations act on one or two logical qubits, so this reduced matrix has fixed small dimension: $2\times 2$ or $3\times 3$ for single-qubit gates, and $6\times 6$ for two-qubit Clifford gates.
Matrix inversion remains essential whenever single-chain faults couple distinct active classes, as occurs within the $X$-type and $Z$-type triples of a $\mathrm{CNOT}$ gate.
Solving the full active block captures this cross-talk at negligible classical cost, and on CSS gates the matrix decouples to leading order into separate $X$-type and $Z$-type subblocks.
Any omitted multi-chain channel then leaves at most the $\mathcal{O}(\alpha p_L^2)$ collision residual bounded by Eq.~\eqref{eq:app_overflow_collision_residual}.

At deeper calibration budgets, a rare class $Q \in \mathcal{Q}_\perp$ can enter $\mathcal{Q}$ adaptively once the count $N_{Q,Q}$ of runs in sector $\mathcal{S}_Q$ that incur logical error $Q$ reaches a cutoff $N_{\min}$ of roughly twenty.
This gate stabilizes the column margin in Eq.~\eqref{eq:app_column_margin_M_free}.
Whenever $p_{Q,Q}$ is of order unity, as for the $Y$ sector in Fig.~\ref{fig:dual_decoder}d, gating on the total sector count $N_Q$ at the same cutoff gives comparable conditioning.
Adaptive truncation thus cancels the $\mathcal{O}(\alpha p_L^2)$ multi-chain floor at large budgets without risking small-budget instability.

\subsubsection{Conditional Logical Overhead Scaling}
\begin{theorem}[Sector-probability control and linear overhead]
\label{thm:app_asymptotic_overhead}
For positive sector probabilities and an invertible contrast matrix $V$, define the coordinate vector $u=V^{-1}p_0$, total contested probability $t=\sum_{k\ge1}\pi_k<1$, and component sum $U=\sum_{k\ge1}u_k$. Suppose $|u_k|\le L\pi_k$ and $t\le Bp_L$. Then the excess sampling overhead satisfies
\begin{align}
C_S^*-1&=\sum_{k\ge1}\frac{(\pi_k+u_k)^2}{\pi_k}
+\frac{(t+U)^2}{1-t}\nonumber\\
&\le(1+L)^2\frac{t}{1-t}.
\label{eq:app_overhead_sum_completed_squares}
\end{align}
If the constants $B$ and $L$ are fixed in the stated limit and $p_L\to0$, then
\begin{equation}
C_S^*=1+(V^{-1}\bar{\bm p})^\top\operatorname{diag}(\bm\pi)^{-1}(V^{-1}\bar{\bm p})+O(p_L^2)
=1+O(p_L).
\label{eq:app_overhead_linear_scaling}
\end{equation}
\end{theorem}
\begin{proof}
Expanding $C_S^*=(1+U)^2/(1-t)+\sum_{k\ge1} u_k^2/\pi_k$ yields the exact identity in Eq.~\eqref{eq:app_overhead_sum_completed_squares}. Because $|\pi_k+u_k|\le(1+L)\pi_k$ and $|t+U|\le(1+L)t$, the inequality follows immediately. Moreover, the baseline error vector obeys
\begin{equation}
\bar{\bm p}=p_0+V\bm\pi,\qquad V^{-1}\bar{\bm p}=u+\bm\pi.
\label{eq:app_vector_sum_identity}
\end{equation}
The remaining reference-sector contribution $(t+U)^2/(1-t)$ is of order $O(t^2)=O(p_L^2)$, and for $t\le1/2$ the excess cost satisfies $C_S^*-1\le2B(1+L)^2p_L$. These bounds allow either sign of $u_k$ and arbitrarily disparate sector probabilities. A bounded inverse $\|V^{-1}\|$ alone does not control the rare-sector ratios $u_k/\pi_k$, and uniform-in-distance claims require $B$ and $L$ to be independent of distance, as in Corollary~\ref{cor:uniform_distance_overhead}.
\end{proof}

\begin{corollary}[Exactly diagonal contrast]
\label{cor:app_overhead_uncoupled}
Suppose $V=\operatorname{diag}(v_k)$ with diagonal contrast entries $v_k=p_{k,k}-p_{0,k}>0$, marginal error rates $\mu_k=\bar p_k>0$, and consensus-error ratios $\alpha_k=p_{0,k}/\mu_k<1$, where the first index of $p_{k,a}$ labels the sector. Under the conditions of Theorem~\ref{thm:app_asymptotic_overhead} and uniform positive lower bounds on $p_{k,k}$ and $1-\alpha_k$,
\begin{equation}
\frac{(\pi_k+u_k)^2}{\pi_k}
=\frac{\mu_k}{v_k(1-\alpha_k)},\qquad
C_S^*=1+\sum_k\frac{\mu_k}{p_{k,k}(1-\alpha_k)}+O(p_L^2).
\label{eq:app_overhead_completed_square}
\end{equation}
\aliaslabel{eq:app_overhead_uncoupled_asymptotic}{\ref{eq:app_overhead_completed_square}}
\end{corollary}
\begin{proof}
Eq.~\eqref{eq:app_vector_sum_identity} gives $\pi_k=(1-\alpha_k)\mu_k/v_k$ and $\pi_k+u_k=\mu_k/v_k$. Because $v_k=p_{k,k}-\alpha_k\mu_k$, replacing $v_k$ with $p_{k,k}$ incurs an $O(\mu_k^2)$ correction under the stated uniform lower bounds. Summing over $k$ and applying Theorem~\ref{thm:app_asymptotic_overhead} yields Eq.~\eqref{eq:app_overhead_completed_square}; when off-diagonal entries are present in arbitrarily rare sectors, the full matrix formula in Eq.~\eqref{eq:app_overhead_linear_scaling} should be used instead.
\end{proof}

\subsection{Finite Calibration Errors and Consensus Suppression}
\label{subsec:app_proof_distance_suppression}\label{subsec:app_residual_error_distance_suppression}
\label{subsec:coupled_sector_derivations}\label{sec:appendix_residual_error}
We now prove Theorem~\ref{thm:app_distance_suppression} (Eq.~\eqref{eq:four_sector_distance_suppression2} of the main text), showing that componentwise relative estimation errors in the dual-decoder protocol are attenuated by the small consensus logical error norm $\|p_0\|_1$.
When overflow gating restricts the linear solve to an active subset $\mathcal{Q}$ as in Section~\ref{subsubsec:app_contrast_truncation}, the analysis below applies directly to the retained submatrix $V_{\mathcal{Q}\mathcal{Q}}$.
The uniform relative-error bounds are then needed only on well-sampled active sectors at $\mathcal{O}((\varepsilon^2 p_L)^{-1})$ calibration cost, while omitted multi-chain sectors contribute the separate $\mathcal{O}(\alpha p_L^2)$ collision residual of Eq.~\eqref{eq:app_overflow_collision_residual}.
In sector-first coordinates $p_{k,a}$, suppose that the inverse contrast matrix satisfies $\|V^{-1}\|_1\le K_0$, the consensus error norm satisfies $a_0=\|p_0\|_1\le1$, and the sector probabilities satisfy $\pi_k>0$ and $\sum_k\hat\pi_k=1$. Assume the relative error bounds
\[
|\hat\pi_k/\pi_k-1|\le\varepsilon,\qquad
|\hat p_{k,a}-p_{k,a}|\le\varepsilon p_{k,a},\qquad
\varepsilon\le[4(1+6K_0)]^{-1}.
\]
Define the calibrated coordinate vector $\hat u=\hat V^{-1}\hat p_0$ together with the probability-weighted filter coefficients $c_0=1+\mathbf1^\top\hat u$ and $c_k=-\hat u_k$ ($k\ge1$). The raw sector weights $\hat w_k=c_k/\hat\pi_k$ then satisfy $\sum_k c_k=1$ and $\sum_k c_k\hat p_k=0$. With the probability ratios $t_k=\pi_k/\hat\pi_k$, the true population mean weight $b$, the raw residual error vector $r$, and the population-normalized residual vector $p_{\rm res}$ take the form
\begin{equation}
b=\sum_kt_kc_k,\quad r=\sum_kt_kc_kp_k
=\sum_kc_k(t_kp_k-\hat p_k),\qquad p_{\rm res}=r/b.
\label{eq:app_p_res_def}
\end{equation}
Normalizing the raw residual vector $r$ by $b$ isolates the traceless channel error $p_{\rm res}$ analyzed in Proposition~\ref{thm:app_exact_differential} and Theorem~\ref{thm:app_distance_suppression}.

By contrast, the raw unnormalized estimator of Eq.~\eqref{eq:estimator} incurs both this channel residual and a first-order shift in its mean weight $b$. Its expected superoperator decomposes into a scalar trace factor and a traceless Pauli-channel residual,
\begin{equation}
\mathcal{E}_{\mathrm{VEC}} = b\Id + \sum_{\alpha} r_{\alpha}(\cE_{\alpha} - \Id),
\label{eq:app_raw_expected_map}
\end{equation}
which bounds the total observable bias by $|b - 1| + 2\|r\|_1$ (Eq.~\eqref{eq:raw_calibration_bias}). Under probability conservation across all sectors, the consensus-probability shift is controlled by the contested share $t = 1 - \pi_0$, yielding the trace-drift bound
\begin{align}
|\hat{\pi}_0 - \pi_0| &= \Bigl|\sum_{k\ge 1}(\hat{\pi}_k - \pi_k)\Bigr| \le \varepsilon t, \nonumber \\
|b - 1| &\le \frac{\varepsilon t}{\hat{\pi}_0}\bigl(1 + \|\hat{u}\|_1\bigr) + 2\varepsilon\|\hat{u}\|_1 = \mathcal{O}(\varepsilon p_L)
\label{eq:app_mean_weight_drift_bound}
\end{align}
whenever $t = \mathcal{O}(p_L)$ and $a_0 = \mathcal{O}(p_L)$. Whereas raw unnormalized estimation retains this $\mathcal{O}(\varepsilon p_L)$ trace drift, population normalization ($p_{\rm res} = r/b$) isolates the smaller $\mathcal{O}(\varepsilon \alpha p_L)$ channel residual directly. In finite-sample experiments, empirical sector-mean or sample-mean normalization (Proposition~\ref{prop:sector_sample_means}) removes $b - 1$ up to exponentially small empty-sector terms or $\mathcal{O}(N^{-1})$ ratio bias.

\begin{proposition}[First differential and finite remainder]
\label{thm:app_exact_differential}
Let $u=V^{-1}p_0$ denote the exact coordinate vector, $\Delta p_k=\hat p_k-p_k$ the centroid shifts, and $\delta_k=\hat\pi_k/\pi_k-1$ the relative probability shifts. The first-order linear variation of $p_{\rm res}$ at exact calibration along this perturbation direction is
\begin{equation}
L=-(1+\mathbf1^\top u)\Delta p_0+
\sum_{k\ge1}u_k[\Delta p_k+(\delta_k-\delta_0)p_k],
\label{eq:app_dp_res_closed_form}
\end{equation}
and there exists a constant $K_2$ depending only on $K_0$ such that $\|p_{\rm res}-L\|_1\le K_2\varepsilon^2a_0$. In the binary two-sector model with reference and contested centroids $p_a$ and $p_c$, contested probability $\pi$, relative centroid errors $\varepsilon_a$ and $\varepsilon_c$, and relative probability shift $\delta_1=\varepsilon_m$, this linear variation simplifies to
\begin{equation}
L=\frac{p_ap_c}{p_c-p_a}
\left[-\varepsilon_a+\varepsilon_c+\frac{\varepsilon_m}{1-\pi}\right].
\label{eq:app_dp_res_binary_reduction}
\end{equation}
\end{proposition}
\begin{proof}
Differentiating $\sum_k c_k\hat p_k=0$ and $t_k=(1+\delta_k)^{-1}$ in Eq.~\eqref{eq:app_p_res_def} at zero calibration error (where $r=0$ and $b=1$, so normalization does not contribute at first order) yields Eq.~\eqref{eq:app_dp_res_closed_form}. Specializing to the binary model with $u=p_a/(p_c-p_a)$ and $\delta_0=-\pi\delta_1/(1-\pi)$ gives Eq.~\eqref{eq:app_dp_res_binary_reduction}.
To establish the finite-error remainder bound, observe that $\|\Delta V\|_1\le2\varepsilon$ implies $\|\hat V^{-1}\|_1\le2K_0$ and $\|\hat u\|_1\le3K_0a_0$, while the matrix inverse-difference identity yields $\|\hat u-u\|_1\le D_0\varepsilon a_0$ with $D_0=2K_0(1+2K_0)$. Expanding $t_k=1-\delta_k+\delta_k^2/(1+\delta_k)$ then gives
\[
\|r-L\|_1\le K_{\rm raw}\varepsilon^2a_0,\qquad
K_{\rm raw}=4D_0+2(1+6K_0),
\]
where the coefficient-variation terms are bounded by $4D_0\varepsilon^2a_0$ and the reciprocal remainders are bounded by $2(1+6K_0)\varepsilon^2a_0$. Because $|b-1|\le2\varepsilon(1+6K_0a_0)\le1/2$ and $\|L\|_1\le2(1+2K_0)\varepsilon a_0$, division by $b$ satisfies $\|p_{\rm res}-L\|_1\le K_2\varepsilon^2a_0$ with $K_2=2K_{\rm raw}+8(1+6K_0)(1+2K_0)$.
\end{proof}

\begin{lemma}[Common coordinate rescaling]
\label{lem:app_uniform_shift_invariance}
If $\hat p_k=A p_k$ for all sectors $k$ with a common invertible matrix $A$ and $\hat\pi=\pi$, then the calibrated sector weights are unchanged wherever the estimated centroids remain admissible.
\end{lemma}
\begin{proof}
Under $\hat p_k=Ap_k$, the contrast matrix transforms as $\hat V=AV$, which implies $\hat V^{-1}\hat p_0=V^{-1}p_0$ and leaves Eq.~\eqref{eq:app_weights_closed_form} invariant; rescaling only a subset of entries does not satisfy this hypothesis.
\end{proof}

\begin{theorem}[Finite-error consensus suppression]
\label{thm:app_distance_suppression}
Under the inverse-norm and componentwise relative-error bounds above, the unnormalized and population-normalized residual error vectors satisfy
\begin{equation}
\|r\|_1\le3(1+6K_0)\varepsilon\|p_0\|_1,\qquad
\|p_{\rm res}\|_1\le6(1+6K_0)\varepsilon\|p_0\|_1.
\label{eq:app_distance_suppression_final}
\end{equation}
\end{theorem}
\begin{proof}
The bound $\|\hat u\|_1\le3K_0a_0$ established in the proof of Proposition~\ref{thm:app_exact_differential} implies $\sum_k|c_k|\le1+6K_0a_0$ and $\sum_k|c_k|\|p_k\|_1\le(1+6K_0)a_0$. Moreover, $|t_k-1|\le2\varepsilon$ implies $\|t_kp_k-\hat p_k\|_1\le3\varepsilon\|p_k\|_1$. Substituting these estimates into Eq.~\eqref{eq:app_p_res_def} and using $b\ge1/2$ yields Eq.~\eqref{eq:app_distance_suppression_final} without requiring any negligible-cross-talk assumption; the constants are uniform wherever $K_0$ is uniform.
\end{proof}
Agreement postselection leaves the unmitigated consensus error $\|p_0\|_1=\alpha p_L$ intact at sampling cost $1/\pi_0$. Theorem~\ref{thm:app_distance_suppression} shows that dual-decoder VEC instead suppresses the population-normalized residual channel error by an additional factor of the relative calibration precision $\varepsilon$ (up to raw trace drift $|b-1|$ or $1/N$ sample-mean ratio bias); across the Fig.~\ref{fig:dual_decoder} data ($K_0 = \|\bm{V}^{-1}\|_1 \in [1.05, 3.1]$), the conservative worst-case prefactor $6(1 + 6K_0)$ lies in $[44, 118]$ for $\varepsilon \le [4(1 + 6K_0)]^{-1} \in [1.3\%, 3.4\%]$, whereas typical first-order residuals are governed by the $\mathcal{O}(1)$ linear sensitivity of Proposition~\ref{thm:app_exact_differential}.

\subsection{Exact Binary Decoder Asymmetry}
\label{subsec:app_decoder_asymmetry_derivation}
To derive the binary asymmetry relations in Eq.~\eqref{eq:dd_overhead_eta_formula} of the main text, we decompose a single binary logical error component into joint real-time and offline failure events.
For a specified binary logical component evaluated in a common true frame, decoder disagreement is given by the binary sum $E_{\rm rt}\oplus E_{\rm off}$. Let $A=\Pr(1,0)$ denote the probability of real-time-only failure, $B=\Pr(0,1)$ the probability of offline-only failure, and $J=\Pr(1,1)$ the probability of simultaneous failure. The marginal error rates are then $\mu=p_L=A+J$ and $\nu=p_L'=B+J$, while the sector shares satisfy $\pi_X=A+B$ and $\pi_I=1-A-B$ with overlap $J=\pi_Ip_I$. Whenever $A>0$, the failure-asymmetry ratio $b$ and the contested error rate $p_X$ obey the exact identities
\begin{equation}
b=\frac BA=\frac{p_L'-\pi_Ip_I}{p_L-\pi_Ip_I},\qquad
p_X=\frac A{A+B}=\frac1{1+b},
\label{eq:app_asymmetry_ratio_derivation}
\end{equation}
\aliaslabel{eq:app_pX_asymmetry_relation}{\ref{eq:app_asymmetry_ratio_derivation}}
holding without approximation for each binary component relative to its own agreement event.

Substituting Eq.~\eqref{eq:app_asymmetry_ratio_derivation} into the two-sector variance formula directly yields the exact sampling overhead and its leading slope $\eta_0$.
For $a=p_I$ and $0\le a<\mu<p_X\le1$, conservation of probability fixes $\pi_X=(\mu-a)/(p_X-a)$ and $\pi_I=(p_X-\mu)/(p_X-a)$. Exact error cancellation then gives the sector weights $w_I=p_X/(p_X-\mu)=1/[1-(1+b)\mu]$ and $w_X=-a/(\mu-a)=-\alpha/(1-\alpha)$, where $\alpha=a/\mu$ is the consensus-error ratio. Because the binary between-sector variance is $D=(\mu-a)(p_X-\mu)$, the sampling overhead evaluates to
\begin{equation}
C=1+\frac{\mu^2}{(\mu-a)(p_X-\mu)}
=1+\frac{(1+b)\mu}{(1-\alpha)[1-(1+b)\mu]}.
\label{eq:app_eta_asymmetry_relation}
\end{equation}
Thus $\eta_0=(1+b)/(1-\alpha)$ is the leading excess-overhead slope, while the finite-noise ratio equals $(C-1)/\mu=\eta_0/[1-(1+b)\mu]$. Whenever $p_X\ge p_{X,0}>0$ and $1-\alpha\ge a_0>0$ hold uniformly for $\mu\le p_{X,0}/2$, subtracting the linear term bounds the quadratic remainder as
\[
0\le C-1-\eta_0\mu\le2\mu^2/(a_0p_{X,0}^2).
\]
Two equally accurate decoders ($\nu=\mu$) yield $b=1$ and $\eta_0=2/(1-\alpha)\ge2$. When the offline reference is asymptotically much more accurate than the real-time decoder ($\mu\to0$ and $\nu/\mu\to0$), both ratios vanish ($\alpha\le\nu/(\pi_I\mu)\to0$ and $b\le\nu/(\mu-\nu)\to0$), achieving the unit-slope limit $\eta_0\to1$.

\begin{figure*}[!htbp]
    \centering
    \includegraphics[width=\linewidth]{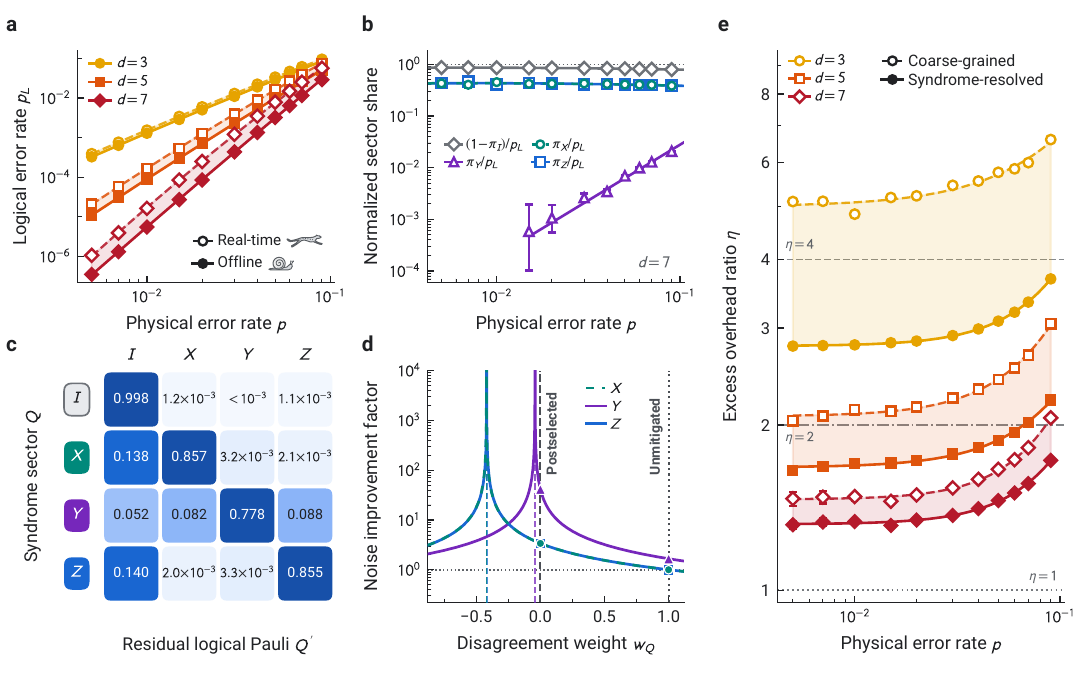}
\caption{\textbf{Code-capacity dual-decoder virtual error cancellation.}
The protocol pairs real-time minimum-weight perfect matching with offline Bravyi--Suchara--Vargo matrix-product-state maximum-likelihood decoding ($\chi_{\max} = 2^d$) and reweights experimental runs in post-processing; simulations model homogeneous two-dimensional code-capacity depolarizing noise on rotated surface codes.
\textbf{(a)}~Logical error rates of real-time matching ($p_L$, open markers) and exact offline decoding ($p_L'$, filled markers) for $d = 3, 5,$ and $7$, with shaded accuracy gaps.
\textbf{(b)}~Sector probabilities $\pi_Q$ (fraction of runs in $\mathcal{S}_Q$) normalized by $p_L$ at $d = 7$ ($p \ge 0.015$ for $\mathcal{S}_Y$, power-law fit over $p \ge 0.04$).
\textbf{(c)}~Conditional probability of residual logical Pauli error $Q'$ after real-time decoding within each syndrome sector $Q$ ($d = 7$, $p = 0.05$).
\textbf{(d)}~Noise improvement factor (baseline-to-residual error ratio in each Pauli channel $X, Y, Z$) versus disagreement weight $w_Q$ ($d = 7$, $p = 0.05$), with zero-bias weights $w_Q^* < 0$ marked by colored dashed lines.
\textbf{(e)}~Excess overhead ratio $\eta = (C - 1)/p_L$ for coarse-grained ($C_S^*$, open markers) and syndrome-resolved ($C_{\mathrm{syn}}^*$, filled markers) mitigation, compared with unconditioned logical PEC ($\eta = 4$), the gap-conditioned single-decoder benchmark ($\eta = 2$), and the ideal offline-decoder limit ($\eta = 1$).
Error bars show $1\text{ s.e.m.}$ ($68\%$ Clopper--Pearson intervals for rare $\mathcal{S}_Y$ counts in \textbf{(b)}) where larger than markers (see Supplementary Section~\ref{sec:methods_numerical}).}
\label{fig:supp_dual_decoder_2d}
\end{figure*}
\clearpage

\section{Noise-Learning Identifiability, Rare-Sector Integration, and Numerical Conditioning}
\label{sec:supp_noise_learning_identifiability}
\label{sec:topological_characters}
\label{sec:symmetry_reduction}

This section provides the mathematical derivations, parameter-recovery criteria, and numerical conditioning certificates supporting \vecTheMethodsSec{} and the dual-decoder and full-circuit protocols (Supplementary Sections~\ref{sec:supp_dual_decoder_protocol} and~\ref{sec:fc-protocol}).
Subsection~\ref{subsec:supp_identifiability_symmetry} analyzes cross-character sample covariances, prediction identifiability, and symmetry-group parameter tying for the Walsh--Hadamard estimator of Supplementary Section~\ref{subsec:methods_noise_learning_benchmark}.
In Subsection~\ref{subsec:supp_character_rank_analysis}, we prove that adjoining single-check (weight-1) syndrome characters with adjacent same-basis two-check modulo-2 (XOR) characters lifts the binary parity-check rank to full real column rank across all data qubits.
Subsection~\ref{subsec:model_sector_integration} formulates importance-sampled digital-twin integration for rare or empty experimental sectors, while Subsection~\ref{subsec:supp_ridge_regularization_conditioning} bounds the perturbation introduced by Tikhonov ridge damping in the syndrome-resolved covariance inversion of Supplementary Section~\ref{sec:appendix_proof_variational_cs_bound}.

\subsection{Character Covariance, Prediction Identifiability, and Symmetry Tying}
\label{subsec:supp_identifiability_symmetry}

Parity checks measured on the same $N_{\mathrm{learn}}$ syndrome records in Eq.~\eqref{eq:methods_fidelity_estimator} of \vecMethodsSec{} are statistically correlated across characters. Writing $Y_\chi=(-1)^{\chi\cdot s}$, the Walsh--Hadamard product identity $Y_\chi Y_\psi=Y_{\chi\oplus\psi}$ gives the exact sample covariance
\begin{equation}
\Cov(\hat F_\chi,\hat F_\psi)=
\frac{F_{\chi\oplus\psi}-F_\chi F_\psi}{N_{\mathrm{learn}}}.
\label{eq:supp_character_covariance}
\end{equation}
In a positive, well-resolved regime, the leading log-covariance divides this expression by $F_\chi F_\psi$. For fixed moments bounded away from zero (with a bounded definition on exceptional samples of negligible probability), the local delta-method expansion gives $\E\hat b_\chi^{\rm raw}=b_\chi+(1-F_\chi^2)/(2N_{\mathrm{learn}}F_\chi^2)+O(N_{\mathrm{learn}}^{-2})$, motivating the leading-order Jensen bias correction in Eq.~\eqref{eq:methods_jensen_shrinkage}. In the weighted non-negative least-squares estimator of Eq.~\eqref{eq:methods_wnnls_estimation}, using bounded diagonal weights $W_{\chi\chi}=\hat F_\chi^2/\max(1-\hat F_\chi^2,1/N_{\mathrm{learn}})$ provides an effective diagonal approximation to generalized least squares that changes inferred-rate standard errors by only $\sim 5\%$ under joint-record bootstrap.

Syndrome-only inference is conditional on a correctly specified physical noise family, which may fix unresolved relative fault rates that cannot be distinguished from syndrome data alone. Let $\phi$ parameterize the allowed noise family and let $L$ denote the residual logical label in the committed frame. Sufficient prediction identifiability requires that any two parameters $\phi$ and $\phi'$ producing identical selected syndrome statistics also yield the same posterior error probabilities $\Pr(L\mid S)$ (or fixed-sector centroids). Identifying the full parameter vector $\phi$ is a stronger sufficient condition, whereas misspecified rate ratios can leave residual bias even in the infinite-sample limit.

For independent binary mechanisms, the linear system $b=A\theta$ under a linear attenuation tie $\theta=C\phi$ is identifiable whenever $AC$ has full column rank. In that case, $AC(\phi-\phi')=0$ forces $\phi=\phi'$, whereas nonlinear rate relationships require global injectivity of the forward map instead of a local rank test alone.
Furthermore, independent fault mechanisms can be merged into a single effective mechanism with occurrence probability $p_{\rm eff}=[1-\prod_j(1-2p_j)]/2$ whenever both their detector signatures and their logical actions agree. Fault mechanisms that share the same detector signature but differ in logical action cannot be merged without a model tie, and any fault with trivial detector signature and nontrivial logical action is invisible to syndromes alone.
Spatial or temporal parameter ties are therefore applied only when justified by the physical noise model. In our inhomogeneous depolarizing benchmark, we tie the three Pauli error rates within each qubit while retaining distinct rates across different qubits, partitioning the fault set into equivalence orbits under a symmetry group $\mathcal G_{\rm sym}$.

Formally, the action of $\mathcal{G}_{\mathrm{sym}}$ on the fault set $\{1, \dots, N_E\}$ induces a partition into $M_{\mathrm{tied}} \le N_E$ disjoint equivalence orbits $\mathcal{O}_1, \dots, \mathcal{O}_{M_{\mathrm{tied}}}$ that reduces the number of free attenuation parameters.
We encode this partition in the symmetry tying matrix $\bm{C} \in \{0, 1\}^{N_E \times M_{\mathrm{tied}}}$, with entries
\begin{equation}
C_{jm} = \begin{cases} 1 & \text{if fault mechanism } j \in \mathcal{O}_m, \\ 0 & \text{otherwise}, \end{cases}
\label{eq:methods_symmetry_tying_matrix}
\end{equation}
such that each fault belongs to exactly one orbit, $\sum_{m=1}^{M_{\mathrm{tied}}} C_{jm} = 1$.
The columns of $\bm{C}$ are mutually orthogonal, satisfying $\bm{C}^\top\bm{C} = \mathrm{diag}(|\mathcal{O}_1|, \dots, |\mathcal{O}_{M_{\mathrm{tied}}}|)$. Parameter tying then constrains the vector $\bm{\theta}$ to the subspace $\bm{\theta} = \bm{C}\bm{\theta}_{\mathrm{tied}}$ with $\bm{\theta}_{\mathrm{tied}} \in \RR_{\ge 0}^{M_{\mathrm{tied}}}$ whenever the model ties the log-fidelity coordinates linearly.
For code-capacity depolarizing noise on a distance-$d = 7$ surface code (Supplementary Fig.~\ref{fig:supp_vlpec_validation_2d}a), this tying compresses $3N_q = 147$ individual Pauli fault rates to $M_{\mathrm{tied}} = 49$ per-qubit error rates.
In Eq.~\eqref{eq:methods_attenuation_linear_system}, our single-check and same-basis-pair characters have the exact modulo-2 support matrix $A_{\rm eff}$ (whereas mixed-basis products would require tracking overlapping $X$ and $Z$ factors that multiply to $Y$ instead of the identity). With $\theta_{\rm tied}=\theta$, the exponential inverse $p_q=(3/4)(1-e^{-\theta_q})$ in Eq.~\eqref{eq:methods_depolarizing_rate_mapping} is exact, avoiding the $O(p_q^2)$ discrepancy that would arise from approximating the isotropic channel by three independent Pauli flips of rate $p_q/3$.

In the three-dimensional circuit-level memory and lattice-surgery benchmarks, mechanism-level identifiability follows directly from hierarchical Boolean M\"obius inversion rather than isotropic Pauli tying. After exact Walsh--Hadamard conversion of disordered single- and two-qubit Pauli channels, every extracted detector-error-model support $E$ has size $1 \le |E| \le 4$, with neither invisible zero-detector faults nor opposite-logical support collisions across all odd distances $d \ge 3$. Forming the alternating subset sum of the unclipped empirical log-moments $b_A = -\ln \Fs(A)$ from Eq.~\eqref{eq:nl_linear_system} of \vecMethodsSec{} over non-empty detector subsets $A \subseteq E$ yields
\begin{equation}
\frac{1}{2^{|E|-1}} \sum_{\emptyset \ne A \subseteq E} (-1)^{|A|+1} b_A
= \theta_E + \sum_{F \supsetneq E} \theta_F.
\label{eq:supp_e0_mobius}
\end{equation}
Ordering the supports $E$ from size four down to one makes Eq.~\eqref{eq:supp_e0_mobius} upper-unitriangular with full column rank, uniquely determining every merged log-attenuation $\theta_E$ without non-negativity clipping. Splitting shared-support mechanisms in proportion to the nominal prior then resolves the $5\%$ Pauli-rate anisotropy without assuming isotropic Pauli error rates.

When device disorder is static across repeated stabilizer cycles, temporal orbit pooling sharpens these finite-sample estimates without introducing structural bias. Each orbit $\mathcal{O}_m$ groups mechanisms strictly across repeated temporal executions of the same gate instruction, target qubits, and Pauli error component, never pooling distinct spatial locations or Pauli components. True fault rates therefore remain uniform within every orbit, and the mutually orthogonal columns of $\bm{C}$ preserve full column rank $M_{\mathrm{tied}}$. This pooling compresses $6{,}351$ mechanisms into $4{,}368$ spacetime translation orbits in the $d = 7$ memory benchmark (Fig.~\ref{fig:vlpec_validation}a) and yields $860$ to $14{,}732$ within-stage orbits across $d = 3$ to $11$ lattice-surgery ladders (Supplementary Fig.~\ref{fig:supp_ladder_learned_scaling}). Temporal symmetry tying thus reduces estimation variance substantially while maintaining zero within-orbit rate bias.

\subsection{Full-Rank Character Construction and Boundary-Check Degeneracy}
\label{subsec:supp_character_rank_analysis}

The precision and identifiability of the weighted non-negative least-squares estimator in Eq.~\eqref{eq:methods_wnnls_estimation} depend on selecting low-weight characters that give $\Amat_{\mathrm{eff}}$ full column rank over $\RR$.
Single-check syndrome characters ($\|\bm{\chi}\|_0 = 1$) leave each Calderbank--Shor--Steane (CSS) stabilizer basis rank-deficient over $\RR$, with null-space dimension $\dim(\Ker_{\RR}(\HH_Z)) = d^2 - (d^2 - 1)/2 = (d^2 + 1)/2$ (equal to $5, 13, 25$ for $d \in \{3, 5, 7\}$).
Even augmenting a single basis with all same-basis weight-2 XOR characters fails to span the full per-qubit column space: the stacked matrix $[\HH_Z; \, \text{pairs}_Z]$ attains rank $d^2 - (d - 1) \in \{7, 21, 43\} < d^2$, with a real nullity of $d - 1 \in \{2, 4, 6\}$.
This degeneracy arises because the $d - 1$ pairs of boundary data qubits $(q_1, q_2)$ supported by weight-2 perimeter $X$-stabilizers $S_X = X_{q_1}X_{q_2}$ intersect every $Z$-stabilizer and $Z_L$ on either both qubits or neither qubit, making their columns in $[\HH_Z; \, \text{pairs}_Z]$ identical.
Because $X_{q_2} = X_{q_1} S_X$ is stabilizer-equivalent to $X_{q_1}$ with identical syndrome and logical action, each boundary pair $(q_1, q_2)$ already forms a single merged detector error model (DEM) mechanism for $X$ errors. Over that reduced mechanism space, $[\HH_Z; \, \text{pairs}_Z]$ attains full column rank $d^2 - (d - 1)$ even under pure $X$-error noise. Combining $Z$- and $X$-stabilizers under isotropic parameter tying further separates the individual perimeter qubits $(q_1, q_2)$ to achieve full per-qubit rank $d^2 \in \{9, 25, 49\}$ without measuring non-commuting logical operators.

To obtain full column rank using low-weight characters, we construct the closed-form stabilizer-only character ensemble $\Xset$ by adjoining all $d^2 - 1$ single-check characters $\bm{e}_u$ with weight-2 XOR characters $\bm{\chi}_{u, v} = \bm{e}_u \oplus \bm{e}_v$ between same-basis checks that share data qubits ($\bm{h}_u \cdot \bm{h}_v = \sum_{q=1}^{d^2} (\HH_{\mathrm{stab}})_{uq} (\HH_{\mathrm{stab}})_{vq} > 0$).
Over the real field, the modulo-2 character sum satisfies
\begin{equation}
(\bm{h}_u \oplus \bm{h}_v)_q = (\bm{h}_u)_q + (\bm{h}_v)_q - 2(\bm{h}_u)_q(\bm{h}_v)_q.
\label{eq:methods_xor_real_identity}
\end{equation}
While the weight-2 XOR rows are linear combinations of the single-check rows over $\mathbb F_2$ (leaving $\operatorname{rank}_{\mathbb F_2}(A_{\rm eff})=d^2-1\in\{8,24,48\}$), over $\mathbb R$ the quadratic overlap term $-2(\bm h_u\odot\bm h_v)$ supplies independent rank, yielding $\operatorname{rank}_{\mathbb R}(A_{\rm eff})=d^2\in\{9,25,49\}$ for $d\in\{3,5,7\}$.
The resulting character counts are $K = 14$ at $d = 3$ (comprising eight single-check characters, three $Z$-check pairs, and three $X$-check pairs, or $K = 20$ when including all same-basis pairs), $K = 54$ at $d = 5$ ($24 + 15 + 15$), and $K = 118$ at $d = 7$ ($48 + 35 + 35$).
At distance $d = 7$, singular value decomposition of $\Amat_{\mathrm{eff}} \in \RR^{118 \times 49}$ yields singular values spanning $\sigma_{\min} \approx 1.4016$ to $\sigma_{\max} \approx 8.3246$, confirming a well-conditioned design matrix with condition number $\kappa_2 \approx 5.94$ ($\kappa_2 \approx 4.53$ under column preconditioning).
Because every selected row has support at most six ($b_\chi\le6\|\theta\|_\infty$), nonzero single-check and pair-character attenuations scale as $\Theta(p)$ with an $\Theta(1)$ log-variance ratio when per-qubit rates are comparable to $p$. This bounded weight keeps the linear system both full-rank and statistically well-resolved.

\subsection{Model Integration for Rare Sectors}
\label{subsec:model_sector_integration}

While empirical sector shares enforce an in-sample mean weight only when all required sectors are occupied (Proposition~\ref{prop:sector_sample_means} in Supplementary Section~\ref{subsec:app_residual_error_matrix}), model-based digital-twin calibration resolves even rare or empty observed sectors through the joint-probability construction below.
When rare disagreement sectors have few or zero experimental counts in a finite batch, we calibrate the sector weights directly from the joint sector--error distribution of the learned physical model $\widehat P(f)$.
For a fixed learned fault law $\widehat P(f)$ and the committed decoder, define the joint model matrix
\begin{equation}
A_{a k}=\sum_f\widehat P(f)\mathbf1_{\{Q(f)=k,L(f)=a\}},\qquad
\pi_k^{\rm model}=\sum_a A_{a k},\qquad Aw=e_I.
\label{eq:model_joint_calibration}
\end{equation}
Here $a$ includes the identity logical label, while the fault configuration $f$ and its true residual error $L(f)$ are evaluated in the digital twin. Eq.~\eqref{eq:model_joint_calibration} uses joint probabilities directly and is equivalent to the centroid constraints of Eq.~\eqref{eq:app_weights_closed_form} in Supplementary Section~\ref{subsec:app_proof_pauli_4sector_invertibility} wherever $\pi_k^{\rm model}>0$, avoiding division by zero when an empirical sample contains no events in sector $k$.

In practice, we evaluate the joint model matrix $A$ by importance sampling fault configurations $f_1,\ldots,f_T$ from a fixed proposal distribution $q(f)>0$ on the model support. The importance-weighted estimator is
\[
\widehat A_{ak}=\frac1T\sum_{i=1}^T\frac{\widehat P(f_i)}{q(f_i)}
\mathbf1_{\{Q(f_i)=k,L(f_i)=a\}}.
\]
This estimator is unbiased for $A$ conditional on the learned model and fixed decoding policy. Choosing a mixture proposal $q=\lambda\widehat P+(1-\lambda)q_{\rm rare}$ with $0<\lambda\le1$ and a normalized rare-event proposal bounds the likelihood ratios by $1/\lambda$; Hoeffding's inequality and a union bound then guarantee simultaneous entrywise error at most
$\delta_A=\lambda^{-1}\sqrt{\log(2(M+1)K/\delta)/(2T)}$ with probability at least $1-\delta$. Any bounded weight table $\|w\|_\infty\le W$ satisfying $\|\widehat A w-e_I\|_1\le\zeta$ has model channel error at most
\begin{equation}
\|Aw-e_I\|_1\le\zeta+(M+1)KW\delta_A,
\label{eq:model_calibration_certificate}
\end{equation}
where the one-norm on the left-hand side equals the diamond-norm channel error for Pauli channels, while model misspecification contributes separately through $(A_{\rm true}-A)w$.

To certify a calibrated weight table before circuit execution, we specify the weight bound $W$, solver tolerance $\zeta$, and target channel-error tolerance in advance. A table is accepted only when the bounded solve meets these tolerances and Eq.~\eqref{eq:model_calibration_certificate} certifies the required accuracy; otherwise model integration is refined or a lower-dimensional feasible partition is selected. Freezing the accepted table and partition independently of the evaluation shots ensures that rare or empty experimental sectors never destabilize the post-processing estimator. Any residual calibration error then enters the single-block and full-circuit bounds of Theorem~\ref{thm:app_distance_suppression} (Supplementary Section~\ref{subsec:app_proof_distance_suppression}) and Proposition~\ref{prop:fc-robust} (Supplementary Section~\ref{sec:fc-robustness}).

Achieving uniform relative precision across rare sectors requires concentrating digital-twin samples on real-time failure and disagreement events rather than drawing uniformly from the full support. Whereas Eq.~\eqref{eq:model_calibration_certificate} provides a conservative additive error certificate, Theorem~\ref{thm:app_distance_suppression} requires bounded relative error $\varepsilon_{\max}$ on both the $\Theta(p_L)$ disagreement sectors and the $\Theta(\alpha p_L)$ consensus logical error rate. In the circuit-level benchmark of Fig.~\ref{fig:vlpec_validation}c, failure-first Monte Carlo subsampling ($f = 5\times 10^{-4}$) on the physical DEM is paired with exact signed common-random-numbers likelihood-ratio reweighting across nested unclipped M\"obius-inversion budgets, while targeted rare-$Y$ importance sampling resolves $\bar{p}_Y$ in Fig.~\ref{fig:vlpec_validation}b. Concentrating digital-twin trajectories on informative fault configurations thus achieves uniform relative accuracy across all four sectors while decoupling offline $\mathcal{D}_{\mathrm{off}}$ evaluations ($\mathcal{O}((\alpha \varepsilon_{\max}^2)^{-1})$ calls) from the $\mathcal{O}((\alpha p_L \varepsilon_{\max}^2)^{-1})$ unweighted real-time screening budget.

\subsection{Ridge Regularization and Covariance Conditioning}
\label{subsec:supp_ridge_regularization_conditioning}

When evaluating the fine-grained benchmark $C_{\mathrm{syn}}^* - 1 = \bar{\bm{p}}^\top \bm{\Sigma}^{-1}\bar{\bm{p}}$ at low error rates where $\kappa(\bm{\Sigma}) > 10^{10}$ or $\lambda_{\min}(\bm{\Sigma}) \le 10^{-14}$, Cholesky inversion is stabilized with adaptive Tikhonov damping $\bm{\Sigma}_{\mathrm{reg}} = \bm{\Sigma} + \max(10^{-12}, 10^{-10}\lambda_{\max})\bm{I}_3$.
In general, ridge regularization $R_\tau=(\Sigma+\tau I)^{-1}$ modifies the target optimum according to
\begin{equation}
C_\tau^{\rm proxy}=1+\mu^\top R_\tau\mu\le C^*,\qquad
C^*-C_\tau^{\rm proxy}=\sum_j\frac{\tau(v_j^\top\mu)^2}{\lambda_j(\lambda_j+\tau)},
\label{eq:supp_ridge_perturbation}
\end{equation}
and if used in the filter $w_\tau=1-\mu^\top R_\tau(\bm p-\mu)$ it leaves a residual $\E[w_\tau\bm p]=\tau R_\tau\mu$ with actual cost $C(w_\tau)=1+\mu^\top R_\tau\Sigma R_\tau\mu$. In our benchmarks, the four-sector coarse-grained weights $w_Q^*$ use resolved sector probabilities and require no ridge penalty. Across the 33 syndrome-resolved $(d,p)$ configurations in Supplementary Fig.~\ref{fig:supp_dual_decoder_2d}e, damping triggers in $C_{\rm syn}^*$ at only two points: $d=7$ at $p=0.005$ and $p=0.007$, with $\tau=10^{-12}$. Solving the unregularized system $\bm\Sigma\bm x=\bar{\bm p}$ via diagonal preconditioning $D_{ii}=\sqrt{\Sigma_{ii}}$ in 80-bit \texttt{longdouble} arithmetic changes $\eta=(C_{\rm syn}^*-1)/p_L$ by only $+0.104\%$ ($0.057\,\mathrm{SE}$) at $p=0.005$ and $+0.033\%$ ($0.017\,\mathrm{SE}$) at $p=0.007$, confirming that numerical conditioning has a negligible impact on the reported overhead curves.

\begin{figure*}[!htbp]
    \centering
    \includegraphics[width=\linewidth]{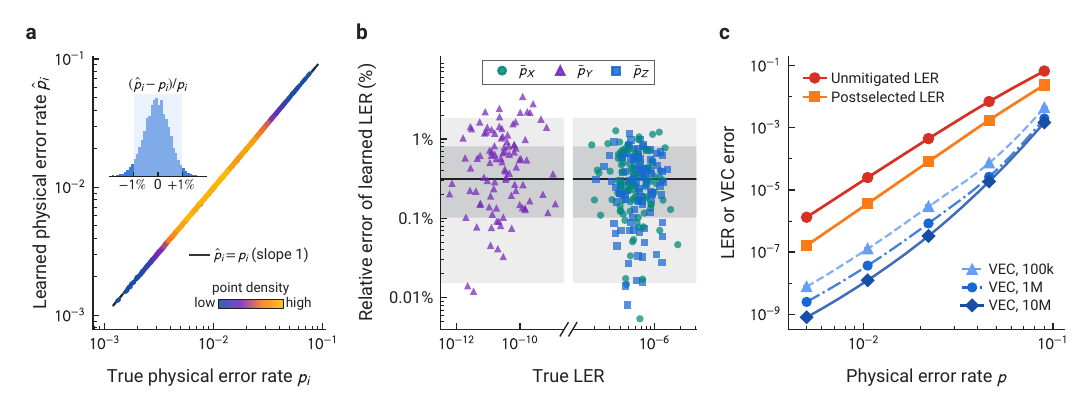}
\caption{\textbf{Code-capacity noise learning and error cancellation.}
Per-qubit error rates learned from stabilizer syndromes calibrate the post-processing weights in a digital twin; simulations model $d = 7$ rotated surface codes under spatially disordered code-capacity depolarizing noise.
\textbf{(a)}~Learned versus true physical error rates across all $49$ data qubits ($0.35\%$ median relative error; signed relative-error distribution in inset) pooled over the three lowest mean error rates $p$ at base sample count $N_0 = 10^7$, for $N_{\mathrm{learn}} = N_0(5\times 10^{-3}/p)$ syndrome runs.
\textbf{(b)}~Relative prediction error of the logical error rate (LER) versus true LER at $p = 5\times 10^{-3}$ and $N_0 = 10^7$, with the $0.31\%$ median (black line) alongside $68\%$ and $95\%$ bands.
\textbf{(c)}~Mean residual VEC error $\langle\|\bm{p}_{\mathrm{res}}\|_1\rangle$ for $N_0 = 10^5, 10^6,$ and $10^7$ ($\text{100k}$, $\text{1M}$, and $\text{10M}$) compared with the unmitigated real-time LER and decoder-agreement postselected LER ($w_Q = 0$); curves show monotone log-space interpolants.
Results summarize $100$ disorder realizations per error rate, with $1\text{ s.e.m.}$ error bars smaller than markers (see Supplementary Section~\ref{sec:methods_numerical}).}
\label{fig:supp_vlpec_validation_2d}
\end{figure*}
\clearpage

\section{Full-Circuit Integration}
\label{sec:supp_full_circuit_integration}

Multi-layer virtual error cancellation composes local per-block weights into a single full-circuit estimator with controlled approximation bias. We first prove a conditional composition theorem for state-independent syndrome instruments and specialize its distance scaling to Clifford circuits under stochastic Pauli noise. The same distance scaling extends to universal circuits with decomposed buffered non-Clifford gate-teleportation gadgets once $b \ge 2d_{\rm circ} - 3$ buffer rounds separate ancilla readout from the conditional Clifford correction. This extension requires only that the pre-feedforward and Clifford sector tables satisfy the marginal calibration and consensus-stability certificates.
Complementing the analytical single-activation proofs, Section~\ref{sec:fc-numerics} verifies exact low-weight support counts in a projected distance-three memory model, while Section~\ref{sec:fc-scaling-simulation} benchmarks multi-distance window-boundary scaling and syndrome-learned digital-twin integration on two-patch lattice-surgery ladders. All standing calibration and implementation conditions are stated explicitly below.

\subsection{Window Schedule, Syndrome Handoff, and Calibration Records}
\label{sec:fc-protocol}
We first formalize the block decomposition of a multi-layer fault-tolerant circuit and the ideal unitary channels associated with each block.
Consider $L\ge1$ sequential layers $\ell\in\{1,\ldots,L\}$ on a common
finite-dimensional logical register, where each layer $\ell$ is
partitioned into $K_\ell\ge1$ complete logical gate blocks
$c\in\{1,\ldots,K_\ell\}$, giving a total circuit volume of
$V=\sum_{\ell=1}^L K_\ell$ blocks. Ordering the blocks chronologically
as $(\ell,c)$, let $\U_{\ell,c}$ denote the ideal unitary channel of
block $(\ell,c)$ extended by the identity on the remaining logical
qubits, with layer channel $\U_\ell=\U_{\ell,K_\ell}\cdots\U_{\ell,1}$
and full circuit channel
$\U=\U_L\cdots\U_1=\prod_{\ell=L}^1\prod_{c=K_\ell}^1\U_{\ell,c}$.
When $K_\ell=1$ for every $\ell$, blocks coincide with layers and
$V=L$. Each block contains complete corrected logical operations, including any
ancilla measurement and conditional correction needed to implement
the intended gate, with physical ancillas and resource factories internal to these
logical maps. The input state $\rho$ is arbitrary and the target observable satisfies
$\|O\|_\infty\le1$. State preparation and terminal measurement are taken to be ideal
in the channel theorem so that the surrounding ideal block maps are invertible, while the state-independent gadget condition in Sec.~\ref{sec:fc-representation}
specifies when measured gadgets admit this block-level description.

\subsubsection{Commit Regions and Open Future Boundaries}
During online execution, the real-time decoder processes overlapping sliding windows of $R$ commit rounds and $b$ buffer rounds.
Let $C_\ell$ denote the detectors in rounds
$(\ell-1)R+1,\ldots,\ell R$, with the final commit region also containing the
terminal readout detectors, and let the visible set $V_\ell$ contain
$C_\ell$ together with the next $b$ buffer rounds. A window that reaches the end of
the experiment includes the terminal detectors and has a closed
future boundary. For all earlier windows, detector constraints beyond $V_\ell$ are omitted, while mechanisms extending past $V_\ell$ connect tentatively to an open future boundary that acts as a decoder boundary condition without cutting the physical circuit.

To prevent double-counting fault mechanisms across overlapping windows, each mechanism is assigned to a unique commit region based on its earliest detector.
In a binary detector error model with check matrix $H$ (suppressing boldface on syndrome, mechanism, and parity-check symbols throughout this section), each mechanism $j$ with nonempty detector support is assigned to the commit region of its earliest detector, denoted by $a(j)$.
Each empty-support mechanism receives a fixed commit-region index $a(j)\in\{1,\ldots,L\}$ as part of the model, for example by its physical time.
Although positive-cost minimum-weight decoders never select such undetected fault mechanisms, their undetected logical faults still enter the physical noise model and circuit-distance condition.

Advancing from window $\ell$ to window $\ell+1$ proceeds by committing only mechanisms assigned to region $\ell$ and subtracting their detector parities from the residual syndrome.
Before processing window $\ell$, let
\begin{equation}
 c^{<\ell}=\bigoplus_{h<\ell}c^{(h)},\qquad
 s^{(\ell)}=s\oplus Hc^{<\ell}.
 \label{eq:fc-handoff-start}
\end{equation}
Only rows of this residual syndrome already measured in $V_\ell$
are needed online. On a feasible branch, the decoder restricts to eligible mechanisms $j$ with $a(j)\ge\ell$ touching $V_\ell$ and returns a tentative vector $m^{(\ell)}$ satisfying the visible detector equations. It commits
\begin{equation}
 c^{(\ell)}_j=
 \begin{cases}
 m^{(\ell)}_j,&a(j)=\ell,\\
 0,&\text{otherwise},
 \end{cases}
 \qquad
 s^{(\ell+1)}=s^{(\ell)}\oplus Hc^{(\ell)}.
 \label{eq:fc-handoff}
\end{equation}
All other tentative assignments are discarded, and the next window
advances by $R$ rounds to reuse the buffer measurements with this
updated syndrome. Because each recovery mechanism is committed at most once, boundary defects
explained by a committed entry are toggled out of $s^{(\ell+1)}$ while unexplained defects are handed off cleanly to the next window.

\begin{lemma}[Consistency of the committed recovery]
\label{lem:fc-commit}
If each tentative recovery satisfies its visible detector equations
and the final window is closed, the accumulated recovery
$c_{\rm rt}=\bigoplus_\ell c^{(\ell)}$ satisfies $Hc_{\rm rt}=s$.
\end{lemma}
\begin{proof}
Within $V_\ell$, any eligible mechanism incident on a detector in
$C_\ell$ is assigned to region $\ell$: earlier-assigned mechanisms are excluded, and a mechanism assigned to a later
region cannot touch an earlier detector. Thus the committed part of
$m^{(\ell)}$ alone satisfies all equations on $C_\ell$.
Eq.~\eqref{eq:fc-handoff} sets their residual syndrome to zero.
Future commits cannot touch these detectors. Applying induction across all windows and the terminal detectors in the final closed window proves the claim.
\end{proof}

This sliding-window handoff imposes causal scheduling constraints on real-time feedforward while leaving the offline reference decoder unconstrained by control deadlines.
The buffer region adds a delay of $b$ extraction rounds beyond
classical decoding and control latency. Any operation whose conditional gate depends on a committed result must therefore be scheduled after
the required measurements and decoding finish, and the buffer cannot
contain measurements whose execution already requires that result.
In post-processing, the offline decoder evaluates the complete recorded history without this
control deadline: its correction $c_{\rm off}$ satisfies the same
full syndrome $Hc_{\rm off}=s$ without altering the real-time recovery already applied to the quantum state. Fixing all tie-breaking, eligibility, assignment,
and boundary rules prior to calibration ensures that any idling, storage, and control locations introduced by this schedule are consistently accounted for in the physical fault model.

\subsubsection{Feasibility on Every Record and an Infeasible-Window Abort Convention}
\label{sec:fc-totality}
Local syndrome consistency alone does not guarantee that every subsequent
window problem is solvable on arbitrary fault configurations, so a complete protocol must specify a fixed convention for unsolvable online windows.
To define the estimator on every shot, the decoder uses a fixed tie-breaking order and a deterministic abort rule.
If an online window syndrome problem is infeasible, we set the abort indicator $A=1$, complete the shot by resetting the logical register to a fixed state $\sigma$, and assign the shot contribution zero.
The aborted shot is retained in the denominator $n_{\rm shots}$ (denoted $N$ in the single-block analysis) without renormalizing by the number of successful runs.
On the usual branch $A=0$, the window procedure is unchanged and the complete record obeys $s=He$, so the offline problem is always solvable ($e$ itself is a candidate) and resolved by a total offline decoder with a fixed secondary ordering.

A simple rank condition on the tail submatrices of $H$ guarantees that this online abort branch is never triggered.
Let $D_{\ge\ell}$ be all detectors in commit regions
$\ell,\ldots,L$, including terminal detectors, and $J_{\ge\ell}$
all mechanisms with commit-region index at least $\ell$. Define over $\mathbb F_2$
\begin{equation}
 H^{\rm tail}_\ell=H[D_{\ge\ell},J_{\ge\ell}].
 \label{eq:fc-tail-matrix}
\end{equation}
\begin{lemma}[A sufficient all-record feasibility certificate]
\label{lem:fc-totality}
If every $H^{\rm tail}_\ell$ has full row rank, every tentative
window problem is feasible for every residual syndrome on its visible
detectors. With exact solvers and the closed terminal window,
the committed recovery is syndrome-consistent on every record.
\end{lemma}
\begin{proof}
Extend any visible residual syndrome arbitrarily to $D_{\ge\ell}$.
Surjectivity of $H^{\rm tail}_\ell$ supplies a solution on
$J_{\ge\ell}$. Removing its mechanisms that do not touch $V_\ell$
does not change the visible syndrome and leaves a feasible tentative
solution. Apply Lemma~\ref{lem:fc-commit} successively.
\end{proof}
Full row rank of $H^{\rm tail}_\ell$ is a sufficient condition for all-record feasibility (verified for our surface-code memory graph in Section~\ref{sec:fc-numerics}), and whenever a low-weight buffer certificate applies (Lemma~\ref{lem:fc-buffer}), infeasibility is likewise ruled out on all low-weight configurations.
When the abort rule is included on general graphs, any reset and early-termination noise is absorbed into the input-independent replacement state $\sigma$ and the total fault count, avoiding any need to exclude rare high-weight syndrome records from the probability space.

To incorporate the abort branch into the block-level channel theorem, we assign every block the
observable sector $Q_{\ell,c}=\bot$ on an online abort and enforce
$w_{\ell,c}(\bot)=0$ exactly (including during approximate calibration).
Writing $\mathcal R_\sigma(X)=\Tr(X)\sigma$ for the trace-preserving reset channel, on any abort branch satisfying the state-independent instrument condition of
Sec.~\ref{sec:fc-representation} we set
\begin{equation}
 \mathcal C_\omega=\mathcal R_\sigma,\qquad
 \M_{\ell,c}=\mathcal R_\sigma\quad(1\le\ell\le L,\ 1\le c\le K_\ell).
 \label{eq:fc-abort-attribution}
\end{equation}
Because the reset map $\mathcal R_\sigma$ is idempotent, the ordered block product equals $\mathcal C_\omega$ with all $V$ blocks marked active ($N=V$), completing the formal definition of the protocol across all syndrome records.

\subsubsection{Noise Learning Across Window Boundaries}
Because the overlapping sliding windows record the full, unbroken detector history $s$, linear-inversion noise learning operates directly across window boundaries without splitting boundary-straddling fault mechanisms.
In the independent Bernoulli detector error model (DEM) of Eq.~\eqref{eq:nl_linear_system} in Supplementary Section~\ref{sec:learning_noise_model}, writing
$a_{\chi j}=\chi^Th_j\bmod2$ gives
\begin{equation}
 F_\chi=\E(-1)^{\chi^Ts}
 =\prod_j(1-2p_j)^{a_{\chi j}},\qquad
 -\log F_\chi=\sum_j a_{\chi j}\theta_j,
 \quad \theta_j=-\log(1-2p_j).
 \label{eq:fc-learning}
\end{equation}
Including single-detector and adjacent pair characters whose detectors straddle the window boundaries ensures that boundary-straddling fault mechanisms enter a single global design matrix with shared physical parameters. Full-circuit DEM identifiability is then preserved whenever the combined design matrix has full column rank (Supplementary Section~\ref{sec:supp_noise_learning_identifiability}).

When multiple Pauli fault alternatives $a\ne0$ occur as mutually exclusive events at independent physical locations $j$, the Walsh--Hadamard moments take a location-factored form instead of an independent-mechanism product.
Specifically, the exact moment for mutually exclusive local branches is
\begin{equation}
 F_\chi=\prod_j\left[1-2\sum_{a\ne0}
 p_{j,a}(\chi^Th_{j,a}\bmod2)\right],
 \label{eq:fc-location-moment}
\end{equation}
so linear attenuation inversion is applied either to the location-level model (as in Eq.~\eqref{eq:methods_attenuation_linear_system}) or to a justified factorized DEM in the small-rate regime where all moments are positive.
In either case, the digital twin uses the learned physical model to calibrate the attributed block channels. Any residual calibration inaccuracy then enters additively through the robustness bound of Sec.~\ref{sec:fc-robustness}.

\subsubsection{Sectors and the Circuit Estimator}
We now define the full-circuit product estimator $\hat{\langle O\rangle}_{\rm VEC}$ formed from the observed block sectors $Q_{\ell,c}$ and the real-time corrected terminal readout $\hat o$.
Let $S$ denote the observed syndrome record used for sector selection,
excluding the terminal logical outcome $\hat o$ and including the
abort indicator $A$. On successful shots, $\hat o\in[-1,1]$ is the
outcome of the target observable in the frame determined by the
real-time recovery used during execution (corresponding to $m_i\in\{+1,-1\}$ for a projective Pauli readout in the main text). The offline reference decoder
is used solely to select the block disagreement sectors $Q_{\ell,c}=f_{\ell,c}(S)$ (with $Q_{\ell,c}=0$ denoting the consensus sector $\mathcal S_I$) without replacing the committed output frame. For each shot and across $n_{\rm shots}$ runs, we define
\begin{equation}
 \OVEC=\hat o\prod_{\ell=1}^{L}\prod_{c=1}^{K_\ell}w_{\ell,c}(Q_{\ell,c}),\qquad
 \hat{\langle O\rangle}_{\rm VEC}
 =\frac1{n_{\rm shots}}\sum_{i=1}^{n_{\rm shots}}\OVEC^{(i)},
 \label{eq:fc-estimator-main}
\end{equation}
where aborted shots ($A=1$) contribute zero to the numerator while remaining in $n_{\rm shots}$. Because digital-twin weight estimation and experimental evaluation use the exact same window schedule, decoders, and attribution rule, any fault pattern near a window boundary is assigned to the same observed sector in both calibration and execution.

\subsection{A Common Channel Description for All Logical Operations}
\label{sec:fc-representation}
\subsubsection{State-Independent Syndrome Instrument}
To analyze multi-block circuits containing both unitary gates and measured gadgets, we represent the completed fault-tolerant execution by a syndrome-resolved quantum instrument on the logical register.
All maps act on the common logical register. The map
$\mathcal J_s$ denotes the output of the completed
protocol after all real-time frame corrections and immediately before the terminal logical measurement (retaining terminal detector parities only when their joint map satisfies Eq.~\eqref{eq:fc-instrument}).
Let $\omega$ denote a complete realization of the physical faults and any internal classical random variables, drawn from an input-independent distribution $\mu_p$ with normalized completely positive, trace-preserving (CPTP) channels
$\mathcal C_\omega$. We require $\mathcal J_s$ to satisfy
\begin{equation}
 \mathcal J_s=\E_p[\mathbf1_{\{S=s\}}\mathcal C_\omega],
 \label{eq:fc-instrument}
\end{equation}
enforcing agreement of both the conditional quantum channel and the probability of obtaining record $s$ across all input states. Any measurement branch with state-dependent outcome probabilities must therefore be grouped into a complete corrected gadget before applying Eq.~\eqref{eq:fc-instrument}.

An exact error attribution decomposes each realization's channel $\mathcal C_\omega$ into an ordered product of block-level CPTP maps $\M_{\ell,c}(\omega)$ that links the product estimator directly to the expected composition of weighted block maps.
Such a decomposition assigns CPTP maps $\M_{\ell,c}(\omega)$ to the ideal block positions while preserving both the observed record $S$ and the ordered channel product,
\begin{equation}
 \left(\prod_{\ell=L}^{1}\prod_{c=K_\ell}^{1}\M_{\ell,c}(\omega)\right)=\mathcal C_\omega
 \quad\hbox{for every supported }\omega.
 \label{eq:fc-attribution}
\end{equation}
This identity holds even when the assigned block differs from the original physical spacetime location of a fault, provided its net action in the full circuit is preserved. Define the weighted block maps
$\B_{\ell,c}=w_{\ell,c}(Q_{\ell,c})\M_{\ell,c}$ and write their chronological product as
$\B_V\cdots\B_1:=\prod_{\ell=L}^{1}\prod_{c=K_\ell}^{1}\B_{\ell,c}$.
Combining Eq.~\eqref{eq:fc-instrument} with $\E[\hat o\mid\omega]=\Tr[O\mathcal C_\omega(\rho)]$ yields the exact expectation identity
\begin{equation}
 \E\OVEC
 =\Tr\!\left[O\,\E(\B_V\cdots\B_1)(\rho)\right],
 \label{eq:fc-estimator-identity}
\end{equation}
where the expectation is taken over the true correlated joint distribution of all $V$ blocks across all input states $\rho$.

\subsubsection{A Certificate for the State-Independent Instrument Representation}
\label{sec:fc-branch-certificate}
Because Eq.~\eqref{eq:fc-instrument} requires the trace $\Tr\mathcal J_s(\rho)=\Pr(S=s)$ to be independent of the logical input state $\rho$, we provide a sufficient gadget-level certificate that guarantees this property.
\begin{proposition}[Normalized outcome-unbiased branches]
\label{prop:fc-branches}
Suppose physical fault realizations $f$ have input-independent
probabilities $\pi_p(f)$. For each $f$, let
$\{\mathcal G_{f,s}\}_s$ be the completely positive (CP) instrument of the corrected
experiment, including its completion on an online abort. If
\begin{equation}
 \mathcal G_{f,s}^{\dagger}(I)=a_{f,s}I,
 \qquad a_{f,s}\ge0,\qquad \sum_s a_{f,s}=1,
 \label{eq:fc-scalar-effects}
\end{equation}
then Eq.~\eqref{eq:fc-instrument} holds with realization
$\omega=(f,s)$, probability $\mu_p(f,s)=\pi_p(f)a_{f,s}$, and normalized channel
$\mathcal C_{f,s}=\mathcal G_{f,s}/a_{f,s}$ on non-zero-probability
branches. The physical fault count is $F(f,s)=F(f)$.
\end{proposition}
\begin{proof}
The normalized maps are CP and trace-preserving by
Eq.~\eqref{eq:fc-scalar-effects}. Their probabilities are independent
of the input. Summing $\pi_p(f)a_{f,s}\mathcal C_{f,s}$ over $f$
reproduces the actual instrument for $s$, also on inputs entangled
with a reference. Zero-probability branches can be omitted.
\end{proof}
Proposition~\ref{prop:fc-branches} can be verified sequentially across adaptive gadgets:
after conditioning on the earlier classical record and physical
fault realization, each retained branch need only act as a scalar effect $a_{f,s}I$
on its incoming logical state. Products of these scalar branch
probabilities then define an input-independent record distribution $\mu_p(f,s)$, while compositions
of the normalized maps $\mathcal C_{f,s}$ remain CPTP across the entire circuit.

\subsubsection{Exact Reassignment and Corrected Gadgets}
Reassigning a fault from one spacetime location to another block requires conjugating its error map through the intervening ideal unitary operations.
For a unitary channel $\mathcal V$ and any CPTP error map
$\mathcal E$, the exact conjugation identity
\begin{equation}
 \mathcal V\mathcal E
 = (\mathcal V\mathcal E\mathcal V^{-1})\mathcal V
 \label{eq:fc-transport}
\end{equation}
moves an error across an ideal operation while keeping the conjugated
map $\mathcal V\mathcal E\mathcal V^{-1}$ CPTP and preserving chronological order.
Because propagating a Pauli error through a non-Clifford gate produces a non-Pauli unitary error (such as $TXT^\dagger=(X+Y)/\sqrt2$), full-circuit attribution and calibration must either keep errors on the Clifford side of such gates or calibrate the full conjugated error superoperator.

Gate-teleportation gadgets satisfy the scalar-effect condition of Proposition~\ref{prop:fc-branches} because their projective ancilla outcomes occur with equal probability $1/2$ for every data input state.
As a concrete example, consider a $T$-gate gadget prepared with a noisy magic state $Z^a|A\rangle$ for $|A\rangle=T|+\rangle$ and phase bit $a\in\{0,1\}$. The circuit applies a data-to-ancilla controlled-X (CNOT) gate, measures the ancilla in the $Z$ basis with true outcome $m\in\{0,1\}$, and applies the conditional correction $S^{m\oplus\xi}$ subject to a readout-bit flip $\xi\in\{0,1\}$.
Direct evaluation of the Kraus operators yields
\begin{equation}
 K_{a,\xi,m}=\frac{e^{i\phi_{a,m}}}{\sqrt2}
     S^{(m\oplus\xi)-m}Z^aT,
 \qquad K_{a,\xi,m}^\dagger K_{a,\xi,m}=I/2.
 \label{eq:fc-gadget}
\end{equation}
Both measurement branches $m\in\{0,1\}$ therefore occur with input-independent probability $1/2$ (even for inputs entangled with a reference system), while a readout flip $\xi=1$ induces a non-Pauli $S^{\pm 1}$ feedforward error on the data qubit.

\subsubsection{Pauli and Classical-Bit Cancellation Before Feedforward}
\label{sec:fc-seam}
Placing a virtual pre-feedforward reference cut immediately after the ancilla measurement and before the classically controlled $S$ gate allows linear reweighting to cancel both data-qubit Pauli faults and wrong-feedforward bit flips.
Keeping the committed measurement outcome in a classical register $C$, we write the ideal pre-feedforward instrument as
\begin{equation}
\mathcal J_{\rm id}(\rho)=\frac12\sum_{m=0}^1
S^{-m}T\rho T^\dagger S^m\otimes|m\rangle\langle m|_C.
\label{eq:fc-seam-ideal}
\end{equation}
We then define the joint Pauli and classical bit-flip error maps $\mathcal E_{P,b}=\operatorname{Ad}_P\otimes\operatorname{Ad}_{X_C^b}$ for $P\in\{I,X,Y,Z\}$ and $b\in\{0,1\}$, where $X_C$ flips the classical bit in $C$. Suppose the retained syndrome or sector label $s$ has input-independent probability $\pi_s$ and conditional joint map
\begin{equation}
\mathcal J_s=\sum_{P,b}q_s(P,b)\mathcal E_{P,b}\mathcal J_{\rm id},
\qquad q_s(P,b)\ge0,\quad\sum_{P,b}q_s(P,b)=1,
\label{eq:fc-seam-noise}
\end{equation}
across all input states. Linear inversion of the joint distribution $q_s(P,b)$ then recovers the ideal $T$ gate after feedforward as follows.

\begin{proposition}[Pre-feedforward channel cancellation]
\label{prop:fc-seam}
If the scalar weights solve
\begin{equation}
\sum_s\pi_sw(s)q_s(P,b)=\mathbf1_{(P,b)=(I,0)},
\label{eq:fc-seam-cancellation}
\end{equation}
then the reweighted corrected gadget is the ideal $T$ channel. This protocol uses the same linear reweighting construction on the joint Pauli and bit-error coefficients, while feasibility and uniform weight bounds are additional calibration conditions.
\end{proposition}
\begin{proof}
Summing Eq.~\eqref{eq:fc-seam-noise} with these weights gives $\sum_s\pi_sw(s)\mathcal J_s=\mathcal J_{\rm id}$. Define the fixed linear feedforward map
$\mathcal F(\sigma\otimes|c\rangle\langle c|)=S^c\sigma S^{-c}$, extended linearly over classical blocks. Then $\mathcal F\mathcal J_{\rm id}=\operatorname{Ad}_T$ by direct substitution, and $\mathcal F$ commutes with the scalar weighted sum. Thus bit-error contributions, including the wrong-$S$ effect they would cause, are canceled at the instrument level.
\end{proof}

Any stochastic Pauli fault occurring on the magic-state ancilla, data qubit, or CNOT gate prior to ancilla readout maps directly into the joint Pauli--bit form of Eq.~\eqref{eq:fc-seam-noise}.
In particular, for the gadget in Eq.~\eqref{eq:fc-gadget}, a resource phase fault $Z^a$ and readout flip $\xi$ produce the instrument $\mathcal J_{a,\xi}=\mathcal E_{Z^a,\xi}\mathcal J_{\rm id}$. More generally, propagating any pre-readout Pauli fault through the Clifford CNOT yields a data Pauli operator $P$ and an ancilla Pauli operator $Q$ immediately before measurement. An ancilla $X$ or $Y$ flips the classical readout bit ($b=1$), whereas an ancilla $Z$ contributes only a branch phase.

When the gadget experiences only resource phase errors $Z^a$ and readout flips $\xi$ whose syndrome $s$ is independent of the fair measurement outcome $m$, branch averaging simplifies the error structure. Averaging over $m\in\{0,1\}$ after the $S^{m\oplus\xi}$ correction reduces the non-Pauli feedforward error to an effective Pauli $Z$-dephasing channel.
Specifically, because $(\operatorname{Ad}_S+\operatorname{Ad}_{S^\dagger})/2=(\Id+\operatorname{Ad}_Z)/2$, the branch-averaged channel has effective $Z$-error probability
\[
q_{\rm eff}(s)=\Pr(a=1,\xi=0\mid s)+\tfrac12\Pr(\xi=1\mid s),
\]
allowing a single scalar channel-cancellation condition to remove both resource phase faults and wrong-feedforward branches simultaneously.

As an explicit solvable example of this control-bit cancellation, suppose a perfect quantum teleportation gadget transmits its measurement outcome $m$ through three independent classical binary channels of flip probability $0<p<1/2$ and commits the majority-vote bit for feedforward. The two pairwise parities of the three received bits form an input-independent syndrome. Its zero-syndrome class occurs with probability $\pi_0=(1-p)^3+p^3$, majority-error posterior $p^3/\pi_0$, and effective $Z$-error rate $q_0=p^3/(2\pi_0)$, whereas its merged non-zero-syndrome class occurs with probability $\pi_1=3p(1-p)$, decoding failure posterior $p$, and effective $Z$-error rate $q_1=p/2$. The exact two-sector weights
\begin{equation}
w_0=\frac1{(1-p)(1-2p)},\qquad
w_1=-\frac{p}{3(1-p)^2(1-2p)}
\label{eq:fc-repetition-control-weights}
\end{equation}
satisfy $\pi_0w_0+\pi_1w_1=1$ and $\pi_0w_0q_0+\pi_1w_1q_1=0$ while remaining uniformly bounded as $p\to0$, providing a concrete demonstration of how post-processing reweighting cancels erroneous non-Clifford corrections.

In the full-circuit composition theorem below, each ideal block $\U_{\ell,c}$ represents a complete corrected unitary operation (such as $\operatorname{Ad}_T=\mathcal F\mathcal J_{\rm id}$), while the internal pre-feedforward cut establishes its local calibration in Eq.~\eqref{eq:fc-calibration}. Any additional faults occurring during the conditional $S$ gate or subsequent idling are attributed to the appropriate stage and included in the circuit's marginal weight equations and single-activation conditions.

While Eq.~\eqref{eq:fc-transport} guarantees that any individual error map can be propagated algebraically through ideal gates, applying the full-circuit bias bound requires a further certificate. The chosen block attribution must preserve both the ordered channel product and the single-activation property at the target calibration granularity, as we prove for Clifford circuits under stochastic Pauli noise in Sec.~\ref{sec:fc-distance}.

\subsection{Calibration and Fault-Order Hypotheses}
\label{sec:fc-assumptions}
We now state the five standing hypotheses \ref{ass:fc-record}--\ref{ass:fc-stability} under which the full-circuit relative bias bound in Theorem~\ref{thm:fc-relative} (Eq.~\eqref{eq:full_circuit_relative_bias} of the main text) holds.
All expectations and probabilities refer to the completed,
attributed experiment (including its infeasible-window abort convention), with measurable sector and map assignments defined on every supported event. All constants are independent of sufficiently small $p>0$ at fixed circuit size, code distance, and window schedule.
\begin{enumerate}[label=(A\arabic*),ref=(A\arabic*),leftmargin=*]
\item\label{ass:fc-record} \textbf{State-independent instrument and exact attribution.}
Eqs.~\eqref{eq:fc-instrument} and \eqref{eq:fc-attribution} hold;
the ideal block maps $\U_{\ell,c}$ are unitary and the assigned maps $\M_{\ell,c}$ are CPTP.
The terminal estimator and abort sectors follow
Sec.~\ref{sec:fc-protocol}.

\item\label{ass:fc-calibration} \textbf{Exact marginal calibration.}
The real weights $w_{\ell,c,p}(q)$ may depend on $p$ but are fixed
before evaluation; the subscript $p$ is suppressed. They obey the
complete channel equations
\begin{equation}
 \E[w_{\ell,c}(Q_{\ell,c})\M_{\ell,c}]=\U_{\ell,c},
 \qquad |w_{\ell,c}(q)|\le K<\infty.
 \label{eq:fc-calibration}
\end{equation}
The bound is uniform for sufficiently small $p$, and
$w_{\ell,c}(\bot)=0$ when abort sectors occur. Taking the trace
shows $\E w_{\ell,c}=1$. Weights are determined under
the marginal law induced by the full experiment, including
boundary-spanning faults; this condition does not assume independence of windows or blocks.

\item\label{ass:fc-activation} \textbf{Single activation on low-weight faults.}
Define
\begin{equation}
 I_{\ell,c}=\mathbf1_{\{Q_{\ell,c}\ne0\ \mathrm{or}\ \M_{\ell,c}\ne\U_{\ell,c}\}},
 \qquad N=\sum_{\ell=1}^L\sum_{c=1}^{K_\ell} I_{\ell,c}.
 \label{eq:fc-activity}
\end{equation}
If $F$ is the number of faulty physical locations, then for integers
$r\ge1$ and $t=r+k$ with $k\ge1$,
\begin{equation}
 F<t\quad\Longrightarrow\quad N\le1.
 \label{eq:fc-light}
\end{equation}
The condition concerns activity after attribution. An event may cross
several windows or blocks and still activate one block. Isolated decoder disagreements without
logical errors are allowed, including below order $p^r$. For $V\ge2$, the online abort indicator
activates all blocks; any claimed low-weight certificate must
therefore exclude online aborts on that range.

\item\label{ass:fc-scale} \textbf{Fault tail and logical-error scale.}
For positive constants $A_t,C_E,c_-,c_+$,
\begin{equation}
 \Pr(F\ge t)\le A_t p^t,\quad
 \bar p_L:=\Pr(\mathcal C_\omega\ne\U)\le C_Ep^r,
 \quad c_-p^r\le p_L\le c_+p^r.
 \label{eq:fc-scales}
\end{equation}
The per-block reference $p_L$ is specified independently of the target
observable; $V p_L$ may equal $\bar p_L$ when the latter has order
$\Theta(p^r)$. An observable's unmitigated bias can vanish by
cancellation and cannot generally replace this normalization.

\item\label{ass:fc-stability} \textbf{Stable consensus weight.}
The consensus entries satisfy
\begin{equation}
 |w_{\ell,c}(0)-1|\le\kappa_{\ell,c}\Pr(\M_{\ell,c}\ne\U_{\ell,c}),
 \qquad \kappa_{\ell,c}\le\kappa<\infty.
 \label{eq:fc-stability}
\end{equation}
Section~\ref{sec:fc-calibration-certificate} gives a sufficient
linear-algebra certificate. Bounded weights by themselves do not
imply this condition.
\end{enumerate}

The fault-tail condition $\Pr(F\ge t)\le A_t p^t$ in \ref{ass:fc-scale} holds automatically for any finite circuit with mutually independent local physical fault locations, without assuming that the attributed block channels factorize.
Specifically, suppose a fixed finite
set of physical locations has uncoupled fault
indicators $B_j$ satisfying $\Pr(B_j=1)\le a_jp$ and total fault count $F\le F_{\max}:=\sum_j B_j$ (allowing for early termination). With $\Lambda=\sum_j a_j$, Markov's inequality yields
\begin{equation}
 \Pr(F\ge t)\le\E\binom{F_{\max}}t
 =\sum_{|T|=t}\prod_{j\in T}\Pr(B_j=1)
 \le\frac{(\Lambda p)^t}{t!},
 \label{eq:fc-tail}
\end{equation}
where mutually exclusive Pauli alternatives at a single location count as one indicator $B_j$ (and conditioning on an acceptance event of probability at least $a_0>0$ simply multiplies this tail bound by at most $a_0^{-1}$).

\subsection{Bound on the Correlated Contribution}
\label{sec:fc-proof}
Because exact marginal calibration \ref{ass:fc-calibration} cancels every term that is linear in a single block's deviation from its ideal unitary channel, the residual full-circuit bias is governed entirely by multi-block activations ($N\ge2$) and second-order consensus-weight shifts.
To bound these terms without assuming that block maps commute or factorize, we decompose each weighted block map $\B_{\ell,c}$ around its consensus weight $u_{\ell,c}=w_{\ell,c}(0)$ by setting
\begin{equation}
 \begin{split}
 u_{\ell,c}&=w_{\ell,c}(0),\qquad H_{\ell,c}=\B_{\ell,c}-u_{\ell,c}\U_{\ell,c},\qquad
 s=\sum_{\ell,c}|u_{\ell,c}-1|,\\
 h&=2K,\qquad \eta=\Pr(N\ge2),\qquad
 C_V=(1+h)^V-1-hV.
 \end{split}
 \label{eq:fc-decomposition}
\end{equation}
Whenever block $(\ell,c)$ is inactive ($I_{\ell,c}=0$), $H_{\ell,c}=0$ identically, while the CPTP diamond-norm bound and
Eq.~\eqref{eq:fc-calibration} imply
\begin{equation}
 \normd{H_{\ell,c}}\le h I_{\ell,c},
 \qquad \E H_{\ell,c}=(1-u_{\ell,c})\U_{\ell,c},
 \label{eq:fc-H}
\end{equation}
isolating all active-block fluctuations into the operators $H_{\ell,c}$.

\begin{theorem}[Activation bound]
\label{thm:fc-activation}
Under \ref{ass:fc-record} and \ref{ass:fc-calibration},
\begin{align}
 \normd{\E(\B_V\cdots\B_1)-\U}
 &\le e^s\left\{\frac{s^2}{2}
       +\E[(1+h)^N-1-hN]\right\}\nonumber\\
 &\le e^s\left(\frac{s^2}{2}+C_V\eta\right).
 \label{eq:fc-activation-bound}
\end{align}
\end{theorem}
\begin{proof}
Expand $\prod_{\ell=L}^{1}\prod_{c=K_\ell}^{1}(u_{\ell,c}\U_{\ell,c}+H_{\ell,c})$ over subsets of
block positions occupied by $H_{\ell,c}$, retaining the original order.
The empty and singleton terms, after expectation and subtraction of
$\U$, equal $V_0\U$, where
\begin{align}
 V_0&=\prod_{\ell,c} u_{\ell,c}+
       \sum_{\ell,c}(1-u_{\ell,c})\prod_{(\ell',c')\ne(\ell,c)}u_{\ell',c'}-1\nonumber\\
 &=\sum_{|T|\ge2}(1-|T|)\prod_{(\ell,c)\in T}(u_{\ell,c}-1).
 \label{eq:fc-vacuum-expansion}
\end{align}
The elementary symmetric polynomial of degree $j$ in
$|u_{\ell,c}-1|$ is at most $s^j/j!$, so
\begin{equation}
 |V_0|\le\sum_{j\ge2}(j-1)s^j/j!
 \le\tfrac12s^2e^s.
 \label{eq:fc-V0}
\end{equation}
For a subset $T$ with at least two block positions, submultiplicativity of
the diamond norm bounds its expected contribution by
\begin{equation}
 \left(\prod_{(\ell,c)\notin T}|u_{\ell,c}|\right)
 h^{|T|}\E\prod_{(\ell,c)\in T}I_{\ell,c}
 \le e^s h^{|T|}\E\prod_{(\ell,c)\in T}I_{\ell,c}.
 \label{eq:fc-subset}
\end{equation}
A realization with $N$ active blocks has $\binom Nj$ active subsets
of size $j$. Summing Eq.~\eqref{eq:fc-subset} yields
$e^s\E[(1+h)^N-1-hN]$. This expression vanishes for $N\le1$ and
is at most $e^s C_V$ otherwise, proving both bounds.
\end{proof}

\begin{lemma}[Control of the consensus-shift contribution]
\label{lem:fc-vacuum}
Under \ref{ass:fc-record} and \ref{ass:fc-stability},
\begin{equation}
 s\le\kappa(\bar p_L+V\eta).
 \label{eq:fc-vacuum-control}
\end{equation}
\end{lemma}
\begin{proof}
Let $n_E=\sum_{\ell,c}\mathbf1_{\{\M_{\ell,c}\ne\U_{\ell,c}\}}$.
On $N\le1$, at most one block map differs from its ideal value.
Canceling the ideal unitary channels before and after that position
in Eq.~\eqref{eq:fc-attribution} shows
$n_E=\mathbf1_{\{\mathcal C_\omega\ne\U\}}$ on this event.
Always $n_E\le V$. Thus $\E n_E\le\bar p_L+V\eta$;
summing Eq.~\eqref{eq:fc-stability} proves the claim. A disagreement-only
activation does not contribute to $n_E$.
\end{proof}

\begin{theorem}[Full-circuit relative residual]
\label{thm:fc-relative}
Under \ref{ass:fc-record}--\ref{ass:fc-stability}, put
$\beta=\kappa(C_Ep^r+VA_t p^t)$ and
$\mathcal E=\E(\B_V\cdots\B_1)-\U$. Then
\begin{equation}
 \boxed{\normd{\mathcal E}
 \le e^\beta\left(\frac{\beta^2}{2}+C_V A_t p^t\right)
 =\mathcal O(p^{\min(2r,t)}).}
 \label{eq:fc-main-theorem}
\end{equation}
For every input $\rho$ and $\|O\|_\infty\le1$,
\begin{equation}
 \boxed{
 \frac{|\E\OVEC-\Tr[O\U(\rho)]|}{2Vp_L}
 =\mathcal O(p^r)+\mathcal O(p^{t-r})
 =\mathcal O(p^{\min(r,k)}),\qquad t=r+k.}
 \label{eq:fc-relative}
\end{equation}
\end{theorem}
\begin{proof}
Eq.~\eqref{eq:fc-light} gives
$\eta\le\Pr(F\ge t)\le A_t p^t$. Lemma~\ref{lem:fc-vacuum}
then gives $s\le\beta$. Substitute these inequalities into
Theorem~\ref{thm:fc-activation}. Finally,
$|\Tr[O\mathcal E(\rho)]|\le\normd{\mathcal E}$,
Eq.~\eqref{eq:fc-estimator-identity} identifies the numerator, and
$2Vp_L\ge 2Vc_-p^r$ gives Eq.~\eqref{eq:fc-relative}.
\end{proof}

Eq.~\eqref{eq:fc-relative} decomposes the normalized residual bias per block into two physical contributions: a multi-block activation term of relative order
$\eta/(Vp_L)=\mathcal O(p^k)$ and a quadratic consensus-weight shift of relative order $\mathcal O(p^r)$. Increasing the single-activation fault margin $k$ therefore suppresses the relative bias as $\mathcal O(p^k)$ up to $k=r$, where it saturates at the $\mathcal O(p^r)$ consensus-shift floor. This bound controls the complete multi-stage diamond-norm error, including all higher-order joint cumulants across the $V$ blocks, at fixed circuit volume.

\subsubsection{Calibration Error and Imperfect Consensus Stability}
\label{sec:fc-robustness}
When the block weights are obtained from a learned noise model instead of exact marginal channels, local inversion errors $\Delta_{\ell,c}$ and consensus-weight slack $\nu_{\ell,c}$ enter additively into the full-circuit composition bound.
Here weights remain fixed before circuit evaluation (for example, from an independent training split) and may depend on the physical noise parameter $p$.

\begin{proposition}[Robust composition bound]
\label{prop:fc-robust}
Retain \ref{ass:fc-record}, \ref{ass:fc-activation} and
\ref{ass:fc-scale}. Replace exact calibration by real weights with
$|w_{\ell,c}(q)|\le K$, retaining the exact zero weight on an online abort, and define
\begin{equation}
 \Delta_{\ell,c}=\E[w_{\ell,c}(Q_{\ell,c})\M_{\ell,c}]-\U_{\ell,c},
 \qquad \varepsilon_{\rm cal}=
       \sum_{\ell,c}\normd{\Delta_{\ell,c}}.
 \label{eq:fc-calibration-residual}
\end{equation}
Suppose, allowing a nonnegative stability slack,
\begin{equation}
 |w_{\ell,c}(0)-1|\le\kappa\Pr(\M_{\ell,c}\ne\U_{\ell,c})+\nu_{\ell,c},
 \quad \nu=\sum_{\ell,c}\nu_{\ell,c},\qquad
 \beta_\nu=\kappa(C_Ep^r+VA_t p^t)+\nu.
 \label{eq:fc-stability-slack}
\end{equation}
Then the same estimator satisfies
\begin{equation}
 \normd{\E(\B_V\cdots\B_1)-\U}
 \le e^{\beta_\nu}\left(
       \frac{\beta_\nu^2}{2}+C_V A_t p^t
       +\varepsilon_{\rm cal}\right).
 \label{eq:fc-robust-bound}
\end{equation}
In particular, $\nu=\mathcal O(p^r)$ and
$\varepsilon_{\rm cal}=\mathcal O(p^{\min(2r,t)})$ preserve the relative
order in Eq.~\eqref{eq:fc-relative}.
\end{proposition}
\begin{proof}
The pointwise bound $\normd{H_{\ell,c}}\le h I_{\ell,c}$ is unchanged,
whereas $\E H_{\ell,c}=(1-u_{\ell,c})\U_{\ell,c}+\Delta_{\ell,c}$.
In the empty and singleton terms of the ordered expansion, the new
terms are
\begin{equation}
 \sum_{\ell,c}\left(\prod_{(\ell',c')\ne(\ell,c)}u_{\ell',c'}\right)
 \U_{>\ell,c}\Delta_{\ell,c}
 \U_{<\ell,c},
 \label{eq:fc-calibration-linear-terms}
\end{equation}
where $\U_{>\ell,c}$ and $\U_{<\ell,c}$ denote the ordered products
of ideal block channels after and before block $(\ell,c)$.
Their total diamond norm is at most
$e^s\varepsilon_{\rm cal}$. The higher-subset bound is unchanged.
The proof of Lemma~\ref{lem:fc-vacuum} now gives
$s\le\kappa(\bar p_L+V\eta)+\nu\le\beta_\nu$.
Substitution proves the inequality. Under the stated accuracy
conditions, its right-hand side is $\mathcal O(p^{\min(2r,t)})$; division
by $2Vp_L\ge2Vc_-p^r$ proves the last assertion.
\end{proof}
For $\nu=\mathcal O(p^r)$, an arbitrary residual $\varepsilon_{\rm cal}$ contributes
$\mathcal O(\varepsilon_{\rm cal}/p^r)$ to the normalized relative bound. Suppose the weight tables are learned on an independent characterization sample such that a uniform bound
$B_{\rm good}$ holds with probability at least $1-\delta_{\rm bad}$, with $|w_{\ell,c}(q)|\le K$ enforced on every table. Averaging over the sample distribution then yields the unconditional error bound
\begin{equation}
 B_{\rm good}+\delta_{\rm bad}(K^V+1),
 \label{eq:fc-calibration-event}
\end{equation}
since the weighted and ideal circuit channels have diamond norms at most $K^V$ and $1$, respectively, on the exceptional calibration event.

\subsection{A Distance-Based Certificate for Windowed Attribution}
\label{sec:fc-distance}
We now specialize the general composition framework to Clifford circuits under stochastic Pauli noise and prove that connected-cluster co-attribution satisfies the single-activation hypothesis Eq.~\eqref{eq:fc-light} up to the circuit fault distance.
This section establishes the distance-based scaling law $\epsilon=\mathcal O(p^{\lfloor d/2\rfloor})$ stated in the main text, while the general multi-block theorem of Sec.~\ref{sec:fc-proof} continues to apply to any ideal unitary sequence once its state-independent instrument, attribution, and calibration hypotheses are certified.

\subsubsection{Cycles and Block Assignment}
In a Clifford circuit under stochastic Pauli noise, the discrepancies between the true fault configuration $e$, the committed real-time recovery $c_{\rm rt}$, and the offline recovery $c_{\rm off}$ decompose into closed detector cycles in $\ker H$.
On a successful branch with fixed binary detector matrix $H$ and syndrome-consistent recoveries $Hc_{\rm rt}=Hc_{\rm off}=He$, the three discrepancy vectors
\begin{equation}
 x_e=c_{\rm rt}\oplus e,\qquad
 x_Q=c_{\rm rt}\oplus c_{\rm off},\qquad
 z=x_e\oplus x_Q=c_{\rm off}\oplus e
 \label{eq:fc-cycles}
\end{equation}
all belong to $\ker H$. Defining two mechanisms to be adjacent whenever they share an
observed detector (excluding virtual boundary nodes), every connected
component of a cycle in $\ker H$ is itself a detector cycle because any shared detector would merge two components.
Each observed detector belongs to a unique block $(\ell,c)$ in chronological order, and each connected component is assigned to the block containing its latest detector. A fixed tie-breaking order resolves ties within the latest round, while undetected fault mechanisms receive a fixed block assignment.

To relate cycle weights to logical failures, we define the circuit fault distance $d_{\rm circ}$ using the additive logical Pauli map $\Phi:\ker H\to\mathbb F_2^{2n}$.
Specifically, let
\begin{equation}
 d_{\rm circ}=\min\{|x|:Hx=0,\ \Phi(x)\ne0\},
 \label{eq:fc-distance}
\end{equation}
where $|x|$ counts elementary mechanisms and every configuration of $F$ physical faults satisfies $|e|\le F$ (so one physical fault is not split into multiple independently counted edges). Because Pauli frame propagation across Clifford gates acts as an invertible linear map on Pauli labels, whether a cycle $x\in\ker H$ has a nontrivial logical action $\Phi(x)\ne0$ is invariant under the choice of block frame.

Disagreement-anchored co-attribution assigns both decoder disagreements and digital-twin residual errors to blocks by clustering $x_Q$ and $x_e\cup x_Q$ on the detector graph.
Each nontrivial connected component of $x_Q$ is assigned to the block
$(\ell,c)$ containing its latest detector, summing its label into the
observed sector $Q_{\ell,c}$. To attribute residual logical errors in the digital twin, we form the connected clusters $J$ of the union $\supp(x_e)\cup\supp(x_Q)$.
Whenever a cluster $J$ contains a nonzero $x_Q$ piece, its total $x_e$ label is attributed to the latest block containing a detector of those nonzero pieces.
When no nonzero $x_Q$ piece is present, we assign the $x_e$ label to the latest block of its $x_Q$ support (or of its $x_e$ support when the $x_Q$ support is empty).
This rule defines the block error label $E_{\ell,c}$ and the attributed channel
$\M_{\ell,c}=\mathcal P_{E_{\ell,c}}\U_{\ell,c}$.

Because every error component is attributed to one block and Pauli frame propagation across Clifford gates preserves the ordered Pauli product up to overall phase, this co-attribution rule satisfies the exact full-circuit channel identity in Eq.~\eqref{eq:fc-attribution}.
The observed post-processing disagreement labels $Q_{\ell,c}$ depend only on $x_Q$, and the number of active blocks $N$ counts blocks with $(Q_{\ell,c},E_{\ell,c})\ne(0,0)$ in accordance with Eq.~\eqref{eq:fc-activity}.

\begin{lemma}[Single activation from three discrepancy bounds]
\label{lem:fc-single}
If each nontrivial cycle has at least $d_{\rm circ}$ mechanisms and
\begin{equation}
 |x_e|<2d_{\rm circ},\qquad |x_Q|<2d_{\rm circ},\qquad |z|<2d_{\rm circ},
 \label{eq:fc-three}
\end{equation}
disagreement-anchored co-attribution leaves at most one active block.
\end{lemma}
\begin{proof}
Two disjoint nontrivial cycles require at least $2d_{\rm circ}$ mechanisms.
Thus $x_e$ and $x_Q$ each have at most one nonzero component.
If both exist, suppose they lie in different clusters of the union.
Restrictions of $x_e,x_Q,z$ to each union cluster are cycles,
because distinct clusters have disjoint detector supports.
On each of those two clusters, $z=x_e\oplus x_Q$ has nonzero
logical label. Their combined size is at least $2d_{\rm circ}$, contradicting
the bound on $z$. The residual error and decoder disagreement therefore belong to the same
union cluster, and co-attribution assigns their nonzero labels to
the same block. If either the residual error or the disagreement is absent, at most
one nonzero label remains.
\end{proof}

Writing $\epsilon=|\E\OVEC-\Tr[O\U(\rho)]|/(2Vp_L)$, Lemma~\ref{lem:fc-single} immediately yields the distance-dependent suppression corollary.
\begin{corollary}[Distance-dependent relative suppression]
\label{cor:fc-distance}
Suppose $|e|\le F$ and, for every supported configuration with
$F<d_{\rm circ}$, both decoders return successful syndrome-consistent
recoveries satisfying
\begin{equation}
 |c_{\rm rt}|\le F,\qquad |c_{\rm off}|\le F.
 \label{eq:fc-decoder-condition}
\end{equation}
Then co-attribution satisfies $F<d_{\rm circ}\Rightarrow N\le1$.
The online abort branch is absent on this range; on other configurations
use the completed protocol of Sec.~\ref{sec:fc-totality}.
If the remaining hypotheses of Theorem~\ref{thm:fc-relative} hold
with $r=\lceil d_{\rm circ}/2\rceil$ and $d_{\rm circ}\ge2$, then
\begin{equation}
 t=d_{\rm circ},\qquad k=d_{\rm circ}-r=\lfloor d_{\rm circ}/2\rfloor,\qquad
 \epsilon=\mathcal O(p^{\lfloor d_{\rm circ}/2\rfloor}).
 \label{eq:fc-distance-bound}
\end{equation}
\end{corollary}
\begin{proof}
Every discrepancy in Eq.~\eqref{eq:fc-cycles} has at most $2F<2d_{\rm circ}$
mechanisms, so Lemma~\ref{lem:fc-single} applies. Substitute $t=d_{\rm circ}$
and $r=\lceil d_{\rm circ}/2\rceil$ into Eq.~\eqref{eq:fc-relative}; since
$d_{\rm circ}-r\le r$, the indicated exponent follows.
\end{proof}

\subsubsection{A Sufficient Buffer Certificate}
\label{sec:fc-buffer-certificate}
To guarantee that the sliding-window real-time recovery satisfies the weight bound $|c_{\rm rt}|\le F$ in Eq.~\eqref{eq:fc-decoder-condition}, we establish a sufficient condition on the buffer depth $b$.
We work in a fixed finite binary detector error model in which each window returns an unweighted minimum-weight recovery on its eligible mechanisms (connected through shared observed detectors inside the window).
Let $\delta_\ell$ denote the smallest cardinality of a connected set $P$ of eligible mechanisms satisfying $H_{V_\ell}\mathbf1_P=0$ that touches both the commit region $C_\ell$ and a detector beyond $V_\ell$. We set $\delta_\ell=\infty$ when no such set exists, as in the closed terminal window.
Thus $\delta_\ell$ measures the minimum size of a buffer-spanning relative cycle that can reach from the region being committed out to future rounds not yet seen by window $\ell$.

\begin{lemma}[Buffered decoding on low-weight configurations]
\label{lem:fc-buffer}
Let $f_*\ge0$ be an integer. Assume $s=He$ and $|e|\le F$ on
every configuration being certified, with $F\le f_*$. Use the
eligibility, assignment, and handoff rules of Sec.~\ref{sec:fc-protocol},
a closed terminal window, and an exact unweighted minimum-weight solver
on every feasible window. If $\min_\ell\delta_\ell>2f_*$, no
window is infeasible on this range, and the complete committed
recovery is a global minimum-weight correction. In particular,
$|c_{\rm rt}|\le|e|\le F$.
\end{lemma}
\begin{proof}
Inductively suppose all decisions on earlier-assigned mechanisms, including zero entries,
agree with some global minimum correction $c_*$. This holds before the first
window. Restrict the uncommitted part of $c_*$ to mechanisms touching
$V_\ell$ to obtain a feasible tentative recovery $m_*$. Since the
true error is a feasible global correction,
$|m_*|\le|c_*|\le|e|\le F$. Window minimality gives
$|m^{(\ell)}|\le|m_*|$.

Each connected component $P$ of
$m^{(\ell)}\oplus m_*$ is a cycle on the visible detectors and
has $|P|\le2F\le2f_*$. By the hypothesis, a component that touches
the commit region cannot extend beyond the open future window boundary.
Eligible mechanisms also have no detectors in earlier commit regions.
Consequently such a component is a cycle of the full detector model.
Exchanging it in the tentative correction or in $c_*$ preserves the
corresponding syndrome constraints. Window minimality gives
$|m^{(\ell)}\cap P|\le|m_*\cap P|$, while the global minimum-weight property gives
the reverse inequality. The two counts are equal. Exchanging all
such components in $c_*$ therefore preserves global optimality and
makes it agree with the new commit, without changing earlier commits.
This proves the induction. After the final closed window, the full
committed vector equals a global minimum correction. Positive costs
exclude undetected mechanisms from a minimum correction; they
therefore do not leave undecided recovery entries.
\end{proof}

\begin{lemma}[Temporal-locality bound]
\label{lem:fc-buffer-length}
Suppose each mechanism with nonempty detector support has all its
detectors in one round or two consecutive rounds. Undetected
fault mechanisms cannot occur in a buffer-spanning relative cycle $P$. For a nonterminal window with
$b$ buffer rounds, $\delta_\ell\ge b+2$.
\end{lemma}
\begin{proof}
Consider a set $P$ counted by $\delta_\ell$ and let $u$ be the
earliest detector round in $P$. Because $P$ meets the commit region,
$u\le\ell R$. At an earliest detector, the visible cycle condition
requires at least two incident mechanisms. Both start in round $u$
and end no later than $u+1$. Connectivity to a detector beyond the
window's last visible round $\tau=\ell R+b$ requires a fault mechanism
crossing each successive cut $v\mid v+1$, for
$v=u+1,\ldots,\tau$. Temporal locality makes these $\tau-u$
mechanisms distinct from each other and from the first two. Hence
$|P|\ge2+\tau-u\ge b+2$.
\end{proof}

Combining Lemmas~\ref{lem:fc-buffer} and~\ref{lem:fc-buffer-length} determines the buffer length $b$ needed to certify any target fault weight $f_*$.
Specifically, $\delta_\ell\ge b+2>2f_*$ yields the sufficient condition
\begin{equation}
 b\ge2f_*-1\quad\Longrightarrow\quad
 |c_{\rm rt}|\le F\quad\hbox{for every }F\le f_*.
 \label{eq:fc-buffer-length}
\end{equation}
For $d_{\rm circ}\ge2$ and $r=\lceil d_{\rm circ}/2\rceil$, choosing $f_*=r$ and $b\ge2r-1$ (equal to $b\ge d_{\rm circ}$ for odd $d_{\rm circ}$)
certifies the $t=r+1$ ($k=1$) regime whenever the offline recovery obeys
the same weight bound. Certifying the full $t=d_{\rm circ}$ ($k=\lfloor d_{\rm circ}/2\rfloor$) regime across all $F<d_{\rm circ}$ is guaranteed by
\begin{equation}
 f_*=d_{\rm circ}-1,\qquad b\ge2d_{\rm circ}-3,
 \label{eq:fc-buffer-distance}
\end{equation}
although shorter buffers can suffice on specific detector graphs.%

This unweighted minimum-weight exchange argument extends directly to log-likelihood weighted matching in the low-error regime $p\to0$.
For any fixed finite DEM with mechanism probabilities $p_j=\lambda_jp+\mathcal O(p^2)$ ($\lambda_j>0$), the negative-log-odds edge weights satisfy $w_j=\log(1/p)+\mathcal O(1)$ uniformly across all supports.
Whenever the strict integer margin $b+2>2f_*$ holds, any buffer-spanning relative cycle reaching the commit region has total weight strictly exceeding $2W(e)$ at sufficiently small $p$, and any global minimum-weight recovery also minimizes fault count because each extra mechanism incurs a dominant $\log(1/p)$ penalty.

\subsubsection{Scope of the Distance Guarantee}
We conclude this subsection by summarizing the scope of the distance-based single-activation guarantee and its relationship to weaker buffer certificates and non-Clifford circuits.
While any global unweighted minimum-weight recovery automatically obeys $|c|\le|e|\le F$, windowed real-time decoding requires the accumulated vector $c_{\rm rt}=\bigoplus_\ell c^{(\ell)}$ across all windows, instead of each tentative window solution in isolation, to satisfy Eq.~\eqref{eq:fc-decoder-condition}. Lemmas~\ref{lem:fc-buffer} and~\ref{lem:fc-buffer-length} certify this total weight bound whenever $b\ge 2f_*-1$.

Even when the recovery weight bounds $|c_{\rm rt}|\le F$ and $|c_{\rm off}|\le F$ hold only up to the minimum uncorrectable fault weight $F\le r=\lceil d_{\rm circ}/2\rceil$ (for example, at buffer depth $b=d_{\rm circ}$ for odd $d_{\rm circ}\ge3$), Lemma~\ref{lem:fc-single} still guarantees $2F\le2r<2d_{\rm circ}$. The single-activation threshold is therefore at least $t=r+1$ ($k=1$, giving $\epsilon=\mathcal O(p)$). Because Eq.~\eqref{eq:fc-three} is a sufficient condition for $N\le1$, direct simulation of the attributed clusters can establish higher suppression exponents $k>\lfloor d_{\rm circ}/2\rfloor$ than guaranteed by the worst-case distance bound alone, as demonstrated in Section~\ref{sec:fc-scaling-simulation}.

The reference logical error scaling $p_L=\Theta(p^{\lceil d_{\rm circ}/2\rceil})$ follows from the definition of $d_{\rm circ}$ whenever minimum-weight fault supports are physically realizable. Under an unweighted minimum-weight decoder, fewer than $\lceil d_{\rm circ}/2\rceil$ faults cannot produce a cycle of weight $d_{\rm circ}$ and therefore never fail. Splitting a minimum-weight nontrivial cycle of length $d_{\rm circ}$ into two pieces of sizes $\lfloor d_{\rm circ}/2\rfloor$ and $\lceil d_{\rm circ}/2\rceil$ produces two syndromes that force a deterministic decoder to fail on at least one of the two halves, yielding $p_L=\Theta(p^{\lceil d_{\rm circ}/2\rceil})$.

Lemma~\ref{lem:fc-single} and Corollary~\ref{cor:fc-distance} extend from Clifford blocks under stochastic Pauli noise to universal circuits with decomposed buffered gate-teleportation gadgets (Fig.~\ref{fig:circuit_integration}a) whenever $b \ge 2d_{\rm circ} - 3$ buffer rounds separate ancilla readout from the conditional correction. Consider a $T$-gate gadget decomposed into magic-state input $|T\rangle$, Clifford CNOT, and destructive ancilla readout $M_Z$ at round $\tau_M$ preceding the pre-feedforward cut of Proposition~\ref{prop:fc-seam}. Following $b \ge 2d_{\rm circ} - 3$ Clifford QEC buffer rounds, the conditional Clifford gate $S^c$ executes at round $\tau_S \ge \tau_M + b + 1$. Assume stochastic Pauli faults on $|T\rangle$ and Clifford gates, outcome-symmetric readout errors independent of the unbiased measurement outcome $m \in \{0, 1\}$, and two-round temporally local detector mechanisms. Proposition~\ref{prop:fc-branches} then certifies input-independent branch weights, while Lemmas~\ref{lem:fc-buffer} and~\ref{lem:fc-buffer-length} ensure $|x_e|, |x_Q|, |z| \le 2F < 2d_{\rm circ}$ for every configuration of $F < d_{\rm circ}$ physical faults. Any closed cycle $C$ with $u_{\max}(C) - u_{\min}(C) \ge 2$ requires at least two incident mechanisms at both endpoints plus one distinct mechanism crossing each intermediate round cut. Every nontrivial connected cluster $C$ of size $|C| \le 2d_{\rm circ} - 2$ therefore obeys
\begin{equation}
 u_{\max}(C) - u_{\min}(C) \le \max(1,\, |C| - 2) \le \max(1,\, 2d_{\rm circ} - 4) < 2d_{\rm circ} - 2 \le \tau_S - \tau_M
 \label{eq:fc-seam-span-separation}
\end{equation}
for all $d_{\rm circ} \ge 2$. Consequently, no active discrepancy cluster at $F < d_{\rm circ}$ can connect pre-readout or readout detectors ($\le \tau_M$) to post-buffer detectors ($\ge \tau_S$). When $u_{\max}(C) < \tau_S$, the cluster closes prior to $S^c$ and enters the eight-component joint Pauli--bit pre-feedforward instrument $\mathcal{J}_s$ of Eq.~\eqref{eq:fc-seam-noise} without crossing $S^c$. Conversely, when $u_{\max}(C) \ge \tau_S$, we have $u_{\min}(C) > \tau_M$, leaving $c = m$ uncorrupted and conjugating any pre-$S^c$ buffer Pauli across $S^m$ into a pure Pauli operator $S^m P S^{-m}$ attributed to the block containing $u_{\max}(C)$. Connected-cluster co-attribution thus satisfies both \ref{ass:fc-record} (Eq.~\eqref{eq:fc-attribution}) and \ref{ass:fc-activation} ($N \le 1$ for $F < d_{\rm circ}$), so Theorem~\ref{thm:fc-relative} yields $\epsilon = \mathcal{O}(p^{\lfloor d_{\rm circ}/2\rfloor})$ with $k = \lfloor d_{\rm circ}/2\rfloor$ under \ref{ass:fc-calibration} and \ref{ass:fc-stability}.

\subsection{Calibration Solvability and a Stability Certificate}
\label{sec:fc-calibration-certificate}
We next formulate the block-level channel inversion equations in a general superoperator basis and provide a linear-algebraic certificate (Proposition~\ref{prop:fc-stability}) for the consensus-weight stability hypothesis \ref{ass:fc-stability}.
For each block $(\ell,c)$, define the subnormalized sector maps
\begin{equation}
 \mathcal A_{\ell,c,q}
 =\E[\mathbf1_{\{Q_{\ell,c}=q\}}\M_{\ell,c}].
 \label{eq:fc-sector-maps}
\end{equation}
Exact marginal calibration in Eq.~\eqref{eq:fc-calibration} is the real linear system
$\sum_q w_{\ell,c}(q)\mathcal A_{\ell,c,q}=\U_{\ell,c}$ in complete channel coordinates, including the trace component $\E w_{\ell,c}=1$ without post-hoc empirical normalization. This linear system is solvable whenever the ideal unitary channel $\U_{\ell,c}$ lies in the linear span of the sector maps $\{\mathcal A_{\ell,c,q}\}_q$.

Suppressing the block index $(\ell,c)$, we choose real linear coordinates $v_j$ on Hermiticity-preserving superoperator differences scaled so that $|v_j(\mathcal D)|\le\normd{\mathcal D}$. The matrix entries are then
\begin{equation}
 A_{0q}=\Pr(Q=q),\qquad
 A_{jq}=\E[\mathbf1_{\{Q=q\}}v_j(\M-\U)]\quad(j\ge1),
 \label{eq:fc-calibration-matrix}
\end{equation}
so that the complete block calibration system takes the compact matrix form $Aw=e_0$ with $e_0=(1,0,\ldots,0)^T$.

\begin{proposition}[Stable extraction of the consensus entry]
\label{prop:fc-stability}
If there is a row vector $\lambda^T$ such that
\begin{equation}
 \lambda^TA=e_{\mathrm{cons}}^T,\qquad
 \sum_{j\ge1}|\lambda_j|\le B<\infty,
 \label{eq:fc-calibration-certificate}
\end{equation}
uniformly for small $p$, where $e_{\mathrm{cons}}$ selects $q=0$, then
every exact solution obeys
$|w(0)-1|\le2B\Pr(\M\ne\U)$.
\end{proposition}
\begin{proof}
The trace entry of $e_0-A\mathbf1$ is zero. Hence
\begin{equation}
 w(0)-1
 =\lambda^TA(w-\mathbf1)
 =-\sum_{j\ge1}\lambda_j\E v_j(\M-\U).
\end{equation}
Both maps are CPTP, so
$|v_j(\M-\U)|\le2\mathbf1_{\{\M\ne\U\}}$.
Taking absolute values proves the stated bound.
\end{proof}
When the matrix $A$ is square and invertible, $\lambda^T$ is the consensus ($q=0$) row of $A^{-1}$, and Proposition~\ref{prop:fc-stability} requires only that its non-trace entries remain uniformly bounded by $B$ as $p\to0$.
In the four-sector Pauli realization of Supplementary Section~\ref{subsec:app_proof_pauli_4sector_invertibility}, Theorem~\ref{thm:app_levy_desplanques_pauli} and Proposition~\ref{thm:app_variational_collapse} establish this bound directly whenever $\pi_0\ge\pi_*>0$. Because $\bar{\bm p}=\bm p_0+\bm V\bm\pi_{1:3}$ for the contrast matrix $\bm V$, the exact consensus weight satisfies
\[
|w_0-1|=\frac{|\mathbf1^T\bm V^{-1}\bar{\bm p}|}{\pi_0}
\le\frac{\|\bm V^{-1}\|_1}{\pi_*}\|\bar{\bm p}\|_1,
\]
where $\|\bar{\bm p}\|_1=\Pr(\mathcal M\ne\mathcal U)$, certifying \ref{ass:fc-stability} with $\kappa=\|\bm V^{-1}\|_1/\pi_*$ whenever $\|\bm V^{-1}\|_1$ and $\pi_*$ are uniformly bounded.

\subsubsection{What a Calibration Certificate Must Establish}
\label{sec:fc-calibration-existence}
Certifying a block calibration scheme for full-circuit composition requires verifying three properties under the attributed marginal law of the circuit: existence of a solution to $Aw=e_0$ (with $w(\bot)=0$), uniform boundedness $|w(q)|\le K$, and the consensus-stability bound of Proposition~\ref{prop:fc-stability}.
For a gate-teleportation gadget calibrated at its internal pre-feedforward cut (Sec.~\ref{sec:fc-seam}), the columns of $A$ must resolve the joint Pauli and classical-bit coefficients $\pi_sq_s(P,b)$ of Eq.~\eqref{eq:fc-seam-cancellation} in addition to data-qubit Pauli errors.

These solvability requirements place concrete design constraints on the sector partition when errors are propagated or reattributed across blocks.
A single unpartitioned sector ($Q=0$ only) with residual channel $(1-\epsilon)\Id+\epsilon\mathcal X$ for $0<\epsilon<1$ and $\mathcal X\ne\Id$ forces $w(0)=1$ by trace preservation and cannot cancel the error $\mathcal X$. Propagating faults through non-Clifford gates can likewise generate additional superoperator coordinates that require a correspondingly richer sector basis whenever errors are reassigned across gate boundaries.

Proposition~\ref{prop:fc-stability} also quantifies how approximate inversion residuals $\zeta=Aw-e_0$ propagate into the consensus-stability slack $\nu_{\ell,c}$ of Proposition~\ref{prop:fc-robust}.
Applying the identity $w(0)-1=\lambda^T(Aw-A\mathbf1)$ yields
\begin{equation}
 |w(0)-1|\le2B\Pr(\M\ne\U)+|\lambda^T\zeta|,
 \label{eq:fc-approx-vacuum}
\end{equation}
where $\lambda^TA\mathbf1=1$ implies $|\lambda_0|\le1+2B\Pr(\M\ne\U)\le1+2B$.
If the block residual $\Delta=\E[w\M]-\U$ has diamond norm $\delta=\normd{\Delta}$, then $|\zeta_0|=|\E w-1|\le\delta$ and
\begin{equation}
 |\zeta_j|=|v_j(\Delta-\zeta_0\U)|\le2\delta\quad(j\ge1),
 \qquad |\lambda^T\zeta|\le(1+4B)\delta,
 \label{eq:fc-slack-from-calibration}
\end{equation}
showing that a uniform extraction certificate $\lambda^T$ automatically converts a diamond-norm error bound $\delta$ into the consensus-stability slack $\nu_{\ell,c}\le(1+4B)\delta$.

\subsection{Memory Experiment and Connection to the Assumptions}
\label{sec:fc-numerics}
To verify the algebraic window-handoff, tail-rank, and single-activation certificates on a concrete fault-tolerant circuit, we perform an exact low-weight fault enumeration on a distance-three rotated surface-code logical-$Z$ memory over nine stabilizer rounds.
The nine rounds are partitioned into $L=3$ commit regions of $R=3$ rounds each.
Online decoding uses sliding-window minimum-weight perfect matching (MWPM) with $b=3$ buffer rounds, whereas offline decoding runs global matching on the complete syndrome record.
Both solvers use unit edge costs and distinct secondary tie-breaking perturbations of total magnitude below $1/4$ in \texttt{Stim}~1.16.0 and \texttt{PyMatching}~2.4.0~\cite{gidney2021stim,higgott2025blossom}.

We construct the projected $Z$-detector graph directly from the physical circuit's gate and measurement faults and verify its circuit distance and tail-matrix ranks.
The physical circuit experiences depolarizing faults after Hadamard and CNOT gates and bit-flip faults before measurements. Projecting each single-fault signature onto the $Z$-check detectors and terminal logical-$Z$ flip produces graphlike edges with at most two detector endpoints spanning at most two consecutive rounds. Merging identical detector and logical signatures yields $40$ detectors and $151$ canonical edges with verified circuit distance $d_{\rm circ}=d=3$. The three tail matrices in Eq.~\eqref{eq:fc-tail-matrix} attain full row ranks of $40$ at window one, $28$ at window two, and $16$ at window three, certifying by Lemma~\ref{lem:fc-totality} that every residual window syndrome is feasible and the online abort branch never occurs. Because $F$ physical faults produce a canonical edge support of size $f\le F$ after modulo-2 cancellation, enumerating canonical supports of size $f$ exhaustively covers all physical fault configurations up to order $F\le f$.

\begin{table}[htbp]
\centering\small
\begin{tabular}{rrrrr}
\toprule
$f$ & Supports checked & Logical-$Z$ failures &
\multicolumn{2}{c}{Supports with $N\ge2$}\\
\cmidrule(lr){4-5}
&&& Separate & Co-attributed\\
\midrule
0 & 1 & 0 & 0 & 0\\
1 & 151 & 0 & 0 & 0\\
2 & 11,325 & 490 & 0 & 0\\
\bottomrule
\end{tabular}
\caption{Exact support counts for the fixed projected memory model.
Both decoding conventions and attribution are held fixed. The counts
are unweighted and do not estimate probabilities or bias coefficients.}
\label{tab:fc-counts}
\end{table}

Table~\ref{tab:fc-counts} summarizes the exhaustive enumeration of all canonical fault supports up to weight $f=2$ under both separate attribution and disagreement-anchored co-attribution.
For all $f\le2$, both real-time and offline recoveries have size at most $f$ and all three discrepancies in Eq.~\eqref{eq:fc-cycles} have weight at most $4<2d_{\rm circ}=6$.
This bound confirms the sufficient certificate of Lemma~\ref{lem:fc-single} and Corollary~\ref{cor:fc-distance} across the entire $F<d_{\rm circ}=3$ range, giving $t\ge d_{\rm circ}=3$ and $k\ge\lfloor d_{\rm circ}/2\rfloor=1$.

Cross-checking these canonical support enumerations against direct physical circuit injections and full $X+Z$ patterns confirms the consistency of the projected model.
All $151$ single-fault representatives and $32$ reproducible pairs%
were injected directly into the physical \texttt{Stim} circuit, reproducing the exact projected detector signature and logical flip in all $183$%
tests.
Furthermore, we enumerated all $11{,}477$%
projected canonical $Z$-supports through $f\le2$%
and all $1{,}274{,}407$ full $X+Z$ mechanism configurations through $f\le2$.
Every tested configuration confirmed exact syndrome matching $Hc_{\rm rt}=Hc_{\rm off}=He\pmod2$ and 100\% agreement between sliding-window $\Drt$ ($b=2d-3=3$) and global uncorrelated $Z$-MWPM through $f\le2$, providing a complete finite-graph verification of our windowed decoding and co-attribution rules.

\subsection{Finite-Circuit Bias Evaluation}
\label{sec:fc-scaling-simulation}
\subsubsection{Exact-DEM Sliding-Window Boundary-Bias Scaling (Fig.~\ref{fig:circuit_integration}b)}
\label{subsec:fc-scaling-panel-a}
The multi-distance scaling curves in Fig.~\ref{fig:circuit_integration}b estimate the normalized mitigation bias $\epsilon$ of Eq.~\eqref{eq:full_circuit_relative_bias} in the nearest-window approximation $\epsilon \approx |\epsilon_{\mathrm{adj}}|$, with adjacent-boundary contribution $\epsilon_{\mathrm{adj}} = \sum_{s=1}^{L-1} \Cov_P(H_s, H_{s+1}) / (2 V p_L)$ (Eq.~\eqref{eq:fc-numerical-cov}; chip-averaged signed value, plotted in magnitude), for a two-patch rotated surface-code $Z_A Z_B$ lattice-surgery ladder across odd code distances $d \in \{3, 5, 7, 9\}$. Each ladder executes $K_{\mathrm{surg}} = 5$ consecutive merge-and-split operations across $L = V = 11$ sliding windows in \texttt{Stim}~1.16.0~\cite{gidney2021stim}, alternating six separated single-patch QEC stages of $d$ rounds with five merged two-patch surgery stages of $d$ rounds plus one terminal readout round. With commit length $R = d$ and buffer length $b = 2d - 3 \in \{3, 7, 11, 15\}$, the circuit spans $11d + 1 \in \{34, 56, 78, 100\}$ stabilizer rounds and $L - 1 = 10$ inter-window boundaries. Baseline SI1000 superconducting circuit noise~\cite{gidney2021fault} assigns two-qubit Clifford depolarization at rate $p$, pre-measurement bit flips at $5p$, post-reset bit flips at $2p$, and data- and resonator-idling errors at $p/10$, with zero single-qubit gate error. Each round is compiled into ten layers (H, CZ, H, CZ, CZ, H, CZ, H, measurement, reset); in every layer, each qubit that has already been used and is not acted on receives single-qubit depolarization at $p/10$, including during measurement and reset layers. The initial reset of all qubits and the terminal data readout carry the same $2p$ and $5p$ bit flips, and the seam qubits, prepared and measured in the $X$ basis at each merge, carry phase flips at these rates. Every gate, idle, measurement, and reset instruction is disordered independently with $25\%$ relative rate spread and $5\%$ Pauli-rate anisotropy across $M = 80$ equal-budget random device instances per $(d, p)$ point (indexed by deterministic seed $10{,}000 d + m$ for $m \in \{0, \dots, 79\}$). This eleven-window ladder benchmark thus tests multi-window composition under realistic spatial heterogeneity across four code distances.

Exact Bernoulli factorization of every single-qubit and two-qubit stochastic Pauli channel eliminates $\mathcal{O}(p^2)$ mutually-exclusive-channel approximations and graphlike hyperedge splitting before extracting the unsplit detector error model (\texttt{decompose\_\allowbreak errors=False}). From the anticommutation indicator $A_{PQ}$ and the Pauli-channel eigenvalues $\lambda_P$, we solve
\begin{equation}
\lambda_P = 1 - 2\sum_{Q:\,\{P, Q\}=0} p_Q, \qquad
\log(1 - 2 q_Q) = \sum_P (A^{-1})_{QP} \log \lambda_P
\label{eq:fc-walsh-factorization}
\end{equation}
for the exact independent mechanism probabilities $q_Q$, folding terminal readout flips into the joint $Z_A Z_B$ observable parity. Across $d \in \{3, 5, 7, 9\}$, this eleven-window extraction produces the unsplit physical mechanism counts ($N_{\mathrm{err}}$), total detectors ($N_{\mathrm{det}}$), $Z$-stabilizer checks ($N_D$), and parity-combined graphlike $Z$-detector edges ($N_E$) summarized below.
\begin{center}
{\small\setlength{\tabcolsep}{5.5pt}
\begin{tabular}{ccccc}
\toprule
$d$ & Rounds ($11d + 1$) & $N_{\mathrm{err}}$ & $N_{\mathrm{det}}$ ($Z$-checks $N_D$) & Graphlike $Z$-edges $N_E$ \\
\midrule
$3$ & $34$ & $8{,}297$ & $568$ ($312$) & $1{,}295$ \\
$5$ & $56$ & $48{,}053$ & $2{,}760$ ($1{,}464$) & $6{,}917$ \\
$7$ & $78$ & $143{,}561$ & $7{,}632$ ($3{,}984$) & $20{,}051$ \\
$9$ & $100$ & $319{,}109$ & $16{,}240$ ($8{,}400$) & $43{,}865$ \\
\bottomrule
\end{tabular}}
\end{center}
Factorizing the microscopic channels in advance thus preserves exact fault correlations across over $3\times 10^5$ spacetime mechanisms at $d = 9$.

Real-time and offline decoding operate on these graphlike and unsplit detector models, respectively, to classify cross-decoder discrepancies by window. The real-time decoder $\mathcal{D}_{\mathrm{rt}}$ executes sliding-window MWPM in \texttt{PyMatching}~2.4.0~\cite{higgott2025blossom} on the $Z$-detector graph with $R = d$ and $b = 2d - 3$. The offline reference $\mathcal{D}_{\mathrm{off}}$ executes global two-pass correlated MWPM (\texttt{enable\_correlations=True}) on the unsplit $X+Z$ model and compares odd-parity discrepancy components on the one-hop shared-$Z$-detector graph against global uncorrelated $Z$-MWPM. Explicitly, the cluster-retention rule reads as follows. Let $c_{\mathrm{corr}}$ and $c_Z$ denote the correlated and uncorrelated corrections on the graphlike $Z$ edges. The edges of $c_{\mathrm{corr}} \cup c_Z$ are grouped into connected components, two edges being adjacent when they share a $Z$ detector. A component is a candidate when $c_{\mathrm{corr}} \oplus c_Z$ flips $Z_A Z_B$ on it and its edge counts $n_{\mathrm{corr}}$ (from $c_{\mathrm{corr}}$) and $n_Z$ (from $c_Z$) satisfy $n_{\mathrm{corr}} \le \max(r, n_Z)$ and $n_{\mathrm{corr}} \le n_Z + 1$, with $r = (d + 1)/2$. Both counts are then augmented by the edges of the other components of the same run that end within three hops of the candidate's detectors (a hop joins two $Z$ detectors that share a model edge) and whose latest window differs from the candidate's latest window; the candidate is retained when the augmented counts satisfy $n'_{\mathrm{corr}} \le \max(r, n'_Z)$. On retained components $\mathcal{D}_{\mathrm{off}}$ applies $c_{\mathrm{corr}}$, and everywhere else $c_Z$; it thereby forgoes part of the accuracy gain of unrestricted correlated matching. Connected-cluster co-attribution then assigns each retained cluster to the window of its latest active detector using a union-cluster tie-breaker. Here a detector is active when it fired or is touched by the real-time correction, and the residual-error parity of each connected component of the union of discrepancy and residual-error edges is assigned to the latest active window of that component's discrepancy edges, or of its residual-error edges when it contains no discrepancy. This rule suppresses weight-$3$ multi-window false activations at $d = 3$ down to the $10^{-6}$ level without altering the single-activation threshold $k = \lfloor d/2 \rfloor$.

Symmetric two-fold cross-fitting eliminates finite-sample ratio bias and spatial-disorder one-body residuals on each device instance $m$. We partition each chip's parallel Monte Carlo shards into two equal disjoint folds $\mathcal{F}_0$ and $\mathcal{F}_1$ without pooling parameters across distinct windows. On calibration fold $\mathcal{F}_a$ ($a \in \{0, 1\}$), each window $\ell \in \{1, \dots, 11\}$ is calibrated individually from its empirical $2\times 2$ joint table $\hat{P}_\ell^{(a)}(q, e)$ and disagreement frequency $\hat{\pi}_\ell^{(a)}$ via
\begin{equation}
\hat{u}_\ell^{(a)} \equiv \hat{w}_\ell^{(a)}(0) = \frac{\hat{P}_\ell^{(a)}(1, 1)}{\hat{\Delta}_\ell^{(a)}}, \qquad
\hat{w}_\ell^{(a)}(1) = -\frac{\hat{P}_\ell^{(a)}(0, 1)}{\hat{\Delta}_\ell^{(a)}}, \qquad
\hat{\Delta}_\ell^{(a)} = \bigl(1 - \hat{\pi}_\ell^{(a)}\bigr)\hat{P}_\ell^{(a)}(1, 1) - \hat{\pi}_\ell^{(a)}\hat{P}_\ell^{(a)}(0, 1).
\label{eq:fc-numerical-single-window-weights}
\end{equation}
Defining the signed calibrated window factor on fold $\mathcal{F}_a$ ensures exact unit empirical calibration means according to
\begin{equation}
\hat{g}_\ell^{(a)}(Q_\ell, e_\ell) = \hat{w}_\ell^{(a)}(Q_\ell)(-1)^{e_\ell}, \qquad
\hat{\E}_{\mathcal{F}_a}\bigl[\hat{w}_\ell^{(a)}(Q_\ell)\bigr] = 1, \qquad
\hat{\E}_{\mathcal{F}_a}\bigl[\hat{g}_\ell^{(a)}\bigr] = 1.
\label{eq:fc-numerical-fold-identities}
\end{equation}
Conditioned on fold $\mathcal{F}_a$, the out-of-fold product-weight expectation under the true eleven-window law $P(\bm{Q}, \bm{e})$ and the normalized mitigation bias take the form
\begin{equation}
\mu_{\mathrm{VEC}}^{(a)} = \sum_{\bm{Q}, \bm{e}} P(\bm{Q}, \bm{e}) \prod_{\ell=1}^{L} \hat{g}_\ell^{(a)}(Q_\ell, e_\ell), \qquad
\epsilon = \frac{|\mu_{\mathrm{VEC}} - 1|}{2 V p_L}, \qquad
V p_L = \sum_{\ell=1}^{L} \Pr(e_\ell = 1).
\label{eq:fc-numerical-full-bias}
\end{equation}
Here $V p_L$ is evaluated on calibration fold $\mathcal{F}_a$, and the two cross-fit directions $(0 \to 1)$ and $(1 \to 0)$ are averaged symmetrically on each chip before computing the unweighted mean and standard error across all $M = 80$ chips. Cross-fitting thus isolates genuine inter-window boundary correlations from finite-sample single-window calibration noise.

The same twenty $(d, p)$ ensembles also test whether circuit-level sampling overhead factorizes across correlated sliding windows. Comparing the empirical joint and factorized circuit overheads,
\begin{equation}
\widehat{C}_{\mathrm{tot}} = \hat{\E}\!\left[\prod_{\ell=1}^{L} \hat{w}_\ell(Q_\ell)^2\right], \qquad
\widehat{C}_{\mathrm{prod}} = \prod_{\ell=1}^{L} \hat{\E}\bigl[\hat{w}_\ell(Q_\ell)^2\bigr],
\label{eq:fc-numerical-ctot-comparison}
\end{equation}
shows that $|\widehat{C}_{\mathrm{tot}} - \widehat{C}_{\mathrm{prod}}| / \widehat{C}_{\mathrm{prod}}$ stays within $0.000\%$ to $0.545\%$ at $d = 3$, $0.000\%$ to $0.045\%$ at $d = 5$, $0.000\%$ to $0.009\%$ at $d = 7$, and $0.000\%$ to $0.002\%$ at $d = 9$. These tight relative bounds numerically confirm the squared-weight factorization theorem of Eq.~\eqref{eq:fc-sqweight-factorization-bound}.

Expanding each calibrated window factor around its consensus baseline via $\hat{H}_\ell^{(a)}(0, 0) = 0$ isolates the $L - 1 = 10$ adjacent window-boundary covariances from the higher-order and non-adjacent remainder $\mathcal{R}_L$,
\begin{equation}
\hat{g}_\ell^{(a)} = \hat{u}_\ell^{(a)} + \hat{H}_\ell^{(a)} = 1 + \hat{Z}_\ell^{(a)}, \qquad
\mu_{\mathrm{VEC}} - 1 = \sum_{s=1}^{L-1} \Cov_P(H_s, H_{s+1}) + \mathcal{R}_L.
\label{eq:fc-numerical-cov}
\end{equation}
Partitioning the $N_{\mathrm{err}}$ independent Bernoulli mechanisms into $L = 11$ disjoint sets $R_1, \dots, R_{11}$ by maximum flipped detector layer factorizes the full-circuit law into an exact product measure. For each target window boundary $s \in \{1, \dots, 10\}$, we define the left block $B_{L(s)}$ ($R_1$ at $s = 1$, or $R_{s-1} \cup R_s$ for $s \ge 2$), the right block $B_{R(s)}$ ($R_{s+1} \cup R_{s+2}$ for $s \le 9$, or $R_{11}$ at $s = 10$), and the exterior complement $B_{\mathrm{bg}(s)}$. This choice keeps both adjacent non-target boundaries $s - 1$ and $s + 1$ uncut across the twelve-block family. Splitting each evaluation fold into two disjoint donor half-banks $(\mathcal{B}_L, \mathcal{B}_R)$ and drawing an unbiased background configuration $D_{\mathrm{bg}}^{(s)} \sim P_{B_{\mathrm{bg}(s)}}$ yields the three-region Hoeffding splice and Radon--Nikodym weight
\begin{align}
Y^{(s)} &= D_L\big|_{B_{L(s)}} + D_R\big|_{B_{R(s)}} + D_{\mathrm{bg}}^{(s)}\big|_{B_{\mathrm{bg}(s)}}, \nonumber\\
\Cov_P(H_s, H_{s+1}) &= \E\Bigl[H_s\bigl(Y^{(s)}\bigr) H_{s+1}\bigl(Y^{(s)}\bigr) - H_s(D_L) H_{s+1}(D_R)\Bigr], \nonumber\\
W_s &= \frac{\hat{P}_{L(s),k_L}}{q_{L(s),k_L}}\frac{\hat{P}_{R(s),k_R}}{q_{R(s),k_R}}.
\label{eq:fc-splice-cov}
\end{align}
Every half-bank shot is classified across all twelve blocks into five disjoint boundary-activity strata $k \in \{0, 1, 2, 3, 4\}$ with global-fold frequencies $\hat{P}_{B,k}$ and non-empty proposals $q_{B,k}$, avoiding per-chunk emptiness bias or reservoir truncation. All spliced realizations are decoded across every boundary-activity stratum under exact weights $W_s$ (empirical mean $W_s \in [0.985, 1.018]$) with zero-mean single-donor block-restriction control variates that hold the maximum chip-level null-test residual below $0.011\%$ of $2 V p_L$.

Executing this two-stage splicing estimator across all twenty $(d, p)$ ensembles ($p / 10^{-3} \in \{1, 1.5, 2, 2.5, 3, 4\}$ at $d = 3$, $\{1.5, 2, 2.5, 3, 4\}$ at $d = 5$, $\{2, 2.5, 3, 3.5, 4\}$ at $d = 7$, and $\{2.5, 3, 3.5, 4\}$ at $d = 9$) resolves power-law suppression of the adjacent-boundary bias with code distance. Across $M = 80$ chips per point, Stage~A decodes $N_A \in [7.17\times 10^6, 1.92\times 10^8]$ eleven-window runs per point ($8.96\times 10^4$ to $2.40\times 10^6$ per chip), while Stage~B evaluates $10 n_B \in [9.22\times 10^7, 1.02\times 10^9]$ three-region boundary splices per point ($1.15\times 10^6$ to $1.28\times 10^7$ per chip). Every simulated point resolves a negative signed bias at $2.1\sigma$ to $59.5\sigma$ significance. Weighted power-law fits $|\epsilon_{\mathrm{adj}}| \propto p^k$ yield $k = 2.37(10)$ at $d = 3$ ($\chi^2/\mathrm{dof} = 47.66/4$ from high-$p$ subleading curvature), $k = 3.92(16)$ at $d = 5$ ($\chi^2/\mathrm{dof} = 3.33/3$), $k = 4.8(3)$ at $d = 7$ ($\chi^2/\mathrm{dof} = 0.35/3$), and $k = 6.4(8)$ at $d = 9$ ($\chi^2/\mathrm{dof} = 1.28/2$). At every simulated distance, the empirical boundary-bias exponent strictly exceeds the worst-case analytical lower bound $\lfloor d/2 \rfloor$.

A three-part remainder audit across $d \in \{3, 5, 7, 9\}$ tests whether $\mathcal{R}_L$ in Eq.~\eqref{eq:fc-numerical-cov} alters these scaling laws (Supplementary Table~\ref{tab:fc-remainder-audit}). First, four-point shared-interior telescopic splicing across all $21$ non-adjacent window pairs ($|\ell' - \ell| \ge 2$, $M = 24$ chips) confirms that non-adjacent covariances are consistent with zero at low and intermediate $p$ and stay within $\sim 10\%$ of $|\epsilon_{\mathrm{adj}}|$ even at $(d = 3, p = 4\times 10^{-3})$. Second, three-donor splicing across all nine connected three-window cumulants $\kappa_3(H_s, H_{s+1}, H_{s+2})$ remains below $4\%$ of $|\epsilon_{\mathrm{adj}}|$ across resolved points ($< 1.9\sigma$ from zero at $d \ge 7$), while disconnected four-or-more-window products contribute below $0.033\%$ everywhere. Third, brute-force full-product Monte Carlo ($5.12\times 10^8$ runs at $L = 11, d = 3$ and up to $5.76\times 10^8$ runs at $L = 3, d \in \{3, 5\}$) agrees within $0.6\sigma$ with the boundary-splice estimator on the same $16$ chips ($-0.0967(88)\%$ direct versus $-0.0917(30)\%$ splice at $p = 3\times 10^{-3}$, and $-0.1905(85)\%$ versus $-0.1948(85)\%$ at $p = 4\times 10^{-3}$). Together, at $d = 3$ these checks confirm that adjacent-boundary covariances capture the full multi-window mitigation bias within statistical uncertainty; at $d \ge 5$ and the highest noise rates, the non-adjacent audit has uncertainties comparable to $|\epsilon_{\mathrm{adj}}|$ itself (Supplementary Table~\ref{tab:fc-remainder-audit}), so there Fig.~\ref{fig:circuit_integration}b reports the adjacent-boundary contribution rather than a verified total.

\begin{table}[t]
\centering
\caption{\textbf{Multi-term audit of the full-circuit remainder $\mathcal{R}_L / (2 V p_L)$ in Eq.~\eqref{eq:fc-numerical-cov}.} Representative benchmark points from the three-part verification suite cover (i) the sum over all $21$ non-adjacent window pairs ($|\ell' - \ell| \ge 2$, $M = 24$ chips), (ii) the sum over all nine connected three-window cumulants $\kappa_3(H_s, H_{s+1}, H_{s+2})$ ($M = 16$--$24$ chips) and disconnected $\ge 4$-window boundary products ($M = 80$ chips), and (iii) brute-force eleven-window full-product Monte Carlo (MC; $5.12\times 10^8$ runs across $M = 16$ chips at $d = 3, K_{\mathrm{surg}} = 5$). All normalized bias contributions are reported in percent ($\pm 1\sigma$ across chips).}
\label{tab:fc-remainder-audit}
{\footnotesize\setlength{\tabcolsep}{3pt}
\begin{tabular}{cccccc}
\toprule
$(d,\, p)$ & Adj.\ boundary $\epsilon_{\mathrm{adj}}$ (\%) & Non-adj.\ $21$ pairs (\%) & Connected $\kappa_3$ (\%) & Disconnected $\ge 4$-win (\%) & Direct $11$-win MC (\%) \\
\midrule
$(3,\, 1.0\times 10^{-3})$ & $-0.00990(61)$ & $+0.0004(60)$ & $-0.000008(6)$ & $+9.5(15)\times 10^{-9}$ & --- \\
$(3,\, 2.5\times 10^{-3})$ & $-0.0586(18)$ & $-0.0029(137)$ & $-0.00070(51)$ & $+2.24(15)\times 10^{-6}$ & --- \\
$(3,\, 3.0\times 10^{-3})$ & $-0.0932(18)$ & --- & --- & $+7.94(31)\times 10^{-6}$ & $-0.0967(88)$ \\
$(3,\, 4.0\times 10^{-3})$ & $-0.1984(33)$ & $-0.0206(104)$ & $-0.0074(32)$ & $+6.51(24)\times 10^{-5}$ & $-0.1905(85)$ \\
$(5,\, 2.0\times 10^{-3})$ & $-0.00663(82)$ & $+0.0012(85)$ & $-0.000020(35)$ & $+6.4(29)\times 10^{-9}$ & --- \\
$(5,\, 4.0\times 10^{-3})$ & $-0.0947(43)$ & $-0.030(17)$ & $+0.0026(36)$ & $+1.17(10)\times 10^{-5}$ & --- \\
$(7,\, 2.5\times 10^{-3})$ & $-0.00482(114)$ & $+0.0039(41)$ & $-0.000056(18)$ & $+2.0(19)\times 10^{-9}$ & --- \\
$(7,\, 4.0\times 10^{-3})$ & $-0.0465(48)$ & $+0.026(31)$ & $+0.0023(28)$ & $+2.16(43)\times 10^{-6}$ & --- \\
$(9,\, 3.0\times 10^{-3})$ & $-0.00542(125)$ & $-0.00047(89)$ & $-0.000067(36)$ & $+5.9(15)\times 10^{-9}$ & --- \\
$(9,\, 4.0\times 10^{-3})$ & $-0.0322(51)$ & $+0.091(40)$ & $+0.00106(67)$ & $+2.4(30)\times 10^{-7}$ & --- \\
\bottomrule
\end{tabular}
}
\end{table}

\subsubsection{End-to-End Learned-Twin Scaling Across Distance and Depth (Supplementary Fig.~\ref{fig:supp_ladder_learned_scaling})}
\label{subsec:fc-scaling-panel-bc}

\begin{figure}[!htbp]
    \centering
    \includegraphics[width=130mm]{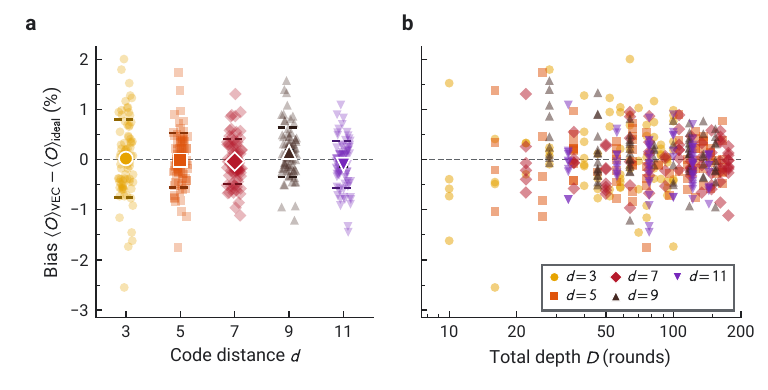}
    \caption{\textbf{Mitigation accuracy with syndrome-learned noise models.}
    Calibrating two-sector VEC from a syndrome-learned digital twin eliminates residual logical parity bias without systematic drift across either code distance or circuit depth.
    \textbf{(a)}~Mitigated parity bias $\langle O\rangle_{\mathrm{VEC}} - \langle O\rangle_{\mathrm{ideal}}$ ($O \in \{Z_A Z_B, X_A X_B\}$) versus code distance $d$; small markers show individual circuits, large markers with error bars show distance means and standard errors, and horizontal dashed lines mark the $\pm 1\sigma$ single-circuit spread.
    \textbf{(b)}~Same mitigated parity bias versus total circuit depth $D$ ($10$ to $176$ rounds across $1$ to $16$ merges), sharing the vertical axis of \textbf{(a)}.
    Simulations span $400$ disordered SI1000 two-patch lattice-surgery ladders ($80$ per odd distance $d \in \{3, 5, 7, 9, 11\}$) decoded with real-time MWPM and offline correlated MWPM at physical error rates tuned to $\langle O\rangle_{\mathrm{unmit}} \approx e^{-1}$. Each circuit uses $N_{\mathrm{learn}} = 999{,}960$ unlabeled syndrome-learning runs, $T_{\mathrm{MC}} = 10^7$ digital-twin calibration runs, and $N_{\mathrm{hw}} = 99{,}960$ independent evaluation runs.}
    \label{fig:supp_ladder_learned_scaling}
\end{figure}

Whereas Fig.~\ref{fig:circuit_integration}b isolates sliding-window boundary correlations under exact detector models, Supplementary Fig.~\ref{fig:supp_ladder_learned_scaling} tests end-to-end mitigation accuracy when the digital twin is learned from finite-shot syndromes across $400$ lattice-surgery ladders. The benchmark spans $80$ random ladder instances per odd distance $d \in \{3, 5, 7, 9, 11\}$ ($40$ measuring $Z_A Z_B$ and $40$ measuring $X_A X_B$), covering $K_{\mathrm{surg}} \in [1, 16]$ merges and $D \in [10, 176]$ rounds. Each instance follows round-stationary SI1000 noise~\cite{gidney2021fault} with two-qubit rate $p$, single-qubit and pre-round idling rate $p/10$, measurement flip rate $5p$, and reset and resonator-idling rate $2p$, perturbed spatially with $25\%$ rate spread and $5\%$ Pauli-rate anisotropy. For each $(d, K_{\mathrm{surg}})$ pair, we tune $p \in [1.28\times 10^{-3}, 5.97\times 10^{-3}]$ so that the unmitigated logical parity expectation lies near $e^{-1}$ (ensemble mean $0.391$), linking the total real-time flip probability $P_L$ to the no-error probability proxy $F$ in Eq.~\eqref{eq:full_circuit_overhead_scaling} through
\begin{equation}
\langle O\rangle_{\mathrm{unmit}} = 1 - 2 P_L \approx e^{-1}, \qquad
F \equiv \exp(-V p_L) \simeq 1 - P_L.
\label{eq:fc-bc-unmit-vs-survival}
\end{equation}
At small $V p_L$, this operating point obeys $\langle O\rangle_{\mathrm{unmit}} \simeq 2 F - 1$ while matching $F^2$ under independent Poissonian logical flips, holding the macroscopic no-error probability comparable across all $400$ circuits.

On every instance, syndrome-based learning, digital-twin calibration, and hardware evaluation use three mutually independent shot streams without dedicated calibration circuits. First, the full-circuit detector error model is reconstructed from $N_{\mathrm{learn}} = 999{,}960$ fresh syndrome runs via hierarchical inclusion--exclusion inversion with within-stage time-translation orbit pooling (Supplementary Section~\ref{sec:methods_numerical}). Rather than forming an eleven-window product weight, we pair global real-time MWPM $\mathcal{D}_{\mathrm{rt}}$ with global two-pass correlated MWPM $\mathcal{D}_{\mathrm{off}}$ and calibrate one two-sector weight pair $(\hat{w}_0, \hat{w}_1)$ from $T_{\mathrm{MC}} = 10^7$ digital-twin Monte Carlo runs. Scoring the raw unnormalized estimator across $N_{\mathrm{hw}} = 99{,}960$ independent evaluation runs per circuit yields an inverse-variance-weighted ensemble parity bias $\langle O\rangle_{\mathrm{VEC}} - \langle O\rangle_{\mathrm{ideal}}$ of $-0.008 \pm 0.019\%$ (unweighted mean $+0.005 \pm 0.028\%$, pull standard deviation $1.02$), with distance-grouped performance summarized below.
\begin{center}
\small
\begin{tabular}{ccc}
\toprule
Code distance $d$ & Weighted parity bias $\langle O\rangle_{\mathrm{VEC}} - \langle O\rangle_{\mathrm{ideal}}$ (\%) & Mean sampling overhead $\overline{C}$ \\
\midrule
$3$ & $-0.024 \pm 0.078$ & $6.77$ \\
$5$ & $-0.034 \pm 0.048$ & $3.49$ \\
$7$ & $-0.015 \pm 0.041$ & $3.15$ \\
$9$ & $+0.088 \pm 0.040$ & $3.46$ \\
$11$ & $-0.054 \pm 0.034$ & $2.80$ \\
\midrule
All $400$ circuits & $-0.008 \pm 0.019$ & $3.93$ \\
\bottomrule
\end{tabular}
\end{center}
Every distance bin remains statistically consistent with zero bias while the mean sampling overhead drops from $6.77$ at $d = 3$ to $2.80$ at $d = 11$. Syndrome-learned two-sector VEC thus delivers unbiased end-to-end mitigation at modest sampling cost across distances up to $d = 11$ and depths up to $176$ rounds.

\subsubsection{Million-Round Memory Under Smooth Temporal Drift (Supplementary Fig.~\ref{fig:supp_memory_1m_drift})}
\label{subsec:fc-million-round-drift}

\begin{figure}[!htbp]
    \centering
    \includegraphics[width=130mm]{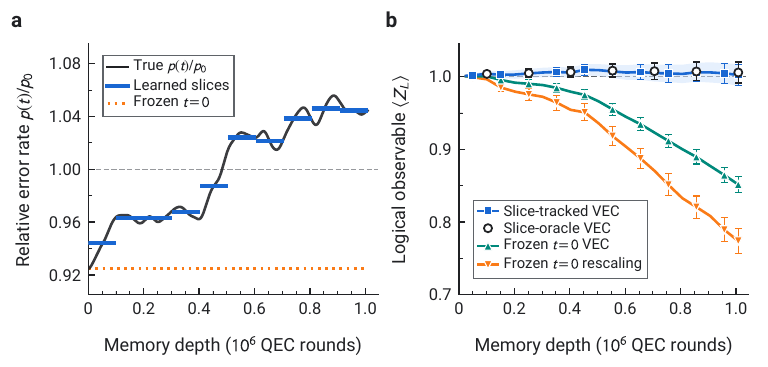}
    \caption{\textbf{Temporal drift tracking across one million stabilizer rounds.}
    Tracking slowly drifting physical error rates across chronological syndrome slices prevents cumulative calibration drift over deep continuous memory streams.
    \textbf{(a)}~Relative physical error rate $p(t)/p_0$ across $1{,}000$ steps of $1{,}008$ stabilizer rounds ($p_0 = 8.85\times 10^{-4}$, $15\%$ static spatial rate spread) on a distance-$7$ rotated surface code, comparing the continuous drift trajectory (black curve) against ten chronological orbit-pooled M\"obius slice estimates ($100{,}800$ rounds per slice; blue segments with $1\text{ s.e.m.}$ bands) and the frozen $t = 0$ model ($0.925$ relative scale; orange dotted line).
    \textbf{(b)}~Logical expectation $\langle Z_L \rangle$ across $1.008\times 10^6$ rounds ($36{,}000$ committed windows at $R = 28$, $b = 14$ across $2.5\times 10^4$ trajectories) comparing ten-slice syndrome-learned VEC (blue squares and $1\text{ s.e.m.}$ shading) and slice-oracle VEC (open circles) against frozen $t = 0$ VEC (teal upward triangles) and frozen $t = 0$ scalar rescaling (orange downward triangles).
    Error bars show total $1\text{ s.e.m.}$ combining shot, syndrome-learning, and digital-twin jackknife uncertainties.}
    \label{fig:supp_memory_1m_drift}
\end{figure}

\begin{figure}[!htbp]
    \centering
    \includegraphics[width=130mm]{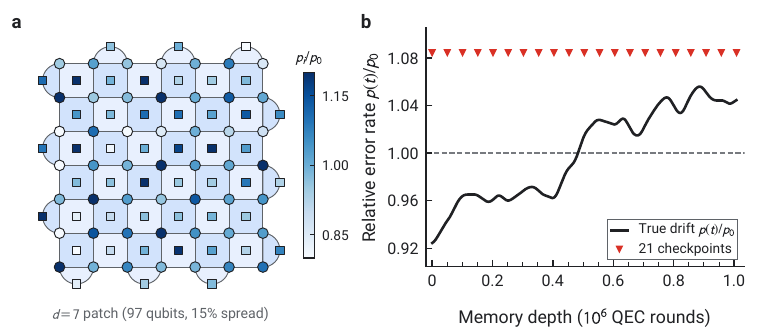}
    \caption{\textbf{Distance-$7$ spatial disorder map and million-round temporal drift schedule.}
    \textbf{(a)}~Layout of the $97$-qubit distance-$7$ rotated surface-code patch ($49$ data qubits as circles, $48$ check qubits as squares) illustrating the $15\%$ static site-to-site error-rate spread $p_i / p_0$ around $p_0 = 8.85\times 10^{-4}$.
    \textbf{(b)}~Continuous chip-mean drift trajectory $p(t)/p_0$ across $1.008\times 10^6$ stabilizer rounds and the $21$ non-destructive \texttt{FlipSimulator} readout checkpoints (red triangles).}
    \label{fig:supp_memory_1m_setup}
\end{figure}

\begin{figure}[!htbp]
    \centering
    \includegraphics[width=130mm]{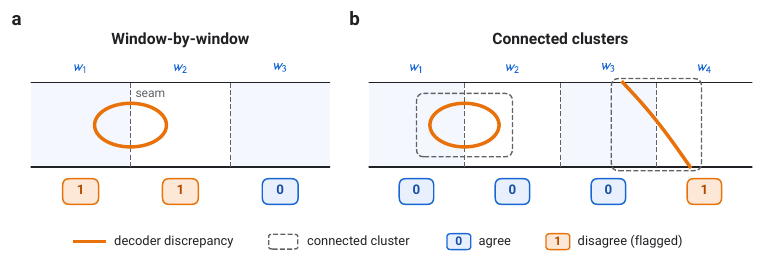}
    \caption{\textbf{Connected-cluster attribution across sliding-window seams.}
    \textbf{(a)}~Evaluating decoder discrepancies ($\mathcal{D}_{\mathrm{rt}} \oplus \mathcal{D}_{\mathrm{off}}$, orange) independently per window $w_k$ causes a trivial closed loop straddling a window seam to be misclassified as a logical disagreement in both adjacent windows ($1$).
    \textbf{(b)}~Grouping connected discrepancy clusters (dashed boxes) across adjacent windows recognizes closed loops as net-trivial ($0$, agree) and attributes each patch-spanning logical chain once to its latest active window ($1$, disagree).}
    \label{fig:supp_seam_clusters}
\end{figure}

Fig.~\ref{fig:memory_1m} and Supplementary Fig.~\ref{fig:supp_memory_1m_drift} evaluate continuous sliding-window mitigation in the regime where cumulative logical error probability approaches order unity over a million stabilizer cycles. On a single $d = 7$ rotated surface-code patch ($97$ qubits, $559$ circuit fault locations per round under SI1000 noise with $p_0 = 8.85\times 10^{-4}$, resonator idling at $p_0/10$, noisy state preparation at $2p_0$, noisy readout at $5p_0$, and $15\%$ static spatial rate spread; Supplementary Fig.~\ref{fig:supp_memory_1m_setup}a), we stream $2.5\times 10^4$ uninterrupted trajectories to $1{,}008{,}000$ rounds in \texttt{Stim}'s \texttt{FlipSimulator}, branching virtual state copies at $21$ checkpoints for destructive $Z_L$ readout without interrupting the main stream (Supplementary Fig.~\ref{fig:supp_memory_1m_setup}b). Over $1{,}000$ steps of $1{,}008$ rounds, the chip-mean error scale $p(t)/p_0$ rises non-monotonically from $0.925$ to $1.044$, peaking at $1.056$ (combining a linear $0.93 \to 1.07$ ramp, zero-mean Gaussian-smoothed Ornstein--Uhlenbeck fluctuations with $2.5\%$ nominal standard deviation, $1.2\%$ realized root-mean-square, and $250$-step correlation length, and a $-3\% \to +3\%$ spatial gradient tilt; Supplementary Fig.~\ref{fig:supp_memory_1m_drift}a). Both sliding-window stream decoders, namely real-time two-pass correlated MWPM and unconditional offline search-based hypergraph decoding via Tesseract~\cite{beni2025tesseract} (beam width $10$, priority-queue limit $10^5$, zero search-limit retries or fallbacks across $1.26\times 10^9$ window decodes), are compiled once from a shared $t = 0$ model learned on a $5\times 10^5$-shot ($70$-round) pre-run syndrome burn-in and held strictly fixed across all $36{,}000$ committed windows ($R = 28$, $b = 14$). Connected discrepancy clusters between the committed real-time and offline recoveries are co-attributed intact to their latest active window (Section~\ref{sec:fc-protocol}; Supplementary Fig.~\ref{fig:supp_seam_clusters}), automatically absorbing trivial zero-parity boundary loops without heuristic lag filters. Partitioning the running stream into ten chronological slices of $100{,}800$ rounds and applying two-fold cross-fitted orbit-pooled M\"obius inversion directly to $2.5\times 10^6$ streamed fourteen-round syndrome blocks per slice ($1.25\times 10^6$ per fold, $100$ open windows per trajectory per slice plus boundary windows) reconstructs each fold's slice model to $4.5\%\text{--}5.1\%$ per-mechanism relative root-mean-square error ($1.2\%\text{--}1.4\%$ pooled root-mean-square error normalized by the mean rate) and $< 0.06\%$ bulk mean-rate bias ($-0.008\%$ signed average across slices).

Tracking these ten in-stream syndrome-learned slices in the offline digital twin removes drift-induced mitigation bias without rebuilding the streaming decoders mid-shot (Supplementary Fig.~\ref{fig:supp_memory_1m_drift}b). At $1{,}008{,}000$ rounds, unmitigated real-time correlated MWPM decays to $\langle Z_L \rangle = 0.391 \pm 0.006$ and standalone windowed offline Tesseract decoding decays to $0.677 \pm 0.005$ (Fig.~\ref{fig:memory_1m}). If the digital twin is held frozen at its initial $t = 0$ model ($0.925$ relative scale), unconditioned scalar Clifford rescaling under-compensates the rising error rate and drifts down to $\langle Z_L \rangle = 0.774 \pm 0.017$ ($-22.6\%$ bias), whereas frozen $t = 0$ dual-decoder VEC drifts about $1.5$ times more gently to $0.852 \pm 0.011$ ($-14.8\%$ bias) because conditioning on real-time/offline disagreements suppresses first-order sensitivity to common-mode rate shifts. Updating the digital twin across the ten cross-fitted syndrome-learned slices via $4.5\times 10^7$ localized eight-window importance-sampled trajectories (importance-weight effective sample size $80\%\text{--}87\%$, pooling the four interior bulk windows of each eight-window block to calibrate unmodified two-sector product weights with mean shot weight $1.013 \pm 0.014$) restores $\langle Z_L \rangle = 1.002 \pm 0.015$ ($\widehat{C}_{\mathrm{tot}} = 2.21$, pooled-model check $1.002 \pm 0.015$; two-slot interior check $1.001 \pm 0.020$), matching the piecewise-constant ten-slice oracle twin ($1.006 \pm 0.015$, paired difference $-0.003 \pm 0.006$). Offline in-stream slice calibration thus decouples continuous real-time decoding from slow hardware drift across million-round fault-tolerant executions.

\subsection{Calibration Granularity and Classical Resources}
\label{sec:fc-resources}
We finally summarize the per-shot classical computational complexity, numerical precision requirements, and multi-block variance scaling of the full-circuit virtual error cancellation (VEC) protocol.
Per shot, the algorithm decodes the syndrome history $S$, extracts the block sectors $Q_{\ell,c}$, and evaluates
the scalar product weight $\prod_{\ell,c} w_{\ell,c}(Q_{\ell,c})$. Its classical runtime per shot is therefore polynomial in the circuit volume whenever offline decoding,
connected-cluster extraction, and table lookup scale algebraically with circuit size, as realized by MWPM on graphlike models~\cite{higgott2025blossom}.

Within each layer $\ell$, both the single-activation bound and marginal calibration operate directly on the individual gate blocks $(\ell,c)$ without assuming that blocks within the same layer are statistically independent. When the attributed map--sector pairs factorize across blocks $(\ell,1),\ldots,(\ell,K_\ell)$, multilinearity of chronological composition simplifies the layer expectation to
\begin{equation}
\E[\mathcal B_{\ell,K_\ell}\circ\cdots\circ\mathcal B_{\ell,1}]
=(\E\mathcal B_{\ell,K_\ell})\circ\cdots\circ(\E\mathcal B_{\ell,1})
=\mathcal U_\ell.
\label{eq:fc-patch-factorization}
\end{equation}
Here all maps act on the common logical register (with a tensor-product symbol $\bigotimes_c$ applying when the blocks act on disjoint qubit subsets), although Theorem~\ref{thm:fc-relative} does not require this within-layer factorization. In particular, for an $n$-logical-qubit layer partitioned into $K_\ell = \mathcal{O}(n)$ one- and two-qubit blocks with $M_{\ell,c} = \mathcal{O}(1)$ active channels each, calibrating the local block tables requires only $\sum_{c=1}^{K_\ell} M_{\ell,c}(M_{\ell,c}+1) = \mathcal{O}(n)$ parameters per layer while Eq.~\eqref{eq:fc-patch-factorization} cancels all independent multi-block fault collisions across the layer to all orders.

Finite floating-point precision in evaluating the product of $J=V$ block weights can be controlled independently of the physical fault correlations.
If $J$ scalar factors satisfy $|w_j|\le K_*$ and each
implemented factor has rounding error at most $u\le1$ (with abort-sector
zeros retained exactly), telescoping the product difference gives
\begin{equation}
 \left|\prod_{j=1}^J\widetilde w_j-\prod_{j=1}^J w_j\right|
 \le Ju(K_*+1)^{J-1}.
 \label{eq:fc-rounding}
\end{equation}
Choosing $u=\mathcal O(p^a/[J(K_*+1)^{J-1}])$ with
$a=\min(2r,t)$ requires only
$\mathcal O(a\log(1/p)+J\log(K_*+1)+\log J)$ bits of precision to preserve the asymptotic bias order of
Eq.~\eqref{eq:fc-main-theorem}, while statistical and model calibration errors enter separately through Proposition~\ref{prop:fc-robust}.

On the quantum resource side, the multi-shot variance of the full-circuit estimator is governed by the second moment of the product weight $W=\prod_{\ell,c} w_{\ell,c}(Q_{\ell,c})$.
For $n_{\rm shots}$ independent evaluation runs, extending the single-block variance bound of Eq.~\eqref{eq:vpec_sampling_overhead_def} gives
\begin{equation}
 \operatorname{Var}\!\left(\frac1{n_{\rm shots}}\sum_{i=1}^{n_{\rm shots}}\OVEC^{(i)}\right)
 \le\frac{\E W^2}{n_{\rm shots}}.
 \label{eq:fc-variance}
\end{equation}
When the squared block weights factorize across blocks ($\E\prod_jw_j^2=\prod_j\E w_j^2$, as holds under mutual independence) and each block has overhead $C_j=\E w_j^2=1+\eta_jp_L+O(p_L^2)$ with uniformly bounded $\eta_j$, the log-overhead sums across the circuit as
\[
\log\E W^2=p_L\sum_{j=1}^V\eta_j+O(Vp_L^2).
\]
For a uniform slope $\eta_j=\eta$ in the regime $Vp_L^2\ll1$, this sum reduces to the exponential scaling $\E W^2=\exp[\eta Vp_L+O(Vp_L^2)]=F^{-\eta}[1+O(Vp_L^2)]$ of Eq.~\eqref{eq:full_circuit_overhead_scaling} in the main text.

Even when physical faults and sliding-window recoveries correlate distinct blocks across window boundaries, the single-activation hypothesis \ref{ass:fc-activation} guarantees that $\E W^2$ agrees with the product of marginal block overheads $C_1 \cdots C_V$ up to subleading $\mathcal{O}(p_L p^{\min(r,k)})$ corrections at fixed circuit volume $V$. Writing each squared block weight around its deterministic consensus value as
\begin{equation}
 w_j(Q_j)^2 = u_j^2 + \Delta_j, \qquad u_j \equiv w_j(0), \qquad \Delta_j = \bigl(w_j(Q_j)^2 - u_j^2\bigr)\mathbf{1}_{\{Q_j \ne 0\}},
 \label{eq:fc-sqweight-decomp}
\end{equation}
we note that $|\Delta_j| \le (K_*^2 + u_j^2) I_j$ vanishes whenever block $j$ is inactive ($I_j = 0$). Under \ref{ass:fc-activation}--\ref{ass:fc-stability} with bounded per-block overhead $C_j = 1 + \mathcal{O}(p_L)$, the consensus shift and fluctuation mean satisfy $|u_j - 1| = \mathcal{O}(p_L)$ and $|\E\Delta_j| = \mathcal{O}(p_L)$. Meanwhile, any fluctuation product over a subset $T$ of $|T| \ge 2$ blocks requires $N \ge 2$ simultaneous activations, which occur with probability $\eta \le A_{r+k} p^{r+k}$. Expanding both expectations in subsets $T \subseteq \{1, \dots, V\}$ cancels the empty ($|T| = 0$) and singleton ($|T| = 1$) terms identically, leaving
\begin{equation}
 \Bigl|\E\Bigl[\prod_{j=1}^V w_j(Q_j)^2\Bigr] - \prod_{j=1}^V \E\bigl[w_j(Q_j)^2\bigr]\Bigr|
 \le \sum_{|T|\ge 2} \Bigl(\prod_{j \notin T} u_j^2\Bigr) \Bigl(\bigl|\E\prod_{j\in T}\Delta_j\bigr| + \prod_{j\in T}|\E\Delta_j|\Bigr)
 = \mathcal{O}\!\bigl(p_L p^{\min(r, k)}\bigr)
 \label{eq:fc-sqweight-factorization-bound}
\end{equation}
at fixed volume $V$. Cross-block correlations in the squared weights are therefore strictly $o(p_L)$ without assuming block independence, as verified numerically in Section~\ref{sec:fc-scaling-simulation}.

\section{Extended Simulation and Noise-Model Details}
\label{sec:supp_simulation_details}

\subsection{Surface-Code Geometry, Detector Counts, and Physical Noise Scope}
\label{subsec:supp_circuits_noise_details}

This subsection records the lattice geometry, detector-error-model dimensions, single-round code-capacity parameters, and physical noise scope supporting the condensed description in \vecSimSubA.
Each rotated surface-code patch of odd distance $d \in \{3, 5, 7, 9, 11\}$ encodes one logical qubit into $N_q = d^2$ data qubits and $d^2 - 1$ measure qubits, with logical operators $X_L$ and $Z_L$ along orthogonal boundaries.
Stabilizer checks split equally into $(d^2 - 1)/2$ generators per basis ($X$ and $Z$), comprising $(d - 1)^2$ weight-4 bulk plaquettes and $2(d - 1)$ weight-2 perimeter checks.
Under the baseline superconducting SI1000 model~\cite{gidney2021fault}, two-qubit Clifford gates incur depolarizing faults at rate $p$ ($p/15$ per non-identity Pauli error), while single-qubit gates and pre-round idling incur rate $p/10$ ($p/30$ per Pauli error).
Ancilla readouts and resets undergo bit flips at rates $5p$ and $2p$.
Resonator idling during measurement and reset is omitted in the memory templates (Figs.~\ref{fig:dual_decoder} and~\ref{fig:vlpec_validation}), set to $p/10$ with noiseless one-qubit gates in Fig.~\ref{fig:circuit_integration}b, set to $p_0/10$ with $p_0/10$ one-qubit gates in Fig.~\ref{fig:memory_1m} and Supplementary Fig.~\ref{fig:supp_memory_1m_drift}, and set to $2p$ in Supplementary Fig.~\ref{fig:supp_ladder_learned_scaling}.
In the single-block memory campaigns of Figs.~\ref{fig:dual_decoder} and~\ref{fig:vlpec_validation}, the physical rate spans $p \in \{1, 2, 3, 5, 7\}\times 10^{-3}$ and $p \in \{1, 2, 3, 5\}\times 10^{-3}$, respectively.
To monitor both logical bases simultaneously, the code block is entangled with an ideal reference qubit, and the joint operators $\bar{Z} Z_{\mathrm{ref}}$ and $\bar{X} X_{\mathrm{ref}}$ are measured noiselessly alongside all stabilizers before and after $T = d$ noisy rounds.
The resulting spacetime detector counts $N_D = (T + 1)(d^2 - 1)$ and unsplit hypergraph fault-mechanism counts $N_E$ across $d \in \{3, 5, 7, 9\}$ are listed below.
\begin{center}
{\small\setlength{\tabcolsep}{8pt}
\begin{tabular}{cccc}
\toprule
Code distance $d$ & Noisy rounds $T$ & Spacetime detectors $N_D$ & Hypergraph mechanisms $N_E$ \\
\midrule
$3$ & $3$ & $32$ & $359$ \\
$5$ & $5$ & $144$ & $2{,}107$ \\
$7$ & $7$ & $384$ & $6{,}351$ \\
$9$ & $9$ & $800$ & $14{,}195$ \\
\bottomrule
\end{tabular}}
\end{center}
In the spatially disordered campaigns (Figs.~\ref{fig:vlpec_validation} and~\ref{fig:circuit_integration}b and Supplementary Fig.~\ref{fig:supp_ladder_learned_scaling}), each gate, idle, reset, and readout rate is multiplied by an independent truncated normal factor $\mathcal{N}(1, 0.25^2)$ ($25\%$ spread).
Within the repeated round block, non-identity Pauli error weights are perturbed by $\mathcal{N}(1, 0.05^2)$ ($5\%$ anisotropy).
Every single-qubit and two-qubit stochastic Pauli channel (with mutually exclusive Pauli errors) is then converted into its exact factorized Bernoulli representation via Walsh--Hadamard inversion prior to detector-error-model extraction, capturing device heterogeneity without graphlike hyperedge-splitting approximations.

In the two-dimensional code-capacity model (Supplementary Figs.~\ref{fig:supp_dual_decoder_2d} and~\ref{fig:supp_vlpec_validation_2d}), each data qubit on a distance $d \in \{3, 5, 7\}$ patch undergoes an independent isotropic Pauli channel $\mathcal{E}_p(\rho) = (1 - p)\rho + (p/3)(X\rho X + Y\rho Y + Z\rho Z)$ across $11$ error rates $p \in [0.005, 0.090]$ over one ideal-syndrome cycle.
On every Monte Carlo run across both code-capacity and circuit-level simulations, \texttt{Stim}~\cite{gidney2021stim} samples physical faults to generate the detector syndrome $\bm{s} \in \{0, 1\}^{N_D}$.
It simultaneously records the true logical Pauli frame flip $Q' \in \{I, X_L, Y_L, Z_L\}$ (or binary parity $e \in \{0, 1\}$) from the commutation parities with $Z_L$ and $X_L$.
Because post-processing weights can become negative under finite-sample calibration mismatch, the residual logical channel vector $\bm{p}_{\mathrm{res}}$ forms a signed quasi-probability perturbation quantified by the signed $\ell_1$ norm
\begin{equation}
\|\bm{p}_{\mathrm{res}}\|_1 \equiv \sum_{Q \ne I} \bigl|(\bm{p}_{\mathrm{res}})_Q\bigr|.
\label{eq:methods_signed_residual_norm}
\end{equation}
Eq.~\eqref{eq:methods_signed_residual_norm} bounds worst-case Pauli observable bias through Eq.~\eqref{eq:four_sector_distance_suppression2} and prevents opposite-sign errors across distinct Pauli sectors from canceling one another.

Different simulation suites also address distinct physical noise scopes depending on whether boundary operations lie inside the spacetime detector graph.
Single-block memory circuits (Figs.~\ref{fig:dual_decoder} and~\ref{fig:vlpec_validation}) use noiseless initial state preparation and terminal data readout to isolate the internal $T = d$ syndrome rounds.
Multi-window lattice-surgery circuits (Fig.~\ref{fig:circuit_integration}b and Supplementary Fig.~\ref{fig:supp_ladder_learned_scaling}) instead include noisy ancilla and surgery-seam resets ($2p$) and measurements ($5p$) directly inside the detector graph.
By contrast, physical processes that leave no stabilizer signature or violate the stochastic Pauli model fall outside the syndrome-only Pauli virtual error cancellation (VEC) guarantee.
Such processes include unmonitored terminal data readout flips, untwirled non-Pauli magic-state input errors, coherent unitary rotations without randomized compiling, and unheralded leakage, which require separate readout mitigation, state distillation or twirling, and leakage-reduction units.

\subsection{Offline Ensemble Decoding, Subsampling Schedules, and Two-Dimensional Matrix-Product-State Contraction}
\label{subsec:supp_decoders_details}

Offline three-candidate ensemble decoding places graph and hypergraph decoders on a shared likelihood footing in Figs.~\ref{fig:dual_decoder} and~\ref{fig:vlpec_validation} \vecSimSubB.
Both $\mathcal{D}_{\mathrm{rt}}$ and $\mathcal{D}_{\mathrm{off}}$ are compiled once per nominal error rate $p$ from the homogeneous SI1000 circuit model using graphlike (\texttt{decompose\_errors=True}) and unsplit (\texttt{decompose\_errors=False}) detector error models.
The first candidate runs serial minimum-sum belief propagation (BP) at scaling factor 0.625 and up to 20 iterations, followed by order-3 combination-sweep ordered statistics decoding (OSD) via \texttt{osd\_cs}~\cite{roffe2020decoding} on the unsplit parity-check matrix.
The second lifts the real-time minimum-weight perfect matching (MWPM) recovery of \texttt{PyMatching}, while the third runs two-pass correlated MWPM (\texttt{enable\_correlations=True})~\cite{higgott2022pymatching,higgott2025blossom} after separating $Z$-type and $X$-type detector components within each split instruction.
Because graph matchings represent weight-3 and weight-4 $Y$-type and hook mechanisms as edge pairs, greedy edge-pair recombination merges compatible active pairs into undecomposed hyperedges whenever doing so lowers the unsplit cost $\sum_i \ln[(1 - p_i)/p_i]$.
The third recovery is accepted only when its recombined fault mechanisms reproduce the measured syndrome, and $\mathcal{D}_{\mathrm{off}}$ switches from BP+OSD only when another valid recovery improves the cost by more than $10^{-4}$.
Scoring every candidate under one unsplit objective prevents double-counting of multi-detector faults and reduces false alarms at large distances.

Failure-first Monte Carlo sampling concentrates offline work on informative runs while keeping the $4 \times 4$ joint sector--error table $J_{Q,Q'} = \Pr(\bm{s} \in \mathcal{S}_Q,\, Q')$ unbiased.
Across $N$ runs, empty-syndrome failures ($\bm{s} = \bm{0}$) map deterministically to consensus ($Q = I$), whereas every remaining real-time failure ($Q' \ne I$) is decoded by $\mathcal{D}_{\mathrm{off}}$ with unit weight $1/N$.
Among non-trivial real-time successes ($Q' = I$ and $\bm{s} \ne \bm{0}$), each run is drawn independently with probability $f \in (0, 1]$ and enters the false-alarm entries $J_{Q,I}$ ($Q \ne I$) with weight $1/(f N)$, fixing $J_{I,I} = 1 - \sum_{(Q,Q') \ne (I,I)} J_{Q,Q'}$.
Run counts, subsampling fractions, and $Y$-sector activation regimes appear below.
\begin{center}
{\small\setlength{\tabcolsep}{6pt}
\begin{tabular}{lllll}
\toprule
Benchmark & $d$ & Subsampling fraction $f$ & Runs $N$ per point or chip & Active $Y$-sector solve \\
\midrule
Fig.~\ref{fig:dual_decoder} & $3$ & $1$ (all $p$) & $5.1\times 10^5$ to $3.0\times 10^6$ & $p \ge 2\times 10^{-3}$ \\
Fig.~\ref{fig:dual_decoder} & $5$ & $1$ (all $p$) & $3.0\times 10^5$ to $1.2\times 10^7$ & $p \ge 3\times 10^{-3}$ \\
Fig.~\ref{fig:dual_decoder} & $7$ & $0.30, 0.25$ ($p = 1, 2\times 10^{-3}$); $1$ ($p \ge 3\times 10^{-3}$) & $2.1\times 10^5$ to $3.0\times 10^8$ & $p \ge 5\times 10^{-3}$ \\
Fig.~\ref{fig:dual_decoder} & $9$ & $0.05, 0.10, 0.25$ ($p \le 3\times 10^{-3}$); $1$ ($p \ge 5\times 10^{-3}$) & $1.0\times 10^5$ to $6.6\times 10^9$ & $p = 7\times 10^{-3}$ \\
Fig.~\ref{fig:vlpec_validation}c & $7$ & $5\times 10^{-4}$ (all $p$) & $6.0\times 10^7$ to $1.2\times 10^{10}$ & All chips ($p \ge 10^{-3}$) \\
\bottomrule
\end{tabular}}
\end{center}
When fewer than twenty raw $Y$-disagreement events occur, restricting the linear solve to $\{I, X, Z\}$ with $w_Y = 0$ avoids sparse-count instability (Supplementary Section~\ref{subsubsec:app_contrast_truncation}) while keeping $Y$ errors in $p_L$ and $\|\bm{p}_{\mathrm{res}}\|_1$; this minimum-count cutoff shifts $\eta$ in Fig.~\ref{fig:dual_decoder}f by less than $0.01$ everywhere.
False-alarm counts cancel algebraically out of the projective ratios $w_Q/w_I$ and the residual numerator, entering only the $1 - \mathcal{O}(p_L)$ denominator.

In the two-dimensional code-capacity benchmark (Supplementary Figs.~\ref{fig:supp_dual_decoder_2d} and~\ref{fig:supp_vlpec_validation_2d}), untruncated Bravyi--Suchara--Vargo matrix-product-state contraction~\cite{bravyi2014efficient} evaluates Eq.~\eqref{eq:methods_coset_partition} of \vecMethodsSec{} exactly for $d \in \{3, 5, 7\}$.
Horizontal transfer-matrix sweeps across all $d$ columns retain full bond dimension $\chi_{\max} = 2^d \in \{8, 32, 128\}$ (column bond profiles \texttt{[2,8,8,2]} at $d = 3$, \texttt{[4,32,32,32,32,4]} at $d = 5$, and \texttt{[8,128,128,128,128,128,128,8]} at $d = 7$).
At $d = 3$, the contracted coset posteriors match exhaustive enumeration across all $4^9 = 262{,}144$ Pauli strings to machine precision.
Exact coset ties ($144$ of $256$ syndromes under uniform noise at $d = 3$) follow a fixed label order across calibration and evaluation, and independent-$X/Z$ Calderbank--Shor--Steane (CSS) lookup-table decoding precomputes sub-syndrome tables $\mathcal{T}_X(\bm{s}_Z)$ and $\mathcal{T}_Z(\bm{s}_X)$ of length $(d^2 - 1)/2 \in \{4, 12, 24\}$.
By contrast, Theorem~\ref{thm:salem_surface_gap} conditions solely on minimum-Hamming-weight decisions and normalized complementary weight gaps.
That single-decoder restriction isolates the overhead reduction gained from offline syndrome re-evaluation.

\subsection{Syndrome-Learning Estimators, Proposal Sampling, and Runtime Scaling}
\label{subsec:supp_noise_learning_details}

This subsection records the two-dimensional code-capacity character-regression estimator, proposal-sampling validation, classical runtime scaling, and secondary uncertainty rules supporting \vecSimSubCD.
In the two-dimensional code-capacity benchmark (Supplementary Fig.~\ref{fig:supp_vlpec_validation_2d}; $100$ realizations per rate across $p \in [0.005, 0.090]$ and $N_0 \in \{10^5, 10^6, 10^7\}$), empirical parities in Eq.~\eqref{eq:methods_fidelity_estimator} of \vecMethodsSec{} have variance $(1 - \Fs(\bm{\chi})^2)/N_{\mathrm{learn}}$ and propagated log-variance $(1 - \Fs(\bm{\chi})^2)/(N_{\mathrm{learn}}\Fs(\bm{\chi})^2)$.
To remove the leading $\mathcal{O}(N_{\mathrm{learn}}^{-1})$ logarithm bias (Supplementary Section~\ref{subsec:supp_identifiability_symmetry}), we apply the second-order Jensen correction
\begin{equation}
\hat{b}_{\bm{\chi}} = -\ln\hat{\Fs}(\bm{\chi}) - \frac{1-\hat{\Fs}(\bm{\chi})^2}{2N_{\mathrm{learn}}\hat{\Fs}(\bm{\chi})^2},
\label{eq:methods_jensen_shrinkage}
\end{equation}
whenever $\hat{\Fs}(\bm{\chi}) \ge \max\!\big(10^{-6},\, 3\hat{\sigma}_{F,\bm{\chi}}\big)$, with $\hat{\sigma}_{F,\bm{\chi}} = \sqrt{\max(1-\hat{\Fs}(\bm{\chi})^2, 1/N_{\mathrm{learn}})/N_{\mathrm{learn}}}$.
We adjoin all $d^2 - 1$ single-check syndrome characters of $\HH_{\mathrm{stab}} \equiv [\HH_Z; \, \HH_X]$ with weight-2 modulo-2 characters $\bm{\chi}_{u, v} = \bm{e}_u \oplus \bm{e}_v$ of adjacent same-basis plaquettes sharing data qubits ($\bm{h}_u \cdot \bm{h}_v > 0$).
This construction forms a full-rank character set $\Xset$ of size $K \in \{14, 54, 118\}$ for $d \in \{3, 5, 7\}$ (Supplementary Section~\ref{subsec:supp_character_rank_analysis}) with attenuation system
\begin{align}
\Amat_{\mathrm{eff}} &\equiv [(\Xset\HH_{\mathrm{stab}})\bmod 2] \in \{0, 1\}^{K \times d^2},
\label{eq:methods_effective_design_matrix} \\
\Fs(\bm{\chi}) &= \prod_{q=1}^{d^2}(1-4p_q/3)^{a_{\bm{\chi},q}}, \quad
\bm{b} = \Amat_{\mathrm{eff}}\bm{\theta},
\label{eq:methods_attenuation_linear_system}
\end{align}
where $\theta_q = -\ln(1-4p_q/3)$ for $p_q \in [0, 3/4)$.
Solving the inverse-variance weighted non-negative least-squares objective across connected components of $\Amat_{\mathrm{eff}}^\top\Amat_{\mathrm{eff}}$ recovers per-qubit rates via
\begin{align}
\hat{\bm{\theta}}_{\mathrm{tied}} &= \arg\min_{\bm{\theta}_{\mathrm{tied}} \ge \bm{0}} \big\|\WW^{1/2}(\Amat_{\mathrm{eff}}\bm{\theta}_{\mathrm{tied}} - \hat{\bm{b}})\big\|_2^2,
\label{eq:methods_wnnls_estimation} \\
\hat{p}_q &= \frac{3}{4}\big(1 - e^{-(\hat{\bm{\theta}}_{\mathrm{tied}})_q}\big),
\label{eq:methods_depolarizing_rate_mapping}
\end{align}
with diagonal weights $W_{\bm{\chi}\bm{\chi}} = \hat{\Fs}(\bm{\chi})^2/\max(1-\hat{\Fs}(\bm{\chi})^2, 1/N_{\mathrm{learn}})$.
All fitted parameters remain strictly interior ($\hat{\bm{\theta}}_{\mathrm{tied}} > \bm{0}$), achieving $0.35\%$ median physical rate error and $0.31\%$ median logical rate error at $p = 5\times 10^{-3}$ and $N_0 = 10^7$ (Supplementary Fig.~\ref{fig:supp_vlpec_validation_2d}a,b).

In the $d = 7$ spacetime memory benchmark (Fig.~\ref{fig:vlpec_validation}), unclipped M\"obius DEM inversion at $p = 10^{-3}$ and $N_0 = 10^7$ reconstructs all $6{,}351$ mechanisms across $40$ chips with Pearson correlation $0.999998$ and $0.76\%$ median unsigned error ($68\%$ interval $[0.18\%, 2.00\%]$; Fig.~\ref{fig:vlpec_validation}a).
Propagating the learned rates predicts $\bar{p}_X, \bar{p}_Y, \bar{p}_Z$ to $0.10\%$ median error ($68\%$ interval $[0.03\%, 0.41\%]$ and $95\%$ interval $[0.01\%, 1.59\%]$; Fig.~\ref{fig:vlpec_validation}b).
In Fig.~\ref{fig:vlpec_validation}c, each chip uses $1.2\times 10^{10}$ runs at $p = 10^{-3}$, followed by $1.5\times 10^9$ runs at $2\times 10^{-3}$, then $6\times 10^8$ runs at $3\times 10^{-3}$, and $6\times 10^7$ runs at $5\times 10^{-3}$ across eight groups.
At $(p = 5\times 10^{-3}, N_0 = 10^4)$, where several hundred negative learned rates lower the common-random-numbers (CRN) relative effective sample size (ESS) to $0.008$, sampling $6\times 10^7$ trajectories per chip from $q_j = |\hat{p}_j|$ and reweighting via Eq.~\eqref{eq:methods_signed_crn_ratio} of \vecMethodsSec{} (replacing $p_j$ by $q_j$) produces equal-magnitude signed weights ($\mathrm{ESS} = 0.24$).
Combining that learned-twin table with the independent physical-chip table yields $\mathrm{SNR}_{\mathrm{chip}} = 4.9$, and at $N_0 = 10^5$ proposal sampling reproduces the CRN residuals within statistical uncertainty.

Syndrome learning decouples quantum calibration runs from $1/p_L$, while failure-first subsampling decouples expensive offline $\mathcal{D}_{\mathrm{off}}$ evaluations from the $\mathcal{O}((\alpha p_L \varepsilon_{\max}^2)^{-1})$ unweighted real-time screening budget.
Resolving pooled mechanism rates of orbit size $|\mathcal{O}|$ to relative error $\varepsilon$ requires a learning sample count
\begin{equation}
N_{\mathrm{learn}} = \mathcal{O}\!\left(\frac{1}{|\mathcal{O}|\, p\, \varepsilon^2}\right),
\label{eq:methods_tlearn_orbit_scaling}
\end{equation}
scaling inversely with $p$ instead of $p_L$.
Bit-packed M\"obius inversion with temporal orbit pooling runs in $0.5\text{ s}$ at $d = 3$ ($860$ orbits) up to $49\text{ s}$ and $5.9\text{ GB}$ peak memory at $d = 11$ ($7.09\times 10^5$ mechanisms in $14{,}732$ orbits).
Offline calibration at relative precision $\varepsilon_{\max}$ across $\Theta(p_L)$ disagreement sectors and the $\Theta(\alpha p_L)$ consensus logical error rate takes $\mathcal{O}((\alpha p_L \varepsilon_{\max}^2)^{-1})$ unweighted runs with correlated MWPM ($0.08$ to $5.8\text{ ms}$ per run across $d = 3$ to $9$ in Supplementary Fig.~\ref{fig:supp_ladder_learned_scaling}).
Slower ensemble decoder evaluations require 0.15 to 120 ms per memory run on $d = 3$ to $7$ and 60 ms to 2.55 s on deep $d = 3$ to $9$ ladders.
Failure-first subsampling ($f = 5\times 10^{-4}$ in Fig.~\ref{fig:vlpec_validation}c) restricts those evaluations to real-time failures plus $0.05\%$ of other non-trivial syndromes.
Code-capacity benchmarks (Supplementary Figs.~\ref{fig:supp_dual_decoder_2d} and~\ref{fig:supp_vlpec_validation_2d}, using $5\times 10^7$ to $10^9$ runs across $B \ge 20$ shards) compute Clopper--Pearson Beta $1\sigma$ intervals for rare counts ($N_Y < 100$ in Supplementary Fig.~\ref{fig:supp_dual_decoder_2d}b) and fit $\eta_0 + a p^b$ in Supplementary Fig.~\ref{fig:supp_dual_decoder_2d}e.
Adaptive Tikhonov damping stabilizes fine-grained covariance inversion at low $p$ while shifting $\eta$ by at most $0.104\%$ (Supplementary Section~\ref{subsec:supp_ridge_regularization_conditioning}).

\end{document}